%% file: main.tex
\documentclass[envcountsect]{llncs}
\makeatletter
\renewcommand\subsubsection{\@startsection{subsubsection}{3}{\z@}%
                       {-18\p@ \@plus -4\p@ \@minus -4\p@}%
                       {0.5em \@plus 0.22em \@minus 0.1em}%
                       {\normalfont\normalsize\bfseries\boldmath}}
\makeatother
\usepackage[ruled]{algorithm2e}

\usepackage{caption}
\usepackage{hyperref}
\usepackage{adjustbox,multirow}
\usepackage{verbatim,color}
\usepackage{amsmath,amsfonts,latexsym,verbatim,threeparttable,amsfonts,rotating,pifont}

\input{packages_article}
\input{symbols_env}
\input{symbols}

\makeatletter
\DeclareRobustCommand*\cal{\@fontswitch\relax\mathcal}
\makeatother

\newtheorem{notation}[theorem]{Notation}
\usepackage{color,tabularx}
\usepackage{framed}
\newcolumntype{b}{X}
\newcolumntype{s}{>{\hsize=.5\hsize}X}

\newcommand{\commentAA}[1]{\textcolor{red}{{\sf (Ananya's  Note:} {\sl{#1}} {\sf EON)}}}

\begin{document}

\sloppy

\title{\bf Network Agnostic MPC with Statistical Security}
\author{
Ananya Appan\thanks{Work done as a student at IIIT Bangalore.}\inst{1}
  \and
  Ashish Choudhury \inst{2}}
  
\date{}

\institute{
SAP Labs, Bangalore India\\
\email{ananya.appan@iiitb.ac.in}\and
IIIT Bangalore, India\\
\email{ashish.choudhury@iiitb.ac.in}
}

\maketitle
\vspace*{-0.7cm}
\begin{abstract}
\input{abstractShort}

\end{abstract}

\input{introShort}

\input{prelims}

\input{BA}

\input{ICSig}

\input{statVSS}

\input{PVMT}

\input{statMVSS}

\input{Triples}

\input{mpc}

\input{Impossibility}

\input{Acknowledgements}

\bibliographystyle{plain}

\bibliography{main}

\newpage


\appendix

\input{AppBA}

\input{AppICP}

\input{AppVSS}

\input{AppRec}

\input{AppMDVSS}

\input{AppTriples}

\input{AppMPC}

\end{document}

%% file: packages_article.tex
\usepackage{enumitem}
\usepackage{latexsym}
\usepackage{mathtools}
\usepackage{multirow}
\usepackage{multicol}
\usepackage{mwe}
\usepackage{makecell}
\usepackage{pgfplots}
\usepackage[section]{placeins}
\usepackage{pifont}
\usepackage{ragged2e}
\usepackage{rotating}
\usepackage{subcaption}
\usepackage{tabu}
\usepackage{tabularx}
\usepackage{threeparttable}
\usepackage{tikz}
\usepackage[normalem]{ulem}
\usepackage{url}
\usepackage{verbatim}
\usepackage{wrapfig}
\usepackage{xcolor}
\hypersetup{colorlinks, linkcolor={red!50!black}, citecolor={green!50!black}}
\usepackage{xspace}
\usepackage[framemethod=tikz]{mdframed}
\usepackage{algpseudocode}

%% file: symbols.tex
\newcommand{\SharingSpec}{\ensuremath{\mathbb{S}}}
\newcommand{\Q}{\ensuremath{\mathbb{Q}}}
\newcommand{\F}{\ensuremath{\mathbb{F}}}

\newcommand{\C}{\ensuremath{\mathcal{C}}}
\newcommand{\W}{\ensuremath{\mathcal{W}}}
\newcommand{\GW}{\ensuremath{\bf{\mathcal{GW}}}}
\newcommand{\Order}{\ensuremath{\mathcal{O}}}
\newcommand{\ckt}{\ensuremath{\mathsf{ckt}}}

\newcommand{\defined}{\ensuremath{\stackrel{def}{=}}}

\newcommand{\poly}{\ensuremath{\mathsf{poly}}}
\newcommand{\ShareSpec}{\ensuremath{\mathbb{S}}}

\newcommand{\ICSig}{\ensuremath{\mathsf{ICSig}}}

\newcommand{\Con}{\ensuremath{\mathsf{Con}}}

\newcommand{\sign}[1]{\ensuremath{\langle #1 \rangle}}
\newcommand{\committed}{\ensuremath{\mathsf{committed}}}
\newcommand{\false}{\ensuremath{\mathsf{false}}}
\newcommand{\true}{\ensuremath{\mathsf{true}}}
\newcommand{\qcore}{\ensuremath{q_{\mathsf{core}}}}
\newcommand{\flag}{\ensuremath{\mathsf{flag}}}
\newcommand{\s}{\ensuremath{\mathbf{s}}}
\newcommand{\LMax}{\ensuremath{\mathsf{L_{max}}}}
\newcommand{\ssec}{\ensuremath{\mathsf{ssec}}}
\newcommand{\specialS}{\ensuremath{\mathbf{S}}}
\newcommand{\Agree}{\ensuremath{\mathsf{Agree}}}

\newcommand{\PartySet}{\ensuremath{\mathcal{P}}}
\newcommand{\Partyset}{\ensuremath{\mathcal{P}}}
\newcommand{\Discarded}{\ensuremath{\mathcal{GD}}}
\newcommand{\Selected}{\ensuremath{\mathsf{Selected}}}
\newcommand{\Products}{\ensuremath{\mathsf{SIS}}}
\newcommand{\AdvStruct}{\ensuremath{\mathcal{Z}}}
\newcommand{\Core}{\ensuremath{\mathsf{CORE}}}

\newcommand{\AdvStructure}{\ensuremath{\mathcal{Z}}}
\newcommand{\Z}{\ensuremath{\mathcal{Z}}}

\newcommand{\CoreSet}{\ensuremath{\mathcal{CS}}}

\newcommand{\Hon}{\ensuremath{\mathcal{H}}}

\newcommand{\BroadcastSet}{\ensuremath{\mathcal{BS}}}
\newcommand{\R}{\ensuremath{\mathcal{SV}}}
\newcommand{\LD}{\ensuremath{\mathcal{LD}}}
\newcommand{\ReceiverSet}{\ensuremath{\mathcal{R}}}
\newcommand{\CD}{\ensuremath{\mathcal{CD}}}
\newcommand{\SET}{\ensuremath{\mathsf{SET}}}
\newcommand{\ACC}{\ensuremath{\mathsf{ACC}}}
\newcommand{\FIN}{\ensuremath{\mathsf{FIN}}}
\newcommand{\vals}{\ensuremath{\mathsf{vals}}}
\newcommand{\prop}{\ensuremath{\mathsf{prop}}}
\newcommand{\PrepSet}{\ensuremath{\mathsf{PrepareSet}}}
\newcommand{\PropSet}{\ensuremath{\mathsf{ProposeSet}}}

\newcommand{\CSet}{\ensuremath{\mathcal{C}}}

\newcommand{\errorAICP}{\ensuremath{\mathsf{\epsilon_{ICP}}}}

\newcommand{\D}{\ensuremath{\mathsf{D}}}
\renewcommand{\S}{\ensuremath{\mathsf{S}}}
\newcommand{\INT}{\ensuremath{\mathsf{I}}}
\newcommand{\Receiver}{\ensuremath{\mathsf{R}}}

\newcommand{\BasicMult}{\ensuremath{\Pi_\mathsf{BasicMult}}}

\newcommand{\VSS}{\ensuremath{\Pi_\mathsf{VSS}}}
\newcommand{\MDVSS}{\ensuremath{\Pi_\mathsf{MDVSS}}}
\newcommand{\Rec}{\ensuremath{\Pi_\mathsf{Rec}}}

\newcommand{\Auth}{\ensuremath{\Pi_\mathsf{Auth}}}
\newcommand{\Reveal}{\ensuremath{\Pi_\mathsf{Reveal}}}
\newcommand{\BC}{\ensuremath{\Pi_\mathsf{BC}}}
\newcommand{\BA}{\ensuremath{\Pi_\mathsf{BA}}}
\newcommand{\Vote}{\ensuremath{\Pi_\mathsf{GA}}}
\newcommand{\SBA}{\ensuremath{\Pi_\mathsf{SBA}}}
\newcommand{\CoinFlip}{\ensuremath{\Pi_\mathsf{CoinFlip}}}
\newcommand{\Coin}{\ensuremath{\mathsf{Coin}}}
\newcommand{\ABA}{\ensuremath{\Pi_\mathsf{ABA}}}
\newcommand{\PiPW}{\ensuremath{\Pi_\mathsf{PW}}}
\newcommand{\Acast}{\ensuremath{\Pi_\mathsf{Acast}}}

\newcommand{\PiMPC}{\ensuremath{\Pi_\mathsf{cktEval}}}
\newcommand{\ICP}{\ensuremath{\Pi_\mathsf{ICP}}}
\newcommand{\RecShare}{\ensuremath{\Pi_\mathsf{RecShare}}}
\newcommand{\SVM}{\ensuremath{\Pi_\mathsf{SVM}}}
\newcommand{\Rand}{\ensuremath{\Pi_\mathsf{Rand}}}
\newcommand{\LSh}{\ensuremath{\Pi_\mathsf{LSh}}}
\newcommand{\TripGen}{\ensuremath{\Pi_\mathsf{TripGen}}}
\newcommand{\RandMultCI}{\ensuremath{\Pi_\mathsf{RandMultCI}}}
\newcommand{\Prop}{\ensuremath{\Pi_\mathsf{Prop}}}

\newcommand{\TimeBC}{\ensuremath{T_\mathsf{BC}}}
\newcommand{\TimePW}{\ensuremath{T_\mathsf{PW}}}
\newcommand{\TimeVSS}{\ensuremath{T_\mathsf{VSS}}}
\newcommand{\TimeMDVSS}{\ensuremath{T_\mathsf{MDVSS}}}
\newcommand{\TimeAuth}{\ensuremath{T_\mathsf{Auth}}}
\newcommand{\TimeReveal}{\ensuremath{T_\mathsf{Reveal}}}
\newcommand{\TimeSBA}{\ensuremath{T_\mathsf{SBA}}}
\newcommand{\TimeABA}{\ensuremath{T_\mathsf{ABA}}}
\newcommand{\TimeBA}{\ensuremath{T_\mathsf{BA}}}
\newcommand{\TimeVote}{\ensuremath{T_\mathsf{GA}}}
\newcommand{\TimeCoinFlip}{\ensuremath{T_\mathsf{CoinFlip}}}
\newcommand{\TimeRecShare}{\ensuremath{T_\mathsf{RecShare}}}
\newcommand{\TimeRec}{\ensuremath{T_\mathsf{Rec}}}
\newcommand{\TimeSVM}{\ensuremath{T_\mathsf{SVM}}}
\newcommand{\TimeBasicMult}{\ensuremath{T_\mathsf{BasicMult}}}
\newcommand{\TimeRandMultCI}{\ensuremath{T_\mathsf{RandMultCI}}}
\newcommand{\TimeTripGen}{\ensuremath{T_\mathsf{TripGen}}}

\newcommand{\TimeLSh}{\ensuremath{T_\mathsf{LSh}}}

\newcommand{\TimeRand}{\ensuremath{T_\mathsf{Rand}}}

\newcommand{\OK}{\ensuremath{\mathsf{OK}}}

\newcommand{\Received}{\ensuremath{\mathsf{Received}}}

\newcommand{\Sender}{\ensuremath{\mathsf{Sen}}}
\newcommand{\prepare}{\ensuremath{\mathsf{prepare}}}
\newcommand{\propose}{\ensuremath{\mathsf{propose}}}
\newcommand{\cert}[1]{\ensuremath{\mathcal{C}(#1)}}
\newcommand{\vote}{\ensuremath{\mathsf{vote}}}

\newcommand{\ready}{\ensuremath{\mathsf{ready}}}

\newcommand{\CandidateCoreSets}{\ensuremath{\mathsf{CanCS}}}

\newcommand{\authCompleted}{\ensuremath{\mathsf{authCompleted}}}
\newcommand{\iter}{\ensuremath{\mathsf{iter}}}
\newcommand{\hop}{\ensuremath{\mathsf{hop}}}

\newcommand{\phaseone}{\ensuremath{\mathsf{phI}}}
\newcommand{\phasetwo}{\ensuremath{\mathsf{phII}}}

\newcommand{\Adv}{\ensuremath{\mathsf{Adv}}}

\newenvironment{myitemize}
{\begin{list}{$\bullet$}{ 
\itemindent=-0.1in
\itemsep=0.0in
\parsep=0.0in
\topsep=0.0in
\partopsep=0.0in}}{\end{list}}
\newcounter{itemcount}

\newenvironment{myenumerate}
{\setcounter{itemcount}{0}\begin{list}
{\arabic{itemcount}.}{\usecounter{itemcount} \itemindent=-0.2cm
\itemsep=0.0in
\parsep=0.0in
\topsep=5pt
\partopsep=0.0in}}{\end{list}}

\newenvironment{mydescription}
{\setcounter{itemcount}{0}\begin{list}
{\arabic{itemcount}.}{\usecounter{itemcount} \itemindent=-0.5cm
\itemsep=0.0in
\parsep=0.0in
\topsep=5pt
\partopsep=0.0in}}{\end{list}}

%% file: abstractShort.tex
In this work, we initiate the study of network agnostic MPC protocols with {\it statistical} security. Network agnostic MPC protocols give the best possible security guarantees,
 {\it irrespective} of the underlying network type. While network agnostic MPC
  protocols have been designed earlier with {\it perfect} and {\it computational} security, {\it nothing} is known in the literature regarding the possibility of
  network agnostic MPC protocols with statistical security. We consider the {\it general-adversary} model, where the adversary is characterized by an {\it adversary structure}, which enumerates all possible candidate subsets of corrupt parties. Given an unconditionally-secure PKI setup (a.k.a pseudo-signature setup), 
  known statistically-secure {\it synchronous} MPC (SMPC) protocols are secure against adversary structures satisfying the
  $\Q^{(2)}$ condition, meaning that the union of {\it any two} subsets from the adversary structure {\it does not} cover the entire set of parties.
  On the other hand, known statistically-secure {\it asynchronous} MPC (AMPC) protocols 
  can tolerate $\Q^{(3)}$ adversary structures where the union of {\it any three} subsets from the adversary structure {\it does not} cover the entire set of parties.
  
  Fix a set of $n$ parties $\PartySet = \{P_1, \ldots, P_n \}$ and adversary structures $\Z_s$ and $\Z_a$, satisfying the $\Q^{(2)}$ and $\Q^{(3)}$ conditions respectively, where
   $\Z_a \subset \Z_s$.  Then given an unconditionally-secure PKI, we ask whether it is possible to design a statistically-secure MPC protocol, which is resilient against $\Z_s$ 
   and $\Z_a$ in a {\it synchronous}
   and an {\it asynchronous}
   network respectively, even if the parties in $\PartySet$ are {\it unaware} of the network type. We show that it is possible iff
   $\Z_s$ and $\Z_a$ satisfy the $\Q^{(2, 1)}$ condition, meaning that the union of any two subsets from $\Z_s$ and any one subset from $\Z_a$
   is a proper subset of $\PartySet$. 
     Enroute our MPC protocol, we design several important network agnostic building blocks with the $\Q^{(2, 1)}$ condition, such as Byzantine {\it broadcast},
  {\it Byzantine agreement} (BA), {\it information checking protocol} (ICP), {\it verifiable secret-sharing} (VSS) and secure multiplication protocol, whose
   complexity is polynomial in $n$ and $|\Z_s|$.

%% file: introShort.tex
\section{Introduction}
\label{sec:intro}
A  secure {\it multiparty computation} (MPC) protocol \cite{Yao82,GMW87,BGW88,RB89}  allows a set of $n$ mutually distrusting parties
 $\PartySet = \{P_1, \ldots, P_n \}$ with private inputs to securely compute any known function $f$ of their inputs.
  This is achieved even if a subset of the parties are under the control of a centralized {\it adversary} and behave {\it maliciously} in a Byzantine fashion  
   during the protocol execution.
      In any MPC protocol, the parties need to interact 
      over the underlying communication network. Two types of networks have been predominantly considered. 
   The more popular {\it synchronous} MPC (SMPC) protocols operate over a synchronous network, where 
   every message sent is delivered within a {\it known} $\Delta$ time. 
  Hence, if a receiving party does not receive an expected
   message within $\Delta$ time,
    then it knows that the corresponding sender party is {\it corrupt}.
   The synchronous model {\it does not} capture real world networks like the Internet appropriately, where messages can be {\it arbitrarily} delayed. 
   Such networks are better modelled by the {\it asynchronous} communication model \cite{CanettiThesis}. In any {\it asynchronous} MPC (AMPC) protocol \cite{BCG93,BKR94}, there are {\it no}
   timing assumptions on message delays and messages can be arbitrarily, yet finitely delayed. The only guarantee is that every message sent will be {\it eventually} delivered. 
   The major challenge here is that no participant will know how long it has to wait for an expected message
   and {\it cannot} distinguish a ``slow" sender party from a corrupt sender party. 
    Consequently, in any AMPC protocol, a party {\it cannot} afford to receive messages from all the parties,
   to avoid an endless wait. Hence, to make ``progress", as soon a party receives messages from a ``subset" of the parties, it has to process them 
    as per the protocol, thus ignoring messages from a subset of potentially non-faulty but slow parties. 
     
    SMPC protocols are relatively simpler and enjoy better {\it fault-tolerance} (which is the maximum number of faults tolerable) compared to AMPC protocols.
   However, SMPC protocols become completely insecure even if a {\it single} message
   (from a non-faulty party) gets delayed. AMPC protocols do not suffer from this shortcoming. 
   On the negative side, AMPC protocols are far more complex than SMPC protocols and enjoy poor fault-tolerance.
    Moreover, every AMPC protocol suffers from {\it input deprivation} \cite{BH07} where, to avoid an endless wait, inputs of {\it all} non-faulty parties
   {\it may not} be considered for the computation of $f$.
   \paragraph{\bf Network Agnostic MPC Protocols.}
  There is a third category of protocols called {\it network agnostic} MPC protocols, where
    the parties {\it will not} be knowing the network type
    and the protocol should provide the best possible security guarantees depending upon the network type. 
    Such protocols are practically motivated, since the parties {\it need not} have to worry about the network type.
     \subsection{Our Motivation and Results}
     One of the earliest demarcations made in the literature is to categorize MPC protocols based on the computing power of the underlying adversary. The two main categories are
     {\it unconditionally-secure} protocols, which remain secure even against  {\it computationally-unbounded} adversaries, and
      {\it conditionally-secure} MPC protocols (also called {\it cryptographically-secure}), which remain
     secure {\it only} against {\it computationally-bounded} adversaries \cite{Yao82,GMW87}.    
     Unconditionally-secure protocols can be further categorized as {\it perfectly-secure} \cite{BGW88,BCG93} or {\it statistically-secure} \cite{RB89,BKR94},
     depending upon whether the security guarantees are {\it error-free} or achieved 
          except with a {\it negligible} probability. The fault-tolerance of statistically-secure MPC protocols are significantly {\it better}
      compared to perfectly-secure protocols.
    The above demarcation carries over even for network agnostic MPC protocols. While perfectly-secure and cryptographically-secure network agnostic MPC protocols have been investigated earlier, 
    {\it nothing} is known regarding network agnostic statistically-secure MPC protocols.
     In this work we derive necessary and sufficient condition for such protocols for the {\it first} time.
     \paragraph{\bf Existing Results for Statistically-Secure MPC.}
     Consider the {\it threshold} setting, where the maximum number of corrupt parties under the adversary's control is upper bounded by a given threshold. 
     In this model, it is known that statistically-secure SMPC tolerating up to $t_s$ faulty parties is possible iff $t_s < n/2$ \cite{RB89},
      provided the parties are given some {\it unconditionally-secure} PKI (a.k.a {\it pseudo-signature} setup) \cite{PW96,FitziThesis}.\footnote{The setup realizes
      unconditionally-secure Byzantine agreement  \cite{PeaseJACM80} with $t_s < n/2$.}
      On the other hand, statistically-secure AMPC tolerating up to $t_a$ faulty parties is possible iff $t_a < n/3$ \cite{BKR94,ADS20}.
      
      A more generalized form of corruption is the general adversary model (also called {\it non-threshold} model) \cite{HM97}.
      Here, the adversary is specified through a publicly known {\it adversary structure} $\Z \subset 2^{\PartySet}$, which is the set of all subsets of potentially
      corruptible parties during the protocol execution. The adversary is allowed to choose any one subset from $\Z$ for corruption. There are several ``merits" of studying the
      general adversary model.  For example,
      it provides more flexibility to model corruption in a fine-grained fashion.
       A threshold adversary is always a ``special" type of non-threshold adversary.
      Consequently, a protocol in the non-threshold setting {\it always} implies a protocol in the threshold setting.
      Also, the protocols in this model are relatively simpler and based on simpler primitives, compared to protocols against threshold adversaries based
      on complex properties of bivariate polynomials.  
       The downside is that the complexity of the protocols in the non-threshold model is polynomial in $n$ and $|\Z|$, where the latter could be $\Order(2^n)$ in the {\it worst} case.
      In fact, as noted in \cite{HM97,HM00}, this is {\it unavoidable}.
      
       Following \cite{HM97}, given a subset of parties $\PartySet' \subseteq \PartySet$ and $\Z$, we say that
      $\Z$ satisfies the $\Q^{(k)}(\PartySet', \Z)$ condition, if the union of any $k$ subsets from $\Z$ {\it does not} ``cover" $\PartySet'$. That is, for any subsets $Z_{i_1}, \ldots, Z_{i_k} \in \Z$, 
       the condition $(Z_{i_1} \cup \ldots \cup Z_{i_k}) \subset \PartySet'$ holds. 
       In the non-threshold model, statistically-secure SMPC is possible if the underlying adversary structure $\Z_s$ satisfies the $\Q^{(2)}(\PartySet, \Z_s)$ condition, provided the parties are
       given an unconditionally-secure PKI setup \cite{HT13}, while statistically-secure AMPC requires the underlying adversary structure $\Z_a$ to satisfy the $\Q^{(3)}(\PartySet, \Z_a)$ condition \cite{HT13,ACC22c}.
       \paragraph{\bf Our Results for Network Agnostic Statistically-Secure MPC.}
        We consider the most generic form of corruption and ask the following question:
       \begin{center}
       Given an unconditionally-secure PKI, a {\it synchronous adversary structure} $\Z_s$ and an {\it asynchronous adversary structure} $\Z_a$ satisfying the 
       $\Q^{(2)}(\PartySet, \Z_s)$  and  $\Q^{(3)}(\PartySet, \Z_a)$ conditions respectively, where $\Z_a \subset \Z_s$, does there exist a statistically-secure MPC protocol, which remains secure against
       $\Z_s$ and $\Z_a$ in a synchronous and an asynchronous network respectively?      
       \end{center}
       We answer the above question affirmatively, iff $\Z_s$ and $\Z_a$ satisfy the $\Q^{(2, 1)}(\PartySet, \Z_s, \Z_a)$ condition,  where
       by $\Q^{(k, k')}(\PartySet, \Z_s, \Z_a)$ condition, we mean that for any $Z_{i_1}, \ldots, Z_{i_k} \in \Z_s$ and ${\mathsf Z}_{j_1}, \ldots, {\mathsf Z}_{j_{k'}} \in \Z_a$, the following
       holds:
       \[ (Z_{i_1} \cup \ldots \cup Z_{i_k} \cup {\mathsf Z}_{j_1} \cup \ldots \cup {\mathsf Z}_{j_k'}) \subset \PartySet.\]
       Our results when applied against {\it threshold} adversaries imply that given an unconditionally-secure PKI, and thresholds $0 < t_a <  \frac{n}{3} < t_s < \frac{n}{2}$, network agnostic
        statistically-secure MPC tolerating $t_s$ and $t_a$ corruptions in the {\it synchronous} and {\it asynchronous} network is 
        possible, iff $2t_s + t_a < n$ holds. Our results in the context of relevant literature are summarized in 
        Table \ref{tab:Summary}.
             \begin{table}[!htb]
         	\begin{center}
	\begin{adjustbox}{width=1\textwidth}
			\begin{tabular}{|c|c|c|c|c|}
			\hline
			Network Type & Corruption Scenario & Security & Condition & Reference \\ \hline
			Synchronous & Threshold $(t)$ & Perfect & $t < n/3$ & \cite{BGW88} \\ \hline
			Synchronous & Non-threshold $(\Z)$ & Perfect & $\Q^{(3)}(\PartySet, \Z)$ & \cite{HM97} \\ \hline
			Synchronous & Threshold $(t)$ & Statistical & $t < n/2$ & \cite{RB89} \\ \hline
			Synchronous & Non-threshold $(\Z)$ & Statistical & $\Q^{(2)}(\PartySet, \Z)$ & \cite{HT13} \\ \hline \hline
			Asynchronous & Threshold $(t)$ & Perfect & $t < n/4$ & \cite{BCG93} \\ \hline
			Asynchronous & Non-threshold $(\Z)$ & Perfect & $\Q^{(4)}(\PartySet, \Z)$ & \cite{MSR02} \\ \hline
			Asynchronous & Threshold $(t)$ & Statistical & $t < n/3$ & \cite{BKR94,ADS20} \\ \hline
			Asynchronous & Non-threshold $(\Z)$ & Statistical & $\Q^{(3)}(\PartySet, \Z)$ & \cite{ACC22c} \\ \hline \hline
			Network Agnostic & Threshold $(t_s, t_a)$ & Perfect & $0 < t_a < n/4 < t_s < n/3$ & \cite{ACC22a}\\ 
			                             &                                     & &  and $3t_s + t_a < n$         &       \\ \hline
                         Network Agnostic & Non-threshold $(\Z_s, \Z_a)$ & Perfect & $\Z_a \subset \Z_s, \Q^{(3)}(\PartySet, \Z_s), \Q^{(4)}(\PartySet, \Z_a)$ & \cite{ACC22b} \\ 
			                             &                                     & &  and $\Q^{(3, 1)}(\PartySet, \Z_s, \Z_a)$         &       \\ \hline
                         Network Agnostic & Threshold $(t_s, t_a)$ & Computational & $0 < t_a < n/3 < t_s < n/2$ & \cite{BZL20}\\ 
			                             &                                     & &  and $2t_s + t_a < n$         &       \\ \hline \hline
                          Network Agnostic & Non-threshold $(\Z_s, \Z_a)$ & Statistical & $\Z_a \subset \Z_s, \Q^{(2)}(\PartySet, \Z_s), \Q^{(3)}(\PartySet, \Z_a)$ & This work\\ 
			                             &                                     & &  and $\Q^{(2, 1)}(\PartySet, \Z_s, \Z_a)$         &       \\ \hline          
                             Network Agnostic & Threshold $(t_s, t_a)$ & Statistical & $0 < t_a < n/3 < t_s < n/2$ & This work\\ 
			                             &                                     & &  and $2t_s + t_a < n$         &       \\ \hline
						
				\end{tabular}
       \end{adjustbox}				
	\end{center}
	\caption{\label{tab:Summary}Various conditions for MPC in different settings}
	\end{table}
   \vspace*{-1cm}
 \subsection{Detailed Technical Overview}
  We perform shared {\it circuit-evaluation} \cite{BGW88,RB89},
   where $f$ is abstracted as an arithmetic circuit $\ckt$ over a finite field $\F$ and the goal is to securely evaluate each gate
   in $\ckt$ in a {\it secret-shared} fashion. For every value during the circuit-evaluation, each party holds a share, such that the shares of the corrupt parties {\it do not} reveal any additional information.
    Once the function output is secret-shared, it is publicly reconstructed.    
    We deploy a {\it linear} secret-sharing scheme, which enables the parties to evaluate linear gates in $\ckt$ in a {\it non-interactive} fashion. {\it Non-linear} gates are evaluated using Beaver's method \cite{Bea91} by deploying secret-shared random {\it multiplication-triples} which are generated {\it beforehand}.

     To instantiate the above approach with {\it statistical} security, we need the following ingredients:
     a {\it Byzantine agreement} (BA) protocol \cite{PeaseJACM80}, 
      an {\it information checking protocol} (ICP) \cite{RB89},
     a {\it verifiable secret sharing} (VSS) protocol \cite{CGMA85},
     a {\it reconstruction} protocol and finally, a secure {\it multiplication} protocol.
     All existing statistically-secure SMPC \cite{RB89,CDDHR99,BH06,GSZ20,HT13} and AMPC \cite{BKR94,PCR15,CP17,CP23,ACC22c} protocols have instantiations of the above building blocks,
     {\it either} in the synchronous {\it or} asynchronous setting.
     However, in a network agnostic setting, we face several challenges to instantiate the above building blocks.
     We now take the reader through a detailed tour of the technical challenges and how we deal with them.
     \subsubsection{Network Agnostic BA with $\Q^{(2, 1)}(\PartySet, \Z_s, \Z_a)$ Condition}
     A BA protocol \cite{PeaseJACM80} allows the parties in $\PartySet$ with private input bits 
     to agree on a common output bit ({\it consistency}), which is the input of the
     non-faulty parties, if they have the {\it same} input bit ({\it validity}).
     Given an unconditionally-secure PKI, {\it synchronous} BA (SBA) is possible iff the underlying adversary structure $\Z_s$ satisfies the $\Q^{(2)}(\PartySet, \Z_s)$ condition \cite{PW96,FM98,FitziThesis}, while
     {\it asynchronous} BA (ABA) requires the underlying adversary structure $\Z_a$ to satisfy the $\Q^{(3)}(\PartySet, \Z_a)$ condition \cite{Cho23}.
     Existing SBA protocols become completely {\it insecure} in an asynchronous network. 
     On the other hand, any ABA protocol becomes {\it insecure} when executed in a synchronous network, since
     $\Z_s$ {\it need not}
     satisfy the $\Q^{(3)}(\PartySet,\Z_s)$ condition.      
      Hence, we design a network agnostic BA protocol with $\Q^{(2, 1)}(\PartySet, \Z_s, \Z_a)$ condition. 
      The protocol is obtained by generalizing the blueprint for network agnostic BA against {\it threshold} adversaries, first proposed in \cite{BKL19} and later used in \cite{ACC22a}.
     While \cite{BKL19} proposed it for {\it computational security}
      with conditions $0 < t_a < \frac{n}{3} < t_s < \frac{n}{2}$ and $2t_s + t_a < n$ in the presence of a {\it computationally-secure} PKI, later,
      \cite{ACC22a} modified it for {\it perfect security}
            and used it with conditions $t_a, t_s < n/3$.\footnote{{\it Unlike} computationally-secure BA, the necessary condition for perfectly-secure BA is $t < n/3$ for {\it both} SBA as well as ABA, where $t$ is
            the maximum number of faults.} We replace the computationally-secure PKI with an {\it unconditionally-secure} PKI and generalize the building blocks of \cite{BKL19} against non-threshold adversaries
             and upgrade their security to unconditional-security. 
            Additionally, we also generalize certain 
            primitives from \cite{ACC22a} and adapt them to work with the $\Q^{(2, 1)}(\PartySet, \Z_s, \Z_a)$ condition.
                      Since this part mostly follows the existing works, we refer to Section \ref{sec:BA} for full details.
     \subsubsection{Network Agnostic ICP with $\Q^{(2, 1)}(\PartySet, \Z_s, \Z_a)$ Condition}
     \label{ssec:ICP}
     An ICP \cite{RB89,CDDHR99} is used for authenticating data in the presence of a 
  {\it computationally-unbounded} adversary. 
  In an ICP, there are {\it four} entities, a {\it signer} $\mathsf{S} \in \PartySet$,
   an  {\it intermediary} $\mathsf{I} \in \PartySet$, a {\it receiver}
  $\mathsf{R} \in \PartySet$ and all the parties in $\PartySet$ acting as {\it verifiers}
   (note that $\mathsf{S}, \mathsf{I}$ and $\mathsf{R}$ also act as verifiers).
    An ICP has two sub-protocols, one for the {\it authentication phase} and one for 
    the {\it revelation phase}. 
    
    In the authentication phase, $\mathsf{S}$ has a private input
    $s \in \F$, which it distributes to $\mathsf{I}$ along with some {\it authentication information}.
    Each verifier is provided with some {\it verification information}, followed by the parties verifying whether $\mathsf{S}$ has distributed ``consistent" information. 
    If the verification is ``successful", then the 
    data held by $\mathsf{I}$ at the end of this phase is called  {\it $\mathsf{S}$'s IC-Signature on $s$ for intermediary $\mathsf{I}$ and receiver $\mathsf{R}$}, denoted by
    $\ICSig(\mathsf{S}, \mathsf{I}, \mathsf{R}, s)$. 
    Later, during the revelation phase, $\mathsf{I}$ reveals $\ICSig(\mathsf{S}, \mathsf{I}, \mathsf{R}, s)$ to $\mathsf{R}$,
    who ``verifies" it with respect to the verification information provided by the verifiers and either accepts or rejects
    $s$. We require the same security guarantees from ICP as expected from cryptographic signatures,
    namely {\it correctness} (if $\mathsf{S, I}$ and $\mathsf{R}$ are {\it all} honest, then $\mathsf{R}$ should accept $s$), {\it unforgeability} 
    (a {\it corrupt} $\mathsf{I}$ should fail to reveal an {\it honest} $\mathsf{S}$'s signature on $s' \neq s$) and {\it non-repudiation}
     (if an {\it honest} $\mathsf{I}$ holds some $\ICSig(\mathsf{S}, \mathsf{I}, \mathsf{R}, s)$, then later an {\it honest} $\mathsf{R}$ should accept
     $s$, even if $\mathsf{S}$ is {\it corrupt}).
       Additionally, we need {\it privacy}, guaranteeing that if 
    $\mathsf{S, I}$ and $\mathsf{R}$ are {\it all} honest, then $\Adv$ does not learn 
    $s$.\footnote{IC-signatures are {\it different} from pseudo-signatures.
    Pseudo-signatures are ``transferable", where a party can transfer a signed message to other parties for verification (depending upon the allowed level
    of transferability), while IC-signatures {\it can} be verified only by the designated $\mathsf{R}$ and {\it cannot} be further transferred. 
    Due to the same reason, IC-signatures satisfy the privacy property unlike pseudo-signatures. Most importantly, IC-signatures are generated from the scratch, assuming a pseudo-signature
    setup which is used to instantiate the instances of broadcast in the ICP.}
    
    The {\it only known} instantiation of ICP in the {\it synchronous} network \cite{HT13} is secure against $\Q^{(2)}$ adversary structures and becomes {\it insecure}
    in the {\it asynchronous} setting. On the other hand, the {\it only known} instantiation of ICP in the {\it asynchronous} setting \cite{ACC22c} can tolerate {\it only}
    $\Q^{(3)}$ adversary structures. Our network agnostic
     ICP is a careful adaptation of the {\it asynchronous} ICP of \cite{ACC22c}.
     We first try to {\it naively} adapt the ICP to deal with the {\it network agnostic} setting, followed by the technical problems in the naive adaptation
     and the modifications needed.
     
    During authentication phase, $\mathsf{S}$ embeds $s$ in a random $t$-degree polynomial $F(x)$ at $x = 0$, where 
    $t$ is the cardinality of the maximum-sized subset in $\Z_s$, and gives $F(x)$ to $\mathsf{I}$. In addition, each verifier $P_i$ is given a random {\it verification-point}
    $(\alpha_i, v_i)$ on $F(x)$. 
    To let the parties securely verify that it has distributed consistent information, $\mathsf{S}$
    additionally distributes a random $t$-degree polynomial $M(x)$ to $\mathsf{I}$, while
    each verifier $P_i$ is given a point on $M(x)$ at $\alpha_i$. Each verifier, upon receiving 
    its verification-points, {\it publicly} confirms the same. 
    Upon receiving these confirmations, $\mathsf{I}$ identifies a
    a subset of {\it supporting verifiers} $\R$ which have confirmed the receipt of their verification-points.
    To avoid an endless wait, $\mathsf{I}$ waits until
    $\PartySet \setminus \R \in \Z_s$. 
    After this, the parties {\it publicly} check the consistency of the $F(x), M(x)$ polynomials
    and the points distributed to $\R$, with respect to a {\it random} linear combination of these polynomials and points, where the linear combiner is selected by
     $\mathsf{I}$. 
    This ensures that $\mathsf{S}$ has {\it no} knowledge beforehand about the random combiner 
     and hence, any ``inconsistency" will be detected with a high probability.
      If no inconsistency is detected, the parties proceed to the revelation phase, where 
    $\mathsf{I}$ reveals $F(x)$ to $\mathsf{R}$, while each verifier in $\R$
    reveals its verification-point to $\mathsf{R}$, who accepts $F(x)$ (and hence $F(0)$) if it
    sure that the verification point of at least one {\it non-faulty} verifier in $\R$ is ``consistent" with the revealed $F(x)$. 
    This would ensure that the revealed $F(x)$ is indeed correct with a high probability, since a {\it corrupt} $\mathsf{I}$ will have no information about the verification point of any
    {\it non-faulty} verifier in $\R$, provided $\S$ is {\it non-faulty}. 
    To avoid an endless wait, once $\mathsf{R}$
     finds a subset of verifiers $\R' \subseteq \R$, where
      $\R \setminus \R' \in \Z_s$, whose verification-points are found to be ``consistent" with $F(x)$, it outputs $F(0)$.
   \paragraph{\bf A Technical Problem and Way-out.}
   The protocol outlined above will achieve all the properties in an {\it asynchronous} network, due to the 
   $\Q^{(3)}(\PartySet, \Z_a)$ condition. 
   However, it {\it fails} to satisfy the 
   {\it unforgeability} property in a {\it synchronous} network. Namely, a {\it corrupt} 
   $\mathsf{I}$ may {\it not} include {\it all} the non-faulty verifiers in $\R$ and may purposely {\it exclude} a subset of {\it non-faulty} verifiers belonging to $\Z_s$.
     Let $\Hon_{\R}$ be the set of {\it non-faulty} verifiers in $\R$ and let $\C_{\R}$ be the set of {\it corrupt} verifiers in $\R$. Due to the above strategy, 
   the condition $\Q^{(1)}(\Hon_\R, \Z_s)$ {\it may not} be satisfied and $\R \setminus \C_\R = \Hon_\R \in \Z_s$ {\it may hold}.
   As a result, during the revelation phase, $\mathsf{I}$ may produce an {\it incorrect} $F'(x) \neq F(x)$ and the verifiers in $\C_{\R}$ may change their verification points to 
    ``match" $F'(x)$, 
   while only the verification points of the verifiers in $\Hon_\R$ may turn out to be inconsistent with $F'(x)$. Consequently, $\R' = \C_\R$
   and if $\Hon_\R \in \Z_s$, then clearly $s' = F'(0)$ will be the output
   of $\mathsf{R}$, thus breaking the
   unforgeability property.
   
   To deal with the above issue, we let $\S$ identify and announce $\R$. This ensures that {\it all}
    honest verifiers are present in $\R$, if $\S$ is {\it honest} and the network is {\it synchronous}, provided $\S$ waits for ``sufficient" time to let the verifiers announce the receipt of their verification points.
    Consequently, {\it now} the condition $\Q^{(1)}(\Hon_\R, \Z_s)$ {\it will be} satisfied.
    Hence, if a {\it corrupt} $\mathsf{I}$ reveals an incorrect $F(x)$, then it will not be accepted,
    as the condition $\R \setminus \C_\R \in \Z_s$ {\it no} longer holds. 
      \paragraph{\bf Linearity of ICP.} 
 Our ICP satisfies the {\it linearity} property (which will be useful later in our VSS),
  provided ``special care" is taken while generating the IC-signatures. 
   Consider a {\it fixed} $\mathsf{S}, \mathsf{I}$ and $\mathsf{R}$ and let $s_a$ and $s_b$ be two values, such that $\mathsf{I}$ holds $\ICSig(\mathsf{S}, \mathsf{I}, \mathsf{R}, s_a)$ and 
  $\ICSig(\mathsf{S}, \mathsf{I}, \mathsf{R}, s_b)$, where {\it all} the following conditions are satisfied during the underlying instances of the authentication phase.
  \begin{myitemize}
  \item[--] The set of supporting verifiers $\R$ are the {\it same} during both the instances.
  \item[--] For $i = 1, \ldots, n$, corresponding to the verifier $P_i$, signer $\mathsf{S}$ uses the {\it same} $\alpha_i$, to compute the verification points, during 
   both the instances.
  \item[--] $\mathsf{I}$ uses the {\it same} linear combiner to verify the consistency of the distributed data in both the instances.
  \end{myitemize}
  Let $s \defined c_1 \cdot s_a + c_2 \cdot s_b$, where $c_1, c_2$ are {\it publicly known} constants from $\F$. It then follows that if all the above conditions are satisfied, then $\INT$
  can {\it locally} compute $\ICSig(\mathsf{S}, \mathsf{I}, \mathsf{R}, s)$ from $\ICSig(\mathsf{S}, \mathsf{I}, \mathsf{R}, s_a)$ and $\ICSig(\mathsf{S}, \mathsf{I}, \mathsf{R}, s_b)$,
   while each verifier in $\R$ can {\it locally} compute their corresponding verification information. 
     \subsubsection{Network Agnostic VSS and Reconstruction}
     \label{ssec:VSS}
     In the network agnostic setting, to ensure privacy, 
      all the values during the circuit evaluation need to be secret-shared ``with respect" to $\Z_s$
      {\it irrespective} of the network type. 
     We follow the notion of {\it additive secret-sharing} \cite{ISN87}, also used in the earlier MPC protocols \cite{Mau02,HT13,ACC22c}.
     Given $\Z_s = \{Z_1, \ldots, Z_{|\Z_s|} \}$, we consider the {\it sharing specification} $\ShareSpec_{\Z_s} = \{S_1, \ldots, S_{|\Z_s|} \}$, where each $S_q = \PartySet \setminus Z_q$. 
     Hence there exists at least one subset $S_q \in \ShareSpec_{|\Z_s|}$ which {\it does not} contain any faulty party, {\it irrespective} of the network type (since $\Z_a \subset \Z_s$).
     A value $s \in \F$ is said to be secret-shared, if there exist shares $s_1, \ldots, s_{|\Z_s|}$ which sum up to $s$, such that all (non-faulty) parties in $S_q$ have the share $s_q$.
     We denote a secret-sharing of $s$ by $[s]$, with $[s]_q$ denoting the share corresponding to $S_q$. If $[s]_1, \ldots, [s]_{|\Z_s|}$ are randomly chosen, then the probability
     distribution of the shares learnt by the adversary will be independent of $s$, since at least one share will be {\it missing} for the adversary.
     We also note that the above secret-sharing is {\it linear} since, given secret-sharings $[a]$ and $[b]$ and publicly known constants $c_1, c_2 \in \F$, 
     the condition $c_1 \cdot [a] + c_2 \cdot [b] = [c_1 \cdot a + c_2 \cdot b]$ holds.     
     Consequently, the parties can {\it non-interactively} compute any publicly known linear function of secret-shared values. 
      Unfortunately, the above secret-sharing {\it does not} allow for the
     {\it robust} reconstruction of a secret-shared value. This is because the corrupt parties may produce {\it incorrect} shares at the time of reconstruction.
     To deal with this, we ``augment" the above secret-sharing. As part of secret-sharing $s$,
     we {\it also} have publicly known {\it core-sets}
     $\W_1, \ldots, \W_{|\Z_s|}$, where each $W_q \subseteq S_q$ such that $\Z_s$ satisfies the $\Q^{(1)}(\W_q, \Z_s)$ condition (ensuring $\W_q$ has at least one {\it non-faulty} party). 
     Moreover, each (non-faulty) $P_i \in \W_q$ will have the IC-signature 
     $\ICSig(P_j, P_i, P_k, [s]_q)$ of every $P_j \in \W_q$, for every $P_k \not \in S_q$, such that the underlying IC-signatures satisfy the {\it linearity} property. 
     
     We call this augmented secret sharing as {\it linear secret-sharing with IC-signatures}, which is still denoted as $[s]$. 
           Now to {\it robustly} reconstruct a secret-shared $s$, we ask the parties in $\W_q$ to make public the share $[s]_q$, along
      with the IC-signatures of {\it all} the parties in $\W_q$ on $[s]_q$. 
       Any party $P_k$ can then verify whether $[s]_q$ revealed by $P_i$ is {\it correct} by verifying the IC-signatures.
     If $P_i$ is {\it corrupt} then, due to the {\it unforgeability} if ICP, it will fail to forge IC-signature of a {\it non-faulty} $P_j$ on an {\it incorrect} $[s]_q$.
    On the other hand, a {\it non-faulty} $P_i$ will be able to reveal the correct $[s]_q$ and the IC-signature of {\it every} $P_j \in \W_q$ on $[s]_q$, which are accepted even if
     $P_j$ is {\it corrupt} (follows from {\it non-repudiation} of ICP).  
     
    We design a network agnostic VSS protocol $\VSS$, which allows a designated {\it dealer} $\D \in \PartySet$ with input $s \in \F$ to {\it verifiably} generate
     $[s]$, where $s$ remains private for a {\it non-faulty} $s$. 
     If $\D$ is {\it faulty} then either no non-faulty party obtains any output (if $\D$ {\it does not} invoke the protocol) or there exists {\it some} $s^{\star} \in \F$ such that the parties
     output $[s^{\star}]$. 
     To design $\VSS$, we use certain ideas from the statistically-secure
     {\it synchronous} VSS (SVSS) and {\it asynchronous} VSS (AVSS)
     of \cite{HT13} and \cite{ACC22c} respectively, along with some new counter-intuitive ideas. In the sequel, we first give a brief outline of the SVSS and AVSS of \cite{HT13,ACC22c}, followed by the technical challenges
     arising in the network agnostic setting and how we deal with them.
     \paragraph{\bf Statistically-Secure SVSS of \cite{HT13} with $\Q^{(2)}(\PartySet, \Z_s)$ Condition.}
     The SVSS of \cite{HT13} 
     proceeds as a sequence of {\it synchronized} phases. During the {\it first} phase, $\D$ picks random shares $s_1, \ldots, s_{|\Z_s|}$ which sum up to $s$ and sends $s_q$ to the parties in $S_q$.
     To verify whether $\D$ has distributed consistent shares to the parties in $S_q$, during the {\it second} phase,
      every pair of parties $P_i, P_j \in S_q$ exchange the supposedly common shares received from $\D$,
      along with their respective IC-signatures. That is $P_i$, upon receiving $s_{qi}$ from $\D$, gives $\{\ICSig(P_i, P_j, P_k, s_{qi})\}_{P_k \in \PartySet}$ to $P_j$ while
      $P_j$, upon receiving $s_{qj}$ from $\D$, gives $\{\ICSig(P_j, P_i, P_k, s_{qj})\}_{P_k \in \PartySet}$ to $P_i$.
      Then during the {\it third} phase, the parties in $S_q$ {\it publicly} complain about any ``inconsistency", in response to which $\D$ makes {\it public} the share
      $s_q$ corresponding to $S_q$ during the {\it fourth} phase.
       Hence, by the end of {\it fourth} phase it is ensured that, for every $S_q$, either the share $s_q$ is {\it publicly} known (if any complaint was reported for $S_q$)
       or all (non-faulty) parties in $S_q$ have the same share (along with the respective IC-signatures of each other on it).
      The privacy of $s$ is maintained for a {\it non-faulty} $\D$, since the share $s_q$ corresponding to the set $S_q$ consisting of {\it only non-faulty} parties  is {\it never} made public.
      \paragraph{\bf Statistically-Secure AVSS of \cite{ACC22c} with $\Q^{(3)}(\PartySet, \Z_a)$ Condition.}
     Let  $\Z_a = \{\mathsf{Z}_1, \ldots, \mathsf{Z}_{|\Z_a|} \}$ and $\ShareSpec_{\Z_a} = \{\S_1, \ldots, \S_{|\Z_a|} \}$, where each $\S_q = \PartySet \setminus \mathsf{Z}_q$.     
     The AVSS protocol of \cite{ACC22c} also follows an idea similar to the SVSS of \cite{HT13}. 
      However, now the parties {\it cannot} afford to wait for {\it all} the parties in $\S_q$ to report
     the statuses of pairwise consistency tests, as the corrupt parties in $\S_q$ may {\it never} turn up. 
     Hence {\it instead} of looking for {\it inconsistencies} in $\S_q$, the parties {\it rather} check how many parties in $\S_q$ are reporting the {\it pairwise consistency} of their supposedly
     common share. The idea is that if $\D$ has {\it not} cheated, then a subset of parties $\W_q$ where $\S_q \setminus \W_q \in \Z_a$ should eventually confirm the receipt of a common share
     from $\D$. Hence, the parties check for {\it core-sets} $\W_1, \ldots, \W_{|\Z_a|}$, where each $\S_q \setminus \W_q \in \Z_a$, such that the parties
     in $\W_q$ have confirmed the receipt of a common share from $\D$. 
     Note that {\it irrespective} of $\D$, each $\W_q$ is bound to have {\it at least} one non-faulty party, since $\Z_a$ will satisfy the $\Q^{(1)}(\W_q, \Z_a)$ condition.
     
     The existence of $\W_1, \ldots, \W_{|\Z_a|}$ {\it does not} imply that {\it all} non-faulty parties in $\S_q$ have received a common share, even if $\D$ is {\it non-faulty}, since there {\it might} be
     non-faulty parties outside $\W_q$. Hence, 
     after the confirmation of the sets $\W_1, \ldots, \W_{|\Z_a|}$, the goal is to ensure that every (non-faulty) party in $\S_q \setminus \W_q$
     also gets the common share held by the (non-faulty) parties in $\W_q$. For this, 
     the parties in $\W_q$ reveal their shares to these ``outsider" parties, along with the required IC-signatures.
     The outsider parties then ``filter" out the correctly revealed shares. 
     The existence of at least one non-faulty party in each $\W_q$ guarantees that the shares filtered by the outsider parties are indeed correct.     
     \paragraph{\bf Technical Challenges for Network Agnostic VSS and Way Out.}       
     Since, in our context, the parties will {\it not} be knowing the network type, our approach will be to follow the AVSS of \cite{ACC22c}, where we look for {\it pairwise consistency} of supposedly the common share in each group.
     Namely, $\D$ on having the input $s$, picks random shares $s_1, \ldots, s_{|\Z_s|}$ which sum up to $s$ and distributes $s_q$ to each $S_q \in \ShareSpec_{|\Z_s|}$\footnote{Recall that we need
     $[s]$ with respect to $\Z_s$, {\it irrespective} of the network type.}. The parties in $S_q$ then exchange IC-signed versions of their supposedly common share.
     To avoid an endless wait, 
      the parties can only afford to wait till a subset of parties $\W_q \subseteq S_q$ have confirmed
     the receipt of a common share from $\D$, where $S_q \setminus \W_q \in \Z_s$ holds.
     Unfortunately, $S_q \setminus \W_q \in \Z_s$ {\it need not} guarantee that $\W_q$ has at least one {\it non-faulty} party, since $\Z_s$ {\it need not} satisfy the
     $\Q^{(1)}(\W_q, \Z_s)$ condition, which is {\it desired} as per our semantics of linear secret-sharing with IC-signatures. 
      
      To deal with the above problem, we note that if $\D$ has distributed the shares consistently, then the subset of parties $\specialS \in \ShareSpec_{\Z_s}$ 
      which consists of {\it only} non-faulty parties will publicly report the pairwise consistency of their supposedly common share.
     Hence, we now let $\D$ search for a candidate set $S_p$ of parties from $\ShareSpec_{\Z_s}$ which have publicly confirmed the pairwise consistency of their supposedly common share.
      Once $\D$ finds such a candidate $S_p$, it computes and make public the core-sets $\W_q$ as per the following rules, for $q = 1, \ldots, |\Z_s|$.
      \begin{myitemize}
      \item[--] If all the parties in $S_q$ have confirmed the pairwise consistency of their supposedly common share, then set $\W_q = S_q$. \hfill {\bf (A)}
      \item[--] Else if $\Z_s$ satisfies the $\Q^{(1)}(S_p \cap S_q, \Z_s)$ condition 
       {\it and} 
      the parties in $(S_p \cap S_q)$ have confirmed the consistency of their supposedly common share, then
      set  $\W_q = (S_p \cap S_q)$. \hfill {\bf (B)}
      \item[--] Else set $\W_q = S_q$ and make {\it public} the share $s_q$.  \hfill {\bf (C)}
      \end{myitemize}
      The parties wait till they see $\D$ making public some set $S_p \in \ShareSpec_{\Z_s}$, along with sets $\W_1, \ldots, \W_{|\Z_s|}$. Upon receiving, the parties verify and ``approve" these sets
      as valid, provided {\it all} parties in $S_p$ have confirmed the pairwise consistency of their supposedly common share and if each
      $\W_q$ is computed as per the rule {\bf (A), (B)} or {\bf (C)}. If $\W_1, \ldots, \W_{|\Z_s|}$ are approved, then they indeed {\it satisfy} the requirements of core-sets as per our semantics of
      linear secret-sharing with IC-signatures. While this is trivially true if any $\W_q$ is computed either using rule {\bf (A)} or rule {\bf (B)}, the same holds even if 
      $\W_q$ is computed using rule {\bf (C)}. This is because, in this case, the parties publicly set $[s]_q = s_q$. Moreover, the parties take a ``default" (linear) IC-signature of $s_q$ on the behalf of $S_q$,
      where the IC-signature as well as verification points are {\it all} set to $s_q$.
      
      If $\D$ is {\it non-faulty}, then {\it irrespective} of the network type, it will {\it always} find a candidate $S_p$ and hence, compute and make public
      $\W_1, \ldots, \W_{|\Z_s|}$ as per the above rules. This is because the set $\specialS$ {\it always} constitutes a candidate $S_p$.
      Surprisingly we can show that even if the core-sets are computed with respect to some {\it different} candidate $S_p \neq \specialS$, a {\it non-faulty}
      $\D$ will {\it never} make public the share corresponding to $\specialS$, since the rule {\bf (C)} will {\it not} be applicable over $\specialS$, 
      implying the privacy of $s$.
      If the network is {\it synchronous}, then the parties in $S_p$ {\it as well as} $\specialS$ would report       
      the pairwise consistency of their respective supposedly common share at the {\it same} time.
      This is ensured by maintaining sufficient ``timeouts" in the protocol to report pairwise consistency of supposedly common shares.
      Consequently, rule {\bf (A)} will be applied on $\specialS$.
      For an {\it asynchronous} network, rule {\bf (B)} will be {\it applicable} for $\specialS$, as
      $\Z_s$ will satisfy the $\Q^{(1)}(S_p \cap \specialS, \Z_s)$ condition,  
      due to the $\Q^{(2, 1)}(\PartySet, \Z_s, \Z_a)$ condition and the fact that $\specialS = \PartySet \setminus Z$ for some $Z \in \Z_a$ in the {\it asynchronous} network.
     \subsubsection{Network Agnostic VSS for Multiple Dealers with Linearity}
     \paragraph{\bf Technical Challenge in $\VSS$ for Multiple Dealers.}
    If  {\it different} dealers invoke instances of $\VSS$ to secret-share their inputs, then the linearity property of $[\cdot]$-sharing {\it need not} hold, 
    since the underlying core-sets might be {\it different}.
     In more detail, let $\D_a$ and $\D_b$ be two {\it different} dealers  which invoke instances $\VSS^{(a)}$ and $\VSS^{(b)}$
     to generate $[a]$ and $[b]$ respectively. Let $\W_1^{(a)}, \ldots, \W_{|\Z_s|}^{(a)}$ and $\W_1^{(b)}, \ldots, \W_{|\Z_s|}^{(b)}$ be the underlying core-sets for 
          $\VSS^{(a)}$ and $\VSS^{(b)}$ respectively.  Now consider a scenario where, for some $q \in \{1, \ldots, |\Z_s| \}$, the core-sets 
     $\W_q^{(a)}$ and $\W_q^{(b)}$ are {\it different}, even though 
     $\Z_s$ satisfies the $\Q^{(1)}(\W_q^{(a)} , \Z_s)$ and $\Q^{(1)}(\W_q^{(b)} , \Z_s)$ conditions.
     Let $c \defined a + b$.
     Then the parties in $S_q$ can compute $[c]_q = [a]_q + [b]_q$.
     As part of $[a]$, every (non-faulty) $P_i \in \W_q^{(a)}$ has the IC-signature $\{ \ICSig(P_j, P_i, P_k, [a]_q) \}_{P_j \in \W_q^{(a)}, P_k \not \in S_q}$, while as part of
     $[b]$, every (non-faulty) $P_e \in \W_q^{(b)}$ has the IC-signature $\{ \ICSig(P_d, P_e, P_f, [b]_q) \}_{P_d \in \W_q^{(b)}, P_f \not \in S_q}$, where the underlying IC-signatures satisfy the {\it linearity} property.
     However, since  $\W_q^{(a)} \neq \W_q^{(b)}$, it is {\it not} guaranteed that we have a core-set $\W_q^{(c)}$ as part of $[c]$, where 
     $\Z_s$ satisfies the $\Q^{(1)}(\W_q^{(c)} , \Z_s)$ condition, 
     such that every
     (non-faulty) $P_i \in \W_q^{(c)}$ has the IC-signature $\{ \ICSig(P_j, P_i, P_k, [c]_q) \}_{P_j \in \W_q^{(c)}, P_k \not \in S_q}$.
     If $\W_q^{(a)} = \W_q^{(b)}$, then the parties {\it could} set $\W_q^{(c)}$ to $\W_q^{(a)}$ and the {\it linearity} of IC-signatures would have ensured that 
     every
     (non-faulty) $P_i \in \W_q^{(c)}$ {\it non-interactively} computes $\{ \ICSig(P_j, P_i, P_k, [c]_q) \}_{P_j \in \W_q^{(c)}, P_k \not \in S_q}$ from
     the IC-signatures, held by $P_i$ as part of $[a]$ and $[b]$. In the {\it absence} of any core-set $\W_q^{(c)}$, {\it robust} reconstruction of $[c]_q$ 
     will {\it fail}, which further implies {\it failure} of shared circuit-evaluation of $\ckt$, where the inputs for $\ckt$ are shared by {\it different} parties.
     \paragraph{\bf Way Out.}
     To deal with the above problem, we ensure that the core-sets are {\it common} for {\it all} the secret-shared values during the circuit-evaluation. Namely, there exist {\it global} core-sets
     ${\GW}_1, \ldots, {\GW}_{|\Z_s|}$, which constitute the core-sets for 
     {\it all} the secret-shared values during the circuit-evaluation, where for each ${\GW}_q$,
     $\Z_s$ satisfies the $\Q^{(1)}({\GW}_q, \Z_s)$ condition. 
      Maintaining common core-sets is challenging,
     especially in an {\it asynchronous} network and  $\VSS$ alone is {\it not sufficient} to achieve this goal. Rather we use a {\it different} approach. We generate a ``bunch" of linearly secret-shared
     random values with IC-signatures and common core-sets ${\GW}_1, \ldots, {\GW}_{|\Z_s|}$ in {\it advance} through another protocol called $\Rand$ (discussed in the next section).
     Later, if any party $P_i$ needs to secret-share some $x$, then one of these random values is reconstructed {\it only} towards $P_i$, which uses it as a one-time pad (OTP) and makes public an OTP-encryption of
     $x$. The parties can then take the ``default" secret-sharing of the OTP-encryption with IC-signatures and ${\GW}_1, \ldots, {\GW}_{|\Z_s|}$ as the core-sets and then non-interactively ``remove" the
     pad from the OTP-encryption. This results in $[x]$, with
     ${\GW}_1, \ldots, {\GW}_{|\Z_s|}$ as core-sets. 
          To ensure privacy, we need to generate $L$ random values 
    through $\Rand$, if $L$ is the maximum number of values which need to be secret-shared by different parties during the circuit-evaluation.
     We show that  $L \leq n^3 \cdot c_M + 4n^2 \cdot c_M + n^2 + n$ where $c_M$ is the number of multiplication gates in $\ckt$.
     \subsubsection{Secret-Shared Random Values with Global Core Sets}
     \label{ssec:Rand}
     Protocol $\Rand$ generates linearly secret-shared random values with IC-signatures and {\it common} core-sets.
    We explain the idea behind the protocol for generating one random value.
     The ``standard" way will be 
      to let each $P_i$ pick a random value $r^{(i)}$ and generate $[r^{(i)}]$ by invoking an instance of $\VSS$.
     To avoid an endless wait, the parties only wait for the completion of $\VSS$ instances invoked by a set of dealers $\PartySet \setminus Z$ for some $Z \in \Z_s$.
      To identify 
     the common subset of dealers for which the corresponding $\VSS$ instances have completed, the parties run an instance of {\it agreement on a common subset} (ACS) primitive \cite{BKR94,CanettiThesis}. This involves
     invoking $n$ instances of our network agnostic BA, one on behalf of each dealer, to decide the $\VSS$ instances of which dealers have completed. Let $\C$ be the set of common dealers identified through 
     ACS, where $\PartySet \setminus \C \in \Z_s$. The set $\C$ has at least one {\it non-faulty} party who has shared a random value.
          Hence, the sum of the values shared by the dealers in $\C$ will be random for the adversary.
     \paragraph{\bf Technical Challenges.}
     The above approach {\it fails} in our context due to the following two ``problems" in the protocol $\VSS$, when executed by {\it different} dealers. \\[.2cm]
       \noindent {\bf Problem I}: The first challenge is to maintain the linearity of underlying IC-signatures.
      To understand the issue, consider a triplet of parties $P_i, P_j, P_k$, acting as $\S, \mathsf{I}$ and $\mathsf{R}$ respectively
      in various instances of $\VSS$ invoked by {\it different} dealers.
      Recall that, to maintain the linearity of IC-signatures, it is {\it necessary} that $P_i$ selects the {\it same} set of supporting-verifiers
      $\R$ in {\it all} the instances of  authentication phase involving $P_j$ and $P_k$.
      This is possible {\it only if} $P_i$ 
      knows {\it all} the values on which it wants to generate the IC-signature for $P_j$ 
      and $P_k$ and starts invoking {\it all} the instances of authentication phase.
      Instead, if $P_i$ invokes instances of authentication phase {\it as and when} it has some data to be authenticated for $P_j$ and $P_k$,
      then it {\it may not} be possible to have the {\it same} $\R$ in {\it all} the instances of authentication phase, involving $P_i, P_j$ and $P_k$ in the above roles, especially in an
      {\it asynchronous} network. Since, in $\VSS$, IC-signatures are generated on the supposedly common shares (after receiving them from the underlying dealer) 
      and multiple instances of $\VSS$ are invoked (by {\it different} dealers), 
       this means that $P_i$ should {\it first} have the data from all the dealers for the various instances of $\VSS$ and {\it before} invoking instances of authentication phase to generate
       IC-signatures on these values for $P_j$.
       This {\it may not} be possible, 
       since $P_i$ {\it need not} know {\it beforehand} which dealers it will be receiving shares from as part of $\VSS$. \\[.2cm] 
      \noindent {\bf Way Out.} 
      To deal with the above issue, we now let the dealers {\it publicly} commit their shares for the $\VSS$ instances through {\it secure verifiable multicast} (SVM).
      The primitive allows a designated {\it sender} $\Sender \in \PartySet$ with input $v$ to ``verifiably" send $v$ to a designated set of {\it receivers}  $\ReceiverSet \subseteq \PartySet$, without
      leaking any additional information.
      The verifiability guarantees that even if $\Sender$ is {\it corrupt}, if the non-faulty
       parties in $\ReceiverSet$ get any value from $\Sender$, then it will be {\it common} and {\it all} the (non-faulty) parties in $\PartySet$ will
        ``know" that $\Sender$ has sent some value to $\ReceiverSet$.
         Our instantiation of SVM is very simple: $\Sender$ acts as a dealer and generates $[v]$ through $\VSS$. Once $[v]$ is generated, the parties know
         that $\Sender$ has
        ``committed" to some {\it unknown} value. The next step is to let {\it only} the parties in $\ReceiverSet$ reconstruct $v$.
        
       Using SVM, we now let the various dealers distribute the shares  
         during the underlying instances of $\VSS$ (for $\Rand$) as follows.
        Consider the dealer $P_{\ell}$ who has invoked an instance of $\VSS$ with input $r^{(\ell)}$. For this, it picks 
        random shares
      $r^{(\ell)}_1, \ldots, r^{(\ell)}_{|\Z_s|}$ which sum up to $r^{(\ell)}$.
      Now {\it instead} of directly sending send $r^{(\ell)}_q$ to the parties in $S_q$, it invokes $|\Z_s|$ instances of SVM with input $r^{(\ell)}_1, \ldots, r^{(\ell)}_{|\Z_s|}$ and
      $S_1, \dots, S_{|\Z_s|}$ as the designated set of receivers respectively. 
      This serves {\it two} purposes. First, it guarantees that all the parties in $S_q$ receive a {\it common} share from $P_{\ell}$.
      Second and more importantly, {\it all} the parties in $\PartySet$ will now {\it know} that $P_{\ell}$ has distributed shares to each set from $\ShareSpec_{\Z_s}$.
     The parties then run an instance of ACS to identify a common subset of {\it committed dealers} $\CD \subseteq \PartySet$, where
      $\PartySet \setminus \CD \in \Z_s$, which have invoked the desired instances of SVM and delivered the required shares to each group $S_q \in \ShareSpec_{|\Z_s|}$. 
      The way timeouts are maintained as part of the ACS, it will be ensured that in a {\it synchronous} network, {\it all} {\it non-faulty} dealers are present in $\CD$.
      Once the set $\CD$ is identified, it is guaranteed that every {\it non-faulty} party $P_i$ will have the shares from {\it all} the dealers in $\CD$. And once it has the shares from all the dealers in $\CD$, it starts
      generating the IC-signatures on these shares for the designated parties as part of the $\VSS$ instances corresponding to the dealers in $\CD$ and ensures that all the pre-requisites are satisfied
      to guarantee the linearity of the underlying IC-signatures.
      Now {\it instead} of selecting the set of dealers $\C$ (for $\Rand$) from $\PartySet$, the parties run an instance of ACS over the set of committed dealers $\CD$ to
       select $\C$ where $\CD \setminus \C \in \Z_s$ holds.
       We stress that {\it irrespective} of the network type, the set $\C$ is {\it still} guaranteed to have at least one non-faulty party.
       While this is trivially true in an {\it asynchronous} network where $\Z_a$ satisfies the $\Q^{(1)}(\C, \Z_a)$ condition, the same is true in the {\it synchronous} network because 
       $\CD$ will have {\it all} non-faulty dealers.       \\[.2cm]
      \noindent {\bf Problem II}:
      The second problem (in the proposed $\Rand$) is that      
      the underlying core-sets might be {\it different} for the values shared by the dealers in $\CD$ (and hence $\C$). 
      Instead, we require every dealer in $\CD$ to secret-share random values with {\it common} underlying core-sets. Only then will it be ensured that the random values
      generated through $\Rand$ are secret-shared with common core-sets. \\[.2cm]
      \noindent{\bf Way Out.} Getting rid of the above problem is {\it not} possible if we let every dealer in $\CD$ compute {\it individual} core-sets during their respective instances of $\VSS$, as per the steps of
      $\VSS$. Recall that in $\VSS$, the dealer $\D$ computes the underlying core-sets with respect to the ``first" set of parties $S_p$ from $\ShareSpec_{|\Z_s|}$ which confirm
      the pairwise consistency of their supposedly common share after exchanging IC-signatures on these values.
      As a result, different dealers (in $\Rand$) may end up computing {\it different} core-sets in their instances of $\VSS$ with respect to {\it different}
      candidate $S_p$ sets. 
      To deal with this issue, we {\it instead} let each dealer in $\CD$ {\it continue} computing and publishing {\it different} ``legitimate" core-sets with respect to various ``eligible" candidate 
      $S_p$ sets from $\ShareSpec_{\Z_s}$. The parties run an instance of ACS to identify a {\it common} subset of dealers $\C \in \CD$ where $\CD \setminus \C \in \Z_s$, such that {\it all} the dealers
      have computed and published ``valid" core-sets, computed with the respect to the {\it same} $S_p \in \ShareSpec_{\Z_s}$.
      The idea here is that there always exists a set $\specialS \in \ShareSpec_{\Z_s}$ consisting of {\it only} non-faulty parties.
      So if the set of {\it non-faulty} dealers $\Hon$ in $\CD$ keep computing and publishing {\it all possible} candidate core-sets  in their $\VSS$ instances, then they will
       publish core-sets with respect to $\specialS$. 
      Hence, $\Hon$ and $\specialS$ {\it always} constitute the candidate $\CD$ and the common $S_p$ set.
      
      Note that identifying $\C$ out of $\CD$ through ACS satisfying the above requirements
      is {\it non-trivial} and requires carefully executing
      the underlying instances of BA in  ``two-dimensions".
      We first run $|\Z_s|$ instances of $\BA$, one on the behalf of each set in $\ShareSpec_{\Z_s}$, where the $q^{th}$ instance is executed to decide whether a subset of dealers in $\CD \setminus Z$ for some
      $Z \in \Z_s$
      have published valid core-sets with respect to the set $S_q \in \ShareSpec_{\Z_s}$.
      This enables the parties to identify a {\it common} set $S_{\qcore} \in \ShareSpec_{\Z_s}$, such that it is guaranteed that a subset of 
      dealers in $\CD \setminus Z$ for some
      $Z \in \Z_s$
      have indeed published valid core-sets with respect to the set $S_{\qcore}$.
      Once the set $S_{\qcore}$ is identified, the parties then run $|\CD|$ instances of BA to decide which dealers in $\CD$ have 
      published core-sets with respect to $S_{\qcore}$.
     \subsubsection{Network Agnostic Secure Multiplication}
     \label{ssec:Mult}
    To generate secret-shared random multiplication-triples for evaluating the multiplication gates in $\ckt$ (using Beaver's trick), we need a network agnostic secure multiplication protocol which securely generates
    a secret-sharing of the product of two secret-shared values. 
    The key subprotocol behind our multiplication protocol is a {\it non-robust} multiplication protocol $\BasicMult$ (standing for basic multiplication), which takes inputs 
        $[a]$ and $[b]$ and an {\it existing} set of {\it globally discarded} parties $\Discarded$, which contains only {\it corrupt} parties.
       The protocol securely generates $[c]$  {\it without}
      revealing any additional  information about $a, b$ (and $c$). 
       If {\it no} party in $\PartySet \setminus \Discarded$ cheats, then $c = a \cdot b$ holds.
       The idea behind the protocol is to let each {\it summand} $[a]_p \cdot [b]_q$ be secret-shared by a {\it summand-sharing party}. Then $[a \cdot b]$ can be computed
       from the secret-sharing of each summand, owing to the linearity property. 
           Existing multiplication protocols in the {\it synchronous} and {\it asynchronous} setting \cite{HT13,ACC22c} also use an 
      instantiation of $\BasicMult$, based on the above idea. 
      In the sequel, we recall them, followed by the technical challenges faced in the network agnostic setting and how we deal with them.     
      \paragraph{\bf $\BasicMult$ in the Synchronous Setting with $\Q^{(2)}(\PartySet, \Z_s)$ Condition \cite{HT13}.}
      In \cite{HT13}, each summand $[a]_p \cdot [b]_q$ is {\it statically} assigned to a {\it designated} summand-sharing party through some {\it deterministic} assignment,
       which is possible since
      $[a]_p$ and $[b]_q$ are held by the parties in $(S_p \cap S_q)$. This is {\it non-empty}, since the
      $\Q^{(2)}(\PartySet, \Z_s)$ condition holds.
                 Since the parties in $\Discarded$ are {\it already} known to be corrupted, all the shares $[a]_p, [b]_p$ held by the parties in $\Discarded$
      are {\it publicly reconstructed} and {\it instead} of letting the parties in $\Discarded$ secret-share their assigned summands, the parties take the ``default" secret-sharing
      of these summands. 
        \paragraph{\bf $\BasicMult$ in the Asynchronous Setting with $\Q^{(3)}(\PartySet, \Z_a)$ Condition \cite{ACC22c}.}
       The idea of {\it statically} designating each summand $[a]_p \cdot [b]_q$ to a {\it unique} party in $\PartySet \setminus \Discarded$
        {\it need not} work in the {\it asynchronous} setting, since the designated party may be {\it corrupt} and {\it need not} secret-share any summand, thus resulting in an {\it endless} wait.
        To deal with this challenge, \cite{ACC22c} {\it dynamically} selects the summand-sharing parties for each summand.
        In more detail, let 
        $\Z_a = \{\mathsf{Z}_1, \ldots, \mathsf{Z}_{|\Z_a|} \}$ and $\ShareSpec_{\Z_a} = \{\S_1, \ldots, \S_{|\Z_a|} \}$, where each $\S_r = \PartySet \setminus \mathsf{Z}_r$.
        Since the $\Q^{(3)}(\PartySet, \Z_a)$ condition is satisfied and $\Discarded \in \Z_a$, it follows that $(\S_p \cap \S_q) \setminus \Discarded \neq \emptyset$.
        This implies that there exists at least one non-faulty party in $(\S_p \cap \S_q)$ who can secret-share
        the summand $[a]_p \cdot [b]_q$. Hence, {\it every} party in $\PartySet \setminus \Discarded$ is allowed to secret-share {\it all} the summands it is ``capable" of, with special care taken to ensure
        that each summand $[a]_p \cdot [b]_q$ is considered {\it exactly once}.
        For this, the protocol now proceeds in ``hops", where in each hop all the parties in $\PartySet \setminus \Discarded$ secret-share all the summands they are capable of, {\it but} a {\it single}
        summand sharing party is finally selected for the hop through ACS. Then, all the summands which have been shared by the elected summand-sharing party
        are ``marked" as shared and not considered for sharing in the future hops. Moreover, a party who has already served as a summand-sharing party is {\it not}
         selected in the future hops.
       \paragraph{\bf Technical Challenges in the Network Agnostic Setting.}
       The {\it asynchronous} $\BasicMult$ based on {\it dynamically} selecting summand-sharing parties
         will {\it fail} in the {\it synchronous} network, since the $\Q^{(3)}$ condition {\it need not} be satisfied.
       On the other hand,  {\it synchronous} $\BasicMult$ based on {\it statically} selecting summand-sharing parties will {\it fail} if a designated summand-sharing party {\it does not} secret-share the required summands,
       resulting in an endless wait.
       The way out is to select summand-sharing parties in {\it three} phases. We 
       {\it first} select summand-sharing parties dynamically in {\it hops}, following the approach of \cite{ACC22c}, till we find a subset of parties from $\ShareSpec_{\Z_s}$ which have shared {\it all} the summands 
       they are capable of. Then in the {\it second} phase, the {\it remaining} summands which are {\it not} yet secret-shared are {\it statically} assigned and shared by the respective designated summand-sharing
       parties. To avoid an endless wait in this phase,  
        the parties wait only for a ``fixed" time required for the parties to secret-share the assigned summands (corresponding to the time taken in a {\it synchronous}
       network) and run instances of BA to identify which of the designated summand-sharing parties have shared their summands up during the second phase.
       During the {\it third} phase, any ``leftover" summand which is {\it not yet} shared is {\it publicly} reconstructed by reconstructing the
       corresponding shares and a default sharing is taken for such summands.

       The idea here is the following: {\it all} non-faulty parties will 
        share the summands which are assigned to them, either statically or dynamically, {\it irrespective} of the network type.
        Consequently, the first phase will be always over, since the set consisting of {\it only} non-faulty parties always constitutes a candidate
        set of summand-sharing parties which the parties look for to complete of the first phase.
        Once the first phase is over, the second phase is bound to be over since the parties wait {\it only} for a fixed time.
        The third phase is always bound to be over, once the first two phases are over, since it involves publicly reconstructing the leftover summands.
        The way summands are assigned across the three phases, it will be always guaranteed that every summand is considered for sharing once in exactly one of the three phases and no summand will be left out.
         The crucial point here is that the 
        the shares held {\it only} by the non-faulty parties {\it never} get publicly reconstructed, thus guaranteeing that the adversary {\it does not} learn any additional information about $a$ and $b$. 
        This is obviously true in a {\it synchronous} network because  we always have the {\it second} phase where every non-faulty party 
         who is {\it not} selected as a summand-sharing party during the first phase 
        will get the opportunity
        to secret-share its assigned summands. 
        On the other hand, in an {\it asynchronous} network, it can be shown that {\it all} the summands which involve any share held by the non-faulty parties would have been secret-shared  
        during the {\it first} phase itself. In more detail, let $Z^{\star} \in \Z_a$ be the set of {\it corrupt} parties and let $\Hon = \PartySet \setminus Z^{\star}$ be the set of {\it honest} parties. 
        Moreover, let $S_h \in \ShareSpec_{\Z_s}$ be the group consisting of {\it only} non-faulty parties which 
        hold the shares $[a]_h$ and $[b]_h$.
       Consider an {\it arbitrary} summand $[a]_h \cdot [b]_q$. Suppose the first phase gets over because 
       every party in $S_{\ell} \in \ShareSpec_{\Z_s}$
        has been selected as a summand-sharing party during the {\it first} phase.         
        Then consider the set $(S_{\ell} \cap \Hon  \cap S_q)$, which is {\it not} empty due to
       the $\Q^{(2, 1)}(\PartySet, \Z_s, \Z_a)$ condition. 
       Hence, there exists {\it some} 
      $P_j \in (\Hon \cap S_{\ell} \cap S_q)$,
      who would have shared $[a]_h \cdot [b]_q$ when selected as a summand-sharing party during some hop in the first phase. 
      Due to a similar reason, any summand of the form $[a]_q \cdot [b]_h$ would have been secret-shared during the first phase itself. 
     \subsection{Other Related Works}
     The domain of network agnostic cryptographic protocols is relatively new and almost all the existing works have considered {\it threshold} adversaries. 
       The work of  \cite{BKL21} presents a network agnostic {\it cryptographically-secure} atomic broadcast protocol. 
     The work of \cite{MO21} studies 
     Byzantine fault tolerance and state machine replication protocols for multiple thresholds, including $t_s$ and $t_a$. 
     The work of \cite{GLW22} presents a network agnostic protocol for the task of approximate agreement using the condition
     $2t_s + t_a < n$. The same condition has been used to design a network agnostic distributed key-generation (DKG) protocol in \cite{BCLL22}. 
           A recent work \cite{DL22} has studied the problem of network agnostic perfectly-secure message transmission (PSMT) \cite{DDWY93} over {\it incomplete} graphs.
     \subsection{Open Problems}
     There are several interesting directions to explore for network agnostic MPC protocols. Here we mention few of them.
     It is {\it not} known whether the condition $3t_s + t_a < n$ (resp.~$\Q^{(3, 1)}(\PartySet, \Z_s, \Z_a)$)
     is {\it necessary} for the network agnostic MPC with {\it perfect} security against threshold (resp.~non-threshold) adversary.
     An abundant amount of research effort has been spent to improve both the theoretical as well as practical efficiency of (unconditionally-secure) SMPC and AMPC protocols.
     The works of \cite{ACC22a,ACC22b} and this work just focus on the {\it possibility} of unconditionally-secure network agnostic MPC. 
     Upgrading the efficiency of these protocols to those of state of the art SMPC and AMPC protocols seems to require a significant research effort.
     Even though the complexity of our MPC protocol is polynomial in $n$ and $|\Z_s|$, when instantiated for {\it threshold} adversaries (where $|\Z_s|$ has all subsets of $\PartySet$ of size up to $t_s$),
     it may require an {\it exponential} (in $n$) amount of computation and communication. This is {\it unlike} the case for {\it perfect security}, where we have a network agnostic MPC 
     protocol against threshold adversaries
     with a complexity polynomial (in $n$) \cite{ACC22a}. Hence, designing network agnostic MPC protocol against {\it threshold} adversaries with {\it statistical} security and {\it polynomial} complexity is left as a challenging open 
     problem.     

%% file: prelims.tex
\section{Preliminaries and Definitions}
\label{sec:prelims}
We assume the {\it pair-wise secure channel} model, where 
 the parties in  $\PartySet$ are assumed to be connected 
 by pair-wise secure channels.
    The underlying communication network can be either synchronous or asynchronous, with parties
     being {\it unaware} about the exact network type.
   In a {\it synchronous} network, every message sent is delivered within a {\it known} time
   $\Delta$. 
     In an {\it asynchronous} network, messages can be delayed arbitrarily, but finitely, with every message sent  
      being delivered {\it eventually}. The distrust among $\PartySet$ is modelled by a {\it malicious} (Byzantine)
      adversary $\Adv$, who can corrupt a subset of the parties in $\PartySet$ and force them to behave in any 
      arbitrary fashion during the execution of a protocol.
      The parties {\it not} under the control of $\Adv$ are called {\it honest}.
      We assume
      the adversary to be {\it static}, who
      decides the set of corrupt parties at the beginning of the protocol execution. 
      As our main goal is to show the possibility of statistically-secure network agnostic MPC, we keep the formalities to a bare minimum and prove 
      the security of our protocols using the {\it property-based} 
      definition, by listing the security properties achieved by our protocols. However, 
      our protocols can also be proven to be secure using the more rigorous
      {\it Universal Composability} (UC) definitional framework \cite{Can01}, without affecting their 
      efficiency.

 Adversary $\Adv$ can corrupt any one subset of parties from $\Z_s$ and $\Z_a$ in {\it synchronous}
   and {\it asynchronous} network respectively. 
     The adversary structures are {\it monotone}, implying that if $Z \in \Z_s$ ($Z \in \Z_a$ resp.), then every subset of $Z$ also belongs to
      $\Z_s$ (resp.~$\Z_a$). 
      We assume that $\Z_s$ and $\Z_a$ satisfy the conditions $\Q^{(2)}(\PartySet, \Z_s)$ and $\Q^{(3)}(\PartySet, \Z_a)$ respectively, which are {\it necessary} for
      statistically-secure MPC in the synchronous and asynchronous network respectively.
      Additionally, we assume that $\Z_a \subset \Z_s$. Moreover, 
      $\Z_s$ and $\Z_a$ satisfy the $\Q^{(2, 1)}(\PartySet, \Z_s, \Z_a)$ condition.
      
      In our protocols, 
   all computations are done over a finite field $\F$, where $|\F| > n^5 \cdot 2^{\ssec}$ and $\ssec$ is the underlying statistical security parameter.
   Looking ahead, this will ensure that the error probability in our MPC protocol is upper bounded by $2^{-\ssec}$. 
    Without loss of generality, we assume that each $P_i$ has an input $x_i \in \F$,
      and the parties want to securely compute a function $f:\F^n \rightarrow \F$, 
     represented by an arithmetic circuit $\ckt$ over $\F$, 
      consisting of linear and non-linear (multiplication) gates, where
      $\ckt$ has $c_M$ multiplication gates and a multiplicative depth of $D_M$. 
      
        We assume the existence of an {\it unconditionally-secure public-key infrastructure} (PKI),
        for an unconditionally-secure signature scheme, also called {\it pseudo-signature}  \cite{PW96,FitziThesis}.
        We briefly explain the requirements from such a setup and refer to \cite{FitziThesis} for complete formal details.
        There exists a {\it publicly known} vector of public keys $(pk_1, \ldots, pk_n)$, where 
         each {\it honest}  $P_i$ holds the generated secret key $sk_i$, associated with $pk_i$.\footnote{Corrupt parties may choose their keys arbitrarily.}
          A {\it valid} signature $\tau$ on message $m$ from $P_i$ is one for which $\mathsf{Verify}_{pk_i}(m, \tau) = 1$, where
          $\mathsf{Verify}$ is the {\it verification function} of the underlying signature scheme. 
           For simplicity, we make the standard convention of treating signatures as {\it idealized} objects
           during our protocol analysis; i.e., we assume that the signatures are 
           {\it  perfectly unforgeable} and hence $\Adv$ will fail to forge signature of an honest party on any message, which is not signed by the party.
            We also assume that the signatures are {\it transferable} and any party upon receiving a valid signature from a party can send and get it verified by any other party.
            However, {\it unlike} the standard digital signatures which are {\it computationally-secure} and offer arbitrary number of transfers, pseudo-signatures offer a ``limited"
            number of transfers, which is typically bounded as a function of the number of parties in the protocol where pseudo-signature is used as a primitive.
            We assume that the given setup supports the required number of transfers demanded by our protocols. We use $|\sigma|$ to represent the size of a pseudo-signature in bits.
            If $P_i$ signs a message $m$, then we denote the resultant signed message as $\sign{m}_{i}$. \\[.2cm]
\noindent{\bf  Termination Guarantees of Our Sub-Protocols:}
 As done in \cite{ACC22a,ACC22c},  for simplicity, we will {\it not} be specifying any {\it termination} criteria
 for our sub-protocols. The parties will keep on participating in these sub-protocol instances even {\it after} computing
 their outputs. The termination criteria of our MPC protocol will ensure the 
 termination of {\it all} underlying sub-protocol instances. We will be using an existing {\it randomized} ABA protocol \cite{Cho23}
  which ensures that the
 honest parties (eventually) obtain their respective output {\it almost-surely} with probability $1$.
 This means that the probability that an honest party obtains its output after participating for {\it infinitely} many rounds approaches $1$ {\it asymptotically} \cite{ADH08,MMR15,BCP20}.
   That is:
   \[ \underset{T \rightarrow \infty}{\mbox{lim}} \mbox{Pr}[\mbox{An honest } P_i \mbox{ obtains its output by local time } T] = 1,\]
   where the probability is over the random coins of the honest parties and the adversary in the 
   protocol. 
  The  property of {\it almost-surely} obtaining the output carries over to the ``higher" level protocols, where 
  ABA is used as a building block. 
  We will say that the ``{\it honest parties obtain some output almost-surely from protocol $\Pi$}" to mean that
  every {\it honest} $P_i$ {\it asymptotically} obtains its output in $\Pi$ with probability $1$, in the above sense.

%% file: BA.tex
\section{Network Agnostic Unconditionally Secure Byzantine Agreement}
\label{sec:BA}
 We recall the definition of BA from \cite{ACC22b}, which is 
  adapted from \cite{BKL19,ACC22a}.
\begin{definition}[{\bf BA}]
\label{def:BA}
Let $\Pi$ be a protocol for $\Partyset$, where every party $P_i$ has an input $b_i \in \{0, 1\}$ and a possible output from  $\{0, 1, \bot \}$.
 Moreover, let $\Adv$ be a computationally-unbounded adversary, characterized by adversary structure $\AdvStructure$, where $\Adv$ can corrupt
  any subset of parties from $\AdvStructure$ during the execution of $\Pi$.
     \begin{myitemize}
    \item[--] {\bf $\AdvStructure$-Guaranteed Liveness}: $\Pi$ has $\AdvStructure$-guaranteed liveness
    if all honest parties obtain an output.
       \item[--] {\bf $\AdvStructure$-Almost-Surely Liveness}: $\Pi$ has $\AdvStructure$-almost-surely liveness
       if, almost-surely, all honest parties obtain some output.
      \item[--] {\bf $\AdvStructure$-Validity}: $\Pi$ has $\AdvStructure$-validity if the following holds:
      If all honest parties have input $b$, then every honest party with an output, outputs $b$. 
      \item[--] {\bf $\AdvStructure$-Weak Validity}: $\Pi$ has $\AdvStructure$-weak validity if the following holds:
      If all honest parties have input $b$, then every honest party with an output, outputs $b$ or $\bot$.      
      \item[--] {\bf $\AdvStructure$-Consistency}: $\Pi$ has $\AdvStructure$-consistency
      if all honest parties with an output, output the same value (which can be $\bot$).
      \item[--] {\bf $\AdvStructure$-Weak Consistency}: $\Pi$ has $\AdvStructure$-weak consistency
      if all honest parties with an output, output either a common $v \in \{0, 1 \}$ or $\bot$.
      \end{myitemize}
  $\Pi$ is called a {\it $\AdvStructure$-secure synchronous BA} (SBA) if, in a {\it synchronous} network, it achieves
    {\it $\AdvStructure$-guaranteed liveness}, {\it $\AdvStructure$-validity}, and {\it $\AdvStructure$-consistency}.
   $\Pi$ is called a {\it $\AdvStructure$-secure 
  asynchronous BA} (ABA) if, in an {\it asynchronous network} it has 
  {\it $\AdvStructure$-almost-surely liveness},  {\it $\AdvStructure$-validity} and
   {\it $\AdvStructure$-consistency}.\footnote{The seminal FLP impossibility result \cite{FLP85} rules out the possibility of {\it any} deterministic ABA, where there always exists a ``bad" execution in which
   the honest parties may keep on running the protocol forever, without obtaining any output. To circumvent this, one can opt for randomized ABA protocols and hope that the bad executions occur asmptotically
   with probability $0$.}
 \end{definition}
 To design our network agnostic BA protocol, we will be using a special type of broadcast protocol. We next review the definition of broadcast
     from \cite{ACC22b} which is further
   adapted from \cite{BKL19,ACC22a}.
\begin{definition}[{\bf Broadcast}]
\label{def:BCAST}
Let $\Pi$ be a protocol, where a designated sender $\Sender \in \Partyset$ has input $m \in \{0, 1\}^{\ell}$, and parties obtain a possible output,
  including $\bot$.
  Moreover, let $\Adv$ be a computationally-unbounded adversary, characterized by an adversary structure $\AdvStructure$, where $\Adv$ can corrupt
  any subset from $\AdvStructure$ during $\Pi$.
   \begin{myitemize}
   \item[--] {\bf $\AdvStructure$-Liveness}: $\Pi$ has $\AdvStructure$-liveness if 
    all honest parties obtain some output.
    \item[--] {\bf $\AdvStructure$-Validity}: $\Pi$ has $\AdvStructure$-validity if the following holds:
    if $\Sender$ is {\it honest}, then every honest party with an output, outputs $m$.
     \item[--] {\bf $\AdvStructure$-Weak Validity}: $\Pi$ has $\AdvStructure$-weak validity if the following holds: 
      if $\Sender$ is {\it honest}, then every honest party with an output, outputs either $m$ or $\bot$.
     \item[--] {\bf $\AdvStructure$-Consistency}: $\Pi$ has $\AdvStructure$-consistency if the following holds:
     if $\Sender$ is {\it corrupt}, then every honest party with an output, outputs a common value.
     \item[--] {\bf $\AdvStructure$-Weak Consistency}: $\Pi$ has $\AdvStructure$-weak consistency if the
     following holds: if $\Sender$ is {\it corrupt}, then every honest party with an output, outputs a common
      $m^{\star} \in \{0, 1 \}^{\ell}$ or $\bot$.
  \end{myitemize}
$\Pi$ is called a {\it $\AdvStructure$-secure broadcast} protocol if it has {\it $\AdvStructure$-Liveness}, 
{\it $\AdvStructure$-Validity}, and {\it $\AdvStructure$-Consistency}.
 \end{definition}

We next recall a blueprint for network agnostic BA from \cite{BKL19,ACC22a}.
 \subsection{A Blueprint for Network Agnostic BA \cite{BKL19,ACC22a}}
  To design our network agnostic BA, we {\it assume} the existence of the following sub-protocols.
  \begin{myitemize}
  \item[--] {\it Synchronous BA with Asynchronous Guarantees}: we assume the existence of a protocol $\SBA$, which is a $\Z_s$-secure SBA and which has $\Z_a$-weak validity and
   $\Z_a$-guaranteed liveness in the {\it asynchronous} network. At (local) time $\TimeSBA$, all
   honest parties will have an output (which could be $\bot$), {\it irrespective} of the network type. 
  \item[--]  {\it Asynchronous BA with Synchronous Guarantees}: we assume the existence of a protocol $\ABA$, which is a $\Z_a$-secure ABA. Moreover, in a {\it synchronous}
  network, the protocol has $\Z_s$-validity, with all honest parties computing their output within time $\TimeABA = c \cdot \Delta$ in this case, for some {\it known constant} $c$.\footnote{Thus $\ABA$ has
   $\Z_s$-guaranteed liveness in the {\it synchronous} network,
    {\it if} all honest parties have the {\it same} input. However, if the honest parties have {\it different} inputs, then $\ABA$ {\it need not} provide guaranteed liveness or consistency guarantees in
    the {\it synchronous} network.} 
  \end{myitemize}
   Based on $\SBA$ and $\ABA$, one can design a network agnostic
    BA protocol $\BA$ as follows, following the blueprint of \cite{BKL19,ACC22a}. The parties first invoke an instance of $\SBA$, {\it assuming} a {\it synchronous} network.
   If the network is indeed synchronous, then the (honest) parties should have a {\it binary} output at time $\TimeSBA$. The parties check the same
   and either switch their input to the output of $\SBA$, if it is {\it not} $\bot$, or stick to their original input. The parties then invoke an instance of $\ABA$ with ``updated" inputs and the output of
   $\ABA$ is set to be the overall output. The description of $\BA$, taken from \cite{ACC22a}, is presented in 
   Fig \ref{fig:BA}.
   \begin{protocolsplitbox}{$\BA$}{Network agnostic BA protocol from $\SBA$ and $\ABA$. The above code is executed
 by $P_i$.}{fig:BA}
\justify
\begin{myitemize}
\item[--] On having input $b_i$, participate in an instance of $\SBA$ with input $b_i$ and {\color{red} wait till the local time becomes
  $\TimeSBA$}. Let $v_i$ be the output from $\SBA$
  at time $\TimeSBA$.
 If $v_i \neq \bot$, then set $v_i^{\star} = v_i$.  Else set $v_i^{\star} = b_i$.   
\item[--] Participate in an instance of $\ABA$ with input $v_i^{\star}$. Output the result of $\ABA$, when it is available.
\end{myitemize}
\end{protocolsplitbox}

Theorem \ref{thm:BA}, follows from \cite{ACC22a}, given that $\SBA$ and $\ABA$ achieve the stated properties. For completeness, the theorem is proved in Appendix \ref{app:BA}.
\begin{theorem}
\label{thm:BA}
  Let $\Z_a \subset \Z_s$ such that $\Z_s$ and $\Z_a$ satisfy the 
   conditions $\Q^{(2)}(\PartySet, \Z_s)$ and $\Q^{(3)}(\PartySet, \Z_a)$ respectively.
   Moreover, let
      $\Z_s$ and $\Z_a$ satisfy the $\Q^{(2, 1)}(\PartySet, \Z_s, \Z_a)$ condition.
   Then protocol $\BA$ achieves the following.\footnote{As the number of invocations of $\BA$ in our MPC protocol will be {\it independent} of
  $|\ckt|$, we {\it do not} focus on its {\it exact} complexity. However, we confirm that it will be polynomial in $n$ and $|\Z_s|$.}
\begin{myitemize}
\item[--] {\bf Synchronous Network}: the protocol is a $\Z_s$-secure SBA, where all honest parties get their output at time $\TimeBA = \TimeSBA + \TimeABA$.
\item[--] {\bf Asynchronous Network}: the protocol is a $\Z_a$-secure ABA.
\end{myitemize}
\end{theorem}
   
   We now proceed to instantiate protocols $\SBA$ and $\ABA$.
  \subsection{$\SBA$: Synchronous BA with Asynchronous Weak Validity and Guaranteed Liveness}
   To design protocol $\SBA$, we again follow the blueprint of \cite{ACC22a}, which design $\SBA$ based on {\it three} components.
    \subsubsection{SBA with Asynchronous Guaranteed Liveness}
    The {\it first} component for designing $\SBA$ is an SBA protocol, which has {\it just} guaranteed liveness in an {\it asynchronous} network. 
    In \cite{PW96}, the authors have presented an SBA protocol against {\it threshold} adversaries, tolerating up to $t_s < n/2$ faults.  
  The protocol which we denote as $\PiPW$, modifies the
  Dolev-Strong BA protocol \cite{DS83}, by {\it replacing} digital signatures with
  pseudo-signatures. We note that $\PiPW$ can be easily generalized, if $\Z_s$ satisfies the $\Q^{(2)}(\PartySet, \Z_s)$ condition. 
   To achieve guaranteed liveness in an {\it asynchronous} network, the parties run the protocol till the supposed timeout in the synchronous network and check if any output is computed
   and in case no output is computed, $\bot$ is taken as the output. 
  Since the protocol is an easy generalization of the existing protocol against threshold adversaries,
  protocol $\PiPW$ and proof of Lemma \ref{lemma:PW} are available in Appendix \ref{app:BA}.
     \begin{lemma}
 \label{lemma:PW}
   Protocol $\PiPW$ achieves the following, where 
    $t$ is the cardinality of the maximum-sized subset in $\Z_s$. 
 \begin{myitemize}
 \item[--] {\bf Synchronous Network}: The protocol is a $\Z_s$-secure SBA protocol, where all honest parties compute their output at time $\TimePW \defined (t + 1) \cdot \Delta$. 
 \item[--] {\bf Asynchronous Network}: The protocol achieves $\Z_a$-guaranteed liveness, where all honest parties compute their output at time $\TimePW$. 
 \item[--] {\bf Communication Complexity}: $\Order(n^4 \cdot \ell \cdot |\sigma|)$ bits are sent by the honest parties, if the inputs of the parties are of size $\ell$ bits.
 \end{myitemize}
 \end{lemma}
    \subsubsection{Asynchronous Broadcast with Synchronous Guarantees}
  The {\it second} component for designing $\SBA$ is an {\it asynchronous} broadcast protocol $\Acast$ (also called Acast), which provides
   {\it liveness, validity} and a ``variant" of consistency in a {\it synchronous} 
   network. The variant guarantees that if $\Sender$ is {\it corrupt}
   and the honest parties compute an output in a {\it synchronous} network, then they {\it may not} do it at the {\it same} time
   and there might be a {\it gap} in the time at which the honest parties compute an output.
   In \cite{KF05,ACC22b}, an instantiation of $\Acast$ is provided against 
    $\Q^{(3)}$ adversary structures.\footnote{The protocol is a generalization of the classic Bracha's {\it threshold} Acast protocol \cite{Bra84}, tolerating
    $t < n/3$ corruptions.} 
    Unfortunately, the protocol {\it fails} to provide {\it any} security guarantees in a {\it synchronous} network against $\Q^{(2)}$ 
   adversary structures. So we provide a {\it different} instantiation of Acast for our setting (see Fig \ref{fig:Acast}). The protocol is obtained by generalizing the ideas used in the 
   broadcast protocol of \cite{MR21}. The protocol of \cite{MR21} uses a computational PKI, which we replace with
   an unconditional PKI. The protocol consists of three phases and each (honest) party executes a phase {\it at most} once.
  \begin{protocolsplitbox}{$\Acast(\Sender, m, \Z_s, \Z_a)$}{Asynchronous broadcast with synchronous guarantees. The above code is executed by each
    $P_i \in \PartySet$ including the designated sender $\Sender$}{fig:Acast}
  Each $P_i \in \PartySet$ executes each of the following phases {\it at most} once.
  \begin{myitemize}
  \item {\bf (Propose):} If $P_i = \Sender$, then on having the input $m$, send
    $\sign{(\propose, m)}_{\Sender}$ to all the parties.
  \item {\bf (Vote):} Upon receiving the {\it first} $\propose$ message
  $\sign{(\propose, m)}_{\Sender}$ 
    from $\Sender$ with valid signature, send $\sign{(\propose, m)}_{\Sender}$ to all the parties
    and
    \textcolor{red}{wait till the local time increases by $\Delta$}. If
    $\sign{(\propose, m')}_{\Sender}$ is not received from any party where $m' \neq m$, then send
    $\sign{(\vote, m)}_i$ to all the parties.
  \item {\bf (Output):} Upon receiving a $\vote$ message $\sign{(\vote, m)}_{j}$ with valid
    signature corresponding to every $P_j \in \PartySet \setminus Z$ for some $Z \in \Z_s$, do the following:
     \begin{myitemize}
     \item[--] Let $\cert{m}$ denote the collection of signed $\sign{(\vote, m)}_{j}$ messages.
     \item[--] Send $\cert{m}$ to all the parties and output $m$.
     \end{myitemize}
  \end{myitemize}
\end{protocolsplitbox}

The proof of Lemma \ref{lemma:Acast} is available in Appendix \ref{app:BA}.
 \begin{lemma}
 \label{lemma:Acast}
Protocol $\Acast$ achieves the following properties.
 \begin{myitemize}
  \item[--] Asynchronous Network: The protocol is a $\Z_a$-secure broadcast protocol. 
  \item[--] Synchronous Network: {\bf (a) $\Z_s$-Liveness}: If $\Sender$ is honest, then all honest parties obtain an output
    within time $3\Delta$. {\bf (b) $\Z_s$-Validity}: If $\Sender$ is honest, then every honest party with an output, outputs $m$.
  {\bf (c) $\Z_s$-Consistency}: If $\Sender$ is corrupt and some honest party outputs $m^{\star}$ at time $T$,
   then every honest $P_i$ outputs $m^{\star}$ by the end of time $T+ \Delta$.
   \item[--] Communication Complexity: $\Order(n^3 \cdot \ell \cdot |\sigma|)$ bits are communicated by the honest parties, where
    $\ell$ is the size of $\Sender$'s input.
 \end{myitemize}
 \end{lemma}
 \paragraph{\bf Terminologies for Using $\Acast$.}
  In the protocol $\Acast$, any party from $\PartySet$ can be designated as $\Sender$. In the rest of the paper
  we will say that ``{\it $P_i$ Acasts $m$}" to mean that $P_i$ acts as $\Sender$ and
  invokes an instance of $\Acast$ with input m, and the parties participate in this instance. Similarly,
   ``{\it $P_j$ receives $m$ from the Acast of $P_i$}" means that $P_j$ outputs $m$ in the corresponding instance of $\Acast$.
\subsubsection{Synchronous Broadcast with Asynchronous Guarantees}
  The {\it third} component for designing $\SBA$ is a broadcast protocol $\BC$, which is secure in a {\it synchronous} network and which also provides 
   liveness, weak validity and
  weak consistency in an {\it asynchronous} network.  Note that the guarantees
   of $\BC$ are {\it different} from that of $\Acast$. The design of $\BC$ is based on the idea from \cite{ACC22b},
   by carefully combining protocols $\Acast$ and $\PiPW$. 
   In the protocol, $\Sender$ first Acasts
   its message. If the network is {\it synchronous}, then at time $3\Delta$, all honest parties should have an output.
   To confirm this, the parties start participating in an instance of $\PiPW$, with whatever output has been obtained from 
   the $\Acast$ instance at time $3\Delta$;
   in case no output is obtained, then the input is $\bot$. 
   Finally, at time $3\Delta + \TimePW$, the parties output
     an $m^{\star}$, if it is the output of the $\Acast$ instance {\it as well as} the output of $\PiPW$, else the output of the parties will be $\bot$. 
     We recall the description of $\BC$ from \cite{ACC22b} and present it in Fig \ref{fig:BC}.
      \paragraph {\bf Protocol $\BC$ as a Network Agnostic Secure Broadcast.}
  Protocol $\BC$ {\it only} guarantees {\it weak validity} and {\it weak consistency}  in an {\it asynchronous} network, since
  only a subset of {\it honest} parties may receive $\Sender$'s message from the Acast of $\Sender$ within time $3\Delta + \TimePW$.
  Note that maintaining the time-out is {\it essential}, as we need {\it liveness} from
   $\BC$ ({\it irrespective} of the network type) when 
   used later in protocol $\SBA$. 
       Looking ahead, we will use $\BC$ in our VSS protocol for broadcasting values.
    The weak validity and consistency may lead to a situation
       where, in an {\it asynchronous} network,
      one subset of {\it honest} parties may output a value different from $\bot$ at the end of the time-out $3\Delta + \TimePW$, while 
      others may output $\bot$. For the security of the VSS protocol, we would require the latter category of parties 
      to {\it eventually} output the common non-$\bot$ value if the parties {\it continue} participating in $\BC$.
      Following \cite{ACC22a,ACC22b}, we 
     make a provision for this in $\BC$. Namely,  
      each
    $P_i$ who outputs $\bot$ at time $3\Delta + \TimePW$ ``switches" its output to $m^{\star}$, if
     $P_i$ {\it eventually} receives $m^{\star}$ from $\Sender$'s Acast. 
    We stress that this switching is {\it only} for the parties who obtained $\bot$
    at time $3\Delta + \TimePW$. 
    To differentiate between the two ways of obtaining output, we use the terms {\it regular-mode}
    and {\it fallback-mode}. Regular-mode is the process of deciding the output at time $3\Delta + \TimePW$, while
    fallback-mode is the process of deciding the output beyond time $3\Delta + \TimePW$.

       \begin{protocolsplitbox}{$\BC(\Sender, m, \Z_s, \Z_a)$}{Synchronous broadcast with asynchronous guarantees.}{fig:BC}
\centerline{\underline{(Regular Mode)}}
   \begin{myitemize}
   \item[--]  Sender $\Sender$ on having the input $m$, Acasts $m$.
   \item[--] {\color{red} At time $3\Delta$}, each $P_i \in \Partyset$ participates in an instance of $\PiPW$, where the input of $P_i$ is $m^{\star}$ if
    $m^{\star} \in \{0, 1\}^{\ell}$ is received from the Acast of 
    $\Sender$, else the input is $\bot$.  
    \item[--] {\bf (Local Computation)}: {\color{red} At time $3\Delta + \TimePW$}, each $P_i \in \Partyset$ does the following.
       \begin{myitemize}
        \item[--] If some $m^{\star} \in \{0, 1 \}^{\ell}$ is received from the Acast of $\Sender$ {\it and} $m^{\star}$ is 
        computed as the output during the instance of $\PiPW$, then output $m^{\star}$. 
         Else output $\bot$.
      \end{myitemize}
\end{myitemize}
\centerline{\underline{{\color{blue}(Fallback Mode)}}}
 \begin{myitemize}
 \item[--] {\color{blue} Each $P_i \in \Partyset$ who outputs $\bot$ at time $3\Delta + \TimePW$,
 changes it to $m^{\star}$, if $m^{\star}$ is received from Acast of $\Sender$.
 }
\end{myitemize}
\end{protocolsplitbox}

 Theorem \ref{thm:BC} follows from \cite{ACC22b} and is proved in 
  Appendix \ref{app:BA}. 
\begin{theorem}
\label{thm:BC}
Protocol
  $\BC$ achieves the following,
   with a communication complexity of $\Order(n^4 \cdot \ell \cdot |\sigma|)$ bits,
    where $\TimeBC =  3\Delta + \TimePW$.   
 \begin{myitemize}
   \item[--] {\it Synchronous} network: 
      {\bf (a) $\Z_s$-Liveness}: At time $\TimeBC$, each honest party has an output. 
    {\bf (b) $\Z_s$-Validity}: If $\Sender$ is {\it honest}, then at time $\TimeBC$, each honest party
    outputs $m$.
      {\bf (c) $\Z_s$-Consistency}: If $\Sender$ is {\it corrupt}, then the output of every honest party is the same at 
      time $\TimeBC$.     
     {\bf (d) $\Z_s$-Fallback Consistency}: If $\Sender$ is {\it corrupt}, and some honest 
     party outputs $m^{\star} \neq \bot$ at time
    $T$ through fallback-mode, then every honest party outputs $m^{\star}$ by 
    time $T + \Delta$.    
\item[--] {\it Asynchronous Network}:
   {\bf (a) $\Z_a$-Liveness}: At time $\TimeBC$, each honest party has an output.
    {\bf (b) $\Z_a$-Weak Validity}: If $\Sender$ is {\it honest}, then at  time $\TimeBC$, 
    each honest party
    outputs $m$ or $\bot$.
    {\bf (c) $\Z_a$-Fallback Validity}: If $\Sender$ is {\it honest}, then each honest party
     with output $\bot$ at time $\TimeBC$, eventually outputs 
    $m$ through fallback-mode.
    {\bf (d) $\Z_a$-Weak Consistency}: If $\Sender$ is {\it corrupt}, then there exists an $m^{\star} \neq \bot$, such that 
   at time $\TimeBC$, 
    each honest party
    outputs
     $m^{\star}$ or $\bot$.
    {\bf (e) $\Z_a$-Fallback Consistency}: If $\Sender$ is {\it corrupt}, and some honest party 
    outputs $m^{\star} \neq \bot$ at  time
    $T \geq \TimeBC$, then 
    each honest party eventually outputs $m^{\star}$.
   \end{myitemize}
\end{theorem}
In the rest of the paper, we use the following terminologies while using $\BC$.
\paragraph {\bf Terminologies for $\BC$:}  
 We say that  {\it $P_i$ broadcasts $m$} to mean that $P_i$ invokes an instance of 
 $\BC$ as  $\Sender$
  with input $m$, and the parties participate in this instance. Similarly, we say that {\it $P_j$ receives $m$ from the broadcast
   of $P_i$ through regular-mode (resp.~fallback-mode)},
   to mean that $P_j$ has the output $m$
  at time $\TimeBC$ (resp.~after time $\TimeBC$)
  during the instance of $\BC$. 
 \subsubsection{$\BC \rightarrow \SBA$}
 Finally, using $\BC$ we instantiate protocol $\SBA$, following the
   blueprint of \cite{ACC22b}.
   In the protocol, every party broadcasts its input bit (for $\SBA$) through an instance of $\BC$.
    At time $\TimeBC$, the parties check if ``sufficiently many"  instances of $\BC$ have produced a binary output
     (which should have happened in a {\it synchronous} network) and if so, they output the ``majority" of those values.
     Otherwise, the network is {\it asynchronous}, in which case the parties output $\bot$. 
     The description of $\SBA$ is recalled from \cite{ACC22b} and presented in Fig \ref{fig:SBA}.
\begin{protocolsplitbox}{$\SBA(\Z_s, \Z_a)$}{Synchronous BA with asynchronous guaranteed liveness and weak validity.
  The above code is executed by every $P_i \in \PartySet$.}{fig:SBA}
\justify
\begin{myitemize}
\item[--] 
 On having input $b_i \in \{0, 1 \}$, broadcast $b_i$.
\item[--] For $j = 1, \ldots, n$, let $b_i^{(j)} \in \{0, 1, \bot \}$ be received from the broadcast of $P_j$ through {\color{red} {\bf regular-mode}}. 
  Include $P_j$ to a set $\R$ if $b_i^{(j)} \neq \bot$.  
    \begin{myitemize}
    \item[--] If 
    $\PartySet \setminus \R \in \Z_s$, then compute the output as follows. 
       \begin{myitemize}
       \item[--] If there exists a subset of parties $\R_i \subseteq \R$, such that $\R \setminus \R_i  \in \Z_s$
       and $b_i^{(j)} = b$ for all the parties $P_j \in \R_i$, then output 
        $b$.\footnote{If there are multiple such $\R_i$, then break the tie using some 
        pre-determined rule.}       
       \item[--] Else output $1$.
       \end{myitemize}
    \item[--] Else output $\bot$.
    \end{myitemize}
\end{myitemize}
\end{protocolsplitbox}

Theorem \ref{thm:SBA} follows from \cite{ACC22b} and is  proved in Appendix \ref{app:BA}.
\begin{theorem}
\label{thm:SBA}
Protocol $\SBA$ achieves the following where $\TimeSBA = \TimeBC$, 
 incurring a communication of 
 $\Order(n^5 \cdot |\sigma|)$ bits.   
 \begin{myitemize}
   \item[--] {\it Synchronous Network}: the protocol is a $\Z_s$-secure SBA protocol where honest parties have an output, different from $\bot$, at time $\TimeSBA$.
   \item[--] {\it Asynchronous Network}: the protocol achieves $\Z_a$-guaranteed liveness and $\Z_a$-weak validity, such that all honest parties have an output at (local) time $\TimeSBA$.
  \end{myitemize}
\end{theorem}
 \subsection{$\ABA$: Asynchronous BA with Synchronous Validity}
 To the best of our knowledge, the {\it only} known {\it unconditionally-secure} ABA protocol is due to \cite{Cho23}, which generalizes the framework of
  randomized ABA \cite{Rab83,Ben83,CR93,FM97} against general adversaries. The protocol of \cite{Cho23} uses a {\it graded agreement} (GA) protocol (also known as the vote protocol),
  along with a secure coin-flip protocol. Unfortunately, both these primitives provide security (in an {\it asynchronous} network) against $\Q^{(3)}$ adversary structures
  and {\it fail} to provide any security guarantees against $\Q^{(2)}$ adversary structures in a {\it synchronous} network. Consequently, the ABA protocol of \cite{Cho23} {\it fails}
  to provide any security guarantees in a {\it synchronous} network against $\Q^{(2)}$ adversary structures.\footnote{Recall that we need {\it validity}, coupled with {\it guaranteed liveness}  from
   $\ABA$, when used in our network agnostic BA protocol $\BA$.} We give a different instantiation of $\ABA$ by generalizing a few ideas used in \cite{BKL19}. Our instantiation of $\ABA$ is based
   on following {\it two} components.
   \paragraph{\bf Component I: Asynchronous Graded Agreement with Synchronous Validity.}
   We assume the existence of a GA protocol $\Vote$, where each party has a binary input. The output for each party is a value from $\{0, 1, \bot \}$, along with a grade from $\{0, 1, 2\}$.
   The protocol achieves the following properties.
   \begin{myitemize}
   \item[--] {\it Asynchronous Network} --- The following properties are achieved, even if the adversary corrupts any subset from $\Z_a$: {\bf (a): $\Z_a$-Liveness:} If all honest parties participate in the protocol, then each honest party eventually obtains an output. 
      {\bf (b) $\Z_a$-Graded Validity:} If every honest party's input is $b$, then all honest parties with an output, output $(b, 2)$.
      {\bf (c) $\Z_a$-Graded Consistency:} If two honest parties output grades $g, g'$, then $|g - g'| \leq 1$ holds; moreover,  
      if two honest parties output $(v, g)$ and $(v', g')$ with $g, g' \geq 1$, then $v = v'$.
   \item[--] {\it Synchronous Network} --- The following properties are achieved, even if the adversary corrupts any subset from $\Z_s$: {\bf (a): $\Z_s$-Liveness:} 
   If all honest parties participate in the protocol with the {\it same} input, then after some fixed time $\TimeVote$,
    all honest parties obtain an output.
   {\bf (b) $\Z_s$-Graded Validity:} If every honest party's input is $b$, then all honest parties with an output, output $(b, 2)$.   
   \end{myitemize}
   Note that we {\it do not} require any form of consistency from $\Vote$ in the {\it synchronous} network.
      We give an instantiation of $\Vote$ with the $\Q^{(2, 1)}(\PartySet, \Z_s, \Z_a)$ condition, by generalizing the {\it threshold} GA protocol of \cite{BKL19} with condition $2t_s + t_a < n$, such that
      $\TimeVote = 4 \cdot \Delta$.
    The protocol of \cite{BKL19} uses digital signatures (hence, is {\it computationally secure}),
     which we replace with pseudo-signatures.
       For the description of $\Vote$ and its properties, see Appendix \ref{app:BA}. 
   \paragraph{\bf Component II: Asynchronous Coin-Flipping with Synchronous Liveness.}
   We assume the existence of a $p$-coin-flipping protocol $\CoinFlip$, where $0 < p < 1$ is a parameter.
    In the protocol, the parties participate with random inputs and the output of each party is a bit satisfying the following properties.   
   \begin{myitemize}
   \item[--] {\it Asynchronous Network}: The following properties are achieved even if the adversary corrupts any subset from $\Z_a$: {\bf (a): $\Z_a$-Almost-Surely Liveness:}
   If all honest parties participate in the protocol, then almost-surely, all honest parties eventually get an output.
   {\bf (b): $(\Z_a, p)$-Commonness:} With probability $p$, the output of {\it all} honest parties is a {\it random} bit $b \in \{0, 1 \}$.
   \item[--] {\it Synchronous Network}: The following property is achieved even if the adversary corrupts any subset from $\Z_s$: {\bf (a): $\Z_s$-Guaranteed Liveness:}
   If all honest parties participate in the protocol, then all honest parties get an output, after some {\it fixed} time $\TimeCoinFlip$.
   \end{myitemize}
  In \cite{Cho23}, the authors presented an instantiation of $\CoinFlip$, which achieves {\it $\Z_a$-Almost-Surely Liveness} as well as {\it $(\Z_a, p)$-Commonness} in an asynchronous network,
  where $p = \frac{1}{n}$, provided $\Z_a$ satisfies the $\Q^{(3)}(\PartySet, \Z_a)$ condition. The protocol incurs an {\it expected} communication of $\Order(\poly(n, |\Z_s|, \log{|\F|}))$ bits.
   Interestingly, the protocol also achieves {\it $\Z_s$-Guaranteed Liveness} in a {\it synchronous} network, {\it irrespective} of $\Z_s$.
  Namely, after time $\TimeCoinFlip =  20 \cdot \Delta$, all honest parties will have an output, 
  with (honest)
  parties communicating $\Order(\poly(n, |\Z_s|, \log{|\F|}))$ bits. 
  \subsubsection{\bf $\Vote + \CoinFlip \rightarrow \ABA$}
  Once we have instantiations of $\Vote$ and $\CoinFlip$, we can easily combine it using the framework of \cite{Rab83,Ben83,CR93,FM97} to get the protocol $\ABA$.
  The protocol consists of several iterations, where in each iteration, the parties run two instances of $\Vote$, along with an instance of $\CoinFlip$. 
  Using the first instance of $\Vote$, the parties check if they all have the same input. Independent of this finding, they then run an instance of $\CoinFlip$. Finally, they again run an instance of
  $\Vote$, with inputs being carefully chosen. Namely, if a party obtained an output with the {\it highest} grade from the {\it first} instance of $\Vote$, then it participates with this input,
  else it participates with the coin-output. Finally, based on the output received from the second instance of $\Vote$, the parties update their input for the next iteration as follows:
  if a bit with a non-zero grade is obtained, then it is set as the updated input, else the updated input is set to the input of the second instance of $\Vote$.
  During each iteration, the parties keep a tab on whether they have received an output bit with the {\it highest} grade from the {\it second} instance of $\Vote$, in which case, they
  indicate it to the others by sending a signed $\ready$ message and the bit. Once ``sufficiently many" parties send the same signed $\ready$ bit, it is taken as the output of the protocol.
  
  The idea here is that if the honest parties start an iteration with the {\it same} input bit $b$, then the output of $\CoinFlip$ is {\it not} considered 
   ({\it irrespective} of the network type) and {\it all} instances of $\Vote$ output $(b, 2)$. Thus, all honest parties will send a signed $\ready$ message for $b$. Consequently, all honest parties will output $b$.  
  This ensures {\it validity}, coupled with {\it guaranteed liveness}, both in {\it synchronous} and {\it asynchronous} networks.
  On the other hand, if the honest parties start an iteration with {\it different} inputs, then with probability at least $p \cdot \frac{1}{2} = \frac{1}{n} \cdot \frac{1}{2}$, all of them will have the 
  {\it same} input for the second instance of $\Vote$. And consequently, all honest parties will have the same input from the next iteration onward and consistency is achieved
   (in the {\it asynchronous} network). 
  The description of $\ABA$ based on $\Vote$ and $\CoinFlip$ is recalled from \cite{Cho23} and presented in Fig \ref{fig:ABA}.
  
   \begin{protocolsplitbox}{$\ABA(\PartySet, \Z_s, \Z_a)$}{The ABA protocol. The above code is executed by every $P_i \in \PartySet$ with input $b_i$}{fig:ABA}
 \begin{myitemize}
  \item {\bf Initialisation}: Set $b = b_i$, $\committed = \false$ and $k=1$. Then do the
    following.
    \begin{mydescription}
    \item Participate in an instance of $\Vote$ protocol with input $b$ and {\color{red} wait for time $\TimeVote$}.            
    \item Once an output $(b, g)$ is received from the instance of $\Vote$, participate in an
      instance of $\CoinFlip$ and {\color{red} wait for time $\TimeCoinFlip$}. Let $\Coin_k$ denote the output received from $\CoinFlip$.
    \item If $g < 2$, then set $b = \Coin_k$.
    \item Participate in an instance of $\Vote$ protocol with input $b$ and {\color{red} wait for time $\TimeVote$}. Let $(b', g')$ be
      the output received. If $g' > 0$, then set $b = b'$.
    \item  If $g' = 2$ and $\committed = \false$, then set $\committed = \true$ and send 
      $\sign{(\ready, b)}_i$ to all the parties.
    \item Set $k = k + 1$ and repeat from $1$.
    \end{mydescription}
  \item {\bf Output Computation}:
    \begin{myitemize}
    \item[--] Upon receiving  $\sign{(\ready, b)}_j$ messages with valid signatures corresponding to every $P_j \in \PartySet \setminus Z$ for some $Z \in \Z_s$, send
    $(b,\cert{b})$ to all the parties and
      output $b$. Here $\cert{b}$ denotes the collection of signed $\sign{(\ready, b)}_j$ messages.
    \item[--] Upon receiving $(b, \cert{b})$ where $\cert{b}$ has valid signatures  corresponding to every $P_j \in \PartySet \setminus Z$ for some $Z \in \Z_s$, 
    send $(b, \cert{b})$ to all the parties and
      output $b$.
    \end{myitemize}
  \end{myitemize}
    \end{protocolsplitbox}
    
    Theorem \ref{thm:ABA} follows from \cite{BKL19} and is proved in Appendix \ref{app:BA}.

\begin{theorem}
\label{thm:ABA}
Protocol $\ABA$ achieves the following where $\TimeABA = \TimeCoinFlip + 2\TimeVote + \Delta$.   
 \begin{myitemize}
   \item[--] {\it Synchronous Network}: If all honest parties have the same input $b \in \{ 0, 1\}$, then all honest parties output $b$,
    at time $\TimeABA$. Moreover, $\Order(\poly(n, |\Z_s|, \log{|\F|}))$ bits
   are communicated by the honest parties.
   \item[--] {\it Asynchronous Network}: the protocol is a $\Z_a$-secure ABA, incurring an expected communication of $\Order(\poly(n, |\Z_s|, \log{|\F|}))$ bits.
  \end{myitemize}
\end{theorem}

%% file: ICSig.tex
\section{Network Agnostic Information Checking Protocol}
 In this section, we present our network agnostic ICP protocol (Fig \ref{fig:ICP}). A detailed overview of the protocol has been already presented in Section \ref{ssec:ICP}.
  The protocol consists of two subprotocols $\Auth$ and $\Reveal$, implementing the authentication and revelation phase respectively, 
  where the parties participate in the revelation phase only upon completing the authentication phase.
  During the authentication phase, $\mathsf{S}$ distributes the authentication and verification information, followed by parties publicly verifying the consistency of distributed information and once the consistency is established, the
  authentication phase is over. During the revelation phase, $\mathsf{I}$ reveals the IC-signature which is verified by $\mathsf{R}$ with respect to the verification information revealed by a ``selected" subset of the verifiers.
 
\begin{protocolsplitbox}{$\ICP(\PartySet, \Z_s, \Z_a, \mathsf{S},\mathsf{I},\mathsf{R})$}{The network-agnostic ICP}{fig:ICP}
   \centerline{\underline{Protocol $\Auth(\PartySet,\Z_s, \Z_a, \mathsf{S},\mathsf{I},\mathsf{R},s)$}: $t \defined \max\{ |Z| :  Z \in \AdvStruct_s 	\}$} 
   \justify \justify
\begin{myitemize}
\item[--] {\bf Distributing Data}: $\mathsf{S}$ executes the following steps.
    \begin{myitemize}
    \item Randomly select $t$-degree {\it signing-polynomial} $F(x)$ and $t$-degree {\it masking-polynomial} $M(x)$,
     where $F(0) = s$. 
     For $i = 1,\ldots,n$, randomly select $\alpha_i \in \F \setminus \{0\}$, and compute
    $v_i = F(\alpha_i)$ and $m_i = M(\alpha_i)$.
    \item Send $(F(x), M(x))$ to $\mathsf{I}$. For $i = 1, \ldots, n$, 
    send $(\alpha_i, v_i, m_i)$ to party $P_i$.    
    \end{myitemize}
\item[--] {\bf Confirming Receipt of Verification Points}: Each party $P_i$ (including $\mathsf{S}, \mathsf{I}$ and $\mathsf{R}$), upon receiving $(\alpha_i, v_i, m_i)$ from $\mathsf{S}$, broadcasts $(\Received,i)$. 
\item[--] {\bf Announcing Set of Supporting Verifiers}: only $\S$ does the following.
	\begin{myitemize}
	\item Initialize the set of {\it supporting verifiers} $\R$ to $\emptyset$, and {\color{red} wait till the local time is $\Delta + \TimeBC$}. 
	Upon receiving $(\Received,i)$ from the broadcast of $P_i$, 
	add $P_i$ to $\R$. Once $\PartySet \setminus \R \in \AdvStructure_s$, broadcast the set $\R$.
	\end{myitemize}
\item[--] {\bf Announcing Masked Polynomial}: only $\INT$ does the following.
	\begin{myitemize}
	\item {\color{red} Wait till the local time is $\Delta + 2\TimeBC$}. Upon receiving $\R$ from the broadcast of $\mathsf{S}$ such that $\PartySet \setminus \R \in \Z_s$, 
	wait till $(\Received,i)$ is received from the broadcast of every 
	$P_i \in \R$. Then randomly pick $d \in \F \setminus \{0\}$ and broadcast $(d,B(x))$, where $B(x) \defined dF(x) + M(x)$.
	\end{myitemize}
\item[--] {\bf Announcing Validity of Masked Polynomial}: only $\S$ does the following. 
	\begin{myitemize}
	\item {\color{red} Wait till the local time is $\Delta + 3\TimeBC$}. Upon receiving $(d,B(x))$ from the broadcast of $\mathsf{I}$, 
	broadcast $\OK$, if $B(x)$ is a $t$-degree polynomial and if
	$dv_j + m_j = B(\alpha_j)$ holds 
	  for every $P_j \in \R$. 
	  \end{myitemize}
\item[--] {\bf Deciding Whether Authentication is Successful}: each $P_i \in \PartySet$ (including $\mathsf{S}, \mathsf{I}$ and $\mathsf{R}$) {\color{red} waits till the local time is $\Delta + 4\TimeBC$}.
 Upon receiving $\R$ and $(d, B(x))$ from the broadcast of $\mathsf{S}$ and $\mathsf{I}$ respectively, where $\PartySet \setminus \R \in \Z_s$,
 it set the variable $\authCompleted_{(\S, \INT, \Receiver)}$ to $1$ if $\OK$ is received from the broadcast of $\mathsf{S}$. 
                Upon setting $\authCompleted_{(\S, \INT, \Receiver)}$ to $1$, $\INT$ sets
        $\ICSig(\mathsf{S}, \mathsf{I}, \mathsf{R}, s) = F(x)$. \\[.1cm] 
\end{myitemize}
  \centerline{\underline{Protocol $\Reveal(\PartySet,\AdvStructure_s, \AdvStructure_a, \mathsf{S},\mathsf{I},\mathsf{R},s)$}}
\begin{myitemize}
\item[--] {\bf Revealing Signing Polynomial and Verification Points}:  Each party $P_i$ (including $\mathsf{S}, \mathsf{I}$ and $\mathsf{R}$) 
 does the following, if 
 $\authCompleted_{(\S, \INT, \Receiver)}$ is set to $1$.
    \begin{myitemize}
    \item If $P_i = \mathsf{I}$ then send 
     $F(x)$ to $\mathsf{R}$, if $\ICSig(\mathsf{S}, \mathsf{I}, \mathsf{R}, s)$ is set to $F(x)$ during $\Auth$.
    \item If $P_i \in \R$, then send $(\alpha_i, v_i, m_i)$ to $\mathsf{R}$. 
    \end{myitemize}
\item[--] {\bf Accepting the IC-Signature}: The following steps are executed only by $\mathsf{R}$, if 
$\authCompleted_{(\S, \INT, \Receiver)}$ is set to $1$  during the protocol $\Auth$.
    \begin{myitemize}
               \item[--] {\color{red} Wait till the local time becomes a multiple of $\Delta$}.
          Upon receiving $F(x)$ from $\mathsf{I}$, where $F(x)$ is a $t$-degree polynomial, proceed as follows.
		\begin{myenumerate}
		\item[1.] If $(\alpha_i, v_i, m_i)$ is received from $P_i \in \R$, then {\it accept} $(\alpha_i, v_i, m_i)$ if either 
		$v_i = F(\alpha_i)$ or $B(\alpha_i) \neq dv_i + m_i$, where $B(x)$ is received from the broadcast of $\mathsf{I}$ during $\Auth$. Otherwise,
		{\it reject} $(\alpha_i, v_i, m_i)$.
    	\item[2.] Wait till a subset of parties $\R' \subseteq \R$ is found, such that $\R \setminus \R' \in \AdvStruct_s$,
	and for every $P_i \in \R'$, the corresponding revealed point $(\alpha_i, v_i, m_i)$ is {\it accepted}. Then, output $s = F(0)$.
		\end{myenumerate}	
	\end{myitemize}
\end{myitemize}
\end{protocolsplitbox}
 
 The proof of Theorem \ref{thm:ICP} is available in Appendix \ref{app:ICP}.
 \begin{theorem}
 \label{thm:ICP}
Protocols $(\Auth, \Reveal)$ satisfy the following properties, except with
  probability at most $\errorAICP \defined \frac{nt}{|\F| - 1}$, where
   $t = \max\{ |Z| :  Z \in \Z_s \}$.
    \begin{myitemize}
       \item[--] If $\mathsf{S}, \mathsf{I}$ and $\mathsf{R}$ are {\it honest}, then the following hold.
       \begin{myitemize}
       		\item[--] {\bf $\AdvStruct_s$-Correctness}: In a synchronous network, each honest party sets $\authCompleted_{(\S, \INT, \Receiver)}$ 
   to $1$ during $\Auth$ at time $\TimeAuth = \Delta + 4\TimeBC$. Moreover $\mathsf{R}$ outputs $s$ during $\Reveal$ which takes $\TimeReveal = \Delta$ time.
   			\item[--] {\bf $\AdvStruct_a$-Correctness}:  In an asynchronous network, each honest 
   party eventually sets $\authCompleted_{(\S, \INT, \Receiver)}$ 
   to $1$ during $\Auth$ and $\mathsf{R}$ eventually outputs $s$ during $\Reveal$.
       \item[--] {\bf Privacy}: The view of $\Adv$ is independent of $s$, irrespective of the network.
    \end{myitemize}
     \item[--] {\bf Unforgeability}: If $\mathsf{S}, \mathsf{R}$ are {\it honest}, $\mathsf{I}$ is corrupt
      and if $\mathsf{R}$ outputs $s' \in \F$ during $\Reveal$, then $s' = s$ holds, irrespective of the network type.      
    \item[--] If $\mathsf{S}$ is {\it corrupt}, $\mathsf{I}, \mathsf{R}$ are {\it honest}
    and if $\mathsf{I}$ sets $\ICSig(\mathsf{S}, \mathsf{I}, \mathsf{R}, s) = F(x)$ during $\Auth$, then the following holds.
    \begin{myitemize}
    	\item[--] {\bf $\AdvStruct_s$-Non-Repudiation}:  In a synchronous network, $\mathsf{R}$ outputs $s = F(0)$ during during $\Reveal$, which takes $\TimeReveal = \Delta$ time.    
    	\item[--] {\bf $\AdvStruct_a$-Non-Repudiation}: In an asynchronous network, $\mathsf{R}$ eventually outputs $s = F(0)$ during during $\Reveal$.
    \end{myitemize}
   \item[--] {\bf Communication Complexity}: Irrespective of the network type, $\Auth$ incurs a communication of $\Order(n^5 \cdot \log{|\F|} \cdot |\sigma|)$ bits, while 
  $\Reveal$ incurs a communication of $\Order(n \cdot \log{|\F|})$ bits.

 \end{myitemize}
 \end{theorem}

Looking ahead, in our VSS protocols, there will be several instances of ICP running, with different parties playing the role of $\S, \INT$ and $\Receiver$. It will be convenient to use the following notations while invoking instances of
 ICP. 
 \begin{notation}[\bf for ICP]
 While using $(\Auth, \Reveal)$, we will say that:
 \begin{myitemize}
 \item[--] ``$P_i$ {\it gives}  $\ICSig(P_i, P_j, P_k, s)$ {\it to} $P_j$" to mean that $P_i$ acts as $\mathsf{S}$ and invokes an instance of
  $\Auth$ with input $s$, where $P_j$ and $P_k$ play the role of $\mathsf{I}$ and $\mathsf{R}$ respectively.
 \item[--]  ``$P_j$ {\it receives} $\ICSig(P_i, P_j, P_k, s)$ {\it from} $P_i$" to mean that
  $P_j$, as $\mathsf{I}$, has set 
   $\authCompleted_{(P_i, P_j, P_k)}$ to $1$
   and $\ICSig(P_i, P_j, P_k, s)$ to some $t$-degree polynomial with $s$ as the constant term
     during the instance of $\Auth$, where $P_i$ and $P_k$
   play the role of $\mathsf{S}$ and $\mathsf{R}$ respectively.   
 \item[--] ``$P_j$ {\it reveals} $\ICSig(P_i, P_j, P_k, s)$ {\it to} $P_k$" to mean $P_j$, as $\mathsf{I}$, invokes an instance
  of $\Reveal$, with $P_i$ and $P_k$ playing the role of $\mathsf{S}$
  and $\mathsf{R}$ respectively.
 \item[--] ``$P_k$ {\it accepts} $\ICSig(P_i, P_j, P_k, s)$" to mean that
  $P_k$, as $\mathsf{R}$, outputs $s$ during the instance of $\Reveal$,
  invoked by $P_j$ as $\mathsf{I}$, with $P_i$ playing the role of $\mathsf{S}$.
\end{myitemize} 
 \end{notation}
\subsection{Linearity of IC Signature}
\label{sec:ICPLinearity}
Our ICP satisfies the {\it linearity} property, provided ``special care" is taken while generating the IC-signatures. In more detail, 
 consider a {\it fixed} $\mathsf{S}, \mathsf{I}$ and $\mathsf{R}$ and let $s_a$ and $s_b$ be two values, such that $\mathsf{I}$ has received $\ICSig(\mathsf{S}, \mathsf{I}, \mathsf{R}, s_a)$ and 
  $\ICSig(\mathsf{S}, \mathsf{I}, \mathsf{R}, s_b)$ from $\mathsf{S}$,
   through instances $\Auth^{(a)}$ and $\Auth^{(b)}$ of $\Auth$ respectively, where {\it all} the following conditions are satisfied.
  \begin{myitemize}
  \item[--] Supporting verifiers $\R_a$ and $\R_b$, during $\Auth^{(a)}$ and $\Auth^{(b)}$, are the {\it same}.
  \item[--] For $i = 1, \ldots, n$, corresponding to the verifier $P_i$, $\mathsf{S}$ has used the {\it same} $\alpha_i \in \F \setminus \{0\}$, to compute the verification points, during 
  $\Auth^{(a)}$ and $\Auth^{(b)}$
  \item[--] $\mathsf{I}$ has used the {\it same} linear combiner $d \in \F \setminus \{ 0\}$ during the instances $\Auth^{(a)}$ and $\Auth^{(b)}$, to compute the linearly-combined masked polynomials.
  \end{myitemize}
  Let $s \defined c_1 \cdot s_a + c_2 \cdot s_b$, where $c_1, c_2$ are {\it publicly known} constants from $\F$. It then follows that if all the above conditions are satisfied, then $\INT$
  can {\it locally} compute $\ICSig(\mathsf{S}, \mathsf{I}, \mathsf{R}, s)$ from $\ICSig(\mathsf{S}, \mathsf{I}, \mathsf{R}, s_a)$ and $\ICSig(\mathsf{S}, \mathsf{I}, \mathsf{R}, s_b)$.
   Namely, $\INT$ can set $\ICSig(\mathsf{S}, \mathsf{I}, \mathsf{R}, s) = c_1 \cdot \ICSig(\mathsf{S}, \mathsf{I}, \mathsf{R}, s_a) + c_2 \cdot \ICSig(\mathsf{S}, \mathsf{I}, \mathsf{R}, s_b)$.
   On the other hand, let the verifier $P_i \in \R_a$ hold the verification points $(\alpha_i, v_{a, i}, m_{a, i})$ and $(\alpha_i, v_{b, i}, m_{b, i})$, corresponding to
   $\ICSig(\mathsf{S}, \mathsf{I}, \mathsf{R}, s_a)$ and $\ICSig(\mathsf{S}, \mathsf{I}, \mathsf{R}, s_b)$ respectively. Then 
   $P_i$ can {\it locally} compute $(\alpha_i, v_i, m_i)$ as its verification point corresponding to $\ICSig(\mathsf{S}, \mathsf{I}, \mathsf{R}, s)$, where
   $v_i = c_1 \cdot v_{a, i} + c_2 \cdot v_{b, i}$ and $m_i = c_1 \cdot m_{a, i} + c_2 \cdot m_{b, i}$. During the protocol $\Reveal$,
   to reveal $\ICSig(\mathsf{S}, \mathsf{I}, \mathsf{R}, s)$, 
   the intermediary $\INT$
   can reveal $\ICSig(\mathsf{S}, \mathsf{I}, \mathsf{R}, s)$ to $\Receiver$, while each verifier $P_i \in \R_a$ can reveal the verification information $(\alpha_i, v_i, m_i)$.
   To accept $(\alpha_i, v_i, m_i)$, the receiver $\Receiver$ either checks for the ``consistency" of $(\alpha_i, v_i)$ with $\ICSig(\mathsf{S}, \mathsf{I}, \mathsf{R}, s)$,
   or the ``inconsistency" of masked polynomial $B(x) \defined c_1 \cdot B_a(x) + c_2 \cdot B_b(x)$ with $(d, v_i, m_i)$; here
   $B_a(x)$ and $B_b(x)$ denote the masked polynomials, made public by $\INT$, during the instances $\Auth^{(a)}$ and $\Auth^{(b)}$ respectively, {\it both} computed with respect to the linear
   combiner $d$.

   Looking ahead, we will require the linearity property from ICP, when used in our VSS protocols, where there will be multiple instances of $\Auth$ running, involving the {\it same} 
   $(\S, \INT, \Receiver)$ triplet.
  To achieve this, we will ensure that in all the $\Auth$ instances invoked during VSS
   involving the {\it same} triplet  $(\S, \INT, \Receiver)$, the signer uses the same non-zero evaluation point $\alpha_{\S, \INT, \Receiver, i}$ for the verifier $P_i$, while distributing
   verification information to $P_i$, as part of the respective $\Auth$ instances. Similarly, $\S$ should find and 
   make public a {\it common} set of supporting verifiers $\R$, on behalf of {\it all} the instances of $\Auth$. 
   And finally, $\INT$ should use the same non-zero random linear combiner $d$, to compute the masked polynomials for all the instances of $\Auth$ and once computed, it should {\it together}
   make public
   $d$ and the masked polynomials for {\it all} the instances of $\Auth$. 
   
   In the rest of the paper, we will use the term ``{\it parties follow linearity principle while generating IC-signatures}", to mean
   that the underlying instances of $\Auth$ are invoked as above.
\subsection{Default IC Signature} 
In our VSS protocols, we will also encounter situations where some {\it publicly known} value $s$ and a triplet $(\S, \INT, \Receiver)$ exist. Then $\INT$ can 
 {\it locally} compute $\ICSig(\S, \INT, \Receiver, s)$ by setting $\ICSig(\S, \INT, \Receiver, s)$ to the {\it constant} polynomial $F(x) = s$. {\it Each} verifier $P_i \in \PartySet$ {\it locally} sets $(\alpha_{\S, \INT, \Receiver, i}, v_i, m_i)$ as its verification information, where $v_i = m_i = s$.
   Moreover, the set of supporting verifiers $\R$ is set as $\PartySet$.
     Notice that the way in which $\ICSig(\S, \INT, \Receiver, s)$ and the verification information is set guarantees that, later, if an {\it honest} $\INT$ reveals
   $\ICSig(\S, \INT, \Receiver, s)$ to an {\it honest} $\Receiver$ during $\Reveal$, then $\Receiver$ {\it always} outputs $s$.
   
   In the rest of the paper, we will use the term ``{\it parties set $\ICSig(\S, \INT, \Receiver, s)$ to the default value}", to mean the above.

%% file: statVSS.tex
\section{Network Agnostic Verifiable Secret Sharing (VSS)}
\label{sec:VSS}
This section presents our network-agnostic VSS protocol, which allows a designated dealer to generate a linear secret-sharing with IC-signatures for its input. We first define the notion of linear 
 secret-sharing with IC-signatures.
 \begin{definition}[\bf Linear Secret Sharing with IC-Signatures]
 \label{def:SS}
 A value $s \in \F$ is said to be linearly secret-shared with IC-signatures, if there exist shares $s_1, \ldots, s_{|\Z_s|} \in \F$ where $s = s_1 + \ldots + s_{|\Z_s|}$. Moreover, for $q = 1, \ldots, |\Z_s|$, 
   there exists some publicly-known {\it core-set} $\W_q \subseteq S_q$, such that {\it all} the following hold.
 \begin{myitemize}
 \item[--] $\Z_s$ satisfies the $\Q^{(1)}(\W_q, \Z_s)$ condition and 
  all (honest) parties in the set $S_q$ have the share $s_q$.
 \item[--] Every honest $P_i \in \W_q$ has the IC-signature 
  $\ICSig(P_j, P_i, P_k, s_q)$ of every $P_j \in \W_q$ for every $P_k \not \in S_q$.
   Moreover, if any corrupt $P_j \in \W_q$ has $\ICSig(P_j, P_i, P_k, s'_q)$ of any honest $P_i \in \W_q$ for any $P_k  \not \in S_q$, then
   $s'_q = s_q$ holds. Furthermore, all the underlying IC-signatures satisfy the linearity property.
 \end{myitemize}
 The vector of information corresponding to a linear secret-sharing with IC-signature of $s$ is denoted by $[s]$, which includes the share $s_1, \ldots, s_{|\Z_s|}$, the core sets $\W_1, \ldots, \W_{|\Z_s|}$
  and IC-signatures $\{ \ICSig(P_j, P_i, P_k, s_q) \}_{P_j, P_i \in \W_q, P_k  \not \in S_q}$. For convenience, we denote the $q^{th}$ share of $s$, corresponding to $S_q$, by
  $[s]_q$.
 
 A vector of values $\vec{S} = (s^{(1)}, \ldots, s^{(L)})$ where $L \geq 1$ is said to be linearly secret-shared with IC-signatures, if each $s^{(\ell)} \in \vec{S}$ is linearly secret-shared with IC-signatures
  and if there exist common core sets $\W_1, \ldots, \W_{|\Z_s|}$, corresponding to the secret-sharings $[s^{(1)}], \ldots, [s^{(\ell)}]$.
 \end{definition}
 
  If $\vec{S} = (s^{(1)}, \ldots, s^{(L)})$ are linearly secret-shared with IC-signatures, then the parties can {\it locally} compute any publicly-known function of these secret-shared values.
    In more detail, let    
   $c_1, \ldots, c_{L} \in \F$ be publicly-known constants and let $s \defined c_1 \cdot s^{(1)} + \ldots + c_{L} \cdot s^{(L)}$.
   Then 
    the following holds:
    \[ c_1 \cdot [s^{(1)}] + \ldots + c_{L} \cdot [s^{(L)}] = [s],    \]
    where the core-sets corresponding to $[s]$ are
    $\W_1, \ldots, \W_{|\Z_s|}$.
         And corresponding to each $S_q \in \ShareSpec_{\Z_s}$, the share $[s]_q$
    for each (honest) party in $S_q$ can be computed locally as $c_1 \cdot [s^{(1)}]_q + \ldots + c_{\ell} \cdot [s^{(\ell)}]_q$. 
     Moreover, every (honest) $P_i \in \W_q$ can compute the IC-signature $\ICSig(P_j, P_i, P_k, [s]_q)$
    of every $P_j \in \W_q$ for every $P_k \not \in S_q$, from $\ICSig(P_j, P_i, P_k, [s^{(1)}]_q), \ldots, \ICSig(P_j, P_i, P_k, [s^{(\ell)}]_q)$. 
    In the rest of the paper, we will say that the ``{\it parties in ${\cal P}$
     locally compute $[c_1 \cdot s^{(1)} + \ldots + c_{\ell} \cdot s^{(\ell)}]$ from $[s^{(1)}], \ldots, [s^{(\ell)}]$}" to mean the above.
\subsection{The VSS Protocol}     
 We present a network-agnostic VSS protocol $\VSS$ (Fig \ref{fig:VSS}). In the protocol, there exists a designated {\it dealer} $\D \in \PartySet$ with input $s$ (the protocol can be easily generalized if
  $\D$ has $L$ inputs). The protocol allows $\D$ to ``verifiably" generate a linear secret-sharing of
  $s$ with IC-signatures. In a {\it synchronous} network, the (honest) parties output 
   $[s]$ after a ``fixed" time, while in an {\it asynchronous} network, they do so {\it eventually}, such that 
  $s$ remains private. The verifiability here guarantees that if $\D$ is {\it corrupt} and some honest party gets an output, then there exists some value, say $s^{\star}$ (which could be different from $s$), 
  such that $s^{\star}$ is linearly secret-shared with IC-signatures. Note that in this case, we {\it cannot} bound the time within which $s^{\star}$ will be secret-shared, since a potentially {\it corrupt} $\D$ may delay sending the
  required messages and the parties will {\it not} be knowing the exact network type. A detailed overview of the protocol has been already presented in Section \ref{ssec:VSS} and so we directly present the protocol.

\begin{protocolsplitbox}{$\VSS(\D, \Z_s, \Z_a, s, \SharingSpec_{\Z_s})$}{The network agnostic VSS protocol}{fig:VSS}
\justify
\begin{myitemize}
\item[--] {\bf Distribution of Shares}: $\D$, 
   on having input $s$, randomly chooses $s_1,\dots,s_{|\Z_s|} \in \F$,
    such that $s = s_1 + \dots + s_{|\Z_s|}$. It then sends $s_q$ to all $P_i \in S_q$, for $q = 1, \ldots, |\Z_s|$.
\item[--] {\bf Exchanging IC-Signed Values}: Each $P_i \in \PartySet$ (including $\D$), {\color{red} waits till the local time becomes $\Delta$}. 
 Then, for each $S_q \in \ShareSpec_{\Z_s}$ such that $P_i \in S_q$, 
  upon receiving $s_{qi}$ from $\D$,  give $\ICSig(P_i, P_j, P_k, s_{qi})$ to every $P_j \in S_q$,  for every $P_k \in \PartySet$, such that {\color{blue} the parties
  follow the linearity principle while generating IC-signatures} (see Section \ref{sec:ICPLinearity}).   
\item[--] {\bf Announcing Results of Pairwise Consistency Tests}:
 Each $P_i \in \PartySet$ (including $\D$) {\color{red} waits till the local time becomes $\Delta + \TimeAuth$}
 and then does the following.
   \begin{myitemize}
     \item[--] Upon receiving $\ICSig(P_j, P_i, P_k, s_{qj})$ from $P_j$ for each $S_q \in \ShareSpec$ such that $P_j, P_i \in S_q$, corresponding to every $P_k \in \PartySet$, 
  broadcast $\OK(i, j)$,  if 
   $s_{qi} = s_{qj}$ holds.
    \item[--] Corresponding to every $P_j \in \PartySet$, participate in any instance of $\BC$ initiated by $P_j$ as a sender, to broadcast any $\OK(P_j, \star)$ message.
   \end{myitemize}
\item[--] {\bf Constructing Consistency Graph}: Each $P_i \in \PartySet$ (including $\D$) {\color{red} waits till the local time becomes $\Delta + \TimeAuth + \TimeBC$}
 and then constructs an undirected consistency graph $G^{(i)}$ with $\PartySet$ as the vertex set, where
  the edge $(P_j, P_k)$ is added to $G^{(i)}$, provided $\OK(j, k)$ and $\OK(k, j)$ is received from the broadcast of $P_j$ and $P_k$ respectively (through any mode).
\item[--] {\bf Identification of Core Sets and Public Announcement by the Dealer}: $\D$ {\color{red} waits till its local time is $\Delta + \TimeAuth + \TimeBC$}, and then
 executes the following steps to compute core sets.
	\begin{myitemize}
	\item[--] Once any $S_p \in \ShareSpec_{\Z_s}$
	 forms a clique in the graph $G^{(\D)}$, then for $q = 1, \ldots, |\Z_s|$, compute core-set $\W_{q}$ 
	 and broadcast-set $\BroadcastSet$ with respect to $S_p$ as follows, followed by broadcasting 
	 $(\CandidateCoreSets, \D, S_p, \{ \W_{q}\}_{q = 1, \ldots, |\Z_s|}, \BroadcastSet, \{ s_q\}_{q \in \BroadcastSet})$.\footnote{If there are multiple $S_p$ from $\ShareSpec_{\Z_s}$ which
	 constitute a clique in $G^{(\D)}$, then consider the one with the smallest index.}
	 \begin{myitemize}
	 \item[--] If $S_q$ constitutes a clique in the graph $G^{(\D)}$, then set $\W_{q} = S_q$.
	 \item[--] Else if $(S_p \cap S_q)$ constitutes a clique in $G^{(\D)}$ and $\Z_s$ satisfies the $\Q^{(1)}(S_p \cap S_q, \Z_s)$ condition, 
	 then set $\W_{q} = (S_p \cap S_q)$.
	 \item[--] Else set $\W_{q} = S_q$ and include $q$ to $\BroadcastSet$. 	 
	 \end{myitemize}
         \end{myitemize}   
\item[--] {\bf Identifying Valid Core Sets}: Each $P_i \in \PartySet$ {\color{red} waits till its local time is $\Delta + \TimeAuth + 2\TimeBC$} and then initializes a set $\C_i = \emptyset$. For
 $p = 1, \ldots, |\Z_s|$, party $P_i$ includes $(\D, S_p)$ to $\C_i$ (initialized to $\emptyset$), provided {\it all} the following hold. 
    \begin{myitemize}
   \item[--]   $(\CandidateCoreSets, \D, S_p, \{ \W_{q}\}_{q = 1, \ldots, |\Z_s|}, \BroadcastSet, \{ s_q\}_{q \in \BroadcastSet})$
      is received from the broadcast of $\D$, such that  for $q = 1, \ldots, |\Z_s|$, the following hold.
      \begin{myitemize}
      \item[--] If $q \in \BroadcastSet$, then the set $\W_{q} = S_q$.
      \item[--] If $(q \not \in \BroadcastSet)$, then $\W_{q}$ is either $S_q$ or  $(S_p \cap S_q)$, such that:
         \begin{myitemize}
         \item[--] If $\W_{q} = S_q$, then $S_q$ constitutes a clique in $G^{(i)}$.
         \item[--] Else if $\W_{q} = (S_p \cap S_q)$, then 
      $(S_p \cap S_q)$ constitutes a clique in $G^{(i)}$ and 
      $\Z_s$ satisfies the $\Q^{(1)}(S_p \cap S_q, \Z_s)$ condition.      
         \end{myitemize}
           \end{myitemize}
   \end{myitemize}
\item[--] {\bf Computing Output}: Each $P_i \in \PartySet$ does the following, once $\C_i \neq \emptyset$.
      \begin{myitemize}
       \item[--] For every $S_q \in \ShareSpec_{\Z_s}$ such that $P_i \in \W_q$, corresponding to every 
 $P_j \in \W_q$, reveal $\ICSig(P_j, P_i, P_k, [s]_q)$ to every $P_k \in \S_q \setminus \W_q$ upon computing $[s]_q$ and $\ICSig(P_j, P_i, P_k, [s]_q)$ as follows.
 		\begin{myitemize}
 		\item[--] If $q \in \BroadcastSet$, then set $[s]_q = s_q$, where $s_q$ is received from the broadcast of $\D$,
	    as part of the message
	    $(\CandidateCoreSets, \D, S_p, \{ \W_{q}\}_{q = 1, \ldots, |\Z_s|}, \BroadcastSet, \{ s_q\}_{q \in \BroadcastSet})$. Moreover, for every $P_j \in S_q$ and every $P_k \in \PartySet$, set
	    $\ICSig(P_j, P_i, \allowbreak P_k, [s]_q)$ to the default value. 
	    \item[--] Else, set $[s]_q$ to $s_{qi}$, where $s_{qi}$ was received from $\D$. Moreover, for every $P_j \in \W_q$ and every $P_k \in \PartySet$, set
	    $\ICSig(P_j, P_i, P_k, [s]_q)$ to  $\ICSig(P_j, P_i, P_k, s_{qj})$, received from $P_j$.
 		\end{myitemize}
       \item[--] For every $S_q \in \ShareSpec_{\Z_s}$ such that $P_i \in S_q \setminus W_q$, compute $[s]_q$ as follows.
        \begin{myitemize}
	    \item[--] Check if there exists any $P_j \in \W_q$ and a value $s_{qj}$, such that $P_i$ has accepted
  $\ICSig(P_k, P_j, P_i, s_{qj})$, corresponding to every $P_k \in \W_q$. Upon finding such a $P_j$,
   set $[s]_q = s_{qj}$. 
	    \end{myitemize}
	  \item[--] {\color{red} Wait till the local time becomes $\Delta + \TimeAuth + 2\TimeBC + \TimeReveal$}. Upon setting $\{[s]_q\}_{P_i \in S_q}$ to some value, output $\W_1, \ldots, \W_{|\Z_s|}$,
	     $\{[s]_q\}_{P_i \in S_q}$ and $\ICSig(P_j, P_i, P_k, [s]_q)_{P_j, P_i \in \W_q, P_k \not \in S_q}$.     
     \end{myitemize}
\end{myitemize}
\end{protocolsplitbox}

The properties of $\VSS$ as stated in Theorem \ref{thm:VSS} are proved in Appendix \ref{app:VSS}.
\begin{theorem}
\label{thm:VSS}
Protocol $\VSS$ achieves the following, except with a probability of $\Order(|\ShareSpec_{\Z_s}| \cdot n^2 \cdot \errorAICP)$, where $\D$ has input $s \in \F$ for $\VSS$ and where 
 $\TimeVSS = \Delta + \TimeAuth + 2\TimeBC + \TimeReveal$. 
   \begin{myitemize}
   \item[--] If $\D$ is honest, then the following hold.
         \begin{myitemize}
         \item[--] {\bf $\Z_s$-correctness}: In a synchronous network, the honest parties output $[s]$ at time $\TimeVSS$. 
         \item[--] {\bf $\Z_a$-correctness}: In an asynchronous network, the honest parties eventually output $[s]$.
         \item[--] {\bf Privacy}: Adversary's view remains independent of $s$ in any network.  
         \end{myitemize}   
   \item[--] If $\D$ is corrupt, then the following hold.
          \begin{myitemize}      
           \item[--]  {\bf $\Z_s$-commitment}: In a synchronous network, either no honest party computes any output or there exists some
             $s^{\star} \in \F$, such that the honest parties output  $[s^{\star}]$.
           Moreover, if any honest party computes its output at time $T$, then all honest parties
		       compute their required output  by time $T + \Delta$.  
	          \item[--]  {\bf $\Z_a$-commitment}: In an asynchronous network, either no honest party computes any output or there exists some
             $s^{\star} \in \F$, such that the honest parties eventually output $[s^{\star}]$.      
          \end{myitemize}
   \item[--] {\bf Communication Complexity}: $\Order(|\Z_s| \cdot n^8 \cdot \log{|\F|} \cdot |\sigma|)$ bits are communicated by the honest parties.
   \end{myitemize}
 \end{theorem}

%% file: PVMT.tex
\section{Network Agnostic Reconstruction Protocols and Secure Verifiable Multicast}
\label{sec:PVMT}
 Let $s$ be a value which is linearly secret-shared with IC-signatures and let $S_q \in \ShareSpec_{\Z_s}$. Moreover, let $\ReceiverSet \subseteq \PartySet$
  be a designated set.  Then protocol 
  $\RecShare([s], S_q, \ReceiverSet)$ allows all the (honest) parties in $\ReceiverSet$ to reconstruct the share $[s]_q$ without disclosing any additional information.
   For this, every $P_i \in \W_q$ reveals 
   $[s]_q$ to {\it all} the parties {\it outside} $\W_q$, who are in $\ReceiverSet$
    (the parties in $\W_q$ who are in $\ReceiverSet$ {\it already} have $[s]_q$). To ensure that $P_i$ {\it does not} cheat, $P_i$ actually reveals the IC-signature of {\it every} party in 
  $\W_q$ on the revealed $[s]_q$. The idea here is that since $\W_q$ has {\it at least} one {\it honest} party (irrespective of the network type), 
  a potentially {\it corrupt} $P_i$ will fail to reveal the signature of an {\it honest} party from $\W_q$ on an {\it incorrect} $[s]_q$.
  On the other hand, an {\it honest} $P_i$ will be able to reveal the signature of {\it all} the parties in $\W_q$ on $[s]_q$.
  
  Based on $\RecShare$, we design another protocol $\Rec([s], \ReceiverSet)$, which allows all the (honest) parties in $\ReceiverSet$ to reconstruct $s$. 
 The idea is to run an instance of $\RecShare$ for {\it every} $S_q \in \ShareSpec_{\Z_s}$.
   Since the protocols are standard, we present them 
  and prove there properties (Lemma \ref{lemma:RecShare} and Lemma \ref{lemma:Rec})
  in Appendix \ref{app:Rec}.
\begin{lemma}
\label{lemma:RecShare}
Let $s$ be a value which is linearly secret-shared with IC signatures, let $S_q \in \ShareSpec_{\Z_s}$ be a designated set and let $\ReceiverSet \subseteq \PartySet$ be a designated
 set of receivers. Then protocol $\RecShare$ achieves the following.
  \begin{myitemize}
  \item[--] {\bf $\Z_s$-correctness}: In a synchronous network, all honest parties in $\ReceiverSet$ output $[s]_q$ after time $\TimeRecShare = \TimeReveal$, 
   except with a probability of $\Order(n^2 \cdot \errorAICP)$. 
   \item[--]  {\bf $\Z_a$-correctness}: In an asynchronous network, all honest parties in $\ReceiverSet$ eventually output $[s]_q$, 
   except with a probability of $\Order(n^2 \cdot \errorAICP)$. 
   \item[--] {\bf Privacy}: If $\ReceiverSet$ consists of only honest parties, then the view of the adversary remains independent of $[s]_q$.
   \item[--] {\bf Communication Complexity}: $\Order(|\ReceiverSet| \cdot n^3 \cdot \log{|\F|})$ bits are communicated.   
  \end{myitemize}
\end{lemma}

\begin{lemma}
\label{lemma:Rec}
Let $s$ be a value which is linearly secret-shared with IC signatures and let $\ReceiverSet \subseteq \PartySet$ be a set of designated receivers.
 Then protocol $\Rec$ achieves the following.
  \begin{myitemize}
  \item[--] {\bf $\Z_s$-correctness}: In a synchronous network, all honest parties in $\ReceiverSet$ output $s$ after time $\TimeRec = \TimeRecShare$, 
   except with probability $\Order(|\ShareSpec_{\Z_s}| \cdot n^2 \cdot \errorAICP)$. 
   \item[--]  {\bf $\Z_a$-correctness}: In an asynchronous network, all honest parties in $\ReceiverSet$ eventually output $s$, 
   except with probability $\Order(|\ShareSpec_{\Z_s}| \cdot n^2 \cdot \errorAICP)$. 
    \item[--] {\bf Privacy}: If $\ReceiverSet$ consists of only honest parties, then the view of the adversary remains independent of $s$.
   \item[--] {\bf Communication Complexity}: $\Order(|\Z_s| \cdot |\ReceiverSet| \cdot n^3 \cdot \log{|\F|})$ bits are communicated.   
  \end{myitemize}
\end{lemma}
\subsection{$\VSS$ and Reconstruction Protocol for Superpolynomial $|\Z_s|$}
\label{sec:SuperPolynomial}
From Theorem \ref{thm:VSS} and Lemma \ref{lemma:Rec}, the error probability of $\VSS$ and $\Rec$ depend linearly on $|\ShareSpec_{\Z_s}|$, which is the same as $|\Z_s|$. 
  This is because there are $\Omega(|\Z_s|)$ instances of $\Auth/\Reveal$ in which the 
 {\it unforgeability/non-repudiation} properties might get violated with probability $\errorAICP$. This might be 
  problematic for a ``large-sized" $\Z_s$. To avoid this, we use the idea of {\it local dispute control}
  used in \cite{HT13,ACC22b}, which ensures that {\it irrespective} of the number of instances of $\Auth/\Reveal$, the {\it overall} error probability is {\it only}
  $\Order(n^3 \cdot \errorAICP)$. This is done by ensuring that the {\it unforgeability/non-repudiation} properties get violated {\it only}
  $\Order(n^3)$ times across all these instances. 
   The idea here is that the parties start {\it locally discarding} corrupt parties the ``moment"
  they are caught cheating during {\it any} instance of $\Auth$ or $\Reveal$. 
   Once a party $P_j$ is locally discarded by any party $P_i$, then $P_i$
    ``behaves" as if $P_j$ has {\it certainly} behaved maliciously in all the ``future” instances of $\Auth$ or $\Reveal$, even
    if this is {\it not} the case. 
    This restricts the number of attempts of cheating for the adversary to a ``fixed" number and consequently, 
    the {\it total} error probability of {\it arbitrary} many instances of $\Auth/\Reveal$ will {\it no} longer depend on $|\Z_s|$.
    To incorporate the above idea, each party $P_i$ now maintains a list of {\it locally discarded} parties $\LD^{(i)}$, which 
    it keeps populating across all the instances of $\Auth$ and $\Reveal$, as soon as $P_i$ identifies any party
     cheating. It will be ensured that an {\it honest} $P_i$ {\it never} includes an {\it honest} $P_j$ to $\LD^{(i)}$. 
    We next discuss the modifications in $\Auth$ and $\Reveal$ and 
     how the parties populate their $\LD$ sets across instances of $\Auth$ and $\Reveal$.    
    
\paragraph{\bf Populating $\LD$ Sets During Instances of $\Auth$.}       
      In any instance of $\Auth$, if $P_i$ is present in the corresponding set of supporting verifiers $\R$ (i.e.~$P_i \in \R$),
      then $P_i$ includes the corresponding {\it signer} $P_j$ of the $\Auth$ instance to $\LD^{(i)}$ if the following condition holds during $\Auth$:
      \[ (P_j \mbox{ broadcasts } \OK) \quad \wedge \quad (B(\alpha_i) \neq dv_i + m_i),      \]
      where $B(x)$ is the masked polynomial broadcasted by the corresponding $\INT$ of the $\Auth$ instance. 
       The idea here is that if $P_i$ is {\it honest} and if the above condition holds, then clearly the signer $P_j$ is {\it corrupt} and is trying to break the {\it non-repudiation}
       property.
       Once $P_j \in \LD^{(i)}$, then in any pair of $(\Auth, \Reveal)$ instances involving $P_j$ as the {\it signer},
      if $P_i$ is present in the corresponding set $\R$, then in the $\Reveal$ instance,
    $P_i$ reveals $\bot$ as its verification information to the corresponding {\it receiver}
     $\Receiver$.\footnote{This serves as an indicator for $\Receiver$ that $P_i$ is in conflict with the signer $P_j$.}
    Upon receiving $\bot$ as the verification information, the strategy for $\Receiver$ is to {\it always} accepts it {\it without} doing any verification,
     {\it irrespective} of the polynomial $F(x)$ 
    revealed as $\ICSig$ by the corresponding $\INT$. 
    
    The above modification ensures that if in any instance of $\Auth$ involving a {\it corrupt signer} $P_j$ and an {\it honest} $\INT$, 
    $P_j$ distributes an inconsistent verification point 
    to an {\it honest verifier} $P_i$ from the corresponding $\R$ set and {\it still} broadcasts an $\OK$ message during $\Auth$, 
    then $P_j$ will locally be discarded by the verifier $P_i$,  except with probability $\errorAICP$ (follows from the {\it non-repudiation} property of ICP).
    From then onwards, in {\it all} the instances of $(\Auth, \Reveal)$, 
    involving $P_j$ as the {\it signer}, 
    if the verifier $P_i$ is added to the $\R$ set, then the ``special" 
    verification information revealed by $P_i$ during $\Reveal$ will {\it always} be considered as accepted, irrespective of what
    verification information
     it actually receives from $P_j$ during $\Auth$.
    Hence, $P_j$ will {\it not} have any chance of cheating the verifier $P_i$ in any $\Auth$ instance. By considering all possibilities
    for a {\it corrupt} $\S$ and an {\it honest} verifier $P_i$, along with an {\it honest} $\INT$, it follows that
    except with probability at most $\Order(n^3 \cdot \errorAICP)$,
     the verification-points of {\it all} honest verifiers from corresponding $\R$,
     will be {\it accepted} by every {\it honest} $\Receiver$,
      during {\it all} the instances of $\Reveal$, in any instance of $\VSS$ or $\Rec$.
       Consequently, except with probability at most $\Order(n^3 \cdot \errorAICP)$, the signatures revealed by all {\it honest} $\INT$ will be {\it always accepted}.

      We stress that the above modification {\it does not} help a {\it corrupt} $\INT$ to break the
      {\it unforgeability} property for an {\it honest} $\S$ and an {\it honest} $\Receiver$, with the help
    of potentially {\it corrupt} verifiers.
\paragraph{\bf Populating $\LD$ Sets During Instances of $\Reveal$.}       
Consider an instance of $\Reveal$ involving $P_i$ as $\Receiver$ and $P_j$ as $\INT$. 
  If $P_i$ finds that $P_j$ has tried to {\it forge} signature on an {\it incorrect} value, then $P_i$
   adds $P_j$ to $\LD^{(i)}$. To achieve this goal, during $\Reveal$, 
    $P_i$ (as $\Receiver$) now {\it additionally} checks if there exists a subset of verifiers $\R'' \subseteq \R$, where $\R \setminus \R'' \in \Z_a$, such that
   the verification-points of {\it all} the parties in $\R''$ are {\it rejected}. If such a subset $\R''$ exists, then clearly $P_j$ (as $\INT$) has cheated and tried to break the 
   {\it unforgeability} property, since
   $\R''$ is bound to contain at least one {\it honest} verifier. If the verification point of an honest verifier is {\it rejected}, then clearly $P_j$ is {\it corrupt}.    
   Once $P_j \in \LD^{(i)}$, from then onwards, 
   in any instance of $\Reveal$ involving $P_j$ as $\INT$ and $P_i$ as $\Receiver$, party
   $P_i$ {\it always rejects} any IC-signature revealed by $P_i$. 
   
    The above modification  ensures that if in any instance of $\Reveal$ involving an {\it honest} signer, a {\it corrupt intermediary} $P_j$ and an {\it honest receiver} $P_i$, $P_j$ tries to reveal an incorrect
    signature during $\Reveal$, 
    then except with probability $\errorAICP$, the intermediary $P_j$ will be locally discarded by the receiver $P_i$ (follows from the {\it unforgeability} property of ICP).
     From then onwards, in {\it all} the instances of $(\Auth, \Reveal)$,
    involving $P_j$ as the intermediary and $P_i$ as the receiver, the signature revealed by $P_j$ during $\Reveal$ will {\it always} be rejected, irrespective of what
    data is actually revealed by $P_j$. 
     Hence, by considering all possibilities for a {\it corrupt} $\INT$, 
   {\it honest} $\S$ and {\it honest} $\Receiver$, 
     it follows that except with probability at most $\Order(n^3 \cdot \errorAICP)$,
      no {\it corrupt} $\INT$ will be able to forge an {\it honest} $\S$'s  signature to any {\it honest} $\Receiver$,
       in any instance of $\Reveal$, during any instance of $\VSS$ or $\Rec$.
\subsection{Network Agnostic Secure Multicast}       
Based on protocols $\VSS$ and $\Rec$, we design a secure verifiable multicast protocol $\SVM$. In the protocol, there exists a designed {\it sender} 
 $\Sender \in \PartySet$ with input $v \in \F$ and a designated set of
 {\it receivers} $\ReceiverSet$. The goal is to let every party in $\ReceiverSet$ receive $v$, without revealing any additional information to the adversary.\footnote{Note that the requirements here are different
 from broadcast since we need the {\it privacy} of $v$ if $\Sender$ is {\it honest} and if $\ReceiverSet$ consists of {\it only} honest parties.} 
  While in a {\it synchronous} network, the (honest) parties in $\ReceiverSet$ get $v$ after a ``fixed" time, in an {\it asynchronous} network, they do so eventually.
  Note that if $\Sender$ is {\it corrupt}, then the parties {\it need not} obtain any output, as $\Sender$ {\it may not} invoke the protocol.
  However, if {\it any} honest party in $\ReceiverSet$ computes an output $v^{\star}$ (which could be {\it different} from $v$), then {\it all} honest parties in $\ReceiverSet$ will also output $v^{\star}$.
  The ``verifiability" here guarantees that in case the honest parties in $\ReceiverSet$ get any output, then {\it all} the (honest) parties in $\PartySet$ will be
  ``aware" of this; namely there will be a Boolean variable $\flag^{(\Sender, \ReceiverSet)}$, which all the honest parties will set to $1$. 
  
  The idea behind $\SVM$ is very simple. The parties participate in an instance of $\VSS$, where $\Sender$ plays the role of the {\it dealer} with input $v$.
  Once any (honest) party computes an output during $\VSS$ (implying that $\Sender$ is committed to some value $v^{\star}$ which is the same as $v$ for an {\it honest} $\Sender$),
   then it turns $\flag^{(\Sender, \ReceiverSet)}$ to $1$.
  Once $\flag^{(\Sender, \ReceiverSet)}$ is turned to $1$, the parties invoke an instance of $\Rec$ to let {\it only} the parties in $\ReceiverSet$ reconstruct the committed value.
  Protocol $\SVM$ and proof of its properties (stated in Lemma \ref{lemma:SVM}), are available in Appendix \ref{app:Rec}.
\begin{lemma}
\label{lemma:SVM}
Protocol $\SVM$ achieves the following, where $\Sender$ participates with input $v$ and where each honest party initializes $\flag^{(\Sender, \ReceiverSet)}$ to $0$.
\begin{myitemize}
\item[--] {\bf Synchronous Network}: If $\Sender$ is honest, then all honest parties set $\flag^{(\Sender, \ReceiverSet)}$ to $1$ at time $\TimeVSS$
 and except with probability $\Order(n^3 \cdot \errorAICP)$, all honest parties in $\ReceiverSet$ output $v$, after time
 $\TimeSVM = \TimeVSS + \TimeRec$. Moreover, if $\ReceiverSet$ consists of only honest parties, then the view of $\Adv$ remains independent of $v$.
  If $\Sender$ is corrupt and some honest party sets $\flag^{(\Sender, \ReceiverSet)}$ to $1$, then there exists some $v^{\star}$ such that, 
  except with probability $\Order(n^3 \cdot \errorAICP)$, all honest parties in $\ReceiverSet$ output $v^{\star}$. Moreover, if any honest party sets 
  $\flag^{(\Sender, \ReceiverSet)}$ to $1$ at time $T$, 
  then all honest parties in $\ReceiverSet$ output $v^{\star}$ by time $T + 2\Delta$.
  \item[--] {\bf Asynchronous Network}: If $\Sender$ is honest, then all honest parties eventually set $\flag^{(\Sender, \ReceiverSet)}$ to $1$
  and  except with probability $\Order(n^3 \cdot \errorAICP)$, all honest parties in $\ReceiverSet$ eventually output $v$.
   Moreover, if $\ReceiverSet$ consists of only honest parties, then the view of the adversary remains independent of $v$.
  If $\Sender$ is corrupt and some honest party sets $\flag^{(\Sender, \ReceiverSet)}$ to $1$, then there exists some $v^{\star}$ such that ,
  except with probability $\Order(n^3 \cdot \errorAICP)$, all honest parties in $\ReceiverSet$ eventually output $v^{\star}$.   
 \item[--] {\bf Communication Complexity}: $\Order(|\Z_s| \cdot n^8 \cdot \log{|\F|} \cdot |\sigma|)$ bits are communicated.
\end{myitemize}
\end{lemma}

%% file: statMVSS.tex
\section{Network Agnostic Protocol for Generating Linearly Secret-Shared Random Values with IC-Signatures}
\label{sec:MVSS}
In this section, we present a network agnostic protocol $\Rand$, which allows the parties to jointly generate 
  linear secret-sharing of random values with IC-signatures.
      To design the protocol $\Rand$, we first design a subprotocol $\MDVSS$.
   \subsection{Network Agnostic VSS for Multiple Dealers}
   Protocol $\MDVSS$ (Fig \ref{fig:MDVSS}) is a multi-dealer VSS.
    In the protocol, each party $P_{\ell} \in \PartySet$ participates as a dealer with some input $s^{(\ell)}$.
   Then, {\it irrespective} of the network type, the protocol outputs a {\it common} subset of dealers $\Core \subseteq \PartySet$, which is {\it guaranteed} to have at least one {\it honest} dealer.
   Moreover, corresponding to every dealer $P_{\ell} \in \Core$, there will be some value, say ${s^{\star}}^{(\ell)}$, which will be the same as $s^{(\ell)}$ for an {\it honest} $P_{\ell}$, 
   such that the values
   $\{{s^{\star}}^{(\ell)} \}_{P_{\ell} \in \Core}$ are linearly secret-shared with IC-signatures. While in a {\it synchronous} network,  $\{[{s^{\star}}^{(\ell)}] \}_{P_{\ell} \in \Core}$
   is generated after a ``fixed" time, in an {\it asynchronous} network, $\{[{s^{\star}}^{(\ell)}] \}_{P_{\ell} \in \Core}$ is generated eventually.
   
   The high level overview of $\MDVSS$ has been already discussed in detail in 
    Section \ref{ssec:Rand}.\footnote{Actually, the overview was for the protocol $\Rand$, but the same idea is also used in the protocol $\MDVSS$.}
     The idea is to let every dealer $P_{\ell}$ to invoke an instance
   of $\VSS$ to secret-share its input. However, we need to take special care to ensure that the inputs of all the dealers in $\Core$ are secret-shared with {\it common}
   core-sets. For this, each individual dealer in its instance of $\VSS$ computes and publishes as many ``legitimate" core-sets as possible and the parties run instances of 
   {\it agreement on common subset}
   (ACS) to identify whether ``sufficiently many" dealers have published the same legitimate core-sets in their respective instances of $\VSS$.
   Moreover, to ensure that all the underlying IC-signatures satisfy the linearity property, we first need to identify the dealers who distribute shares as part of their respective 
   $\VSS$ instances. For this, we let each dealer distribute shares in its instance of $\VSS$ through instances of $\SVM$. This enables the parties to identify
   a set of {\it committed dealers} $\CD$ who have indeed distributed shares as part of their $\VSS$ instances through instances of $\SVM$.

\begin{protocolsplitbox}{$\MDVSS(\PartySet, \Z_s, \Z_a, (s^{(1)}, \ldots, s^{(n)}), \SharingSpec_{\Z_s})$}{The statistically-secure VSS protocol for multiple dealers to generate linearly secret-shared values with
 IC-signatures}{fig:MDVSS}
\justify
\begin{myitemize}
\item[--] {\bf Committing Shares}: Each $P_{i} \in \PartySet$ executes the following steps.
  \begin{myitemize}
   \item[--] On having input $s^{(i)}$, randomly choose $s^{(i)}_1,\dots,s^{(i)}_{|\Z_s|}$,
    such that $s^{(i)} = s^{(i)}_1 + \dots + s^{(i)}_{|\Z_s|}$.
    Act as $\Sender$ and invoke instances $\SVM(P_{i}, s^{(i)}_1, S_1), \ldots, \SVM(P_{i}, s^{(i)}_{|\Z_s|}, S_{|\Z_s|})$ of $\SVM$.
    \item[--] Corresponding to every dealer $P_{\ell} \in \PartySet$, participate in the instances 
     of $\SVM$, invoked by $P_{\ell}$ as a $\Sender$
    and {\color{red} wait till the local time becomes $\TimeSVM$}. For $q = 1, \ldots, |\Z_s|$,  
    let $\flag^{(P_{\ell}, S_q)}$ be the Boolean flag, corresponding to the instance $\SVM(P_{\ell}, s^{(\ell)}_q, S_q)$, invoked by $P_{\ell}$.
    \end{myitemize}
   \item[--] {\bf Identifying the Set of Committed Dealers Through ACS}: Each $P_i \in \PartySet$ does the following.
      \begin{myitemize}
       \item[--] For $\ell = 1, \ldots, n$, 
       participate in an instance $\BA^{(\ell)}$ of $\BA$ with input $1$, provided $P_i$ has set
        $\flag^{(P_{\ell}, S_q)} = 1$, for $q = 1, \ldots, |\Z_s|$.
      \item[--] Once there exists a subset of dealers $\CD_i$ where $\PartySet \setminus \CD_i \in \Z_s$, such that 
      corresponding to every dealer $P_{\ell} \in \CD_i$, the instance 
       $\BA^{(\ell)}$ has produced output $1$, then participate with input $0$ in all the BA instances $\BA^{(\star)}$, for which no input is provided yet.
      \item[--] Once all the $n$ instances of $\BA^{(\star)}$ have produced a binary output, set $\CD$ to be the set of dealers $P_{\ell}$, 
       such that
      $\BA^{(\ell)}$ has produced output $1$.
      \end{myitemize}
   \item[--] {\bf Exchanging IC-Signed Values}: Each $P_i \in \PartySet$ {\color{red} waits till the local time becomes $\TimeSVM + 2\TimeBA$}. 
   Then corresponding to each dealer $P_{\ell} \in \CD$, does the following.
      \begin{myitemize}
        \item[--] For each $S_q \in \ShareSpec_{\Z_s}$ such that $P_i \in S_q$, 
          upon computing an output $s^{(\ell)}_{qi}$ during $\SVM(P_{\ell}, s^{(\ell)}_q, S_q)$, 
           give $\ICSig(P_i, P_j, P_k, s^{(\ell)}_{qi})$ to every $P_j \in S_q$,  for every $P_k \in \PartySet$, 
           {\color{blue} where the parties follow the linearity principle while generating IC-signatures}.         
      \end{myitemize}
 \item[--] {\bf Announcing Results of Pairwise Consistency Tests}:
 Each $P_i \in \PartySet$ {\color{red} waits till the local time becomes $\TimeSVM + 2\TimeBA + \TimeAuth$}
 and then does the following, corresponding to each dealer $P_{\ell} \in \CD$.
   \begin{myitemize}
     \item[--] Upon receiving $\ICSig(P_j, P_i, P_k, s^{(\ell)}_{qj})$ from $P_j$ for each $S_q \in \ShareSpec$ such that $P_j, P_i \in S_q$, corresponding to every $P_k \in \PartySet$, 
     broadcast $\OK^{(\ell)}(i, j)$,  if 
   $s^{(\ell)}_{qi} = s^{(\ell)}_{qj}$ holds.
    \item[--] Corresponding to every $P_j \in \PartySet$, participate in any instance of $\BC$ initiated by $P_j$ as a sender, to broadcast any $\OK^{(\ell)}(P_j, \star)$ message.
   \end{myitemize}
 \item[--] {\bf Constructing Consistency Graphs}: Each $P_i \in \PartySet$ {\color{red} waits till the local time becomes $\TimeSVM + 2\TimeBA + \TimeAuth + \TimeBC$}
   and then does the following, corresponding to each dealer $P_{\ell} \in \CD$.
      \begin{myitemize}
        \item[--] Construct an undirected consistency graph $G^{(\ell, i)}$ with $\PartySet$ as the vertex set, where
	  the edge $(P_j, P_k)$ is added to $G^{(\ell, i)}$, provided $\OK^{(\ell)}(j, k)$ and $\OK^{(\ell)}(k, j)$ is received from the broadcast of $P_j$ and $P_k$ respectively (through any mode).
        \end{myitemize}	 
\item[--] {\bf Public Announcement of Core Sets by the Committed Dealers}: Each dealer $P_{\ell} \in \CD$
  {\color{red} waits till its local time is $\TimeSVM + 2\TimeBA + \TimeAuth + \TimeBC$}, and then
 executes the following steps to compute core sets.
	\begin{myitemize}
	\item[--] $\forall S_p \in \ShareSpec_{\Z_s}$, once
	 $S_p$ forms a clique in $G^{(\ell, \ell)}$, then for $q = 1, \ldots, |\Z_s|$, compute core-set $\W^{(\ell)}_{p, q}$ 
	 and broadcast-set $\BroadcastSet^{(\ell)}_p$ with respect to $S_p$ as follows, followed by broadcasting 
	 $(\CandidateCoreSets, P_{\ell}, S_p, \{ \W^{(\ell)}_{p, q}\}_{q = 1, \ldots, |\Z_s|}, \BroadcastSet^{(\ell)}_p, \{ s^{(\ell)}_q\}_{q \in \BroadcastSet^{(\ell)}_p})$.
	 \begin{myitemize}
	 \item[--] If $S_q$ constitutes a clique in the graph $G^{(\ell, \ell)}$, then set $\W^{(\ell)}_{p, q} = S_q$.
	 \item[--] Else if $(S_p \cap S_q)$ constitutes a clique in $G^{(\ell, \ell)}$ and
	 $\Z_s$ satisfies the $\Q^{(1)}(S_p \cap S_q, \Z_s)$ condition, 
	 then set $\W^{(\ell)}_{p, q} = (S_p \cap S_q)$.
	 \item[--] Else set $\W^{(\ell)}_{p, q} = S_q$ and include $q$ to $\BroadcastSet^{(\ell)}_p$. 	 
	 \end{myitemize}
         \end{myitemize}   
  \item[--] {\bf Identifying Valid Core Sets}: Each $P_i \in \PartySet$ {\color{red} waits for time $\TimeSVM + 2\TimeBA + \TimeAuth + 2\TimeBC$} and then initializes a set $\C_i = \emptyset$.
   Corresponding to $P_{\ell} \in \CD$ 
    and $p = 1, \ldots, |\Z_s|$, party $P_i$ includes $(P_{\ell}, S_p)$ to $\C_i$, provided {\it all} the following hold. 
    \begin{myitemize}
   \item[--]  $(\CandidateCoreSets, P_{\ell}, S_p, \{ \W^{(\ell)}_{p, q}\}_{q = 1, \ldots, |\Z_s|}, \BroadcastSet^{(\ell)}_p, \{ s^{(\ell)}_q\}_{q \in \BroadcastSet^{(\ell)}_p})$
   is received from the broadcast of $P_{\ell}$, such that  for $q = 1, \ldots, |\Z_s|$, the following hold.
      \begin{myitemize}
      \item[--] If $q \in \BroadcastSet^{(\ell)}_p$, then the set $\W^{(\ell)}_{p, q} = S_q$.
      \item[--] If $(q \not \in \BroadcastSet^{(\ell)}_p)$, then $\W^{(\ell)}_{p, q}$ is either $S_q$ or  $(S_p \cap S_q)$, such that:
         \begin{myitemize}
         \item[--] If $\W^{(\ell)}_{p, q} = S_q$, then $S_q$ constitutes a clique in $G^{(\ell, i)}$.
         \item[--] Else if $\W^{(\ell)}_{p, q} = (S_p \cap S_q)$, then 
      $(S_p \cap S_q)$ constitutes a clique in $G^{(\ell, i)}$ and 
      $\Z_s$ satisfies the $\Q^{(1)}(S_p \cap S_q, \Z_s)$ condition.      
         \end{myitemize}
           \end{myitemize}
   \end{myitemize}
  \item[--] {\bf Selecting the Common Committed Dealers and Core Sets through ACS}:  Each party $P_i \in \PartySet$ does the following.
      \begin{myitemize}
      \item[--] For $p = 1, \ldots, |\Z_s|$, participate in an instance $\BA^{(1, p)}$ of $\BA$ with input $1$, provided there exists a set of dealers
       ${\cal A}_{p, i} \subseteq \CD$ 
      where $\CD \setminus {\cal A}_{p, i} \in \Z_s$
      and where $(P_{\ell}, S_p) \in \C_i$ for every $P_{\ell} \in {\cal A}_{p, i}$.
      \item[--] Once any instance of $\BA^{(1, \star)}$ has produced an output $1$, participate with input $0$ in all the BA instances $\BA^{(1, \star)}$, for which no input is provided yet.
      \item[--] Once all the $|\Z_s|$ instances of $\BA^{(1, \star)}$ have produced a binary output, set $\qcore$ to be the least index among $\{1, \ldots, |\Z_s| \}$, such that
      $\BA^{(1, \qcore)}$ has produced output $1$.
      \item[--] Once $\qcore$ is computed, then corresponding to each $P_j \in \CD$,
       participate in an instance $\BA^{(2, j)}$ of $\BA$ with input $1$, provided $(P_j, S_{\qcore}) \in \C_i$. 
      \item[--] Once there exists a set of parties ${\cal B}_i \subseteq \CD$,
       such that $\CD \setminus {\cal B}_i \in \Z_s$ and $\BA^{(2, j)}$ has produced output $1$, corresponding to each $P_j \in {\cal B}_i$,
      participate with input $0$ in all the instances of $\BA^{(2, \star)}$, for which no input is provided yet.
      \item[--] Once all the $|\CD|$
       instances of $\BA^{(2, \star)}$ have produced a binary output, include all the parties $P_j$ from $\CD$ in $\Core$ (initialized to $\emptyset$), such that $\BA^{(2, j)}$ has produced output $1$.
      \end{myitemize}
 \item[--] {\bf Computing Output}: Each $P_i \in \PartySet$ does the following, after computing $\Core$ and $\qcore$.
     \begin{myitemize}
     \item[--] If $(\CandidateCoreSets, P_{\ell}, S_{\qcore}, \{ \W^{(\ell)}_{\qcore, q}\}_{q = 1, \ldots, |\Z_s|}, \BroadcastSet^{(\ell)}_{\qcore}, \{ s^{(\ell)}_{q}\}_{q \in \BroadcastSet^{(j)}_{\qcore}})$
   is not yet received from the broadcast of $P_{\ell}$ for for any $P_{\ell} \in \Core$, then wait to receive it
    from the broadcast of $P_{\ell}$ through fallback-mode.   
     \item[--] Once $(\CandidateCoreSets, P_{\ell}, S_{\qcore}, \{ \W^{(\ell)}_{\qcore, q}\}_{q = 1, \ldots, |\Z_s|}, \BroadcastSet^{(\ell)}_{\qcore}, \{ s^{(\ell)}_{q}\}_{q \in \BroadcastSet^{(j)}_{\qcore}})$
      is available for every
     $P_{\ell} \in \Core$, compute $\W_q$ for $q = 1, \ldots, |\Z_s|$ as follows.
        \begin{myitemize}
        \item[--] If $\W^{(\ell)}_{\qcore, q} = S_q$ for every $P_{\ell} \in \Core$, then set $\W_q = S_q$.
        \item[--] Else set $\W_q = (S_{\qcore} \cap S_q)$.
        \end{myitemize}
     \item[--] Corresponding to every $P_{\ell} \in \Core$ and every $S_q \in \ShareSpec_{\Z_s}$ such that $P_i \in S_q$, compute the output as follows.
           \begin{myitemize}
	    \item[--] If $q \in \BroadcastSet^{(\ell)}_{\qcore}$, then set $[s^{(\ell)}]_q = s^{(\ell)}_q$, where $s^{(\ell)}_q$ was received from the broadcast of $P_{\ell}$,
	    as part of $(\CandidateCoreSets, P_{\ell}, S_{\qcore}, \allowbreak \{ \W^{(\ell)}_{\qcore, q}\}_{q = 1, \ldots, |\Z_s|}, \BroadcastSet^{(\ell)}_{\qcore}, \{ s^{(\ell)}_{q}\}_{q \in \BroadcastSet^{(j)}_{\qcore}})$.
	    Moreover, for every $P_j \in \W_q$ and every $P_k \in \PartySet$, set
	    $\ICSig(P_j, P_i, P_k, [s^{(\ell)}]_q)$ to the default value. 
	      \item[--] Else, set $[s^{(\ell)}]_q$ to $s^{(\ell)}_{qi}$, where $s^{(\ell)}_{qi}$ was computed as output during $\SVM(P_{\ell}, s^{(\ell)}_q, S_q)$. Moreover, 
	      if $P_i \in \W_q$, then
	      for every $P_j \in \W_q$ and every $P_k \in \PartySet$, set
	    $\ICSig(P_j, P_i, P_k, [s^{(\ell)}]_q)$ to  $\ICSig(P_j, P_i, P_k, s^{(\ell)}_{qj})$, received from $P_j$.
	    \end{myitemize}
	  Output $\Core$, the core sets $\W_1, \ldots, \W_{|\Z_s|}$, shares
	     $\{[s^{(\ell)}]_q\}_{P_{\ell} \in \Core \; \wedge \; P_i \in S_q}$ and 
	     the IC-signatures $\ICSig(P_j, P_i, P_k, [s^{(\ell)}]_q)_{P_{\ell} \in \Core \; \wedge \; P_j, P_i \in \W_q, P_k \in \PartySet}$.     
     \end{myitemize}
  
\end{myitemize}
\end{protocolsplitbox}

The properties of $\MDVSS$, stated in Theorem \ref{thm:MDVSS} are proved in Appendix \ref{app:MDVSS}.
\begin{theorem}
\label{thm:MDVSS}
Protocol $\MDVSS$ achieves the following where each $P_{\ell}$ participates with input $s^{(\ell)}$ and where  $\TimeMDVSS = \TimeSVM  + \TimeAuth + 2\TimeBC + 6\TimeBA$.
   \begin{myitemize}
   \item[--] {\bf $\Z_s$-Correctness\&Commitment}:   If the network is synchronous, then except with probability $\Order(n^3 \cdot \errorAICP)$,
   at time $\TimeMDVSS$, all honest parties output a common set $\Core \subseteq \PartySet$ 
   such that at least one honest party will be present in $\Core$.
   Moreover, corresponding to every $P_{\ell} \in \Core$, there exists some ${s^{\star}}^{(\ell)}$, where ${s^{\star}}^{(\ell)} = s^{(\ell)}$ for an honest $P_{\ell}$,   
       such that the values $\{  {s^{\star}}^{(\ell)}  \}_{P_{\ell} \in \Core}$ are linearly secret-shared with IC-signatures.
    \item[--] {\bf $\Z_a$-Correctness\&Commitment}:   If the network is asynchronous, then except with probability $\Order(n^3 \cdot \errorAICP)$, almost-surely
   all honest parties output a common set $\Core \subseteq \PartySet$ eventually 
   such that at least one honest party will be present in $\Core$.
   Moreover, corresponding to every $P_{\ell} \in \Core$, there exists some ${s^{\star}}^{(\ell)}$, where ${s^{\star}}^{(\ell)} = s^{(\ell)}$ for an honest $P_{\ell}$,   
       such that the values $\{  {s^{\star}}^{(\ell)}  \}_{P_{\ell} \in \Core}$ are eventually linearly secret-shared with IC-signatures. 
    \item[--] {\bf Privacy}: Irrespective of the network type, the view of the adversary remains independent of $s^{(\ell)}$, corresponding to every honest $P_{\ell} \in \Core$.
   \item[--] {\bf Communication Complexity}: $\Order(|\Z_s|^2 \cdot n^9 \cdot \log{|\F|} \cdot |\sigma|)$ bits are communicated by
   the honest parties. In addition, 
 $\Order(|\Z_s| + n)$ instances of $\BA$ are invoked.
   \end{myitemize}
\end{theorem}

\paragraph{\bf Protocol $\MDVSS$ with $L$ Values for Each Dealer.} In protocol $\MDVSS$, each dealer $P_{\ell} \in \PartySet$ participates with a {\it single} input. 
 Consider a scenario where {\it each} $P_{\ell}$ participates with $L$ inputs $\overrightarrow{S^{(\ell)}} = (s^{(\ell, 1)}, \ldots, s^{(\ell, L)})$, where $L \geq 1$. 
 The goal is to identify a {\it common} subset of dealers $\Core \subseteq \PartySet$ which is {\it guaranteed} to have at least one {\it honest} dealer, irrespective of the network type. Corresponding to every dealer $P_{\ell} \in \Core$, there exist $L$ values, say $\overrightarrow{{S^{\star}}^{(\ell)}} = ({s^{\star}}^{(\ell, 1)}, \ldots, {s^{\star}}^{(\ell, L)})$, 
 which will be the same as $\overrightarrow{S^{(\ell)}}$ for an {\it honest} $P_{\ell}$, where all the values in 
   $\{\overrightarrow{{S^{\star}}^{(\ell)}} \}_{P_{\ell} \in \Core}$ are linearly secret-shared with IC-signatures. To achieve this, we run the protocol $\MDVSS$ with the following modifications,
   so that the number of instances of $\BA$ in the protocol {\it still} remains to be $\Order(|\Z_s| + n)$, which is {\it independent} of $L$.
   
   Corresponding to each $s^{(\ell)} \in \overrightarrow{S^{(L)}}$, the dealer $P_{\ell}$ will pick $|\Z_s|$ random shares (which sum up to $s^{(\ell)}$) and the shares corresponding to the group
   $S_q \in \ShareSpec_{|\Z_s|}$ are communicated through an instance of $\SVM$; hence $|\Z_s| \cdot L$ instances are invoked by $P_{\ell}$ as a $\Sender$.
   Then, while identifying the set of committed dealers $\CD$, parties vote $1$ for $P_{\ell}$  in the instance $\BA^{(\ell)}$ provided the underlying $\flag$ variable is set to $1$ in {\it all}
   the $|Z_s| \cdot L$ instances of $\SVM$ invoked by $P_{\ell}$. The rest of the steps for identifying $\CD$ remains the same. This way, by executing {\it only} $\Order(n)$ instances of $\BA$, we identify
   the set $\CD$.
   
   Next, the parties exchange IC-signatures on their supposedly common shares for each group, corresponding to {\it all} the $L$ values shared by each dealer from $\CD$. 
   However, each $P_i$ now broadcasts a {\it single} $\OK^{(\ell)}(i, j)$ message, corresponding to each $P_j \in S_q$, provided $P_i$ receives IC-signed common share from $P_j$
   on the behalf of {\it all} the $L$ values, shared by $P_{\ell}$. This ensures that, for each $P_{\ell}$, every $P_i$ constructs a {\it single} consistency graph. Next, each dealer $P_{\ell}$
   computes and broadcasts the candidate core-sets and broadcast-sets, as and when they are ready. The parties identify the $\Core$ set by running $\Order(|\Z_s| + n)$  instances of $\BA$. 
   To avoid repetition, we do not present the formal steps of the modified $\MDVSS$ protocol.  
   The protocol incurs a communication of $\Order(|\Z_s|^2 \cdot L \cdot n^9 \cdot \log{|\F|} \cdot |\sigma|)$ bits,
   apart from $\Order(|\Z_s| + n)$ instances of $\BA$.
   \subsection{Protocol for Generating Secret-Shared Random Values}
   Protocol $\Rand$ (Fig \ref{fig:Rand}) allows the parties to jointly generate linear secret-sharings of $L$ values with IC-signatures, where $L \geq 1$,
   which are random for the adversary. For this, the parties invoke an instance of the (modified) $\MDVSS$ where each dealer $P_{\ell} \in \PartySet$ participates
    with a random vector of $L$ values. 
    Let $\Core$ be the set of common dealers identified during the instance of $\MDVSS$. Then for $\mathfrak{l} = 1, \ldots, L$, the parties output the sum of ${\mathfrak{l}}^{th}$
     value shared by {\it all} the dealers in $\Core$. Since there will be at least one {\it honest} dealer in $\Core$ whose shared values will be random for the adversary, 
     it follows that the resultant values also remain random for the adversary. 
     
     In the rest of the paper, we will refer to the core sets $\W_1, \ldots, \W_{|\Z_s|}$ obtained during $\Rand$ as {\it global core-sets} and denote them by $\GW_1, \ldots, \GW_{|\Z_s|}$. 
     From now onwards, all the secret-shared values will be generated with respect to these global core-sets.
     
     \begin{protocolsplitbox}{$\Rand(\PartySet, \Z_s, \Z_a, \SharingSpec_{\Z_s}, L)$}{Protocol for generating linearly secret-shared random values with
 IC-signatures}{fig:Rand}
\justify
 \begin{myitemize}
\item[--] {\bf Secret-Sharing Random Values}:
  Each $P_{\ell} \in \PartySet$ picks $L$ random values $\overrightarrow{R^{(\ell)}} = (r^{(\ell, 1)}, \ldots, r^{(\ell, L)})$ and participates in an instance of $\MDVSS$ with input
 $\overrightarrow{R^{(\ell)}}$ and {\color{red} waits for time $\TimeMDVSS$}.
 \item[--] {\bf Computing Output}: Let $(\Core, \W_1, \ldots, \W_{|\Z_s|}, \{ ([{r^{\star}}^{(\ell, 1)}], \ldots, [{r^{\star}}^{(\ell, L)}])    \}_{P_{\ell \in \Core}} )$
  be the output from the instance of $\MDVSS$. For $\mathfrak{l} = 1, \ldots, L$, the parties locally compute 
  $[r^{(\mathfrak{l})}] = \displaystyle \sum_{P_{\ell} \in \Core} [{r^{\star}}^{(\ell, \mathfrak{l})}]$ from $\{[{r^{\star}}^{(\ell, \mathfrak{l})}] \}_{P_{\ell} \in \Core}$. 
   The parties then output $({\GW}_1, \ldots, {\GW}_{|\Z_s|}, \{[r^{(\mathfrak{l})}] \}_{\mathfrak{l} = 1, \ldots, L})$, where ${\GW}_q = \W_q$ for $q = 1, \ldots, |\Z_s|$.
\end{myitemize}
\end{protocolsplitbox}

Theorem \ref{thm:Rand} follows easily from the above discussion.
\begin{theorem}
\label{thm:Rand}
Protocol $\Rand$ achieves the following where  $\TimeRand = \TimeMDVSS = \TimeSVM  + \TimeAuth + 2\TimeBC + 6\TimeBA$ and $L \geq 1$.
\begin{myitemize}
\item[--] {\bf $\Z_s$-correctness}:  If the network is synchronous, then except with probability $\Order(n^3 \cdot \errorAICP)$,
   at the time $\TimeRand$, there exist values $r^{(1)}, \ldots, r^{(L)}$, which are linearly secret-shared with IC-signatures, where the core-sets are 
   $\GW_1, \ldots, \GW_{|\Z_s|}$.
\item[--] {\bf $\Z_a$-correctness}:  If the network is asynchronous, then except with probability $\Order(n^3 \cdot \errorAICP)$,
  there exist values $r^{(1)}, \ldots, r^{(L)}$, which are almost-surely linearly secret-shared with IC-signatures, where the core-sets are 
   $\GW_1, \ldots, \GW_{|\Z_s|}$.
\item[--] {\bf Privacy}: Irrespective of the network type, the view of the adversary remains independent of $r^{(1)}, \ldots, r^{(L)}$.
\item[--] {\bf Communication Complexity}: The protocol incurs a communication of 
 $\Order(|\Z_s|^2 \cdot L \cdot n^9 \cdot \log{|\F|} \cdot |\sigma|)$ bits,
   apart from $\Order(|\Z_s| + n)$ instances of $\BA$.
\end{myitemize}
\end{theorem}

%% file: Triples.tex
\section{Network Agnostic Protocol for Generating Random Multiplication Triples}
\label{sec:triples}
In this section, we present our network-agnostic triple-generation protocol, which generates random and private multiplication-triples which are linearly secret-shared with IC-signatures.  
  The protocol is based on several sub-protocols which we present next.
  Throughout this section, we will assume the existence of {\it global} core-sets ${\GW}_1, \ldots, {\GW}_{|\Z_s|}$, where $\Z_s$ satisfies the $\Q^{(1)}({\GW}_q, \Z_s)$ condition
  for $q = 1, \ldots, |\Z_s|$.
  Looking ahead, these core-sets will be generated by first running the protocol $\Rand$, using an appropriate value of $L$, 
   which will be determined across all the sub-protocols which we will be discussing next. All the secret-shared values in the various sub-protocols in the sequel will have ${\GW}_1, \ldots, {\GW}_{|\Z_s|}$
   as underlying core-sets.
 \subsection{Verifiably Generating Linear Secret Sharing of a Value with IC-signatures}
 \label{sec:LSh}
 In protocol $\LSh$ (Fig \ref{fig:LSh}), there exists a designated dealer $\D \in \PartySet$ with private input $s$. In addition, there is 
  a random value $r \in \F$, which is linearly secret-shared with IC-signatures,
  such that the underlying core-sets are $\GW_1, \ldots, \GW_{|\Z_s|}$ (the value $r$ will {\it not} be known to $\D$ at the beginning of the protocol).
   The protocol allows the parties to let $\D$ verifiably generate a linear secret-sharing of $s$ with 
  IC-signatures, such that the underlying core-sets are  ${\GW}_1, \ldots, \GW_{|\Z_s|}$, where $s$ remains private for an {\it honest} $\D$. 
   The verifiability guarantees that even if $\D$ is {\it corrupt}, if any (honest) party computes an output, then there exists {\it some} value, say $s^{\star}$, which is linearly secret-shared with IC-signatures,
    such that the underlying core-sets are ${\GW}_1, \ldots, \GW_{|\Z_s|}$.
   
   The protocol idea is very simple and standard. We first let $\D$ reconstruct the value $r$, which is then used as a {\it one-time pad} (OTP) by $\D$ to make public an OTP-encryption of $s$. 
   Then, using the linearity property of secret-sharing, the parties locally remove the OTP from the OTP-encryption.
    \begin{protocolsplitbox}{$\LSh(\D, s, \Z_s, \Z_a, \ShareSpec_{\Z_s}, [r], {\GW}_1, \ldots, {\GW}_{|\Z_s|})$}{VSS for verifiably generating a linear secret-sharing of a value with 
      IC-signatures with respect to given global core-sets}{fig:LSh}
\justify
\begin{myitemize}
\item[--] {\bf Reconstructing the OTP Towards the Dealer}: The parties in $\PartySet$ invoke an instance $\Rec([r], \{\D\})$ of $\Rec$ to let $\D$ reconstruct $r$
 and {\color{red} wait for time $\TimeRec$}.
\item[--] {\bf Making the OTP-encryption Public}: $\D$, upon computing the output $r$ from the instance of $\Rec$, broadcasts $\s = s + r$.
\item[--] {\bf Computing the Output}: The parties in $\PartySet$ {\color{red} wait till the local time becomes 
 $\TimeRec + \TimeBC$}. Then upon receiving $\s$ from the broadcast of $\D$, the parties in $\PartySet$
 locally compute $[\s - r]$ from $[\s]$ and $[r]$. {\color{blue}Here $[\s]$ denotes the {\it default linear secret-sharing} of $\s$ with IC-signatures and 
  core-sets ${\GW}_1, \ldots, {\GW}_{|\Z_s|}$, where $[\s]_1 = \s$ and
 $[\s]_2 = \ldots = [\s]_{|\Z_s|} = 0$, and where the parties set
 $\ICSig(P_j, P_i, P_k, [\s]_q)_{P_j, P_i \in {\GW}_q, P_k \in \PartySet}$ to the default value}. 
  The parties then output $({\GW}_1, \ldots, {\GW_{|\Z_s|}}, [\s - r])$.
\end{myitemize}
\end{protocolsplitbox}

The properties of the protocol $\LSh$ stated in Lemma \ref{lemma:LSh} are proved in Appendix \ref{app:Triples}.

\begin{lemma}
\label{lemma:LSh}
Let $r$ be a random value which is linearly secret-shared with IC-signatures with ${\GW}_1, \ldots, {\GW}_{|\Z_s|}$ being the underlying core-sets.
  Then protocol $\LSh$ achieves the following where $\D$ participates with the input $s$.
\begin{myitemize}  
\item[--] If $\D$ is honest, then the following hold, where $\TimeLSh = \TimeRec + \TimeBC$.
 \begin{myitemize}
    \item[--] {\bf $\Z_s$-Correctness}: If the network is synchronous, then except with probability $\Order(n^3 \cdot \errorAICP)$, the honest parties output $[s]$ at the time $\TimeLSh$, with
     ${\GW}_1, \ldots, {\GW}_{|\Z_s|}$ being the underlying core-sets.
     \item[--] {\bf $\Z_a$-Correctness}: If the network is asynchronous, then except with probability $\Order(n^3 \cdot \errorAICP)$, the honest parties eventually
      output $[s]$, with
     ${\GW}_1, \ldots, {\GW}_{|\Z_s|}$ being the underlying core-sets.
    \item[--] {\bf Privacy}: Irrespective of the network type, the view of the adversary remains independent of $s$.
 \end{myitemize}
  \item[--] If $\D$ is corrupt then either no honest party computes any output or there exists some value, say $s^{\star}$, such that the following hold.
    \begin{myitemize}
     \item[--] {\bf $\Z_s$-Commitment}: If the network is synchronous, then
      except with probability $\Order(n^3 \cdot \errorAICP)$, the honest parties output $[s^{\star}]$, with
     ${\GW}_1, \ldots, {\GW}_{|\Z_s|}$ being the underlying core-sets. Moreover, if any honest party computes its output at the time $T$, then all honest parties will have their respective output by
     the time $T + \Delta$.
        \item[--] {\bf $\Z_a$-Commitment}: If the network is asynchronous, then
      except with probability $\Order(n^3 \cdot \errorAICP)$, the honest parties eventually output $[s^{\star}]$, with
     ${\GW}_1, \ldots, {\GW}_{|\Z_s|}$ being the underlying core-sets.     
    \end{myitemize}
  \item[--] {\bf Communication Complexity}: $\Order(|\Z_s| \cdot n^3 \cdot \log{|\F|} + n^4 \cdot \log{|\F|} \cdot |\sigma|)$ bits are communicated by the honest parties.
\end{myitemize}
\end{lemma}
We end this section with some notations which we use while invoking the protocol $\LSh$ in the rest of the paper.
\begin{notation}[{\bf Notations for Using Protocol $\LSh$}]
\label{notation:LSh}
Let $P_i \in \PartySet$. In the rest of the paper we will say that ``{\it $P_i$ invokes an instance of $\LSh$ with input $s$}" to mean that $P_i$ acts as $\D$ and invokes an instance 
 $\LSh(\D, s, \Z_s, \Z_a, \ShareSpec_{\Z_s}, [r], {\GW}_1, \ldots, {\GW}_{|\Z_s|})$ of $\LSh$. Here, $r$ will be the corresponding random ``pad" for this instance of $\LSh$, which will {\it already} be linearly secret-shared
  with IC-signatures, with ${\GW}_1, \ldots, {\GW}_{|\Z_s|}$ being the underlying core-sets. If there are multiple instances of $\LSh$ invoked by $\D$, then corresponding to each instance, there will be a random secret-shared
  pad available to the parties {\it beforehand}.  The parties will be knowing which secret-shared pad is associated with which instance of $\LSh$. This will be ensured by upper-bounding the {\it maximum} number of 
  $\LSh$ instances $\LMax$ invoked across {\it all} our protocols. The parties then generate $\LMax$ number of linearly secret-shared random values with IC-signatures,
   with ${\GW}_1, \ldots, {\GW}_{|\Z_s|}$ being the underlying core-sets, by running the protocol $\Rand$ beforehand with $L = \LMax$.
\end{notation}

\subsection{Non-Robust Multiplication Protocol}
Protocol $\BasicMult$ (Fig \ref{fig:BasicMult}) takes input $a$ and $b$, which are linearly secret-shared with IC-signatures, with ${\GW}_1, \ldots, {\GW}_{|\Z_s|}$ being the underlying core-sets
 and a {\it publicly known} subset $\Discarded \subset \PartySet$, consisting of {\it only corrupt} parties. The parties output a linear secret-sharing of $c$ with IC-signatures, with 
  ${\GW}_1, \ldots, {\GW}_{|\Z_s|}$ being the underlying core-sets. If all the parties in 
   $\PartySet \setminus \Discarded$ behave {\it honestly}, then $c = a \cdot b$, else $c = a \cdot b + \delta$, where $\delta \neq 0$. 
   Moreover, the adversary does not learn anything additional about $a$ and $b$ in the protocol. 
   The protocol also takes input an {\it iteration number} $\iter$ and all the sets computed in the protocol are tagged with $\iter$.
   Looking ahead, our {\it robust} triple-generation protocol will be executed {\it iteratively}, with each iteration invoking instances of $\BasicMult$.
   
   A detailed overview of the protocol $\BasicMult$ has been already presented in Section \ref{ssec:Mult}. The idea is to let each summand $[a]_p \cdot [b]_q$ be linearly secret-shared
   by exactly one {\it summand-sharing} party. A secret-sharing of $a \cdot b$ then follows from 
   the secret-sharing of each summand $[a]_p \cdot [b]_q$ owing to the linearity property of the secret-sharing.
   To deal with the network agnostic condition, the summand-sharing parties are selected in two phases: {\it first}, we select them {\it dynamically}, {\it without} pre-assigning any summand to any
   designated party. Once there exists a subset of parties from $\ShareSpec_{\Z_s}$ who have served the role of summand-sharing parties, we go to the {\it second} phase,
   where each remaining summand is designated to the left-over parties through some publicly known assignment. Strict timeouts are maintained
    to ensure that we don't stuck forever during the second phase.
    Finally, if there are still any remaining summands which are not yet secret-shared, they are publicly reconstructed and the default sharing is taken on their behalf.
    Throughout, the parties in $\Discarded$ are not let to secret-share any summand, since they are {\it already known} to be {\it corrupt} and at the same time, it is ensured that the shares
    of the honest parties are {\it never} publicly reconstructed.

\begin{protocolsplitbox}{$\BasicMult(\Z_s, \Z_a, \ShareSpec_{\Z_s}, [a], [b], {\GW}_1, \ldots, {\GW}_{|\Z_s|}, \Discarded, \iter)$}{Network-agnostic non-robust multiplication protocol}{fig:BasicMult}
\justify
\begin{myitemize}
\item[--] \textbf{Initialization}: The parties in $\PartySet$ do the following.
    \begin{myitemize}
       \item Initialize the {\it summand-index-set} of indices of {\it all} summands:
        \[\Products_\iter = \{(p, q)\}_{p, q = 1, \ldots, |\SharingSpec_{\Z_s}|}.\]
         \item Initialize the {\it summand-index-set} corresponding to each $P_j \in \PartySet \setminus \Discarded$:
	\[\Products^{(j)}_\iter = \{(p, q)\}_{P_j \in S_p \cap S_q}. \]	
        \item Initialize the {\it summand-index-set} corresponding to each $S_q \in \ShareSpec_{\Z_s}$: 
        \[\Products^{(S_q)}_\iter = \displaystyle \cup_{P_j \in S_q} \Products^{(j)}_\iter.\]
        \item Initialize the set of summands-sharing parties:
         \[\Selected_\iter = \emptyset.\]
        \item Initialize the hop number:
         \[\hop = 1.\]        
       \end{myitemize}
\justify       
\centerline{\underline{\bf Phase I: Sharing Summands Through Dynamic Assignment}} 
\item[--] While there exists {\it no} $S_q \in \ShareSpec_{\Z_s}$, where $\Products^{(S_q)}_\iter  = \emptyset$, the parties do the following:
\begin{myitemize}
	\item \textbf{Sharing Sum of Eligible Summands}: Every $P_i \notin (\Selected_\iter \cup \Discarded)$ invokes an instance $\LSh^{(\phaseone, \hop, i)}$ of $\LSh$ 
	 with input $c^{(i)}_\iter$, where
	  \[ \displaystyle c^{(i)}_\iter = \sum_{(p, q) \in \Products^{(i)}_\iter} [a]_p[b]_q.\]
	Corresponding to every $P_j \notin (\Selected_\iter \cup \Discarded)$, the parties in $\PartySet$ participate in the instance $\LSh^{(\phaseone, \hop, j)}$, if invoked by $P_j$.
	\item \textbf{Selecting Summand-Sharing Party for the Hop Through ACS}: The parties in $\PartySet$ {\color{red} wait for time $\TimeLSh$} and then do the following.
	\begin{myitemize}
	 \item[--] For $j = 1, \ldots, n$, participate in an instance $\BA^{(\phaseone,\hop,j)}$ of $\BA$
	corresponding to $P_j \in \PartySet$ with input $1$ if {\it all} the following hold:
	  \begin{myitemize}
		\item[--] $P_j \notin (\Selected_\iter \cup \Discarded)$; 
		\item[--] An output $[c^{(j)}_\iter]$ is computed during the instance $\LSh^{(\phaseone,\hop,j)}$.
		\end{myitemize}
	  \item[--] Upon computing an output $1$ during the instance
	   $\BA^{(\phaseone,\hop,j)}$ corresponding to some $P_j \in \PartySet$, participate with input $0$
	    in the instances $\BA^{(\phaseone, \hop, k)}$ corresponding to parties $P_k \notin (\Selected_\iter \cup \Discarded)$, for which no input has been provided yet.	
	  \item[--] Upon computing outputs during the instances  $\BA^{(\phaseone,\hop, i)}$ corresponding to each $P_i \notin (\Selected_\iter \cup \Discarded)$, 
	  let $P_j$ be the least-indexed party, such that the output $1$ is computed during the instance $\BA^{(\phaseone, \hop, j)}$. Then update the following. 
        \begin{myitemize}
		\item[--] $\Selected_\iter = \Selected_\iter \cup \{P_j\}$.
		\item[--] $\Products_\iter = \Products_\iter \setminus \Products^{(j)}_\iter$.
		\item[--] $\forall P_k \in \PartySet \setminus \{\Discarded \cup \Selected_\iter \}$:
		$\Products^{(k)}_\iter = \Products^{(k)}_\iter \setminus \Products^{(j)}_\iter$.
		\item[--] For each $S_q \in \ShareSpec_{\Z_s}$, $\Products^{(S_q)}_\iter = \Products^{(S_q)}_\iter \setminus \Products^{(j)}_\iter$.
		\item[--] Set $\hop = \hop + 1$.	
		\end{myitemize}
	\end{myitemize}
\end{myitemize}
\end{myitemize}
\justify
\centerline{\underline{\bf Phase II: Sharing Remaining Summands Through Static Assignment}} 
\justify
\begin{myitemize}
\item {\bf Re-assigning the Summand-Index-Set of Each Party}: Corresponding to each $P_j \in \PartySet \setminus \Selected_\iter$,  the parties in $\PartySet$
  set $\Products^{(j)}_\iter$ as
     \[ \Products^{(j)}_\iter = \displaystyle \Products_\iter \cap \{(p, q)\}_{P_j = \min(S_p \cap S_q)}, \] 
     where $\min(S_p \cap S_q)$ denotes the minimum indexed party in $(S_p \cap S_q)$.
 \item {\bf Sharing Sum of Assigned Summands}: Every party $P_i \notin (\Selected_\iter \cup \Discarded)$
   invokes an instance $\LSh^{(\phasetwo, i)}$ of $\LSh$ with input $c^{(i)}_\iter$, where
    \[ \displaystyle c^{(i)}_\iter = \sum_{(p, q) \in \Products^{(i)}_\iter} [a]_p[b]_q.\]
    Corresponding to every $P_j \in \PartySet \setminus (\Selected_\iter \cup \Discarded)$,
     the parties in $\PartySet$ participate in the instance $\LSh^{(\phasetwo, j)}$, if invoked by $P_j$.    
\item {\bf Agreeing on the Summand-Sharing parties of the Second Phase}: The parties in $\PartySet$
   {\color{red} wait for $\TimeLSh$ time after the beginning of the second phase}.
    Then for each $P_j \in \PartySet$, participate in an instance $\BA^{(\phasetwo, j)}$ of $\BA$ with input $1$, if {\it all} the following hold, otherwise 
   participate with input $0$.
	\begin{myitemize}
	\item[--] $P_j \notin (\Selected_\iter \cup \Discarded)$;
	\item[--] An output $[c^{(j)}_\iter]$ is computed during the instance $\LSh^{(\phasetwo,j)}$.
	\end{myitemize}
\item {\bf Updating the Sets for the Second Phase}: Corresponding to each $P_j \notin (\Selected_\iter \cup \Discarded)$, such that $1$ is computed as the output during
 $\BA^{(\phasetwo,j)}$, update
     \begin{myitemize}
     \item[--] $\Products_\iter = \Products_\iter \setminus \Products^{(j)}_\iter$;
     \item[--] $\Selected_\iter = \Selected_\iter \cup \{P_j\}$.  \justify
     \end{myitemize}
\end{myitemize}     
\centerline{\underline{\bf Phase III: Reconstructing the Remaining Summands}} 
\justify
\begin{myitemize}
\item {\bf Reconstructing the Remaining Summands and Taking the Default Sharing}: The parties in $\PartySet$ do the following.
	\begin{myitemize}
	\item[--] Corresponding to each $[a]_p$ such that $(p, \star) \in \Products_\iter$, participate in the instance 
	$\RecShare([a], S_p, \PartySet)$ of $\RecShare$ to publicly reconstruct $[a]_p$
	\item[--] Corresponding to each $[b]_q$ such that $(\star, q) \in \Products_\iter$, participate in the instance
	 $\RecShare([b], S_q, \PartySet)$ of $\RecShare$ to publicly reconstruct $[b]_q$.
         \item[--] Corresponding to every $P_j \in \PartySet \setminus \Selected_\iter$, take the {\it default linear secret-sharing} of 
       the public input $c^{(j)}_\iter$ with IC-signatures and core-sets ${\GW}_1, \ldots, {\GW}_{|\Z_s|}$, 
        where\footnote{The default linear secret-sharing of $c^{(j)}_\iter$ is computed in a similar way as done in the protocol 
       $\Rand$ for $\s$ (see Fig \ref{fig:Rand}).}
     \[ \displaystyle c^{(j)}_\iter = \sum_{(p, q) \in \Products^{(j)}_\iter} [a]_p[b]_q.\]
	\end{myitemize}
\item\textbf{Output Computation}:  The parties output
$({\GW}_1, \ldots, {\GW}_{|\Z_s|},  [c^{(1)}_\iter], \ldots, [c^{(n)}_\iter], \allowbreak [c_\iter])$, where $c_\iter \defined c^{(1)}_\iter + \ldots + c^{(n)}_\iter$. 
\end{myitemize}
\end{protocolsplitbox}

The properties of the protocol $\BasicMult$ are claimed in the following lemmas, which are proved in Appendix \ref{app:Triples}.
 \begin{lemma}
\label{lemma:BasicMultFuture}
During any instance $\BasicMult(\Z_s, \Z_a, \ShareSpec_{\Z_s}, [a], [b], {\GW}_1, \ldots, \allowbreak {\GW}_{|\Z_s|}, \Discarded, \iter)$ of $\BasicMult$,
 if $P_j \in \Selected_{\iter}$ then $P_j \not \in \Discarded$, irrespective of the network type. 
 \end{lemma}  

\begin{lemma}
\label{lemma:BasicMultHopTermination}
Suppose that no honest party is present in $\Discarded$. If the honest parties start participating during hop number $\hop$ of Phase I of $\BasicMult$ with iteration number $\iter$, 
 then except with probability $\Order(n^3 \cdot \errorAICP)$,
 the hop takes $\TimeLSh + 2\TimeBA$ time to complete in a synchronous network, or almost-surely completes eventually in an asynchronous network.
\end{lemma}

\begin{lemma}
\label{lemma:BasicMultTermination}
If no honest party is present in $\Discarded$, then in protocol $\BasicMult$, except with probability $\Order(n^3 \cdot \errorAICP)$,
 all honest parties compute some output by the time $\TimeBasicMult = (2n + 1) \cdot \TimeBA + (n + 1) \cdot \TimeLSh + \TimeRec$ in a synchronous network, or almost-surely,
  eventually in an asynchronous network.
\end{lemma}

\begin{lemma}
\label{lemma:BasicMultPrivacy}
If no honest party is present in $\Discarded$, then the view of the adversary remains independent of $a$ and $b$ throughout the protocol, irrespective of the network type.
\end{lemma}

\begin{lemma}
\label{lemma:BasicMultCorrectness}
If no honest party is present in $\Discarded$ and if all parties in $\PartySet \setminus \Discarded$ behave honestly, then in protocol $\BasicMult$, the honest parties output
 a linear secret-sharing of $a \cdot b$ with IC-signatures, with ${\GW}_1, \ldots, {\GW}_{|\Z_s|}$ being the underlying core-sets, irrespective of the network type.
\end{lemma}

\begin{lemma}
\label{lemma:BasicMultCommunication}
Protocol $\BasicMult$ incurs a communication of $\Order(|\Z_s| \cdot n^5 \cdot \log{|\F|} + n^6 \cdot \log{|\F|} \cdot |\sigma|)$
 bits and makes $\Order(n^2)$ calls to $\BA$.
\end{lemma}

As a corollary of Lemma \ref{lemma:BasicMultCommunication}, we can derive the following corollary, which determines the {\it maximum} number of instances of $\LSh$ which are invoked during an instance
 of $\BasicMult$. Looking ahead, this will be useful to later calculate the maximum number of instances of $\LSh$ which need to be invoked as part of our final multiplication protocol.
 This will be further useful to determine the number of linearly secret-shared values with IC-signatures and core-sets ${\GW}_1, \ldots, {\GW}_{|\Z_s|}$, which need to be generated through the
  protocol $\Rand$ beforehand.
\begin{corollary}
\label{cor:BasicMultUpperBound}
During any instance of $\BasicMult$, there can be at most $n^2 + n$ instances of $\LSh$ invoked.
\end{corollary}

\paragraph{\bf Protocol $\BasicMult$ for $L$ Pairs of Inputs.}
Protocol $\BasicMult$ can be easily generalized, if there are $L$ pairs of inputs $\{(a^{\ell}, b^{\ell}) \}_{\ell = 1, \ldots, L}$, all of which are linearly secret-shared with IC-signatures, with  
 ${\GW}_1, \ldots, {\GW}_{|\Z_s|}$ being the underlying core-sets. However, with a slight modification during Phase I and Phase II, we can ensure that the number of instances of $\BA$ remain {\it only} 
 $\Order(n^2)$, which is {\it independent} of $L$. Consider Phase I. During hop number $\hop$, every party $P_i \not \in (\Selected_\iter \cup \Discarded)$ invokes $L$ instances of $\LSh$ to linearly secret-share
  $L$ candidate summand-sums. Now while selecting the summand-sharing party through ACS for this hop, the parties vote for a candidate $P_j \not \in (\Selected_\iter \cup \Discarded)$, provided an output is computed in
  {\it all} the $L$ instances of $\LSh$ invoked by $P_j$. Consequently, the number of instances of $\BA$ during Phase I will be $\Order(n^2)$. 
  Similarly during Phase II, each party {\it outside} $(\Selected_\iter \cup \Discarded)$ invokes $L$ instances of $\LSh$ to linearly secret-share $L$ candidate re-assigned summand-sums. And then the parties vote for a
  candidate $P_j$ as a summand-sharing party, if an output is computed in {\it all} the $L$ instances of $\LSh$ invoked by $P_j$. 
  Finally, during Phase III, the default linear secret-sharing with IC-signatures is taken for the sum of all the summands, which are not yet secret-shared by any party, by making public all these summands.
  The resultant protocol incurs a communication of $\Order(|\Z_s|^3 \cdot n^4 \cdot L \cdot \log{|\F|} + |\Z_s| \cdot n^5 \cdot L \cdot \log{|\F|} + n^6 \cdot L \cdot \log{|\F|} \cdot |\sigma|)$
   bits and makes $\Order(n^2)$ calls to $\BA$. We also note that there will be at most $n^2 \cdot L + n \cdot L$ instances of $\LSh$ invoked in the generalized protocol.
   To avoid repetition, we do not present the steps of the generalized protocol here.
   \subsection{Network Agnostic Random Triple Generation with Cheater Identification}
   \label{ssec:RandMultCI}
The network-agnostic protocol $\RandMultCI$ (Fig \ref{fig:RandMultCI}) takes an iteration number $\iter$ and a publicly known subset of parties $\Discarded$, who
 are guaranteed to be {\it corrupt}. If {\it all} the parties in $\PartySet \setminus \Discarded$
 behave {\it honestly}, then the protocol outputs a random linearly secret-shared multiplication-triple with IC-signatures, with
  ${\GW}_1, \ldots, {\GW}_{|\Z_s|}$ being the underlying core sets. 
  Otherwise, with a high probability, the honest parties identify a {\it new corrupt} party, which is added to $\Discarded$. 
  
  Protocol $\RandMultCI$ is based on \cite{HT13} and 
    consists of two stages: during the {\it first} stage, the parties jointly generate a pair of random values, which are
   linearly secret-shared with IC-signatures, with ${\GW}_1, \ldots, {\GW}_{|\Z_s|}$ being the underlying core sets. 
   During the second stage, the parties run an instance of $\BasicMult$ to compute the product of the pair of secret-shared random values from the first stage.
   To check whether any cheating has occurred during the instance of $\BasicMult$, the parties then run a probabilistic test, namely 
   the ``sacrificing trick" \cite{DPSZ12},
    for which the parties need {\it additional} secret-shared random values, which 
   are generated during the first stage itself.

\begin{protocolsplitbox}{$\RandMultCI(\PartySet, \Z_s, \Z_a, \ShareSpec_{\Z_s}, {\GW}_1, \ldots, {\GW}_{|\Z_s|}, \Discarded, \iter)$}{ Network-agnostic protocol for generating secret-shared random multiplication-triple with cheater 
 identification.}{fig:RandMultCI}
\justify
\begin{myitemize}
\item[--] {\bf Generating Linear Secret Sharing of Random Values with IC-signatures}: Each $P_i \in \PartySet$ does the following.
     \begin{myitemize}
        \item[--]  Invoke instances of $\LSh$ with randomly chosen inputs $a^{(i)}_\iter, b^{(i)}_\iter, b'^{(i)}_\iter, r^{(i)}_\iter \in \F$.
        \item[--] Corresponding to every $P_j \in \PartySet$, participate in the instances of $\LSh$ invoked by $P_j$ (if any)
         and {\color{red} wait for time $\TimeLSh$}.               
        Initialize a set $\CSet_i = \emptyset$ {\color{red} after local time $\TimeLSh$} and include $P_j$ in 
            $\CSet_i$, if any output is computed in {\it all} the instances of $\LSh$ invoked by $P_j$.  
        \item[--] Corresponding to every $P_j \in \PartySet$,
             participate in an instance of $\BA^{(j)}$ of $\BA$ with input $1$, if $P_j \in \CSet_i$.
            \item[--] Once $1$ has been computed as the output from instances of $\BA$ corresponding to a set of parties
            in $\PartySet \setminus Z$ for some $Z \in \Z_s$, participate 
            with input $0$ in all the $\BA$ instances $\BA^{(j)}$, such that            
            $P_j \not \in \CSet_i$.
            \item[--] Once a binary output is computed in all the instances of $\BA$ corresponding to the parties in $\PartySet$,
             compute $\CoreSet$, which is the set of parties $P_j \in \PartySet$, such that
                           $1$ is computed as the output in the instance $\BA^{(j)}$. 
       \end{myitemize}   
   Once $\CoreSet$ is computed, the parties in $\PartySet$ locally compute $[a_\iter], [b_\iter], [b'_\iter]$ and $[r_\iter]$ from
    $\{[a^{(j)}_\iter] \}_{P_j \in \CoreSet}, \{[b^{(j)}_\iter] \}_{P_j \in \CoreSet}, \{[b'^{(j)}_\iter] \}_{P_j \in \CoreSet}$ and $\{[r^{(j)}_\iter] \}_{P_j \in \CoreSet}$ respectively as follows:
   \[ [a_\iter] = \displaystyle \sum_{P_j \in \CoreSet} [a^{(j)}_\iter], \;  [b_\iter] = \displaystyle \sum_{P_j \in \CoreSet} [b^{(j)}_\iter],  \; 
    [b'_\iter] = \displaystyle \sum_{P_j \in \CoreSet} [b'^{(j)}_\iter], \;
   [r_\iter] = \displaystyle \sum_{P_j \in \CoreSet} [r^{(j)}_\iter].\]
\item[--] {\bf Computing Secret-Shared Products}: The parties in $\PartySet$ do the following.
       \begin{myitemize}
       \item[--] Participate in instances $\BasicMult(\Z_s, \Z_a, \ShareSpec_{\Z_s}, [a_\iter], [b_\iter], {\GW}_1, \ldots, {\GW}_{|\Z_s|}, \allowbreak \Discarded, \iter)$
       and $\BasicMult(\Z_s, \Z_a, \ShareSpec_{\Z_s}, [a], [b'_\iter], {\GW}_1, \ldots, {\GW}_{|\Z_s|}, \Discarded, \iter)$ of $\BasicMult$
       to compute the outputs $({\GW}_1, \ldots, {\GW}_{|\Z_s|}, [c^{(1)}_\iter] , \ldots, [c^{(n)}_\iter], [c_\iter])$ and
       $({\GW}_1, \ldots, {\GW}_{|\Z_s|}, [c'^{(1)}_\iter] , \ldots, [c'^{(n)}_\iter], [c'_\iter])$ respectively.        
         Let $\Selected_{\iter, c}$ and $\Selected_{\iter, c'}$ be the summand-sharing parties for the two instances respectively. 
         Moreover, for each $P_j \in \Selected_{\iter, c}$, let $\Products^{(j)}_{\iter, c}$
          be the set of ordered pairs of indices corresponding to the summands whose sum has been shared by $P_j$ 
          during the instance $\BasicMult(\Z_s, \Z_a, \ShareSpec_{\Z_s}, [a_\iter], [b_\iter], {\GW}_1, \ldots, {\GW}_{|\Z_s|}, \Discarded, \iter)$. 
          And similarly,  for each $P_j \in \Selected_{\iter, c'}$, let $\Products^{(j)}_{\iter, c'}$
          be the set of ordered pairs of indices corresponding to the summands whose sum has been shared by $P_j$ 
          during the instance $\BasicMult(\Z_s, \Z_a, \ShareSpec_{\Z_s}, [a], [b'_\iter], {\GW}_1, \ldots, {\GW}_{|\Z_s|}, \Discarded, \iter)$.         
    \end{myitemize}
\item[--] {\bf Error Detection in the Instances of $\BasicMult$}: The parties in $\PartySet$ do the following.
	\begin{myitemize}
	   \item[--] Upon computing outputs from the instances of $\BasicMult$, participate in an instance $\Rec([r_\iter], \PartySet)$ of $\Rec$ to publicly reconstruct $r_\iter$. 
	\item[--] Locally compute $[e_{\iter}] \defined r_\iter[b_\iter] + [b'_\iter]$ from $[b_\iter]$ and $[b'_\iter]$. Participate in an instance $\Rec([e_\iter], \PartySet)$ of $\Rec$ 
	to publicly reconstruct $e_\iter$.		
	\item[--] Locally compute $[d_{\iter}] \defined e_\iter [a_\iter] - r_\iter [c_\iter] - [c'_\iter]$ from $[a_\iter], [c_\iter]$ and $[c'_\iter]$.
	 Participate in an instance $\Rec([d_\iter], \PartySet)$ of $\Rec$ to publicly reconstruct $d_\iter$.	
	\item[--] {\bf Output Computation in Case of Success:} 
	If $d_\iter = 0$, then set the boolean variable $\flag_\iter=0$ and output ${({\GW}_1, \ldots, {\GW}_{|\Z_s|}, [a_\iter], [b_\iter], [c_\iter])}$.
	\item[--] {\bf Cheater Identification in Case of Failure:} If $d_\iter \neq 0$, then set the boolean variable 
	$\flag_\iter = 1$ and proceed as follows.
	\begin{myitemize}
	\item[--] For each $S_q \in \ShareSpec_{\Z_s}$, participate in instances $\RecShare([a_\iter], S_q, \allowbreak \PartySet), \RecShare([b_\iter], S_q, \PartySet)$ and 
	$\RecShare([b'_\iter], S_q, \PartySet)$
	of $\RecShare$ to publicly reconstruct the shares $\{ [a_\iter]_q, [b_\iter]_q, [b'_\iter]_q\}_{S_q \in \ShareSpec_{\Z_s}}$.
	In addition, for $i = 1, \ldots, n$, participate in instances $\Rec(c^{(i)}_\iter, \PartySet)$ and $\Rec(c'^{(i)}_\iter, \PartySet)$ to publicly reconstruct $c^{(i)}_\iter$ and $c'^{(i)}_\iter$.
	\item Set $\displaystyle \Discarded = \Discarded \cup \{P_j\}$, if $P_j \in \Selected_{\iter, c} \cup \Selected_{\iter, c'}$ 
	and the following holds for $P_j$:
	\[ \displaystyle r_{\iter} \cdot c^{(j)}_\iter + c'^{(j)}_\iter \neq r_{\iter} \cdot \sum_{(p, q) \in \Products^{(j)}_{\iter, c}} [a_\iter]_p[b_\iter]_q 
	  + \sum_{(p, q) \in \Products^{(j)}_{\iter, c'}}[a_\iter]_p [b'_\iter]_q.  \]
	\end{myitemize}  
	\end{myitemize}
\end{myitemize}
\end{protocolsplitbox}

The properties of the protocol $\RandMultCI$ are claimed in the following lemmas, which are proved in Appendix \ref{app:Triples}.
\begin{lemma}
  \label{lemma:RandMultCIACS}
 In protocol $\RandMultCI$, the following hold.
 \begin{myitemize}
 \item[--] {\bf Synchronous Network}: Except with probability $\Order(n^3 \cdot \errorAICP)$, honest parties will have linearly secret-shared $a_\iter, b_\iter, b'_\iter$
 and $r_\iter$ with IC-signatures, with ${\GW}_1, \ldots, {\GW}_{|\Z_s|}$ being the underlying core-sets, by the time $\TimeLSh + 2\TimeBA$. Moreover, adversary's view is independent of
  $a_\iter, b_\iter, b'_\iter$ and $r_\iter$. 
 \item[--] {\bf Asynchronous Network}: Except with probability $\Order(n^3 \cdot \errorAICP)$, almost-surely, honest parties will eventually have linearly secret-shared $a_\iter, b_\iter, b'_\iter$
 and $r_\iter$ with IC-signatures, with ${\GW}_1, \ldots, {\GW}_{|\Z_s|}$ being the underlying core-sets. Moreover, adversary's view is independent of
  $a_\iter, b_\iter, b'_\iter$ and $r_\iter$. 
 \end{myitemize}
 \end{lemma}
 
 \begin{lemma}
\label{lemma:RandMultCITermination}
Consider an arbitrary $\iter$, 
  such that all honest parties participate in the instance $\RandMultCI(\PartySet, \Z_s, \Z_a, \ShareSpec_{\Z_s}, {\GW}_1, \ldots, {\GW}_{|\Z_s|}, \Discarded, \iter)$,
   where
   no honest party is present in $\Discarded$. 
   Then except with probability $\Order(n^3 \cdot \errorAICP)$, all honest parties
   reconstruct a (common) value $d_\iter$ and set $\flag_\iter$ to a common Boolean value, 
    at the time $\TimeLSh + 2\TimeBA + \TimeBasicMult + 3\TimeRec$ in a synchronous network, or eventually in an asynchronous network. 
\end{lemma}

\begin{lemma}
\label{lemma:RandMultCIHonestBehaviour}
Consider an arbitrary $\iter$, 
  such that all honest parties participate in the instance $\RandMultCI(\PartySet, \Z_s, \Z_a, \ShareSpec_{\Z_s}, {\GW}_1, \ldots, {\GW}_{|\Z_s|}, \Discarded, \iter)$,
   where
   no honest party is present in $\Discarded$.
   If no party in $\PartySet \setminus \Discarded$ behaves maliciously, then $d_\iter = 0$
   and 
    the honest parties output $([a_\iter], [b_\iter], [c_\iter])$     
     at the time $ \TimeLSh + 2\TimeBA + \TimeBasicMult + 3\TimeRec$ in a synchronous network or eventually in an asynchronous network, where $c_\iter = a_\iter \cdot b_\iter$ 
     and where ${\GW}_1, \ldots, {\GW}_{|\Z_s|}$ are
           the underlying core-sets
\end{lemma}   

\begin{lemma}
\label{lemma:RandMultCICorruptBehaviour}
Consider an arbitrary $\iter$, 
  such that all honest parties participate in the instance $\RandMultCI(\PartySet, \Z_s, \Z_a, \ShareSpec_{\Z_s}, {\GW}_1, \ldots, {\GW}_{|\Z_s|}, \Discarded, \iter)$,
   where
   no honest party is present in $\Discarded$. 
   If $d_\iter \neq 0$, then except with probability $\Order(n^3 \cdot \errorAICP)$, 
   the honest parties update $\Discarded$ by adding a new maliciously-corrupt
   party in $\Discarded$, either at the time $\TimeRandMultCI = \TimeLSh + 2\TimeBA + \TimeBasicMult + 4\TimeRec$ in a synchronous network or eventually in an asynchronous network.
   \end{lemma}

\begin{lemma}
\label{lemma:RandMultCICorrectness}
Consider an arbitrary $\iter$, 
  such that all honest parties participate in the instance $\RandMultCI(\PartySet, \Z_s, \Z_a, \ShareSpec_{\Z_s}, {\GW}_1, \ldots, {\GW}_{|\Z_s|}, \Discarded, \iter)$,
   where
   no honest party is present in $\Discarded$. 
   If $d_\iter = 0$, then the honest parties output linearly secret-shared
  $(a_\iter, b_\iter, c_\iter)$ with IC-signatures with ${\GW}_1, \ldots, {\GW}_{|\Z_s|}$ being
           the underlying core-sets, 
  at the time $\TimeLSh + 2\TimeBA + \TimeBasicMult + 3\TimeRec$ in a synchronous network or eventually in an asynchronous network where, 
   except with probability  $\frac{1}{|\F|}$, the condition $c_\iter = a_\iter \cdot b_\iter$ holds. 
   Moreover, the view of $\Adv$ will be independent of $(a_\iter, b_\iter, c_\iter)$.
\end{lemma}

\begin{lemma}
\label{lemma:RandMultCICommunication}
Protocol $\RandMultCI$ incurs a communication of $\Order(|\Z_s| \cdot n^5 \cdot \log{|\F|} + n^6 \cdot \log{|\F|} \cdot |\sigma|)$
 bits and makes $\Order(n^2)$ calls to $\BA$.
\end{lemma}

\paragraph{\bf Protocol $\RandMultCI$ for $L$ Triples.}
Protocol $\RandMultCI$ can be easily generalized to generate $L$ triples with cheater identification, such that the number of instances of $\BA$ is {\it independent} of $L$.
 To begin with, every party $P_i$ now picks $3L + 1$ random values $(r_\iter^{(i)}, \{a_\iter^{(\ell, i)}, b_\iter^{(\ell, i)}, {b'}_\iter^{(\ell, i)} \}_{\ell = 1, \ldots, L})$ 
  and invokes $3L + 1$ instances of $\LSh$ to generate linear secret-sharing of these values with IC-signatures, with
 ${\GW}_1, \ldots, {\GW}_{|\Z_s|}$ being the underlying core-sets. Next, to determine $\CoreSet$, the parties invoke {\it only} $n$ instances of $\BA$, where the $j^{th}$ instance is used
 to decide whether $P_j$ should be included in $\CoreSet$, the criteria being whether some output is computed in {\it all} the $3L + 1$ instances of $\LSh$ invoked by $P_j$. 
 Once $\CoreSet$ is decided, the parties locally compute $\{ [a_\iter^{(\ell)}] \}_{\ell = 1, \ldots, L}$, $\{ [b_\iter^{(\ell)}] \}_{\ell = 1, \ldots, L}$, $\{ [{b'}_\iter^{(\ell)}] \}_{\ell = 1, \ldots, L}$
 and $[r_\iter]$ from $\{ [a_\iter^{(\ell, j)}] \}_{P_j \in \CoreSet, \ell = 1, \ldots, L}$, $\{ [b_\iter^{(\ell, j)}] \}_{P_j \in \CoreSet, \ell = 1, \ldots, L}$, 
 $\{ [{b'}_\iter^{(\ell, j)}] \}_{P_j \in \CoreSet, \ell = 1, \ldots, L}$ and $\{ [r_\iter^{(j)}] \}_{P_j \in \CoreSet}$ respectively.
 The parties then invoke two instances of (generalized) $\BasicMult$ protocol with $L$ pairs of secret-shared inputs to compute
 $\{ [c_\iter^{(\ell)}] \}_{\ell = 1, \ldots, L}$ and $\{ [{c'}_\iter^{(\ell)}] \}_{\ell = 1, \ldots, L}$ respectively. Note that this requires {\it only} $\Order(n^2)$ instances of $\BA$, apart from $n^2 \cdot L + n \cdot L$
  instances of $\LSh$. The rest of the protocol steps are then generalized to deal with $L$ inputs. 
  The resultant protocol incurs a communication of 
  $\Order(|\Z_s| \cdot n^5 \cdot L \cdot \log{|\F|} + n^6 \cdot L \cdot \log{|\F|} \cdot |\sigma|)$
   bits and makes $\Order(n^2)$ calls to $\BA$. We also note that there will be at most $n^2 \cdot L + 4n \cdot L + n$ instances of $\LSh$ invoked overall in the generalized protocol.
   To avoid repetition, we do not present the steps of the generalized protocol here.
   \subsection{The Multiplication-Triple Generation Protocol}
Protocol $\TripGen$ for generating a single secret-shared multiplication-triple is presented in Fig \ref{fig:TripGen}. The idea of the protocol is very simple and based on \cite{HT13}.
 The parties iteratively run instances of $\RandMultCI$, till they hit upon an instance when {\it no} cheating is detected. 
  Corresponding to each ``failed" instance of $\RandMultCI$, the parties keep updating the set $\Discarded$. Since after each failed instance the set $\Discarded$ is updated with
   one {\it new} corrupt party, there will be at most $(t + 1)$ iterations, where $t$ is the cardinality of the largest-sized subset in $\Z_s$. 
\begin{protocolsplitbox}{$\TripGen(\PartySet, \Z_s, \Z_a, \ShareSpec_{\Z_s}, {\GW}_1, \ldots, {\GW}_{|\Z_s|})$}{Network-agnostic protocol to generate a linear
 secret sharing with IC-signature of a single random multiplication-triple.}{fig:TripGen}
\justify
\begin{myitemize}
\item[--] \textbf{Initialization}: The parties in $\PartySet$ initialize $\Discarded = \emptyset$ and $\iter = 1$.
\item[--] {\bf Triple Generation with Cheater Identification}: The parties in $\PartySet$ participate in an instance
 $\RandMultCI(\PartySet, \Z_s, \Z_a, \ShareSpec_{\Z_s}, {\GW}_1, \ldots, {\GW}_{|\Z_s|}, \Discarded, \iter)$ of $\RandMultCI$ and wait for its completion.
  Upon computing output from the instance, the parties proceed as follows.
	\begin{myitemize}
	\item {\bf Positive Output}: If the Boolean variable $\flag_\iter$ is set to $0$ during the instance of 
	$\RandMultCI$, then 
	 output ${({\GW}_1, \allowbreak \ldots,  {\GW}_{|\Z_s|}, [a_\iter], [b_\iter], [c_\iter])}$, computed during the instance of $\RandMultCI$.
	\item {\bf Negative Output}: Else set $\iter = \iter+1$ and go to the step labelled {\bf Triple Generation with Cheater Identification}. 
	\end{myitemize}
\end{myitemize}
\end{protocolsplitbox}

The properties of the protocol $\TripGen$ are claimed in the following lemmas, which are proved in Appendix \ref{app:Triples}.
\begin{lemma}
\label{lemma:TripGenTermination}
Let $t$ be the size of the largest set in $\Z_s$. Then except with probability $\Order(n^3 \cdot \errorAICP)$, the honest parties compute an output during 
  $\TripGen$, by the time $\TimeTripGen = (t + 1) \cdot \TimeRandMultCI$ in a synchronous network, or almost-surely, eventually in an asynchronous network, where
  $\TimeRandMultCI = \TimeLSh + 2\TimeBA + \TimeBasicMult + 4\TimeRec$.
\end{lemma}

\begin{lemma}
\label{lemma:TripGenCorrectnessPrivacy}
If the honest parties output ${({\GW}_1,  \ldots,  {\GW}_{|\Z_s|}, [a_\iter], [b_\iter], [c_\iter])}$ during the protocol $\TripGen$,
 then $a_\iter, b_\iter$ and $c_\iter$ are linearly secret-shared with IC-signatures, with
  ${\GW}_1,  \ldots,  {\GW}_{|\Z_s|}$ being the underlying core-sets. Moreover, $c_\iter = a_\iter b_\iter$ holds, except with probability $\frac{1}{|\F|}$. 
   Furthermore, the view of the adversary remains independent of $a_\iter, b_\iter$ and $c_\iter$.
\end{lemma}

\begin{lemma}
\label{lemma:TripGenCommunication}
Protocol $\TripGen$ incurs a communication of $\Order(|\Z_s| \cdot n^6 \cdot \log{|\F|} + n^7 \cdot \log{|\F|} \cdot |\sigma|)$
 bits and makes $\Order(n^3)$ calls to $\BA$.
\end{lemma}

\paragraph{\bf Protocol $\TripGen$ for Generating $L$ Multiplication-Triples.}
To generate $L$ multiplication-triples, the parties now need to invoke an instance of the generalized (modified) $\RandMultCI$ protocol in each iteration,
 which generates $L$ triples with cheater identification.
 The rest of the protocol steps remain the same. To avoid repetition, we do not present the formal details here.
  The protocol incurs a communication of $\Order(|\Z_s| \cdot n^6 \cdot L \cdot \log{|\F|} + n^7 \cdot L \cdot \log{|\F|} \cdot |\sigma|)$
   bits and makes $\Order(n^3)$ calls to $\BA$.
   \paragraph{\bf On the Maximum Number of Calls of $\LSh$ in $\TripGen$.}
   As discussed in the previous section, each instance of $\RandMultCI$ for $L$ triples requires at most
   $n^2 \cdot L + 4n \cdot L + n$ instances of $\LSh$. Now as there can be up to $t + 1 \approx n$ such instances of $\RandMultCI$ in the protocol $\TripGen$, it follows that
   at most $n^3 \cdot L + 4n^2 \cdot L + n^2$ instances of $\LSh$ are invoked in the protocol $\TripGen$ for generating $L$ multiplication-triples.

%% file: mpc.tex
\section{Network Agnostic Circuit-EvaluationProtocol}
\label{sec:mpc}
 The network-agnostic circuit-evaluation protocol $\PiMPC$ is presented in Fig \ref{fig:MPC}. 
  The idea behind the protocol is to perform shared circuit-evaluation, where each value remains linearly secret-shared with IC-signatures and common core-sets
  Once the function-output is secret-shared, it is publicly reconstructed. To achieve this goal, the parties first secret-share their respective inputs for the function $f$ through instances of $\LSh$.
   The parties then agree on a common subset of parties $\CoreSet$, where $\PartySet \setminus \CoreSet \in \Z_s$, such that the inputs of the parties in $\CoreSet$ are linearly secret-shared
   with IC-signatures. If the network is {\it synchronous} then it will be ensured that {\it all honest} parties are present in $\CoreSet$
   and hence the inputs of {\it all honest} parties are considered for the circuit-evaluation.
      The linearity of secret-sharing ensures that the linear gates in $\ckt$ are evaluated non-interactively, while Beaver's trick is deployed for evaluating multiplication gates in $\ckt$. 
   For the latter, the parties need to have $c_M$ number of random multiplication-triples apriori, which are linearly secret-shared with IC-signatures. 
   This is achieved by apriori calling the protocol $\TripGen$ with $L = c_M$, which in turn will require {\it at most} $n^3 \cdot c_M + 4n^2 \cdot c_M + n^2$ number of instances of $\LSh$.
   As the {\it total} number of instances of $\LSh$ across $\TripGen$ and the input-phase is at most $n^3 \cdot c_M + 4n^2 \cdot c_M + n^2 + n$, 
  the parties first invoke an instance of $\Rand$   
   by setting $L = n^3 \cdot c_M + 4n^2 \cdot c_M + n^2 + n$, to generate
  these many linearly secret-shared random values with IC-signatures, with ${\GW}_1,  \ldots,  {\GW}_{|\Z_s|}$ being the underlying core-sets.
  This will ensure that all the values during the circuit-evaluation are linearly secret-shared with IC-signatures, with ${\GW}_1,  \ldots,  {\GW}_{|\Z_s|}$ being the underlying core-sets.

  Notice that if the network is {\it asynchronous} then different parties may be in the different phases of the protocol. And consequently, a party upon reconstructing the function output {\it cannot}
   afford to immediately terminate, as its presence may be required in the other parts of the protocol. Hence there is also a {\it termination phase}, which is executed concurrently, where the parties check if
   it is ``safe" to terminate the protocol.

\begin{protocolsplitbox}{$\PiMPC(\PartySet, \Z_s, \Z_a, \ShareSpec_{\Z_s}, \ckt, c_M)$}{Network agnostic secure circuit-evaluation protocol.}{fig:MPC}
	\justify
\begin{center}\underline{\bf Pre-Processing Phase}\end{center}
\begin{myitemize}
\item[--] {\bf Generating Linearly Secret-Shared Random Values with IC-signatures}: The parties invoke an instance $\Rand(\PartySet, \Z_s, \Z_a, \SharingSpec_{\Z_s}, L)$
  where $L = n^3 \cdot c_M + 4n^2 \cdot c_M + n^2 + n$
  and compute output 
  $({\GW}_1, \ldots, {\GW}_{|\Z_s|}, \{[r^{(\mathfrak{l})}] \}_{\mathfrak{l} = 1, \ldots, L})$
\item[--] {\bf Generating Linearly Secret-Shared Random Multiplication-Triples with IC-signatures}:
 Upon computing an output in the instance of $\Rand$,
 the parties invoke an instance 
 $\TripGen(\PartySet, \Z_s, \Z_a, \ShareSpec_{\Z_s}, {\GW}_1, \ldots, {\GW}_{|\Z_s|})$ with $L = c_M$
  and compute output ${({\GW}_1, \ldots,  {\GW}_{|\Z_s|}, \{[a^{(\ell)}], [b^{(\ell)}], [c^{(\ell)}]\}_{\ell = 1, \ldots, c_M})}$.
  {\color{blue} During the instance of $\TripGen$, the secret-shared values $\{[r^{(\mathfrak{l})}] \}_{\mathfrak{l} = 1, \ldots, n^3 \cdot c_M + 4n^2 \cdot c_M + n^2}$
  are used as the corresponding pads in the underlying instances of $\LSh$, invoked as part of $\TripGen$.}
\end{myitemize}
\begin{center}\underline{\bf Input Phase}\end{center}
Upon computing an output during the instance of $\TripGen$, each $P_i \in \PartySet$ does the following.
     \begin{myitemize}
        \item[--]  On having the input $x^{(i)}$, invoke an instance of $\LSh$ with input
         $x^{(i)}$.\footnote{{\color{blue} The secret-shared random value
         $r^{(n^3 \cdot c_M + 4n^2 \cdot c_M + n^2 + j)}$ serves as the pad for the instance of $\LSh$ invoked by $P_j$ in this phase.}}
        \item[--] Corresponding to every $P_j \in \PartySet$, participate in the instance of $\LSh$ invoked by $P_j$ 
         and {\color{red} wait for local time $\TimeLSh$ after starting the input phase}. Then               
        initialize a set $\CSet_i = \emptyset$ and include $P_j$ in 
            $\CSet_i$, if any output is computed in the instance of $\LSh$ invoked by $P_j$.  
        \item[--] Corresponding to every $P_j \in \PartySet$,
             participate in an instance of $\BA^{(j)}$ of $\BA$ with input $1$, if $P_j \in \CSet_i$.
            \item[--] Once $1$ has been computed as the output from instances of $\BA$ corresponding to a set of parties
            in $\PartySet \setminus Z$ for some $Z \in \Z_s$, participate 
            with input $0$ in all the $\BA$ instances $\BA^{(j)}$, such that            
            $P_j \not \in \CSet_i$.
            \item[--] Once a binary output is computed in all the instances of $\BA$ corresponding to the parties in $\PartySet$,
             compute $\CoreSet$, which is the set of parties $P_j \in \PartySet$, such that
                           $1$ is computed as the output in the instance $\BA^{(j)}$. 
       \end{myitemize}   
   Once $\CoreSet$ is computed, corresponding to every $P_j \not \in \CoreSet$, the parties set $x^{(j)} = 0$ and {\color{blue} take the {\it default}
   linear secret-sharing of $0$ with IC-signatures, with ${\GW}_1, \ldots, {\GW}_{|\Z_s|}$ being the underlying core-sets.}
\begin{center}\underline{\bf Circuit-Evaluation Phase}\end{center}
\begin{myitemize}
 \item[--] Evaluate each gate $g$ in the circuit
 according to the topological ordering as follows, depending upon the type.
        \begin{myitemize}
           \item[--] {\bf Addition Gate}: If $g$ is an addition gate with inputs $x, y$ and output $z$, then the parties in $\PartySet$ locally compute $[z] = [x + y]$ from $[x]$ and $[y]$.
             \item[--] {\bf Multiplication Gate}: If $g$ is the $\ell^{th}$ multiplication gate with inputs $x, y$ and output $z$, where $\ell \in \{1, \ldots, M \}$, then 
             the parties in $\PartySet$ do the following:
               \begin{myitemize}
               \item [--] Locally compute $[d^{(\ell)}] = [x - a^{(\ell)}]$ from $[x]$ and $[a^{(\ell)}]$ and 
               $[e^{(\ell)}] = [y - b^{(\ell)}]$ from $[y]$ and $[b^{(\ell)}]$.
               \item [--] Participate in instances $\Rec([d^{(\ell)}], \PartySet)$ and $\Rec([e^{(\ell)}], \PartySet)$ of $\Rec$ 
               to publicly reconstruct 
                 $d^{(\ell)}$ and $e^{(\ell)}$, where $d^{(\ell)} \defined x - a^{(\ell)}$ and $e^{(\ell)} \defined y - b^{(\ell)}$.
               \item [--] Upon reconstructing $d^{(\ell)}$ and $e^{(\ell)}$, {\color{blue} take the default linear secret-sharing of $d^{(\ell)} \cdot e^{(\ell)}$ with IC-signatures, with
               ${\GW}_1, \ldots, {\GW}_{|\Z_s|}$ being the underlying core-sets}. Then 
                locally compute
               $[z] = [d^{(\ell)} \cdot e^{(\ell)}] + d^{(\ell)} \cdot [b^{(\ell)}] + e^{(\ell)} \cdot [a^{(\ell)}] + [c^{(\ell)}]$ from
               $[d^{(\ell)} \cdot e^{(\ell)}], [a^{(\ell)}], [b^{(\ell)}]$ and $[c^{(\ell)}]$.               
               \end{myitemize}
          \item[--] {\bf Output Gate}: If $g$ is the output gate with output $y$, then participate in an instance $\Rec([y], \PartySet)$ of $\Rec$ to publicly
          reconstruct $y$.
       \end{myitemize}
\end{myitemize}
\begin{center}\underline{\bf Termination Phase}\end{center}
Every $P_i \in \PartySet$ concurrently executes the following steps during the protocol:
   \begin{myitemize}
   \item[--] Upon computing the circuit-output $y$, send the message $(\ready, P_i, y)$ to every party in $\PartySet$.
      \item[--] Upon receiving the message $(\ready, P_j, y)$ from a set of parties ${\cal A}$ such that $\Z_s$ satisfies  $\Q^{(1)}({\cal A}, \Z_s)$ condition, 
       send $(\ready, P_i, y)$  to every party in $\PartySet$, provided no $(\ready, P_i, \star)$ message has been sent yet.
       \item[--] Upon receiving the message $(\ready, P_j, y)$ from a set of parties $\W$ such that $\PartySet \setminus \W \in \Z_s$, output $y$ and terminate.
   \end{myitemize}
\end{protocolsplitbox}

The properties of the protocol $\PiMPC$ stated in Theorem \ref{thm:MPC}, are proved in Appendix \ref{app:MPC}.
 \begin{theorem}
 \label{thm:MPC}
 Let $\Z_s$ and $\Z_a$ be monotone adversary structures where $\Z_a \subset \Z_s$,  the set
  $\Z_s$ satisfy the condition $\Q^{(2)}(\PartySet, \Z_s)$, the set $\Z_a$ satisfy the condition $\Q^{(3)}(\PartySet, \Z_a)$ and $\Z_s, \Z_a$ together satisfy the 
   $\Q^{(2, 1)}(\PartySet, \Z_s, \Z_a)$ condition. 
 Let $\F$ be a finite field such that $|\F| \geq n^5 \cdot 2^{\ssec}$ where
  $\ssec$ is the statistical security parameter. Moreover, let $y = f(x^{(1)}, \ldots, x^{(n)})$ be a publicly known function over $\F$ represented by an arithmetic circuit $\ckt$ over $\F$,
  where each $P_i \in \PartySet$ has the input $x^{(i)}$. Furthermore, let $c_M$ and $D_M$ be the number of multiplication gates and the multiplicative depth of $\ckt$ respectively. 
  Then given an unconditionally-secure PKI,  
    protocol $\PiMPC$ achieves the following, where every honest $P_j$ participates with the input $x^{(j)}$ and where $t$ denotes the cardinality of the maximum sized subset in $\Z_s$.
  \begin{myitemize}
  \item[--] {\bf Synchronous Network}: Except with probability $2^{-\ssec}$, all honest parties output $y = f(x^{(1)}, \ldots, x^{(n)})$ at the time 
   $[(3n + 5) t^2 + (74n + 140)t + 69n + D_M + 438] \cdot \Delta$, where $x^{(j)} = 0$, for every $P_j \not \in \CoreSet$ and where every honest party is present in $\CoreSet$, such that
    $\PartySet \setminus \CoreSet \in \Z_s$.
  \item[--] {\bf Asynchronous Network}: Except with probability $2^{-\ssec}$, almost-surely, all honest parties eventually output $y = f(x^{(1)}, \ldots, x^{(n)})$ 
   where $x^{(j)} = 0$, for every $P_j \not \in CoreSet$ and where $\PartySet \setminus \CoreSet \in \Z_s$.
 \end{myitemize}
 The protocol incurs a communication of
   $\Order(|\Z_s|^2 \cdot n^{12} \cdot \log{|\F|})$
   bits and makes $\Order(n^3)$ calls to $\BA$.
   Moreover, irrespective of the network type, the view of the adversary remains independent of the inputs of the honest parties in $\CoreSet$.
 \end{theorem}

%% file: Impossibility.tex
\section{Impossibility Result}
\label{sec:impossibility}
Here we show the necessity of the $\Q^{(2, 1)}(\PartySet, \Z_s, \Z_a)$ condition for network agnostic MPC. In fact we show that the condition is even necessary for network agnostic BA. For this, we generalize the
 impossibility proof of \cite{BKL19} which shows the impossibility of network agnostic BA against threshold adversaries if $2t_s + t_a \geq n$.
 \begin{theorem}
 \label{thm:Impossibility}
  Let $\Z_s$ and $\Z_a$ satisfy the 
 $\Q^{(2)}(\PartySet, \Z_s)$  and  $\Q^{(3)}(\PartySet, \Z_a)$ conditions respectively, where $\Z_a \subset \Z_s$.\footnote{The necessity of the 
  $\Q^{(2)}(\PartySet, \Z_s)$  and  $\Q^{(3)}(\PartySet, \Z_a)$ conditions follow from the existing results on the impossibility of unconditionally-secure SBA and ABA respectively, without
   these conditions. The condition $\Z_a \subset \Z_s$ is also necessary since any potential corrupt subset which is tolerable in an {\it asynchronous} network should also be tolerable if
   the network is {\it synchronous}.}
  Moreover, let the parties have access to the setup of an unconditional PKI.
  Furthermore, 
   let $\Pi$ be an $n$-party protocol, which is a $\Z_s$-secure SBA protocol in the synchronous network and which is a $\Z_a$-secure ABA protocol in the asynchronous network (as per Definition \ref{def:BA}).
  Then $\Pi$ exists only if $\Z_s$ and $\Z_a$ satisfy the $\Q^{(2, 1)}(\PartySet, \Z_s, \Z_a)$ condition.
 \end{theorem}
 \begin{proof}
 The proof is by contradiction. Let $\Pi$ exist, even if $\Z_s$ and $\Z_a$ {\it do not} satisfy the $\Q^{(2, 1)}(\PartySet, \Z_s, \Z_a)$ condition.
  Then, there exist sets, say $Z_0, Z_1 \in \Z_s$ and $Z_2 \in \Z_a$ such that $Z_0 \cup Z_1 \cup Z_2  \supseteq \PartySet$ holds. 
   For simplicity and without loss of generality, assume that $Z_0, Z_1$ and $Z_2$ are {\it disjoint}. 
   Now consider the following executions of $\Pi$.
    In {\it all} these executions, parties in $Z_0$ participate with input $0$, and parties in $Z_1$ participate with input $1$.
\begin{myitemize}
\item[--] {\bf Execution $E_1$}: In this execution, the network is {\it synchronous}. All the parties in $Z_0$ are corrupted by the adversary and simply {\it abort},
 and all parties in $Z_2$ participate with input $1$. Since {\it all} the {\it honest} parties (namely the parties in $Z_1 \cup Z_2$) have input $1$, from the {\it $\Z_s$-validity} of $\Pi$ in the
  {\it synchronous} network, the parties in $Z_1 \cup \Z_2$ should output $1$ after some fixed time, say $T_1$. 
\item[--] {\bf Execution $E_2$}: In this execution, the network is {\it synchronous}. All the parties in $Z_1$ are corrupted by the adversary and simply {\it abort},
 and all parties in $Z_2$ participate with input $0$. Since {\it all} the {\it honest} parties (namely the parties in $Z_0 \cup Z_2$) have input $0$, from the {\it $\Z_s$-validity} of $\Pi$ in the
  {\it synchronous} network, the parties in $Z_0 \cup \Z_2$ should output $0$ after some fixed time, say $T_2$. 
\item[--] {\bf Execution $E_3$}: In this execution, the network is {\it asynchronous}, the adversary corrupts all the 
 parties in $Z_2$ and behave as follows:
 the communication between the parties in $Z_0$ and $Z_1$ is delayed by at least time $\mbox{max}(T_1, T_2)$.
  The adversary communicates with parties in $Z_0$ and $Z_1$, such that the views of the parties in $Z_0$ and $Z_1$ are identical to $E_1$ and $E_2$ respectively.
  For this, the adversary runs $\Pi$ with input $0$ when interacting with the parties in $Z_0$ and runs 
  $\Pi$ with input $1$ when interacting with the parties in $Z_1$. 
   Hence, the parties in $Z_0$ output $0$, while the parties in $Z_1$ output $1$, which violates the {\it $\Z_a$-consistency} of $\Pi$ in the {\it asynchronous} network. This is a contradiction and hence, $\Pi$
   {\it does not} exist.
\end{myitemize}

 \end{proof}

%% file: Acknowledgements.tex
\section{Acknowledgements}
\label{sec:impossibility}

We would like to thank Anirudh Chandramouli for several helpful discussions during the early stages of this research.

%% file: AppBA.tex
\section{Properties of the Network Agnostic BA Protocol}
\label{app:BA}
In this section, we prove the properties of the network agnostic BA protocol (see Fig \ref{fig:BA} for the protocol).
 We first formally present the sub-protocol $\PiPW$ and prove its properties.
\subsection{Protocol $\PiPW$: Synchronous BA with Asynchronous Guaranteed Liveness}
Protocol $\PiPW$ is presented in Fig \ref{fig:PW}, where for simplicity we assume that the input of each party is a bit.
 The protocol is very simple.
 Each party $P_i$ uses a Dolev-Strong (DS) style protocol \cite{DS83}
  to broadcast its input $b_i$. The protocol runs for $t+1$ ``rounds", where $t$ is the size of the largest set in $\Z_s$. 
  Each party $P_i$ accumulates values on the behalf of every party $P_j$ in a set $\ACC_{ij}$. A bit $b$ is added to $\ACC_{ij}$ during round $r$ only if $P_i$ receives signatures on $b$ from $r$ distinct parties 
  {\it including} $P_j$. Party $P_i$ computes the final output by taking the ``majority" among the accumulated values. This is done by computing a final set $\FIN_i$ of values based on each $\ACC_{ij}$ set.
  Since the DS protocol is designed for the {\it synchronous} network, for convenience,
  we present the protocol $\PiPW$ in a round-based fashion, where the parties set the duration of each round to $\Delta$ and {\it will know} the
  beginning and end of each round. 
  
\begin{protocolsplitbox}{$\PiPW(\PartySet, \Z_s, \Z_a)$}{ Synchronous BA with asynchronous guaranteed liveness. The above code is executed by every party $P_i$ with input $b_i$.}{fig:PW}
\justify
\begin{myitemize}
\item[--] {\bf Initialization}: Initialize the following sets.
	\begin{myitemize}
	\item[--] For $j = 1, \ldots, n$: the set of values accumulated on the behalf of $P_j$, $\ACC_{ij} = \emptyset$.
	\item[--] $t$: size of the largest set in $\Z_s$.
	\item[--] The final set of values to be considered while computing the output $\FIN_i = \emptyset$.
	\end{myitemize}
\item[--] {\bf Round 0}: On having input bit $b_i$, sign $b_i$ to obtain the signature $\sigma_{ii}$. Set $\SET_i = \{\sigma_{ii}\}$ and send $(b_i, i, \SET_i)$ to every party $P_j \in \PartySet$.
\item[--] {\bf Round $r$ = $0$ to $t + 1$}: On receiving $(b, j, \SET_j)$ in round $r$, check if {\it all} the following hold.
	\begin{myitemize}
	\item[--] $\SET_j$ contains valid signatures on $b$ from $r$ {\it distinct} parties, {\it including} $P_j$; 
	\item[--] $b \notin \ACC_{ij}$.
	\end{myitemize}
	If all the above hold, then do the following.	
	\begin{myitemize}
	\item[--] Add $b$ to $\ACC_{ij}$.
	\item[--] If $r \neq (t + 1)$, then compute a signature $\sigma_{ij}$ on $b$ and send $(b,j, \SET_j \cup \{ \sigma_{ij} \})$ to every party $P_k \in \PartySet$.
	\end{myitemize}
\item[--] {\bf Output Computation}: At time $(t+1) \cdot \Delta$, do the following.
	\begin{myitemize}
	\item[--] If $\ACC_{ij}$ contains exactly one value $b$, then add $(j,b)$ to $\FIN_i$.  Else, add $(j,\bot)$ to $\FIN_i$.
	\item[--] If there exists a set $Z \in \Z_s$ and a value $b \neq \bot$ such that, for every $P_j \in \PartySet \setminus Z$, $(j,b)$ belongs to $\FIN_i$, then output $b$. Else, output $\bot$.
	\end{myitemize}
\end{myitemize}
\end{protocolsplitbox}

We next prove the properties of the protocol $\PiPW$, which are a straightforward generalization of the properties of the DS protocol.
\begin{lemma}
\label{lemma:PW-sync-consistency}
Protocol $\PiPW$ achieves $\Z_s$-Consistency in a synchronous network.
\end{lemma}
\begin{proof}
We claim that each {\it honest} party $P_i$ computes the same set $\ACC_{ij}$ corresponding to {\it every} 
party $P_j$, by time $(t + 1) \cdot \Delta$. Assuming the claim is true, 
 the proof then follows from the fact that $\FIN_i$ is computed {\it deterministically} at the time $(t + 1) \cdot \Delta$,
  based on the sets $\ACC_{ij}$. 
  To prove the claim, consider an {\it arbitrary} honest party $P_i$ and an {\it arbitrary} $P_j$.
   If $P_i$ includes $b$ to $\ACC_{ij}$, then we show that by the time $(t + 1) \cdot \Delta$, the value $b$ will be present in the set $\ACC_{kj}$ of {\it every} honest party $P_k$. 
   For this, we consider the following two cases.
	\begin{myitemize}
	\item[--] {\bf Case 1 - $P_i$ added $b$ to $\ACC_{ij}$ during round $r$ where  $r \leq t$}:
	 In this case, $P_i$ must have received $(b, j, \SET_j)$ in round $r$, where $\SET_j$ contained valid signatures on $b$ from $r-1$ distinct parties {\it apart} from $P_j$. Party
	  $P_i$ then computes $\sigma_{ij}$, adds this to $\SET_j$, and sends $(b,j,\SET_j)$ to every party. When $P_k$ receives this during round $r+1$, it will find that $\SET_j$ contains $r$ valid signatures on $b$ apart from party $P_j$'s, including $\sigma_{ij}$. Hence, $P_k$ will add $b$ to $\ACC_{kj}$. Since $r+1 \leq t + 1$, this happens by time $(t + 1) \cdot \Delta$.
	\item[--] {\bf Case 2 - $P_i$ added $b$ to $\ACC_{ij}$ during round $r = t + 1$}:
	 In this case $P_i$ must have received $(b,j,\SET_j)$ in round $t + 1$, where $\SET_j$ contained valid signatures on $b$ from $t$ distinct parties {\it apart} from $P_j$. 
	 This means that, in total, $P_i$ has received valid signatures on $b$ from $t+1$ distinct parties. Among these, at least one party, say $P_h$, must be honest, as there can be at most $t$ corrupt
	 parties. This means that $P_h$ must have added $b$ to $\ACC_{hj}$ during some round $r'$, where $r' \leq t$. 
	 Thus, as argued in the previous case, each party $P_k$ also adds $b$ to $\ACC_{kj}$ by round $t + 1$ and hence time $(t + 1) \cdot \Delta$.
	\end{myitemize}
\end{proof}

\begin{lemma}
\label{lemma:PW-sync-validity}
Protocol $\PiPW$ achieves $\Z_s$-Validity in a synchronous network.
\end{lemma}
\begin{proof}
Suppose that all honest parties have the same input $b$. Corresponding to each {\it honest} party $P_j$, every {\it honest} 
$P_i$ sets $\ACC_{ij} = \{b\}$. This is because $P_i$ would receive a valid signature on $b$ from $P_j$ during round $0$, and adds $b$ to $\ACC_{ij}$. Further, $P_i$ will not add any $b' \neq b$ to $\ACC_{ij}$ during
 any of the rounds, since the adversary {\it cannot} forge a signature on $b'$ on the behalf of $P_j$. Thus, for each honest $P_j$, party 
  $P_i$ adds $(j,b)$ to $\FIN_i$. Let $Z^{\star} \in \Z_s$ be the set of {\it corrupt} parties and let $\Hon = \PartySet \setminus Z^{\star}$ be the set of {\it honest} parties. 
  Hence corresponding to every $P_j \in \Hon$, then value $(j, b)$ is added to $\FIN_i$ of every $P_i \in \Hon$ by time $(t + 1) \cdot \Delta$. 
  Moreover, since $\Q^{(2)}(\PartySet, \Z_s)$
  conditions holds, $\PartySet \setminus Z^{\star} \not \in \Z_s$. Consequently, as per the ``majority" rule, every party in $\Hon$ outputs $b$.  
\end{proof}

\begin{lemma}
\label{lemma:PW-liveness}
In protocol $\PiPW$, irrespective of the network type, all honest parties obtain an output at the time $(t+1)\cdot \Delta$.
\end{lemma}
\begin{proof}
The proof follows from the fact that irrespective of the network type, the parties compute an output (which could be $\bot$) at the local time $(t + 1) \cdot \Delta$.
\end{proof}

If the inputs of the parties are of size $\ell$ bits, then we invoke $\ell$ instances of $\PiPW$. The following lemma describes the communication cost incurred while doing this.

\begin{lemma}
\label{lemma:PW-communication}
If the inputs of the parties are of size $\ell$ bits, then 
 protocol $\PiPW$ incurs a communication of $\Order(n^4 \cdot \ell \cdot |\sigma|)$ bits from the honest parties.
\end{lemma}
\begin{proof}
During round $0$, each party signs its input and sends this to every other party, incurring a total communication of
 $\Order(\ell \cdot n^2 \cdot |\sigma|)$ bits. During the next $t$ rounds, 
  each party $P_i$ sends $(b, j,\SET_j)$ to every other party $P_k$ at most {\it once}. This is because $P_i$ sends this only if $b \notin \ACC_{ij}$ holds, and does not do this once it adds $b$ to $\ACC_{ij}$. Considering all possibilities for $b$, $i$, $j$ and $k$, and taking into account that $\SET_j$ will contain $\Order(n)$ signatures, the communication cost of this will be $\Order(\ell \cdot n^4 \cdot |\sigma|)$ bits.
\end{proof}

The proof of Lemma \ref{lemma:PW} now follows from Lemma \ref{lemma:PW-sync-consistency}-\ref{lemma:PW-communication}.
\subsection{Protocool $\Acast$: Asynchronous Broadcast with Synchronous Guarantees}
In this section, we prove the properties of the protocol $\Acast$ (see Fig \ref{fig:Acast} for the protocol description). \\~\\
\noindent {\bf Lemma \ref{lemma:Acast}.}
{\it Protocol $\Acast$ achieves the following properties.
 \begin{myitemize}
  \item[--] Asynchronous Network: The protocol is a $\Z_a$-secure broadcast protocol. 
  \item[--] Synchronous Network: {\bf (a) $\Z_s$-Liveness}: If $\Sender$ is honest, then all honest parties obtain an output
    within time $3\Delta$. {\bf (b) $\Z_s$-Validity}: If $\Sender$ is honest, then every honest party with an output, outputs $m$.
  {\bf (c) $\Z_s$-Consistency}: If $\Sender$ is corrupt and some honest party outputs $m^{\star}$ at time $T$,
   then every honest $P_i$ outputs $m^{\star}$ by the end of time $T+ \Delta$.
   \item[--] Communication Complexity: $\Order(n^3 \cdot \ell \cdot |\sigma|)$ bits are communicated by the honest parties, where
    $\ell$ is the size of $\Sender$'s input.
 \end{myitemize}
 }
\begin{proof}
We first consider a {\it synchronous} network, followed by an {\it asynchronous} network.
\paragraph{\bf Properties in the Synchronous Network.}
 Let $Z^{\star} \in \Z_s$ be the set of {\it corrupt} parties and let $\Hon = \PartySet \setminus Z^{\star}$ be the set of {\it honest} parties.
 Suppose $\Sender$ is {\it honest}. Then by time $\Delta$, each party in $\Hon$ receives $\sign{(\propose, m)}_\Sender$ from $\Sender$
 and no honest party ever receives $\sign{(\propose, m')}_\Sender$ from any party by time $2\Delta$, for any $m' \neq m$, since signature of an {\it honest} $\Sender$ {\it cannot} be forged. 
  Thus, every $P_j \in \Hon$ 
  sends $\sign{(\vote, m)}_j$ by time $2\Delta$.
  Consequently, by time $3\Delta$, every $P_i \in \Hon$ will have a quorum $\cert{m}$ of legitimately signed $\sign{(\vote, m)}_j$ messages corresponding to every $P_j \in \Hon$. 
  The parties $P_k \in Z^{\star}$ may send signed $\sign{(\vote, m')}_k$ messages where $m' \neq m$ and consequently
  the parties in $\Hon$ may also have a quorum $\cert{m'}$ of legitimately signed $\sign{(\vote, m')}_k$ messages corresponding to every $P_k \in Z^{\star}$.
  However, since $\Z_s$ satisfies the $\Q^{(2)}(\PartySet, \Z_s)$ condition, $\PartySet \setminus Z^{\star} = \Hon \not \in \Z_s$.
  Consequently, the parties in $\Hon$ outputs $m$ by time $3\Delta$, as the condition for outputting an $m' \neq m$ will be never satisfied for the parties in $\Hon$.
   This proves the $\Z_s$-liveness and $\Z_s$-validity. 

We now consider $\Sender$ to be {\it corrupt}. We first show that no two parties in $\Hon$ can vote for different messages. On the contrary,
  let $P_i \in \Hon$ sends $\sign{\vote, m')}_i$ at time $T_i$, and let $P_j \in \Hon$ sends $\sign{(\vote, m'')}_j$ at time $T_j$,
  where $T_j \geq T_i$. This implies that $P_i$ must have received $\sign{(\propose, m')}_\Sender$ from $\Sender$ within time $T_i - \Delta$, and would have sent
  $\sign{(\propose, m')}_\Sender$ to $P_j$. And $P_j$ would have received $\sign{(\propose, m')}_\Sender$ from $P_i$ within time $T_i$. Now since $T_i \leq T_j$, 
  it implies that $P_j$ {\it would not} have sent $\sign{(\vote, m'')}_j$ at time $T_j$, and this is a contradiction. 
  
  Now based on the above fact,  we proceed to prove that $\Z_s$-consistency holds. 
   Let $P_h \in \Hon$ outputs $m^\star$ at time $T$. This implies that at time $T$, there exists a subset $Z_{\alpha} \in \Z_s$, such that
   $P_h$ has a quorum $\cert{m^{\star}}$ of legitimately signed $\sign{(\vote, m)}_{j}$ messages, corresponding to every $P_j \in \PartySet \setminus Z_{\alpha}$. 
   Now since $\Z_s$ satisfies the $\Q^{(2)}(\PartySet, \Z_s)$ condition, it follows that $\Hon \cap (\PartySet \setminus Z_{\alpha}) \neq \emptyset$. This implies that there exists at least
   one party in $\Hon$, say $P_k$, who has voted for $m^{\star}$ by sending a $\sign{(\vote, m)}_{k}$ message. Consequently, no other party in $\Hon$ every votes for $m^{\star \star} \neq m^{\star}$.
   The parties in $Z^{\star}$ may vote for $m^{\star \star} \neq m^{\star}$. But since $\PartySet \setminus Z^{\star} \not \in \Z_s$, it follows that no party in $\Hon$ will ever have a sufficiently large
   quorum of legitimately signed vote messages for $m^{\star \star}$ to output $m^{\star \star}$.
   Since $P_h$ sends $\cert{m^\star}$ to all parties at time $T$, every other party in $\Hon$ will receive $\cert{m^\star}$ by time $T + \Delta$.
   Consequently, all the parties in $\Hon$ will output $m^{\star}$, latest by time $T + \Delta$.
\paragraph{\bf Properties in the Asynchronous Network.}
We now consider an asynchronous network. Let $Z^{\star} \in \Z_a$ be the set of {\it corrupt} parties and let
 $\Hon = \PartySet \setminus Z^{\star}$ be the set of {\it honest} parties. We first consider an {\it honest} $\Sender$. 
  Each party in $\Hon$ eventually receives $\sign{(\propose, m)}_\Sender$ from $\Sender$. Furthermore, no party in $\Hon$ ever receives $\sign{(\propose, m')}_\Sender$ from any party, 
  for any $m' \neq m$,  since the signature of an honest $\Sender$ cannot be forged. Hence, each party in $\Hon$ eventually sends a signed $\vote$ message for $m$, which is eventually
  delivered to every party in $\Hon$. The parties in $Z^{\star}$ may send signed $\vote$ messages for $m' \neq m$. 
  However, since $\Z_a$ satisfies the $\Q^{(3)}(\PartySet, \Z_a)$ condition, it follows that each party in $\Hon$ eventually outputs $m$ and the conditions for outputting $m' \neq m$ will be never satisfied
   for any party in $\Hon$. This proves the $\Z_a$-liveness and $\Z_a$-consistency.

Next, consider a {\it corrupt} $\Sender$.
  Let $P_h \in \Hon$  outputs $m^\star$. This implies there exists a subset $Z_{\alpha} \in \Z_s$, such that
   $P_h$ has a quorum $\cert{m^{\star}}$ of legitimately signed $\vote$ messages for $m^{\star}$, corresponding to 
   every party in $\PartySet \setminus Z_{\alpha}$. 
   Now consider an {\it arbitrary} $P_i \in \Hon$, where $P_i \neq P_h$. We claim that for any $Z \in \Z_s$, party 
  $P_i$ will {\it never} have a 
   quorum $\cert{m^{\star \star}}$ of legitimately signed $\vote$ messages for $m^{\star \star}$, corresponding to 
   the parties in $\PartySet \setminus Z$.
   On the contrary, let  $P_i$ eventually have a 
   quorum $\cert{m^{\star \star}}$ of legitimately signed $\vote$ messages for $m^{\star \star}$, corresponding to 
   every party in $\PartySet \setminus Z_{\beta}$, for some $Z_{\beta} \in \Z_s$.
   Now since the $\Q^{(2, 1)}(\PartySet, \Z_s, \Z_a)$ condition is satisfied, it follows that $\Hon \cap (\PartySet \setminus Z_{\alpha}) \cap (\PartySet \setminus Z_{\beta}) \neq \emptyset$. 
   This implies that there exists at least one party from $\Hon$, say $P_k$, such that $P_k$ has voted {\it both} for $m^{\star}$, as well as $m^{\star \star}$, which is a 
   contradiction. Consequently, $P_i$ will never output $m^{\star \star}$. 
   Now, since $P_h$ sends $\cert{m^\star}$ to all the parties, party $P_i$ eventually receives  $\cert{m^\star}$ and outputs $m^\star$.
   This proves $\Z_a$-consistency.
   
   Finally, the communication complexity follows from the fact that irrespective of the type of network, every party may have to send a quorum of up to $\Order(n)$ signed $\vote$
   messages, to every other party.
 \end{proof}
 \subsection{Properties of the Protocol $\BC$}
 In this section, we prove the properties of the protocol $\BC$ (see Fig \ref{fig:BC} for the formal description). \\~\\
\noindent {\bf Theorem \ref{thm:BC}.}
{\it Protocol  $\BC$ achieves the following,
   with a communication complexity of $\Order(n^4 \cdot \ell \cdot |\sigma|)$ bits,
    where $\TimeBC =  3\Delta + \TimePW$.   
 \begin{myitemize}
   \item[--] {\it Synchronous} network: 
        \begin{myitemize}
           \item[--]{\bf (a) $\Z_s$-Liveness}: At time $\TimeBC$, each honest party has an output. 
	    \item[--] {\bf (b) $\Z_s$-Validity}: If $\Sender$ is {\it honest}, then at time $\TimeBC$, each honest party
    outputs $m$.
           \item[--] {\bf (c) $\Z_s$-Consistency}: If $\Sender$ is {\it corrupt}, then the output of every honest party is the same at 
      time $\TimeBC$.     
            \item[--]   {\bf (d) $\Z_s$-Fallback Consistency}: If $\Sender$ is {\it corrupt} and some honest 
     party outputs $m^{\star} \neq \bot$ at time
    $T$ through fallback-mode, then every honest party outputs $m^{\star}$ by 
    time $T + \Delta$.    
      \end{myitemize}
\item[--] {\it Asynchronous Network}:
   \begin{myitemize}
     \item[--] {\bf (a) $\Z_a$-Liveness}: At time $\TimeBC$, each honest party has an output.
    \item[--] {\bf (b) $\Z_a$-Weak Validity}: If $\Sender$ is {\it honest}, then at  time $\TimeBC$, 
    each honest party
    outputs $m$ or $\bot$.
    \item[--] {\bf (c) $\Z_a$-Fallback Validity}: If $\Sender$ is {\it honest}, then each honest party
     with output $\bot$ at time $\TimeBC$, eventually outputs 
    $m$ through fallback-mode.
    \item[--] {\bf (d) $\Z_a$-Weak Consistency}: If $\Sender$ is {\it corrupt}, then there exists some $m^{\star} \neq \bot$, such that 
   at time $\TimeBC$, 
    each honest party
    outputs either 
     $m^{\star}$ or $\bot$.
    \item[--] {\bf (e) $\Z_a$-Fallback Consistency}: If $\Sender$ is {\it corrupt}, and some honest party 
    outputs $m^{\star} \neq \bot$ at  time
    $T$ where $T \geq \TimeBC$, then 
    each honest party eventually outputs $m^{\star}$.
   \end{myitemize}
   \end{myitemize}
}
\begin{proof}
The $\Z_s$-liveness and $\Z_a$-liveness properties follow from the fact that 
 every honest party outputs something (including $\bot$) at (local)
  time $\TimeBC$, irrespective of the type of the network.
   We next prove the rest of the properties of the protocol in the {\it synchronous} network.
\paragraph{\bf Properties in the Synchronous Network.}
If $\Sender$ is {\it honest}, then due to the {\it $\Z_s$-liveness} and
  {\it $\Z_s$-validity} properties of $\Acast$ in the {\it synchronous} network (see Lemma \ref{lemma:Acast}), all honest
   parties receive $m$ from the Acast of $\Sender$
   at time $3\Delta$. Consequently, all honest parties participate with input $m$ in the instance of $\PiPW$.
  The {\it $\Z_s$-guaranteed liveness} and {\it $\Z_s$-validity} properties
   of $\PiPW$ in the {\it synchronous} network (see Lemma \ref{lemma:PW}) guarantees that at time $3\Delta + \TimePW$, all
    honest parties will have
  $m$ as the output from the instance
  of $\PiPW$. As a result, all honest parties output $m$ at time $\TimeBC$, thus proving the {\it $\Z_s$-validity} property.
  
   To prove the {\it $\Z_s$-consistency} property, we consider a {\it corrupt} $\Sender$.
  From the {\it $\Z_s$-consistency} property of $\PiPW$ in the {\it synchronous} network (see Lemma \ref{lemma:PW}), 
  all honest parties will have the {\it same} output from the instance of $\PiPW$ at time $\TimeBC$.
    If {\it all} honest parties have the output $\bot$ for $\BC$ at time $\TimeBC$, then {\it $\Z_s$-consistency}
     holds trivially. So, consider the case when
    some {\it honest} party, say $P_i$, has the output $m^{\star} \neq \bot$ for $\BC$ at time $\TimeBC$. 
    This implies that all honest parties have the output $m^{\star}$ from the instance of $\PiPW$. 
    Moreover, at time $3\Delta$, at least one {\it honest} party, say $P_h$, has received $m^{\star}$ from the Acast of $\Sender$. 
    If the latter does not hold, then all honest parties would have participated
     with input $\bot$ in the instance of $\PiPW$, and from the {\it $\Z_s$-validity} of $\PiPW$
      in the {\it synchronous}
    network (see Lemma \ref{lemma:PW}), all honest  parties would compute $\bot$ as the output during the instance of $\PiPW$, which is a contradiction.
    Since $P_h$ has received
     $m^{\star}$ from $\Sender$'s Acast at time $3\Delta$, it follows from the {\it $\Z_s$-consistency} property of
      $\Acast$ in the {\it synchronous}
    network (see Lemma \ref{lemma:Acast}) that {\it all} honest parties will receive
     $m^{\star}$ from $\Sender$'s Acast by time $4\Delta$. Moreover, $4\Delta < 3\Delta + \TimePW$ holds.
    Consequently, at time $3 \Delta + \TimeBC$,
     {\it all} honest parties will have $m^{\star}$ from $\Sender$'s Acast {\it and} as the output of
    $\PiPW$, implying that all honest parties output $m^{\star}$ for $\BC$.   
  
   We next prove the  {\it $\Z_s$-fallback consistency} property for which
    we again consider a {\it corrupt} $\Sender$. 
    Let $P_h$ be an {\it honest} party who outputs $m^{\star} \neq \bot$ at time $T$ through fallback-mode.
    Note that $T > \TimeBC$, as the output during the fallback-mode is computed only after time $\TimeBC$.  
     We also note that {\it each} honest party has output $\bot$ at 
    time $\TimeBC$. This is because, from the proof of the {\it $\Z_s$-consistency} property of 
    $\BC$ (see above), if any {\it honest} party has an output 
    $m' \neq \bot$ at time $\TimeBC$, then {\it all} honest parties (including $P_h$) must have computed the output 
    $m'$ at time $\TimeBC$. Hence, $P_h$ will never change its output to
    $m^{\star}$.\footnote{Recall that in the protocol $\BC$, the parties who obtain
    an output different from $\bot$ at time $\TimeBC$, never change their output.}  
    Now since $P_h$ has obtained the output $m^{\star}$, it implies that at time $T$, it has received $m^{\star}$ from the Acast of
    $\Sender$. It then follows from the {\it $\Z_s$-consistency} of 
    $\Acast$ in the {\it synchronous} network that every honest party will also receive 
    $m^{\star}$ from the Acast of $\Sender$, latest by time $T + \Delta$ and output $m^{\star}$.
    This completes the proof of all the properties in the {\it synchronous} network.
    \paragraph{\bf Properties in the Asynchronous Network.}
      The {\it $\Z_a$-weak validity} property follows from the {\it $\Z_a$-validity} property of 
    $\Acast$ in the {\it asynchronous} network (see Lemma \ref{lemma:Acast}), which ensures that
    no honest party ever receives an $m'$ from the Acast of $\Sender$, where $m' \neq m$.
    So, if at all any honest party outputs a value different from $\bot$ at time $\TimeBC$, it has to be $m$.
       The {\it $\Z_a$-weak consistency} property follows using similar arguments
   as used to prove {\it $\Z_s$-consistency} in the {\it synchronous} network;
   however we now rely on the {\it $\Z_a$-validity} and 
     {\it $\Z_a$-consistency}
     properties of $\Acast$ in the asynchronous network (see Lemma \ref{lemma:Acast}). 
     The latter property ensures  that
    for a {\it corrupt} $\Sender$, two different honest parties never end up receiving $m_1$ and $m_2$ from the Acast of $\Sender$, where $m_1\neq m_2$. 
    
     For the {\it $\Z_s$-fallback validity} property, consider an {\it honest} $\Sender$, and let 
     $P_i$ be an arbitrary {\it honest} party
      who outputs $\bot$ at (local) time $\TimeBC$. Since the parties keep on participating in the protocol beyond time $\TimeBC$,
       it follows from the {\it $\Z_a$-liveness} and
        {\it $\Z_a$-validity} properties of $\Acast$ in the {\it asynchronous} network (see Lemma \ref{lemma:Acast}) that party $P_i$ will {\it eventually} 
     receive $m$ from the Acast of $\Sender$, by executing the steps of the fallback-mode of $\BC$. 
     Consequently, party $P_i$ eventually changes its output from
     $\bot$ to $m$.

For the {\it $\Z_a$-fallback consistency} property, we consider a {\it corrupt} $\Sender$. Let 
     $P_j$ be an {\it honest} party who outputs some $m^{\star}$ different from $\bot$ at time $T$, where $T \geq \TimeBC$.
      This implies that $P_j$ has
      obtained $m^{\star}$ from the Acast of $\Sender$.
           Now, consider an arbitrary {\it honest} $P_i$. From the
            {\it $\Z_a$-liveness} and {\it $\Z_a$-weak consistency} properties of 
            $\BC$ in {\it asynchronous} network proved above,
     it follows that $P_i$ outputs either $m^{\star}$ or $\bot$ at local time $\TimeBC$. 
     If $P_i$ has output $\bot$, then from the {\it $\Z_a$-consistency}
     property of $\Acast$ in the {\it asynchronous} network (see Lemma \ref{lemma:Acast}), it follows that 
     $P_i$ will also eventually obtain $m^{\star}$ from the Acast of $\Sender$, by executing the steps of the fallback-mode of $\BC$. 
     Consequently, party $P_i$ eventually changes its output from
     $\bot$ to $m^{\star}$.
     
       The {\it communication complexity} (both in the synchronous as well as asynchronous network)
      follows from the communication complexity of $\PiPW$ and $\Acast$.
\end{proof} 
  \subsection{Properties of the Protocol $\SBA$}
   In this section, we prove the properties of the protocol $\SBA$ (see Fig \ref{fig:SBA} for the formal description). \\~\\
   \noindent {\bf Theorem \ref{thm:SBA}.}
   {\it Protocol $\SBA$ achieves the following where $\TimeSBA = \TimeBC$,  
 incurring a communication of 
 $\Order(n^5 \cdot |\sigma|)$ bits.   
 \begin{myitemize}
   \item[--] {\it Synchronous Network}: the protocol is a $\Z_s$-secure SBA protocol where honest parties have an output, different from $\bot$, at time $\TimeSBA$.
   \item[--] {\it Asynchronous Network}: the protocol achieves $\Z_a$-guaranteed liveness and $\Z_a$-weak validity, such that all honest parties have an output at (local) time $\TimeSBA$.
  \end{myitemize}
}
\begin{proof}
 The communication complexity simply follows from the fact that $n$ instances of $\BC$ are invoked in the protocol.
The {\it guaranteed liveness}, both in the synchronous and asynchronous network trivially follows from 
 the {\it $\Z_s$-liveness} and {\it $\Z_a$-liveness} of $\BC$ (see Theorem \ref{thm:BC}), which ensures that all the $n$ instances of $\BC$ produce {\it some} output
  within (local) time $\TimeBC$, both in a synchronous as well as an asynchronous network, through {\it regular} mode.
  Hence, at local time $\TimeSBA = \TimeBC$, all honest parties will have some output.
   We next prove the rest of the properties in a {\it synchronous} network.
\paragraph{\bf Properties in the Synchronous Network.}
Let $Z^{\star} \in \Z_s$ be the set of {\it corrupt} parties and let $\Hon = \PartySet \setminus Z^{\star}$ be the set of {\it honest} parties. 
 In a {\it synchronous} network, the instances of $\BC$ corresponding to the senders in $\Hon$, 
  result in an output different from $\bot$ for all honest parties (follows from the {\it $\Z_s$-validity} property of $\BC$ in the {\it synchronous}
  network, Theorem \ref{thm:BC}). Hence $\Hon \subseteq \R$ will hold.
    Moreover, 
  from the {\it $\Z_s$-consistency} property of $\BC$ in the {\it synchronous} network (Theorem \ref{thm:BC}),
   all the parties in $\Hon$ obtain a common output from the $i^{th}$ instance of $\BC$, for $i = 1, \ldots, n$, at time $\TimeBC$. 
  Hence, all honest parties output the {\it same} value,
  different from $\bot$, at time $\TimeBC$, proving the {\it $\Z_s$-consistency} of $\SBA$.
   Finally, if all parties in $\Hon$ have the same input bit $b$, then only the instances of $\BC$ corresponding to the parties in $\R \setminus \Hon$ may output
   $\bar{b} = 1 - b$. However, $\R \setminus \Hon \in \Z_s$. Moreover, $\R \setminus Z^{\star} \not \in \Z_s$ (since $\Z_s$ satisfies the $\Q^{(2)}(\PartySet, \Z_s)$ condition). 
   It then follows that all honest parties output $b$, proving {\it $\Z_s$-validity} of $\SBA$.
\paragraph{\bf Properties in the Asynchronous Network.}
 Let $Z^{\star} \in \Z_a$ be the set of {\it corrupt} parties and let $\Hon = \PartySet \setminus Z^{\star}$ be the set of {\it honest} parties.
   Suppose all the parties in $\Hon$ have the same input bit $b$. Let 
  $P_i \in \Hon$ be an {\it arbitrary} party, that obtains an output $c$, {\it different} from $\bot$, at time $\TimeBC$. This implies that there exists a subset of parties $\R$ for $P_i$, where $\PartySet \setminus \R \in \Z_s$, such that
  $P_i$ has obtained a Boolean output $b_i^{(j)}$ from the $\BC$ instances, corresponding to every $P_j \in \R$. Moreover, there also exists a subset of parties
    $\R_i \subseteq \R$, where $\R \setminus \R_i \in \Z_s$, such that
  the output $b_i^{(j)} = c$, corresponding to every $P_j \in \R_i$. Now since the $\Q^{(2, 1)}(\PartySet, \Z_s, \Z_a)$ 
   condition is satisfied, it follows that $\Z_a$ satisfies the $\Q^{(1)}(\R_i, \Z_a)$ condition and hence 
  $\R_i \cap \Hon \neq \emptyset$. Consequently, $c = b$ holds. 
  This proves the {\it $\Z_a$-weak validity} in the asynchronous network.  
\end{proof}

\subsection{Protocol $\Vote$: Asynchronous Graded Agreement with Synchronous Validity}
 To design protocol $\Vote$, we first design a sub-protocol $\Prop$ for proposing values. 
 \subsubsection{$\Prop$: A Network Agnostic Protocol for Proposing Values}
 Protocol $\Prop$ takes an input value from each party from the set $\{0, 1, \lambda\}$
  and outputs a set $\prop$ of proposed values for each party. Liveness is ensured in an {\it asynchronous} network as long as each honest party holds one of {\it two} inputs. In an asynchronous network, it will be ensured that each value in the output $\prop$ must be the input of some {\it honest} party. Moreover, if any two honest parties output a singleton set for $\prop$, then they must output the {\it same} set. In a {\it synchronous} network, validity and liveness are ensured as long as each honest party participates with the same input. Protocol $\Prop$ is presented in figure \ref{fig:Prop}.

\begin{protocolsplitbox}{$\Prop(\PartySet, \Z_s, \Z_a)$}{Sub-protocol to propose values. The above code is executed by every party $P_i$ with input $v \in \{0, 1, \lambda\}$}{fig:Prop}
\justify
\begin{myenumerate}
\item Set $\vals = \prop = \emptyset$.
\item On having the input $v \in \{0, 1, \lambda\}$, send $(\prepare, v)$ to every party $P_j \in \PartySet$.
\item On receiving $(\prepare,b)$ for some $b \in \{0,1,\lambda\}$ from a set of parties $S^{(b)}$ satisfying $\Q^1(S^{(b)}, \Z_s)$ condition, 
 send $(\prepare,b)$ to all $P_j \in \PartySet$, if not sent earlier.
\item Upon receiving the message $(\prepare, b)$ for some $b \in \{0,1,\lambda\}$ from parties in set $\PrepSet^{(b)} = \PartySet \setminus Z$ for some $Z \in \Z_s$, set $\vals = \vals \cup \{b\}$. 
\item Upon adding the first value $b$ to $\vals$, send $(\propose,b)$ to every party $P_j \in \PartySet$.
\item Upon receiving $(\propose, b)$ messages from a set of parties $\PropSet = \PartySet \setminus Z$ for some $Z \in \Z_s$ on values $b \in \vals$, let
 $\prop \subseteq \vals$ be the set of values carried by those messages. 
  Output $\prop$.
\end{myenumerate}
\end{protocolsplitbox}

The guarantees provided by $\Prop$ are proven in a series of lemmas below. In the below proofs, we assume that $Z^\star$ is the set of corrupt parties.

\begin{lemma}
\label{prop-async-consistency}
Suppose that the network is asynchronous. If two honest parties $P_i$ and $P_j$ output $\{b\}$ and $\{b'\}$ respectively, then $b = b'$.
\end{lemma}
\begin{proof}
 Since $P_i$ outputs $\{b\}$, it must have received $(\propose, b)$ from a set of parties $\PropSet = \PartySet \setminus Z$ for some $Z \in \Z_s$.
  Similarly, since $P_j$ outputs $\{b'\}, it$ must have received $(\propose, b')$ from a set of parties $\PropSet' = \PartySet \setminus Z'$ for some $Z' \in \Z_s$.
  Let $\PropSet_\Hon$ and $\PropSet_\Hon$ be the set of {\it honest} parties in $\PropSet$ and $\PropSet'$ respectively. Since $\Z_s, \Z_a$ satisfy the $\Q^{(2, 1)}(\PartySet, \Z_s, \Z_a)$
  condition and $Z^{\star} \in \Z_a$, it follows
  that  $\PropSet_\Hon \cap \PropSet'_\Hon \neq \emptyset$. Let $P_h \in \PropSet_\Hon \cap \PropSet'_\Hon$.
  If $b \neq b'$, this would mean that $P_h$ has sent both $(\propose, b)$ and $(\propose, b')$, which is a contradiction, since an honest party sends
   {\it at most} one $\propose$ message as per the protocol.
\end{proof}
\begin{lemma}
\label{prop-async-validity}
Suppose that the network is asynchronous. If no honest party has input $v$, then no honest party outputs $\prop$ containing $v$.
\end{lemma}
\begin{proof}
If $v$ was not input by any honest party, then no honest party sends $(\prepare,v)$ in step $2$. Hence, no honest party receives $(\prepare,v)$ from a set of parties $S^{(v)}$ 
 which satisfies the $\Q^1(S^{(v)}, \Z_s)$ condition during step $3$, since such a set must contain at least one honest party. Consequently, no honest party sends $(\propose,v)$. Thus, no honest party adds $v$ to $\vals$ and no honest party outputs $\prop$ containing $v$.
\end{proof}

The following lemmas help prove liveness.

\begin{lemma}
\label{prop-async-liveness}
Suppose that the network is asynchronous. If all honest parties hold one of two inputs, say $v_0$ and $v_1$, then all honest parties eventually compute
  an output.
\end{lemma}
\begin{proof}
We first show that every honest party eventually sends a $\propose$ message. 
  Let $\Hon_0$ and $\Hon_1$ be the sets of {\it honest} parties holding inputs $v_0$ and $v_1$ respectively
   and let $\Hon$ be the set of {\it honest} parties.
   We know that due to the $\Q^{(2, 1)}(\PartySet, \Z_s, \Z_a)$ condition, $\Z_s$ satisfies the
   $\Q^{(2)}(\Hon, \Z_s)$ condition. Now, consider the following cases. 
	\begin{myitemize}
	\item[--] {\bf Case 1 - $\Hon_0$ satisfies the $\Q^1(\Hon_0,\Z_s)$ condition}: In this case, $\Hon_0$ is a candidate for the set $S^{(v_0)}$, since all 
	the parties in $\Hon_0$ send $(\prepare,v_0)$ in step $2$.
	\item[--] {\bf Case 2 - $\Hon_0$ does not satisfy the $\Q^1(\Hon_0,\Z_s)$ condition}: In this case, $\Z_s$
	 must satisfy the $\Q^1(\Hon_1,\Z_s)$ condition, since $\Z_s$ satisfies the
   $\Q^{(2)}(\Hon, \Z_s)$ condition. Then similar to what was argued in the previous case, $\Hon_1$ is a candidate for the set $S^{(v_1)}$.
	\end{myitemize}
Let $v_b$ be the input corresponding to the candidate set $S^{(v_b)}$. Every honest party will now eventually send $(\prepare, v_b)$ in step $3$. 
 This would mean that the set of honest parties $\Hon = \PartySet \setminus Z^\star$ is a candidate for the set $\PrepSet^{(v_b)}$.
  Thus, every honest party eventually adds $v_b$ to $\vals$ and sends a $\propose$ message for some value $v_b$.
   This way, every honest party eventually receives $\propose$ messages from every other honest party for $v_b$ and thus, the set $\Hon$ also forms a candidate for the set $\PropSet$. Thus, all 
   honest parties eventually compute an output. 
\end{proof}

\begin{lemma}
\label{prop-sync-validity}
If the network is synchronous and if all honest parties participate with input $v$, then all honest parties output $\prop = \{v\}$ at time $2\Delta$.
\end{lemma}
\begin{proof}
Let $\Hon = \PartySet \setminus Z^\star$ be the set of {\it honest} parties. Then 
 in step $2$, every party in $\Hon$ sends $(\prepare,v)$ to every party, which is delivered within $\Delta$ time. 
 Thus, at time $\Delta$, all the parties in $\Hon$ receive $(\prepare,v)$ from the parties in $\Hon$ and add $v$ to $\vals$. Further, no other value $v^{\star}$ will be 
  added to $\vals$, since only the parties in $Z^{\star}$ may send a $\prepare$ message for $v^{\star}$ and $\Z_s$ 
  {\it does not} satisfy the $\Q^{(1)}(Z^\star, \Z_s)$ condition. 
  Thus, all the parties in $\Hon$ send $(\propose, v)$ in step $5$, which gets delivered to all the parties in $\Hon$ at time $2\Delta$.
   Consequently, all the parties in $\Hon$ output $\prop = \{v\}$.
\end{proof}

\begin{lemma}
\label{prop-communication}
Protocol $\Prop$ incurs a communication of $\Order(n^2)$ bits.
\end{lemma}
\begin{proof}
The proof follows from the fact that each party sends $(\prepare,b)$ to every other party at most once for any value of $b \in \{0,1,\lambda\}$.
\end{proof}

\subsubsection{The Graded Agreement Protocol}
We now present protocol $\Vote$ (Figure \ref{fig:GA}) based on protocol $\Prop$. The protocol cleverly ``stitches" together two instances of $\Prop$, by defining the input for the second instance based on the output from the first instance. Each party, with input either $0$ or $1$, participates in the first instance of $\Prop$ with their input. Since the parties participate with one of two inputs, this instance will
  eventually complete (in an {\it asynchronous} network), 
  and the parties obtain an output, say $\prop_1$. Only if $\prop_1$ is a singleton set, say $\{b\}$, for some party, then that party participates in the second instance of $\Prop$ with the input $b$. Otherwise, it participates in the second instance with a {\it default} input of $\lambda$. Since no two honest parties can output different singleton sets for $\prop_1$, this ensures that each honest party participates with an input of either $b$ or
   $\lambda$ (in an {\it asynchronous} network).
    Thus, the second instance of $\Prop$ also eventually completed with an output, say $\prop_2$. This also ensures that $\prop_2$ can contain only values $b$ and $\lambda$. If $\prop_2$ contains {\it only} $b$ for some honest party, then that party outputs $b$ with a grade of $2$. If $\prop_2$ contains $b$ along with $\lambda$, then the party outputs $b$ with a grade of $1$. 
    Else, if $\prop_2$ contains only $\lambda$, then the party outputs $\bot$ with a grade of $0$.  If the network is {\it synchronous} and all honest parties start the protocol with the same input, then
     {\it both} $\prop_1$
    as well as $\prop_2$ will be a {\it singleton} set containing that value and hence all honest parties will output that value with the highest grade.

\begin{protocolsplitbox}{$\Vote(\PartySet, \Z_s, \Z_a)$}{Asynchronous graded agreement with synchronous validity. The above code is executed by every party $P_i$ with input $v \in \{0,1\}$.}{fig:GA}
\justify
\begin{myenumerate}
\item On having the input $v \in \{0, 1 \}$, set $b_1 = v$. Participate in an instance of the protocol $\Prop$ with input $b_1$ and wait for its completion. Let 
 $\prop_1$ be the output computed during the instance of $\prop$.
\item If $\prop_1 = \{b\}$ for some $b \in \{0, 1 \}$, then set $b_2 = b$. Else, set $b_2 = \lambda$. Then
 participate in an instance of the protocol $\Prop$ with input $b_2$ and wait for its completion. Let $\prop_2$ be the output computed from this instance of $\prop$.
\item If $\prop_2 = \{b'\}$ and $b' \neq \lambda$, then output $(b', 2)$. If $\prop_2 = \{b', \lambda\}$ where $b' \in \{0, 1 \}$, then output $(b', 1)$. Else, if $\prop_2 = \{\lambda\}$, then 
output $(\bot, 0)$.
\end{myenumerate}
\end{protocolsplitbox}

We now proceed to prove the properties of the protocol $\Vote$.
\begin{lemma}
\label{lemma:GASyncProperties}
Protocol $\Vote$ achieves the following in a synchronous network, where $\TimeVote = 4\Delta$.
\begin{myitemize}
  \item[--] {\bf (a) $\Z_s$-Liveness}: 
   If all honest parties participate in the protocol with the {\it same} input, then at time $\TimeVote$,
    all honest parties obtain an output.
   \item[--] {\bf (b) $\Z_s$-Graded Validity}: If every honest party's input is $b$, then all honest parties with an output, output $(b, 2)$.   
 \end{myitemize}  
\end{lemma}
\begin{proof}
If all honest parties participate in the protocol $\Vote$ with the same input $b$, then
 from Lemma \ref{prop-sync-validity}, all honest parties output $\prop_1 = \{b\}$ at time $2\Delta$.
   Thus, all honest parties participate with input $b_2 = b$ in the second instance of $\Vote$ and,
    once again from Lemma \ref{prop-sync-validity}, output $\prop_2 = \{b\}$ at time $4 \Delta$. Thus, all honest parties output $(b, 2)$.
\end{proof}
\begin{lemma}
\label{lemma:GAAsyncProperties}
Protocol $\Vote$ achieves the following in an asynchronous network. 
  \begin{myitemize}
     \item[--] {\bf (a) $\Z_a$-Liveness}:  If all honest parties participate in the protocol with a binary input, then each honest party eventually obtains an output. 
      \item[--] {\bf (b) $\Z_a$-Graded Validity}:  If every honest party's input is $b$, then all honest parties with an output, output $(b, 2)$.
       \item[--] {\bf (c) $\Z_a$-Graded Consistency}:  If two honest parties output grades $g, g'$, then $|g - g'| \leq 1$ holds; moreover,  
      if two honest parties output $(v, g)$ and $(v', g')$ with $g, g' \geq 1$, then $v = v'$.
  \end{myitemize}   
\end{lemma}
 \begin{proof}
 Since each honest party participates with a binary input, 
  from Lemma \ref{prop-async-liveness}, each party eventually outputs some value for $\prop_1$ during the first instance of $\prop$. Now there are two possible cases.
	\begin{myitemize}
	\item[--] {\bf Case 1 - Some honest party outputs $\{b\}$ as its value for $\prop_1$ where $b \in \{0, 1 \}$}: 
	From Lemma \ref{prop-async-consistency}, no honest party can output $\{b'\}$, where $b \neq b'$, as $\prop_1$. 
	Thus, each honest party participates with input either $b$ or $\lambda$ for the second instance of $\prop$. 
	\item[--] {\bf Case 2 - No honest party outputs $\{b\}$ as its value for $\prop_1$ for any $b \in \{ 0, 1\}$}: In this case, all honest parties participate with input $\lambda$ in the second instance of $\prop$.
	\end{myitemize}
In either case, the honest parties participate in the second instance of $\Prop$ with no more than {\it two} different inputs. 
Thus, from Lemma \ref{prop-async-liveness}, all parties eventually compute some value for $\prop_2$ during the second instance of $\prop$ and hence compute some 
 output for protocol $\Vote$. This proves the $\Z_a$-Liveness.
 
 We next prove the $\Z_a$-Graded Consistency. We first show that the grades output by any two parties differ by at most $1$. For this, 
  suppose that some {\it honest} party $P_i$ outputs $(b, 2)$. We show that no other honest party $P_j$ can output $(\bot, 0)$. 
   Since $P_i$ output $(b, 2)$, from Lemma \ref{prop-async-consistency}, $P_j$ cannot output $\prop_2 = \{\lambda\}$. Thus, $P_j$ cannot output $(\bot, 0)$.
 Next, we show that any two {\it honest} parties which output {\it non-zero} grades must output the same value. Similar to what was argued for the proof of $\Z_a$-Liveness,
  there exists a bit $b$ such that each honest party participates in $\Prop$ with input $b$ or $\lambda$ during step $2$. 
  Thus, $\prop_2 \subseteq \{b,\lambda\}$ for every honest party. This means that any honest party which outputs a non-zero grade must output it along with the bit $b$.
  
  We finally prove the $\Z_a$-Graded Validity. Suppose that each honest party participates with the same input bit $b$. 
   From Lemma \ref{prop-async-liveness}, we know that all honest parties output some value for $\prop_1$. From Lemma \ref{prop-async-validity}, all honest parties must output $\prop_1 = \{b\}$. 
    Hence, all honest parties participate in $\Prop$ in step $2$ with input $b$. By the same argument, all honest parties output $\prop_2 = \{b\}$. Hence, all honest parties output $(b,2)$.
 \end{proof}
\begin{lemma}
\label{GA-communication}
Protocol $\Vote$ incurs a communication of $\Order(n^2)$ bits.
\end{lemma}
\begin{proof}
The proof follows from Lemma \ref{prop-communication}, since $\Prop$ is invoked twice in the protocol.
\end{proof}
\subsection{Properties of the Protocol $\ABA$}
In this section, we prove the properties of the protocol $\ABA$ (see Fig \ref{fig:ABA} for the protocol steps). We start with the properties in the {\it asynchronous} network first, which mostly follows from \cite{Cho23}
 and are recalled from \cite{Cho23}.
 We start with the validity property.
 \begin{lemma}
   \label{lemma:ABAAsyncValidity}
    In protocol $\ABA$, if the network is asynchronous and all honest parties have the same input bit $b$, then
    all honest parties eventually output $b$. 
\end{lemma} 
\begin{proof}
   Let $Z^{\star} \in \Z_a$ be the set of {\it corrupt} parties.
    If every honest party has the same input bit $b$,
    then from the {\it $\Z_a$-Graded Validity} of $\Vote$ in the {\it asynchronous} network (Lemma \ref{lemma:GAAsyncProperties}),
     all honest parties eventually output $(b, 2)$ at the end of the first as well as the second instance of the $\Vote$ protocol
   during the first iteration. Consequently, every {\it honest} party eventually sends a signed $(\ready, b)$ message to all the parties and only the parties
    in $Z^{\star}$ may send a signed $(\ready, \overline{b})$ message. It now follows easily from the steps of the output computation stage that {\it no} honest party ever
    sends a signed $(\ready, \overline{b})$ message and all hence honest parties eventually output $b$.
\end{proof}
We next prove the {\it consistency} property. 
\begin{lemma}
 \label{lemma:ABAAsyncConsistency}
    In protocol $\ABA$, if the network is asynchronous and if any honest party outputs $b$, then
    every other honest party eventually outputs $b$.
\end{lemma} 
\begin{proof}
   Let $P_i$ be the {\it first} {\it honest} party who sends a signed $\ready$ message for some bit $b \in \{0, 1 \}$, during some iteration, say iteration $r$.
   We show that {\it no} honest party ever sends a signed $(\ready, \overline{b})$ message during iteration $r$ or in the subsequent iterations. 
        Since $P_i$ has sent a signed $\ready$ message for $b$, it implies that $P_i$ outputs $(b, 2)$ in the second instance of the $\Vote$ protocol during iteration $r$ and sets $\committed$ to $\true$.
     Then, from the {\it $\Z_a$-Graded Consistency} of 
     $\Vote$ in the {\it asynchronous} network (Lemma \ref{lemma:GAAsyncProperties}),
      every other honest party
    outputs either $(b, 2)$ or $(b, 1)$ in the second instance of the $\Vote$ protocol during iteration $r$. Consequently, no other 
    honest party sends the signed $(\ready, \overline{b})$ message during iteration $r$. Also, from the protocol steps, all honest parties 
    update their input
    to $b$ for the next iteration. This further implies that all honest parties will continue to input $b$
     to each subsequent invocation of $\Vote$, ignoring the output of $\CoinFlip$, for as long as they continue running. 
     Consequently, no honest party ever sends a signed $(\ready, \overline{b})$ message.
    
    Now let some {\it honest} party, say $P_h$, computes the output $b$ during iteration $k$. This
    implies that $P_h$ receives the signed $(\ready, b)$ message from a set of parties, say ${\cal T}$, such that
    $\PartySet \setminus {\cal T} \in \Z_s$. The set ${\cal T}$ is bound to have at least one {\it honest} party, due to the $\Q^{(2, 1)}(\PartySet, \Z_s, \Z_a)$ condition, implying that
    at least one honest party has sent a signed $(\ready, b)$ message, either during the iteration $k$ or some previous iteration.
    From the protocol steps, $P_h$ sends $\C(b)$, the set of signed $(\ready, b)$ messages of the parties in ${\cal T}$ to all other parties, which get eventually delivered. 
     Moreover, as shown above,
    no honest party will ever send a signed $(\ready, \overline{b})$ message.     
        Consequently, every honest party
        eventually receives sufficiently many numbers of signed $(\ready, b)$ messages 
        and outputs $b$. 
       \end{proof}

We next prove that at the end of each iteration, the updated value of all honest parties will be the same with the probability 
 at least $\frac{1}{2n}$.
\begin{lemma}
	\label{lemma:ABACorrectness}
	In protocol $\ABA$, if the network is asynchronous and if 
         all honest parties participate during iteration $k$, then with probability at least $\frac{1}{2n}$, all honest parties have the same updated bit
		$b$ at the end of iteration $k$.
\end{lemma}
\begin{proof}
 To prove the lemma statement, we consider an event $\Agree$, which denotes that all honest parties have the same input for the second instance
   of $\Vote$ during iteration $k$. If the event $\Agree$ occurs, then from the {\it $\Z_a$-Graded Validity} of $\Vote$ in the {\it asynchronous} 
    network (Lemma \ref{lemma:GAAsyncProperties}), 
    all honest parties will have the same
   updated bit at the end of iteration $k$. We show that the event $\Agree$ occurs during iteration $k$ with a probability of at least $\frac{1}{2n}$.
    For this, we consider two different possible cases with respect to the output from the first instance of $\Vote$ during iteration $k$.
   \begin{myitemize}
   \item[--] {\bf Case I: No honest party obtains an output $(b, 2)$ for any $b \in \{0, 1 \}$ during the first instance of $\Vote$.}
   In this case, all honest parties set the output from the instance of $\CoinFlip$ during iteration $k$ as the input for the second instance of $\Vote$.
   From the {\it $(\Z_a, p)$-commonness} of $\CoinFlip$ in {\it asynchronous} network \cite{Cho23}, 
    all honest parties will have the same output bit $\Coin_k$ from the instance of $\CoinFlip$ with a probability of at least
   $p = \frac{1}{n} > \frac{1}{2n}$.   
   \item[--] {\bf Case II: Some honest party obtains an output $(b, 2)$ during the first instance of $\Vote$.}
   In this case, the {\it $\Z_a$-Graded Consistency} of $\Vote$ in the {\it asynchronous} 
    network (Lemma \ref{lemma:GAAsyncProperties})   
    ensure that all honest parties obtain the output
   $(b, 2)$ or $(b, 1)$ from the first instance of $\Vote$. Moreover, from the protocol steps, the output of the
    instance of $\CoinFlip$ during iteration $k$ is {\it not} revealed, until the first honest party generates an output from the first instance of $\Vote$ during iteration
    $k$. Consequently, the output bit $b$ from the first instance of $\Vote$ is {\it independent} of the output of $\CoinFlip$. 
    From the {\it $(\Z_a, p)$-commonness} of $\CoinFlip$ in {\it asynchronous} network \cite{Cho23}, 
    all honest parties will have the same output bit $\Coin_k$ from the instance of $\CoinFlip$ with a probability of at least
   $p = \frac{1}{n}$.  Then the probability that $\Coin_k = b$ holds is at least $\frac{1}{2} \cdot \frac{1}{n} = \frac{1}{2n}$ and with this probability,
    all honest
     parties will have the same input for the second instance of $\Vote$.     
   \end{myitemize}   
\end{proof}
We next derive the expected number of iterations required in the protocol $\ABA$ for the honest parties to produce an output. 
 This automatically gives the expected running time in an asynchronous network, since each iteration takes a constant time.
\begin{lemma}
	\label{lemma:ConditionalABATermination}
	If the network is asynchronous, then 
	in protocol $\ABA$, it requires expected $\Order(n^2)$
	iterations for the honest parties to compute an output.
\end{lemma}
\begin{proof}
   To prove the lemma, we need to derive the expected number of iterations, until all the honest parties have the same input during the second instance
   of $\Vote$ of an iteration. This is because once all the honest parties have the same input during the second instance
   of $\Vote$ of an iteration, then all honest parties will set $\committed$ to $\true$ at the end of that iteration and start sending signed $\ready$ messages, followed
   by computing an output. 
   Let $\tau$ be the random variable
    which counts the number of iterations until all honest parties
    have the same input during the second instance
   of $\Vote$ in an iteration.  
         Then the probability that $\tau = k$ is given as: 
\begin{align*}
\Pr(\tau = k) &= \Pr(\tau \neq 1)\cdot\Pr (\tau \neq 2 \mid \tau \neq 1)\cdot \ldots 
                     \cdot Pr(\tau \neq (k-1) \mid \tau \neq 1 \cap \ldots \cap  \\
                    & \tau \neq (k-2))  \cdot  \Pr(\tau = k \mid \tau \neq 1 \cap \ldots \cap \tau \neq    (k-1)).
\end{align*}
    From Lemma \ref{lemma:ABACorrectness}, 
    every multiplicand on the right-hand side in
    the above equation, except the last one, is upper bounded by $(1 - \frac{1}{2n})$
    and the last multiplicand is upper bounded by $\frac{1}{2n}$. Hence, we get
    \[ \Pr(\tau = k) \leq (1 - \frac{1}{2n})^{k-1}(\frac{1}{2n}).\] 
 Now the expected
    value $E(\tau)$ of $\tau$ is computed as follows:
    \begin{align*}
        E(\tau) &=  \sum_{k=0}^{\infty} \tau \cdot \Pr(\tau = k) \\
                &\leq \sum_{k=0}^{\infty}k(1 - \frac{1}{2n})^{k-1}(\frac{1}{2n})\\
                &= \frac{1}{2n}\sum_{k=0}^{\infty}k(1 - \frac{1}{2n})^{k-1}\\
                &= \frac{1}{1 - (1 - \frac{1}{2n})} + \frac{1 - \frac{1}{2n}}{\Big(1 - (1 -
                \frac{1}{2n})\Big)^2}\\
                &= 2n + 4n^2 - 2n = 4n^2
    \end{align*}
    The expression for $E(\tau)$ is a sum of $AGP$ up to infinite terms, which is
    given by $\frac{a}{1 - r} + \frac{dr}{(1-r)^2}$, where $a=1$, $r= 1 -
    \frac{1}{2n}$ and $d = 1$. Hence, we have $E(\tau) \leq 4n^2$.
\end{proof}

We finally prove the properties of the protocol $\ABA$ in a {\it synchronous} network.
\begin{lemma}
\label{lemma:ABASynchronous}
If the network is synchronous and if all honest parties have the same input $b \in \{ 0, 1\}$ during $\ABA$, then all honest parties output $b$,
    at time $\TimeABA = \TimeCoinFlip + 2\TimeVote + \Delta$. 
\end{lemma}
\begin{proof}
Let $Z^{\star} \in \Z_s$ be the set of {\it corrupt} parties and let $\Hon = \PartySet \setminus Z^{\star}$ be the set of {\it honest} parties. 
 If all the parties in $\Hon$ participate with input $b$, then from the {\it $\Z_s$-liveness}
 and {\it $\Z_s$-Graded Validity} of $\Vote$ in the {\it synchronous} network (Lemma \ref{lemma:GASyncProperties}), all the parties in $\Hon$ output
 $(b, 2)$ during the first instance of $\Vote$ at time $\TimeVote$ in the first iteration. The {\it $\Z_s$-Guaranteed Liveness} of $\Vote$ in the {\it synchronous} network \cite{Cho23}
 ensures that all honest parties compute some output from the instance of $\CoinFlip$ during the first iteration at the time $\TimeVote + \TimeCoinFlip$.
  Since the parties in $\Hon$ output $(b, 2)$ during the first instance of $\Vote$, they participate with input $b$ during the second instance of $\Vote$.
  Consequently,  from the {\it $\Z_s$-liveness}
 and {\it $\Z_s$-Graded Validity} of $\Vote$ in the {\it synchronous} network, 
  all the parties in $\Hon$ compute the output $(b, 2)$ during the second instance of $\Vote$ in the first iteration at the time $2\TimeVote + \TimeCoinFlip$.
  Hence every party in $\Hon$ sends a signed $\ready$ message for $b$ at the time $2\TimeVote + \TimeCoinFlip$, which gets delivered at the time $\TimeABA$.
  Moreover, only the parties in $Z^{\star}$ may send a signed $\ready$ message for $\overline{b}$.
  Since $\PartySet \setminus \Hon = Z^{\star} \in \Z_s$ and since $\Hon \not \in \Z_s$ (due to the $\Q^{(2)}(\PartySet, \Z_s)$ condition), 
  it follows that all the parties in $\Hon$ will have sufficiently many signed $\ready$ messages for $b$ at the time $\TimeABA$ to output $b$. 
\end{proof}

The proof of Theorem \ref{thm:ABA} now follows from Lemma \ref{lemma:ABAAsyncValidity}-\ref{lemma:ABASynchronous}.
 The communication complexity follows from the communication complexity of $\CoinFlip$ \cite{Cho23} and the communication complexity of $\Vote$ (Lemma \ref{GA-communication})
 and the fact that in a {\it synchronous} network, only a {\it constant} number of invocations of $\Vote$ and $\CoinFlip$ are involved, while in an {\it asynchronous} network,
 there are $\mbox{poly}(n)$ invocations of $\Vote$ and $\CoinFlip$ in {\it expectation}.

%% file: AppICP.tex
\section{Properties of Our Network Agnostic ICP}
\label{app:ICP}
In this section, we prove the properties of our network-agnostic ICP (see Fig \ref{fig:ICP} for the formal details). Throughout this section, we assume that $\Z_s$ and $\Z_a$ satisfy the conditions 
 $\Z_a \subset \Z_s$, $\Q^{(2)}(\PartySet, \Z_s)$, $\Q^{(3)}(\PartySet, \Z_a)$ and $\Q^{(2, 1)}(\PartySet, \Z_s, \Z_a)$. 
 \begin{lemma}
 \label{lemma:ICPCorrectness}
    If $\mathsf{S}, \mathsf{I}$ and $\mathsf{R}$ are {\it honest}, then the following hold during protocol $\Auth$ and $\Reveal$.
       \begin{myitemize}
       		\item[--] {\bf $\AdvStruct_s$-Correctness}: In a synchronous network, each honest party sets $\authCompleted_{(\S, \INT, \Receiver)}$ 
   to $1$ during $\Auth$ at time $\TimeAuth = \Delta + 4\TimeBC$. Moreover $\mathsf{R}$ outputs $s$ during $\Reveal$ which takes $\TimeReveal = \Delta$ time.
   	   \item[--] {\bf $\AdvStruct_a$-Correctness}:  In an asynchronous network, each honest 
   party eventually sets $\authCompleted_{(\S, \INT, \Receiver)}$ 
   to $1$ during $\Auth$ and $\mathsf{R}$ eventually outputs $s$ during $\Reveal$.
      \end{myitemize}      
 \end{lemma}
\begin{proof}
 We first start with the {\it synchronous} network. Let $Z^{\star} \in \Z_s$ be set of {\it corrupt} parties and let 
  $\Hon = \PartySet \setminus Z^{\star}$ be the set of {\it honest} parties.
   During $\Auth$, $\mathsf{S}$ chooses a random $t$-degree signing-polynomial $F(x)$ such that $s = F(0)$ holds,
    a random $t$-degree masking-polynomial $M(x)$, and computes 
    verification points $(\alpha_i, v_i, m_i)$ such that $v_i = F(\alpha_i)$ and $m_i = M(\alpha_i)$ hold. $\mathsf{S}$ then
 sends the signing-polynomial $F(x)$ and masking-polynomial $M(x)$  to $\mathsf{I}$, and the 
 corresponding verification-point $(\alpha_i, v_i, m_i)$ to each verifier $P_i$.
 Consequently, each verifier in $\Hon$ 
  receives its verification-point by time $\Delta$, and indicates this by broadcasting $(\Received,i)$. Since $\PartySet \setminus \Hon = Z^{\star} \in \AdvStructure_s$,
  from the {\it $\Z_s$-validity} of $\BC$ in the {\it synchronous} network (see Theorem \ref{thm:BC}), 
   it follows that at time $\Delta + \TimeBC$, $\mathsf{S}$ will find a set $\R$, such that $\PartySet \setminus \R \in \AdvStructure_s$, where 
   each verifier in $\R$ has indicated that it has received its verification-point. Consequently, $\mathsf{S}$ will broadcast $\R$ at time $\Delta + \TimeBC$. 
   From the {\it $\Z_s$-validity} of $\BC$ in the {\it synchronous} network, $\mathsf{I}$ will receive $\R$ at time $\Delta + 2\TimeBC$. 
   Moreover, due to the {\it $\Z_s$-Consistency} and {\it $\Z_s$-Validity} of $\BC$ in the {\it synchronous} network,
   party $\mathsf{I}$ would have gotten $(\Received,i)$, corresponding to every verifier $P_i \in \R$, by time $\Delta + 2\TimeBC$. Furthermore, $\PartySet \setminus \R \in \Z_s$ will hold.
    Hence, $\mathsf{I}$ will randomly select $d \in \F$, compute $B(x) = dF(x) + M(x)$, and broadcast $(d, B(x))$. From the {\it $\Z_s$-validity} of $\BC$ in the {\it synchronous} network,
     this will be delivered to every honest party, including 
     $\mathsf{S}$, by time $\Delta + 3\TimeBC$.
     Moreover,  $\mathsf{S}$ will find that $B(\alpha_j) = dv_j + m_j$ holds 
     for all the verifiers $P_j \in \R$. Consequently, $\mathsf{S}$ will broadcast an $\OK$ message, which is received by every $P_i \in \Hon$ at time $\Delta + 4\TimeBC$, due to the
      {\it $\Z_s$-validity} of $\BC$ in the {\it synchronous} network.
       Thus, each $P_i \in \Hon$ sets $\authCompleted_{(\S, \INT, \Receiver)}$
        to $1$, while $\mathsf{I}$ additionally sets $\ICSig(\mathsf{S}, \mathsf{I}, \mathsf{R}, s)$ to $F(x)$ at time $\Delta + 4 \TimeBC$.
        
         During $\Reveal$, $\mathsf{I}$ will send $F(x)$ to $\mathsf{R}$, and each verifier $P_i \in \Hon \cap \R$ will send its 
     verification point $(\alpha_i,v_i,m_i)$ to $\mathsf{R}$. These points and the polynomial $F(x)$ are received by
     $\mathsf{R}$ within $\Delta$ time. Moreover, the condition $v_i = F(\alpha_i)$ will hold true for these points, and consequently, these points
      will be {\it accepted}. 
     Since $\R \setminus (\Hon \cap \R) \subseteq Z^{\star} \in \AdvStructure_s$, it follows that at time $\Delta$, receiver
     $\mathsf{R}$ will find a subset $\R' \subseteq \R$ where $\R \setminus \R' \in \AdvStructure_s$, such that the points corresponding to
     all the parties in $\R'$ are accepted.
      This implies that $\mathsf{R}$ will output $s = F(0)$ within time $\Delta$. 
     
The proof for the asynchronous case is similar as above, except that each ``favourable" event occurs {\it eventually}, and follows from the fact that every set in $\Z_a$ is a subset of some set in $\Z_s$. Moreover, we rely on the
 properties of $\BC$ in the {\it asynchronous} network.
\end{proof}

We next prove the privacy property, for which we again need to consider an {\it honest} $\S, \INT$ and $\Receiver$.
 \begin{lemma}
 \label{lemma:ICP-privacy}
 If $\S, \INT$ and $\Receiver$ are honest, then the view of $\Adv$ remains independent of $s$ during $\Auth$ and $\Reveal$, irrespective of the network
  type.
 \end{lemma}
\begin{proof}
We prove privacy in a {\it synchronous} network. The privacy in an {\it asynchronous} network automatically follows, since $\Z_a$ is a subset of $\Z_s$. Let $t = \max\{ |Z| :  Z \in \AdvStruct_s 	\}$ and let $Z^\star \in \AdvStructure_s$ be the set of corrupt parties. 
 For simplicity and without loss of generality, let 
  $|Z^\star| = t$. 
  During $\Auth$, the adversary $\Adv$ learns  $t$ verification-points $\{(\alpha_i, v_i, m_i) \}_{P_i \in Z^\star}$. 
  However, since $F(x)$ is a random $t$-degree polynomial with $F(0) = s$,
  the points $\{(\alpha_i, v_i) \}_{P_i \in Z^{\star}}$ are distributed independently of $s$.
  That is, for every candidate $s \in \F$ from the point of view of $\Adv$, there is a corresponding unique
  $t$-degree polynomial $F(x)$, such that $F(\alpha_i) = v_i$ holds corresponding to every $P_i \in Z^{\star}$.

  During $\Auth$, the adversary $\Adv$ also learns $d$ and the blinded-polynomial $B(x) = dF(x) + M(x)$.
    However, this does not add any new information about $s$ to the view of the adversary. This is because
  $M(x)$ is a random $t$-degree polynomial and $\Adv$ learns $t$ points on $M(x)$, corresponding to the parties in $Z^{\star}$.
    Hence, for every candidate $M(x)$ polynomial from the point of view of
  $\Adv$ where $M(\alpha_i) = m_i$ holds for every $P_i \in Z^{\star}$, there is a corresponding unique 
  $t$-degree polynomial $F(x)$, such that $F(\alpha_i) = v_i$ holds corresponding to every $P_i \in Z^{\star}$, and where
  $dF(x) + M(x) = B(x)$. 
   Finally, $\Adv$ does not learn anything new about $s$ during $\Reveal$, since the verification-points and the signing-polynomial are sent only to $\mathsf{R}$, who is {\it honest}
  as per the lemma conditions.
\end{proof}

We next prove the unforgeability property, for which we have to consider a {\it corrupt} $\INT$.
 \begin{lemma}
 \label{lemma:ICP-unforgeability}
  If $\mathsf{S}, \mathsf{R}$ are {\it honest}, $\mathsf{I}$ is corrupt
      and if $\mathsf{R}$ outputs $s' \in \F$ during $\Reveal$, then $s' = s$ holds except with
    probability at most $\errorAICP \defined \frac{nt}{|\F| - 1}$, where $t = \max\{ |Z| :  Z \in \AdvStruct_s \}$, irrespective of the network type.      
 \end{lemma}
\begin{proof}
Let $\Hon$ be the set of {\it honest} parties in $\PartySet$ and let $Z^{\star} = \PartySet \setminus \Hon$ be the set of {\it corrupt} parties.
 Since $\mathsf{R}$ outputs $s'$ during $\Reveal$, it implies that during $\Auth$, the variable 
 $\authCompleted_{(\S, \INT, \Receiver)}$ is set to $1$ by $\mathsf{R}$. This further implies that
  $\mathsf{S}$ has broadcasted an $\OK$ message during $\Auth$, which also
  implies that during $\Auth$, 
    $\mathsf{I}$ had broadcasted a $t$-degree blinded-polynomial $B(x)$, and $\mathsf{S}$ broadcasted the set $\R$. Furthermore,
    $\mathsf{S}$ has verified that
    $B(\alpha_i) = dv_i + m_i$ holds for every verifier $P_i \in \R$.
    Now during $\Reveal$, if $\mathsf{I}$ sends $F(x)$ as $\ICSig(\mathsf{S}, \mathsf{I}, \mathsf{R}, s)$ to $\mathsf{R}$, then
    $s' = s$ holds with probability $1$.     
    So, consider the case when $\mathsf{I}$ sends $F'(x)$ as $\ICSig(\mathsf{S}, \mathsf{I}, \mathsf{R}, s)$ to $\mathsf{R}$, where 
    $F'(x)$ is a $t$-degree polynomial such that $F'(x) \neq F(x)$ and where $F'(0) = s'$. In this case, we claim that except with probability at most $\frac{nt}{|\F| - 1}$,
    the verification-point of no {\it honest} verifier from $\R$ will get accepted by $\mathsf{R}$ during $\Reveal$, with respect to $F'(x)$. Now, assuming that the claim is true, the 
    proof follows using the following arguments, depending upon the network type.
    \begin{myitemize}
    	\item[--] {\it Synchronous Network}: 
	In this case, all the verifiers in $\Hon$ will be present in $\R$. This is because, each verifier $P_i \in \Hon$ would have received its verification-point from $\S$ during $\Auth$, within time $\Delta$ and 
	indicates this by broadcasting $(\Received, i)$, which is received by $\S$ at time $\Delta + \TimeBC$. Let $\R'$ be the set of verifiers from which $\mathsf{R}$ receives verification points which it accepts. Since the 
	 verification point of {\it none} of the honest verifier will be accepted.
	 Hence $(\Hon \cap \R') = \emptyset$ and so
	  $\Hon \subseteq \R \setminus \R'$ must hold. Since $\Hon$ satisfies the $\Q^1(\PartySet,\Z_s)$ condition, $\R \setminus \R' \in \Z_s$ will {\it never} hold true. Hence, $\mathsf{R}$ will {\it not} output $s' \neq s$. 
  \item[--] {\it Asynchronous Network}: In this case, we first note that $\Z_s$ and $\Z_a$ satisfy the $\Q^{(1, 1)}(\R, \Z_s, \Z_a)$ condition. This is because 
  $\PartySet \setminus \R \in \Z_s$ and $\Z_s$ and $\Z_a$ satisfy the $\Q^{(2, 1)}(\PartySet, \Z_s, \Z_a)$ condition. From the steps of $\Reveal$, it follow that for
  $\Receiver$ to output $F'(0)$, $\Receiver$ should find a subset of verifiers $\R' \subseteq \R$, where $\R \setminus \R' \in \Z_s$, such that the verification-points of all the verifiers in 
  $\R'$ are accepted by $\Receiver$. This further implies that $\R' \cap \Hon \neq \emptyset$, as  $\Z_a$ satisfies the $\Q^{(1)}(\R', \Z_a)$ condition. 
 And hence $\R'$ has at least one {\it honest} verifier, whose verification-point is accepted with respect to $F'(x)$. However, from the above claim, it is not possible and hence
 $\mathsf{R}$ will {\it not} output $s' \neq s$. 
        \end{myitemize}
   We now prove the claimed statement. So consider an {\it arbitrary} verifier $P_i \in \Hon \cap \R$ from whom $\mathsf{R}$ receives the verification-point $(\alpha_i,v_i,m_i)$ during $\Reveal$. 
      This point can be {\it accepted} with respect to $F'(x)$, only if either of the following holds.
	\begin{myitemize}
	\item $v_i = F'(\alpha_i)$:
	 This is possible with probability at most $\frac{t}{|\F| - 1}$. 
	 This is because $F'(x)$ and $F(x)$, being distinct $t$-degree polynomials, can have  at most $t$ points in common.
	  And the evaluation-point $\alpha_i$ corresponding to $P_i$, being randomly selected from $\F - \{0\}$, will {\it not} be known to
	 $\mathsf{I}$.
	\item $dv_i + m_i \neq B(\alpha_i)$: This is impossible, as otherwise $\mathsf{S}$ would have {\it not} broadcasted $\OK$ during 
	$\Auth$, which is a contradiction.	
	\end{myitemize}
    As there could be up to $n - 1$ {\it honest} verifiers in $\R$, it follows from the union bound that except with probability at most $\frac{nt}{|\F| - 1}$,
   the verification-point of no {\it honest} verifier from $\R$ will get accepted by $\mathsf{R}$ during $\Reveal$, with respect to $F'(x)$.
\end{proof}
 
 We next prove the non-repudiation property, for which we have to consider a {\it corrupt} $\S$.
 \begin{lemma}
 \label{lemma:ICP-non-repudiation}
 If $\mathsf{S}$ is {\it corrupt}, $\mathsf{I}, \mathsf{R}$ are {\it honest}
    and if $\mathsf{I}$ sets $\ICSig(\mathsf{S}, \mathsf{I}, \mathsf{R}, s)$ during $\Auth$, then the following hold, except with probability at most $\frac{n}{|\F| - 1}$.
    \begin{myitemize}
    	\item[--] {\bf $\AdvStruct_s$-Non-Repudiation}:  In a synchronous network, $\mathsf{R}$ outputs $s$ during $\Reveal$, which takes $\TimeReveal = \Delta$ time.    
    	\item[--] {\bf $\AdvStruct_a$-Non-Repudiation}: In an asynchronous network, $\mathsf{R}$ eventually outputs $s$ during during $\Reveal$.
    \end{myitemize}
 \end{lemma}
\begin{proof}
 Let $\Hon$ be the set of honest parties in $\PartySet$ and $Z^{\star}$ be the set of corrupt parties, where $\Hon = \PartySet \setminus Z^{\star}$.
 Since  $\mathsf{I}$ has set $\ICSig(\mathsf{S}, \mathsf{I}, \mathsf{R}, s)$ during $\Auth$, it implies that $\mathsf{I}$ has set the variable 
 $\authCompleted_{(\S, \INT, \Receiver)}$ to $1$. This further implies that $\mathsf{I}$ has broadcasted
 $(d, B(x))$, where $B(x) = dF(x) + M(x)$, and where $F(x)$ and $M(x)$ are the $t$-degree
  signing and masking-polynomials received by $\mathsf{I}$ from $\mathsf{S}$. 
  Moreover, $\mathsf{I}$ also received the set of supporting verifiers $\R$ from the broadcast of $\mathsf{S}$, and verified that $\PartySet \setminus \R \in \Z_s$ holds.
  Furthermore, $\mathsf{S}$ has broadcasted an $\OK$ message.
   Consequently, from the {\it consistency} properties of $\BC$ (see Theorem \ref{thm:BC}),
    irrespective of the network type, all {\it honest} parties including $\mathsf{R}$ eventually set $\authCompleted_{(\S, \INT, \Receiver)}$ to $1$.
   Moreover,
    $\mathsf{I}$ sets $\ICSig(\mathsf{S}, \mathsf{I}, \mathsf{R}, s)$ to $F(x)$, where
    $s = F(0)$. 
    During $\Reveal$, $\mathsf{I}$ sends
    $F(x)$ to $\mathsf{R}$. Moreover, every verifier $P_i \in \Hon \cap \R$ sends its verification-point 
    $(\alpha_i,v_i,m_i)$ to $\mathsf{R}$. 
    In a {\it synchronous} network, these will be received by $\mathsf{R}$ within time $\Delta$, 
    while in an asynchronous network, these will be eventually received by $\mathsf{R}$. 
    We claim that except with probability at most $\frac{n}{|\F|-1}$, 
    all these verification-points are  accepted by $\mathsf{R}$. 
    Now, assuming that the claim is true, the proof follows from the fact that
    $\Hon \cap \R = \R \setminus Z^{\star}$, and $Z^{\star} \in \Z_s$ holds, {\it irrespective} of the network type (since $\Z_a \subset \Z_s$).
       Consequently, $\mathsf{R}$ accepts the verification-points from a subset of the verifiers
        $\R' \subseteq \R$ where $\R \setminus \R' \in \AdvStructure_s$. And hence it outputs $s$, either within
         time $\Delta$ in a synchronous network, or eventually, in an asynchronous network.
      
     We now proceed to prove the claim. So consider an {\it arbitrary} verifier $P_i \in \Hon \cap \R$ whose verification-point $(\alpha_i, v_i, m_i)$ is received by $\mathsf{R}$
     during $\Reveal$. Now, there are two possible cases, depending upon the relationship that holds between $F(\alpha_i)$ and $v_i$
      during $\Auth$.
     \begin{myitemize}
	\item[--] {\it $v_i = F(\alpha_i)$ holds}: In this case, according to the protocol steps of $\Reveal$, the point $(\alpha_i, v_i, m_i)$ is {\it accepted} by $\mathsf{R}$.
	\item[--] {\it $v_i \neq F(\alpha_i)$ holds}: In this case, we claim that except with probability at most $\frac{1}{|\F| - 1}$, the condition 
	$dv_i + m_i \neq B(\alpha_i)$ will hold, implying that the point $(\alpha_i, v_i, m_i)$ is {\it accepted} by $\mathsf{R}$.
	This is because the {\it only} way $dv_i + m_i = B(\alpha_i)$ holds is when $\mathsf{S}$ distributes $(\alpha_i, v_i, m_i)$ to $P_i$
	where $v_i \neq F(\alpha_i)$ and $m_i \neq M(\alpha_i)$ holds, and $\mathsf{I}$ selects $d = (M(\alpha_i) - m_i)\cdot(v_i - F(\alpha_i))^{-1}$.
	However, $\mathsf{S}$ will {\it not} be knowing the random $d$ from $\F \setminus \{ 0\}$ which $\mathsf{I}$ is going to pick, while distributing
	$F(x), M(x)$ to $\mathsf{I}$, and $(\alpha_i, v_i, m_i)$ to $P_i$. Hence, the probability that $\INT$ indeed selects $d = (M(\alpha_i) - m_i)\cdot(v_i - F(\alpha_i))^{-1}$
	is $\frac{1}{|\F| - 1}$.
	\end{myitemize}
    As there can be up to $n - 1$ {\it honest} verifiers in $\R$, from the union bound, it follows that except with probability at most 
    $\frac{n}{|\F| - 1}$, the verification-point of {\it all} honest verifiers in $\R$ are accepted by $\mathsf{R}$.
\end{proof}

We finally derive the communication complexity.
\begin{lemma}
\label{lemma:ICPCommunicationComplexity}
Protocol $\Auth$ incurs a communication of $\Order(n^5 \cdot \log{|\F|} \cdot |\sigma|)$ bits. Protocol 
  $\Reveal$ incurs a communication of $\Order(n \cdot \log{|\F|})$ bits.
\end{lemma}
\begin{proof}
During $\Auth$, signer $\mathsf{S}$ sends $t$-degree polynomials
  $F(x)$ and $M(x)$ to $\mathsf{I}$, and verification-points to each verifier. This requires a communication of $\Order(n \cdot \log{|\F|})$ bits.
    Intermediary $\mathsf{I}$ needs to broadcast $B(x)$ and $d$ using protocol $\BC$, 
    while $\S$ needs to broadcast 
     the set $\R$ using $\BC$. Moreover, $\mathsf{S}$ may need to broadcast $s$ using $\BC$.
     By substituting the communication cost of $\BC$ (see Theorem \ref{thm:BC}), the overall communication cost of $\Auth$ turns out to be
     $\Order(n^5 \cdot \log{|\F|} \cdot |\sigma|)$ bits. 
     During $\Reveal$, $\mathsf{I}$ may send $F(x)$ to $\mathsf{R}$, and each verifier may send its verification-point to $\mathsf{R}$.
    This incurs a communication of $\Order(n \cdot \log{|\F|})$ bits.
\end{proof}

The proof of Theorem \ref{thm:ICP} now follows easily from Lemma \ref{lemma:ICPCorrectness}-\ref{lemma:ICPCommunicationComplexity}.

%% file: AppVSS.tex
\section{Properties of Network Agnostic VSS}
\label{app:VSS}
In this section, we prove the properties of the protocol $\VSS$ (see Fig \ref{fig:VSS}).  Throughout this section, we assume that $\Z_s$ and $\Z_a$ satisfy the conditions 
 $\Z_a \subset \Z_s$, $\Q^{(2)}(\PartySet, \Z_s)$, $\Q^{(3)}(\PartySet, \Z_a)$ and $\Q^{(2, 1)}(\PartySet, \Z_s, \Z_a)$. 
 We start with the properties in a {\it synchronous} network and first consider an {\it honest}
 $\D$. We first show that an {\it honest} $\D$ will broadcast some candidate core-sets, which will be accepted by all honest parties. Moreover, adversary will not learn
  any additional information about $s$.
 \begin{lemma}
 \label{lemma:VSSSynchronousHonestDealer}
 If the network is synchronous and $\D$ is {\it honest}, 
  participating  in $\VSS$ with input $s$, then all the following hold, where $\Hon$ is the set of honest parties.
    \begin{myitemize}
      \item[--] There exists some $S_p \in \ShareSpec_{\Z_s}$, such that $\D$ broadcasts a message 
  $(\CandidateCoreSets, \D, S_p, \{ \W_{q}\}_{q = 1, \ldots, |\Z_s|}, \BroadcastSet, \{ s_q\}_{q \in \BroadcastSet})$ at time $\Delta + \TimeAuth + \TimeBC$,
   and every $P_i \in \Hon$ includes $(\D, S_p)$ to the set $\C_i$ at time $\Delta + \TimeAuth + 2\TimeBC$. 
  Moreover, all the following hold for $q = 1, \ldots, |\Z_s|$.
    \begin{myitemize}
    \item[--] If $S_q = \Hon$, then $q \not \in \BroadcastSet$.
    \item[--] $\W_{q}$ will be either $S_q$ or $(S_p \cap S_q)$. Moreover, 
    $\Z_s$ will satisfy the $\Q^{(1)}(\W_q, \Z_s)$ condition.
    \item[--] Corresponding to every $S_q \in \ShareSpec_{\Z_s}$, every honest $P_i \in \W_{q}$ will have the IC-signature
     $\ICSig(P_j, P_i, P_k, s_q)$ of every $P_j \in \W_{q}$ for every $P_k \not \in S_q$, such that
    the underlying signatures will satisfy the linearity principle.
    Furthermore, if any corrupt $P_j \in \W_{q}$ has the IC-signature $\ICSig(P_i, P_j, P_k, s'_q)$ of any honest $P_i \in \W_{q}$ for any $P_k \in \PartySet$, then
    $s'_q = s_q$ holds and the underlying signatures will satisfy the linearity principle.
    \item[--] Corresponding to every $S_q \in \ShareSpec_{\Z_s}$, 
    every honest $P_i \in S_q$ will have the share $s_q$, except with a probability $\Order(|\ShareSpec_{\Z_s}| \cdot n^2 \cdot \errorAICP)$, at time $\Delta + \TimeAuth + 2\TimeBC + \TimeReveal$, where
    $s = s_1 + \ldots + s_{|\Z_s|}$.
    \end{myitemize}
 \item[--] The view of the adversary will be independent of $s$.
 \end{myitemize}
 \end{lemma}
\begin{proof}
 Let  $Z^\star \in \Z_s$ be the set of {\it corrupt} parties and let $\Hon = \PartySet \setminus Z^{\star}$ be the set of {\it honest} parties. 
   We note that $\Hon \in \ShareSpec_{\Z_s}$.
   Since $\D$ is honest, it picks the shares $s_1,\dots, s_{|\Z_s|}$ such that $s = s_1 + \dots + s_{|\Z_s|}$ and sends $s_q$ to each party $P_i \in S_q$, corresponding to
   every $S_q \in \ShareSpec_{\Z_s}$. These shares are delivered within time $\Delta$. Now consider an {\it arbitrary} $S_q \in \ShareSpec_{\Z_s}$.
  At time $\Delta$, each party $P_i \in (S_q \cap \Hon)$ starts giving
   $\ICSig(P_i, P_j, P_k, s_{qi})$ to every $P_j \in S_q$, for every $P_k \in \PartySet$, where $s_{qi} = s_q$ holds.
   Moreover, the linearity principle is followed while generating these IC-signatures.
      Then from the {\it $\Z_s$-correctness} of $\Auth$ in the {\it synchronous} network (Theorem \ref{thm:ICP}), 
   it follows that at time $\Delta + \TimeAuth$, each party $P_i \in (S_q \cap \Hon)$ will receive 
   $\ICSig(P_j, P_i, P_k, s_{qj})$ from every $P_j \in (S_q \cap \Hon)$, for every $P_k \in \PartySet$, such that $s_{qj} = s_{qi} = s_q$ holds. 
   Since $S_q$ is arbitrary, it follows that at time $\Delta + \TimeAuth$, every party $P_i \in \Hon$ broadcasts an $\OK(i, j)$ message, corresponding to every 
   $P_j \in \Hon$. From the {\it $\Z_s$-validity} of $\BC$ in the {\it synchronous} network (Theorem \ref{thm:BC}), it follows that
   these $\OK(i, j)$ messages are received by every party in $\Hon$ through {\it regular-mode} at time $\Delta + \TimeAuth + \TimeBC$.
   Consequently, the set $\Hon$ constitutes a clique in the consistency graph of every party in $\Hon$ at time $\Delta + \TimeAuth + \TimeBC$.
   Now since the set $\Hon \in \ShareSpec_{\Z_s}$, it follows that at time $\Delta + \TimeAuth + \TimeBC$, there exists {\it some}
   $S_p \in \ShareSpec_{\Z_s}$, such that $\D$ computes the core-sets $\{ \W_{q}\}_{q = 1, \ldots, |\Z_s|}$ and broadcast-set $\BroadcastSet$, followed by
   broadcasting  $(\CandidateCoreSets, \D, S_p, \{ \W_{q}\}_{q = 1, \ldots, |\Z_s|}, \BroadcastSet, \{ s_q\}_{q \in \BroadcastSet})$ at time $\Delta + \TimeAuth + \TimeBC$.
    Moreover, since $\D$ is {\it honest}, it computes the sets $\{ \W_{q}\}_{q = 1, \ldots, |\Z_s|}$ and $\BroadcastSet$ honestly, satisfying the following conditions, for
   $q = 1, \ldots, |\Z_s|$.
        \begin{myitemize}
	 \item[--] If $S_q$ constitutes a clique in the graph $G^{(\D)}$, then $\W_{q} = S_q$.
	 \item[--] Else if $(S_p \cap S_q)$ constitutes a clique in $G^{(\D)}$ and 
	 $\Z_s$ satisfies the $\Q^{(1)}(S_p \cap S_q, \Z_s)$ condition, 
	then $\W_{q} = (S_p \cap S_q)$.
	 \item[--] Else $\W_{q} = S_q$ and $q \in \BroadcastSet$. 	 
	 \end{myitemize}
      Note that for each $\W_q$, the condition $\Q^{(1)}(\W_q, \Z_s)$  holds. This is obviously true if $\W_{q} = (S_p \cap S_q)$, since in this case $\D$ also checks that
      $\Q^{(1)}(S_p \cap S_q, \Z_s)$ condition holds. On the other hand, even if $\W_q = S_q$, the condition $\Q^{(1)}(\W_q, \Z_s)$  holds, as
      $\Q^{(1)}(S_q, \Z_s)$  holds due to the $\Q^{(2)}(\PartySet, \Z_s)$ condition.
      We also note that if $S_q = \Hon$, then $q \notin \BroadcastSet$ and consequently, $\D$ does not make public the share $s_q$. This is because
      as shown above, at time $\Delta + \TimeAuth + \TimeBC$, the parties in $S_q$ constitute a clique in the graph $G^{(\D)}$.
      
      Since $\D$ broadcasts $(\CandidateCoreSets, \D, S_p, \{ \W_{q}\}_{q = 1, \ldots, |\Z_s|}, \BroadcastSet, \{ s_q\}_{q \in \BroadcastSet})$ at time $\Delta + \TimeAuth + \TimeBC$,
      from the {\it $\Z_s$-validity} of $\BC$ in the {\it synchronous} network, it follows that all the parties in $\Hon$ will receive
      $(\CandidateCoreSets, \D, S_p, \{ \W_{q}\}_{q = 1, \ldots, |\Z_s|}, \BroadcastSet, \{ s_q\}_{q \in \BroadcastSet})$ from the broadcast of
      $\D$, at time $\Delta + \TimeAuth + 2\TimeBC$. To show that every $P_i \in \Hon$ will include $(\D, S_p)$ to the set $\C_i$, we need to show that all the conditions which hold
      for $\D$ in its graph $G^{(\D)}$ at time $\Delta + \TimeAuth + \TimeBC$, are bound to hold for every $P_i \in \Hon$, at the time $\Delta + \TimeAuth + 2\TimeBC$.
      Namely, all the edges which are present in $G^{(\D)}$ at time $\Delta + \TimeAuth + \TimeBC$, are bound to be present in the
      graph $G^{(i)}$ every $P_i \in \Hon$, at the time $\Delta + \TimeAuth + 2\TimeBC$.
      However, this simply follows from the {\it $\Z_s$-validity, $\Z_s$-consistency} and {\it $\Z_s$-fallback consistency} of $\BC$ in the {\it synchronous} network (see Theorem \ref{thm:BC})
      and the fact that edges are added to consistency graphs, based on the receipt of $\OK(\star, \star)$ messages, which are broadcasted through various $\BC$ instances.
      Consequently, any edge $(i, j)$ which is included in $G^{(\D)}$ at the time $\Delta + \TimeAuth + \TimeBC$, is bound to be included in the graph $G^{(i)}$ of every
      $P_i \in \Hon$, latest by time $\Delta + \TimeAuth + 2\TimeBC$.
      
      We next note that corresponding to every $S_q \in \ShareSpec_{\Z_s}$, every honest $P_i \in \W_{q}$ will have the share $s_q$, which is either made public by $\D$
      as part of the $\CandidateCoreSets$ message
      or received from $\D$. Each honest $P_i \in W_q$ will thus set $[s]_q$ to $s_q$ at time $\Delta + \TimeAuth + 2\TimeBC$.
     We next show that each $P_i \in \Hon$ will have $\ICSig(P_j, P_i, P_k, s_q)$ corresponding to
     every $S_q \in \ShareSpec_{\Z_s}$, where $P_i \in \W_q$, for every $P_j \in \W_q$ and every $P_k \not \in S_q$.
      This also is set at time $\Delta + \TimeAuth + 2\TimeBC$. This is because there are two possible cases with respect to $q$.
     If $q \in \BroadcastSet$, then from the protocol steps, $\W_q$ is publicly set to $S_q$ and $\ICSig(P_j, P_i, P_k, s_q)$ is set to the default value.
     On the other hand, if $q \not \in \BroadcastSet$, then also $P_i$ will possess $\ICSig(P_j, P_i, P_k, s_q)$. This is because as per the protocol steps, 
     since $P_i, P_j \in \W_q$, it follows that 
      $P_i$ must have verified that the edge $(P_i, P_j) \in G^{(i)}$, which further implies that $P_i$ has received $\ICSig(P_j, P_i, P_k, s_{qj})$ from $P_j$, where $s_{qj} = s_{qi}$ holds.
      And since $\D$ is {\it honest}, $s_{qi} = s_q$ holds, implying that $\ICSig(P_j, P_i, P_k, s_{qj})$ is the same as $\ICSig(P_j, P_i, P_k, s_q)$.
      On the other hand, in the protocol, $P_i$ gives $\ICSig(P_i, P_j, P_k, s_{qi})$ to every $P_j \in S_q$ for every $P_k \in \PartySet$, where $s_{qi} = s_q$ holds.
      Hence if any corrupt $P_j \in \W_{q}$ has $\ICSig(P_i, P_j, P_k, s'_q)$ of any honest $P_i \in \W_{q}$
     for any $P_k  \in \PartySet$, then
    $s'_q = s_q$ holds.
    
    We now show that, corresponding to each $S_q \in \ShareSpec_{\Z_s}$, every honest party $P_i \in S_q \setminus W_q$ sets $[s]_q$ to $s_q$, except with a probability of $\Order(|\ShareSpec_{\Z_s}| \cdot n^2 \cdot \errorAICP)$, at time $\Delta + \TimeAuth + 2\TimeBC + \TimeReveal$. For this, we first show that $P_i$ sets $[s]_q$ to {\it some} value. 
     Since $\Z_s$ satisfies the $\Q^{(1)}(\W_q, \Z_s)$ condition, the set
      $\W_q$ contains {\it at least} one {\it honest} party, say $P_j$. Since $P_j$ follows the protocol steps honestly, it reveals 
      $\ICSig(P_k, P_j, P_i, [s]_q)$ of every $P_k \in \W_q$ to $P_i$, at time $\Delta + \TimeAuth + 2\TimeBC$. From the {\it $\Z_s$-correctness}  of ICP in the {\it synchronous} network
     (see Theorem \ref{thm:ICP}), it follows that $P_i$ will accept  these signatures
     after time $\TimeReveal$. On the other hand, even if $P_k \in \W_q$ is {\it corrupt}, then also from the {\it $\Z_s$-non-repudiation} property of ICP in the {\it synchronous} network
      (see Theorem \ref{thm:ICP}), it follows that $P_i$ accepts $\ICSig(P_k, P_j, P_i, [s]_q)$, except with a probability $\errorAICP$, after time $\TimeReveal$. As there can be $\Order(n)$ {\it corrupt} parties in $\W_q$, from the union bound, it follows that 
except with a probability $\Order(n \cdot \errorAICP)$, party $P_i$ will find a candidate party from $\W_q$, who reveals $[s]_q$, along with the IC-signature of all the parties in $\W_q$, after time $\TimeReveal$. Now as there can be $\Order(n)$ parties in $S_q \setminus \W_q$, it follows that except with probability $\Order(n^2 \cdot \errorAICP)$, every honest party $P_i \in S_q \setminus \W_q$ will find a candidate party from $\W_q$, who reveals $[s]_q$ along with the IC-signature of all the parties in $\W_q$ at time $\Delta + \TimeAuth + 2\TimeBC + \TimeReveal$. 

We next show that $P_i \in S_q \setminus W_q$ indeed sets $[s]_q$ to $s_q$. Suppose that $P_i$ sets $[s]_q$ to some value $s'$. From the protocols steps, this implies that there exists some $P_j \in \W_q$, such that $P_i$ has accepted
    $\ICSig(P_k, P_j, P_i, s')$ of every $P_k \in \W_q$, revealed by $P_j$.  If $P_j$ is {\it honest}, then indeed $s' = [s]_q$, as one of the IC-signatures
    $\ICSig(P_k, P_j, P_i, s')$ is the same as $\ICSig(P_k, P_j, P_i, [s]_q)$, corresponding to the {\it honest} $P_k \in \W_q$, which is guaranteed to exist.
    So consider the case when $P_j$ is {\it corrupt}.
    Moreover, let $P_k \in \W_q$ be an {\it honest} party (which is guaranteed to exist).
    In order that $s' \neq [s]_q$, it must be the case that $P_i$ accepts  $\ICSig(P_k, P_j, P_i, s')$, revealed by $P_j$.  However, from
    the {\it unforgeability} property of ICP (see Theorem \ref{thm:ICP}), this can happen only with probability $\errorAICP$. Now as there can be up to $\Order(n)$ corrupt parties in $\W_q$, from the union bound, it follows that the probability that $P_i$ outputs $s' \neq [s]_q$ is at most $\Order(n \cdot \errorAICP)$. Since there can be up to $\Order(n)$ parties in $S_q \setminus \W_q$, it follows that except with probability at most 
    $\Order(n^2 \cdot \errorAICP)$, the output of every honest party in $S_q$ is indeed $s_q$. Now, there can be $|\ShareSpec_{\Z_s}|$ possibilities for $S_q$. From the union bound, it follows that corresponding to each $S_q \in \ShareSpec_{\Z_s}$, every honest party $P_i \in S_q \setminus W_q$ sets $[s]_q$ to $s_q$, except with a probability of $\Order(|\ShareSpec_{\Z_s}| \cdot n^2 \cdot \errorAICP)$.

   Finally, the privacy for $s$ follows from the fact that throughout the protocol, the view of the adversary remains independent of the share $s_q$, corresponding to the group $S_q$, where
   $S_q = \Hon$. This is because as shown above, $q \notin \BroadcastSet$ and consequently, $\D$ does not make public the share $s_q$.
   Moreover, during the pairwise consistency tests, the view of the adversary remains independent of $s_q$, when the parties in $\Hon$ exchange IC-signed $s_q$, which
   follows from the {\it privacy} property of ICP (see Theorem \ref{thm:ICP}). Further, while computing the output, $\ICSig(P_k, P_j, P_i, s_q)$ is revealed by party $P_j \in W_q$ only to each party $P_i \in S_q \setminus W_q$. Hence, the adversary does not learn $s_q$.
\end{proof}
An immediate corollary of Lemma \ref{lemma:VSSSynchronousHonestDealer} is that if $\D$ is honest, then the parties output $[s]$ at time $\Delta + \TimeAuth + 2\TimeBC + \TimeReveal$, which follows
 from the definition of $[\cdot]$-sharing.
\begin{corollary}
If the network is synchronous and $\D$ is {\it honest} and 
  participates in $\VSS$ with input $s$, then the parties output $[s]$ at time $\Delta + \TimeAuth + 2\TimeBC + \TimeReveal$, except with a probability of $\Order(|\ShareSpec_{\Z_s}| \cdot n^2 \cdot \errorAICP)$.
\end{corollary}

We next consider a {\it corrupt} $\D$ in the {\it synchronous} network and show that if any honest party computes an output at time $T$, then there exists some $s^{\star} \in \F$ such that the (honest)
 parties output $[s^{\star}]$ by time $T + \Delta$, except with probability $\Order(|\ShareSpec_{\Z_s}| \cdot n^2 \cdot \errorAICP)$.
 \begin{lemma}
 \label{lemma:VSSSynchronousCorruptDealer} 
 If the network is synchronous and $\D$ is corrupt and if any honest party computes an output at time $T$, then there exists some $s^{\star} \in \F$, such that the honest parties output
  $[s^{\star}]$ by time $T + \Delta$, except with a probability of $\Order(|\ShareSpec_{\Z_s}| \cdot n^2 \cdot \errorAICP)$. 
 \end{lemma}
 \begin{proof}
 Let $Z^{\star} \in \Z_s$ be the set of {\it corrupt} parties and let $\Hon = \PartySet \setminus Z^{\star}$ be the set of {\it honest} parties. Let $P_\ell \in \Hon$ be the {\it first} honest party
  which computes an output at time $T$. The way in which the core sets are defined ensures that each set in $\{ \W_{q}\}_{q = 1, \ldots, |\Z_s|}$ must contain at least one honest party. Since $\Reveal$ takes at least $\TimeReveal$ time to complete, this means that there exists $S_p \in \ShareSpec_{\Z_s}$ such that some honest party $P_m$ receives a message
     $(\CandidateCoreSets, \D, S_p, \{ \W_{q}\}_{q = 1, \ldots, |\Z_s|}, \BroadcastSet, \{ s_q\}_{q \in \BroadcastSet})$ from the broadcast
     of $\D$ by the time $T - \TimeReveal$ and includes $(\D, S_p)$ to $\C_m$.
     This further implies that $P_m$ has verified that the following hold, for 
   $q = 1, \ldots, |\Z_s|$, by time $T - \TimeReveal$.
    \begin{myitemize}
       \item[--] If $q \not \in \BroadcastSet$, then $\W_{q}$ is either $S_q$ or $(S_p \cap S_q)$. Moreover, 
       $\Z_s$ satisfies the $\Q^{(1)}(\W_q, \Z_s)$ condition 
      and the parties in 
       $\W_q$ constitute a clique in the graph $G^{(m)}$.
    \item[--] If  $q  \in \BroadcastSet$, then $\D$ has made public $s_q$, as part of the $\CandidateCoreSets$ message. Moreover, $\W_q$ is set to $S_q$.
    \end{myitemize}
From the {\it $\Z_s$-consistency} and {\it $\Z_s$-fallback consistency} of $\BC$ in the {\it synchronous} network (see Theorem \ref{thm:BC}), it follows that every party in $\Hon$ will receive
  $(\CandidateCoreSets, \D, S_p, \{ \W_{q}\}_{q = 1, \ldots, |\Z_s|}, \BroadcastSet, \{ s_q\}_{q \in \BroadcastSet})$ from the broadcast
     of $\D$, latest by time $T - \TimeReveal + \Delta$. We next show that each $P_i \in \Hon$ will include $(\D, S_p)$ to $\C_i$, latest by time $T - \TimeReveal + \Delta$. 
     For this, it is enough to show that all the edges which are present in $G^{(m)}$ at time $T- \TimeReveal$, are bound to be present in the
      graph $G^{(i)}$ every $P_i \in \Hon$, by the time $T - \TimeReveal + \Delta$.
      However, this simply follows from the {\it $\Z_s$-validity, $\Z_s$-consistency} and {\it $\Z_s$-fallback consistency} of $\BC$ in the {\it synchronous} network (see Theorem \ref{thm:BC})
      and the fact that edges are added to consistency graphs, based on the receipt of $\OK(\star, \star)$ messages, which are broadcasted through various $\BC$ instances.
      Consequently, any edge $(i, j)$ which is included in $G^{(m)}$ at the time $T$, is bound to be included in the graph $G^{(i)}$ of every
      $P_i \in \Hon$, latest by time $T - \TimeReveal + \Delta$.
     
     We next show that by time $T - \TimeReveal + \Delta$, corresponding to every $S_q \in \ShareSpec_{\Z_s}$, every $P_i \in \W_q$ will have a common share, say $s^{\star}_q$.
     If $q \in \BroadcastSet$, this is trivially true, since in this case, $\D$ makes public the share $s_q$ and hence  $s^{\star}_q = s_q$.
     On the other hand, consider the case when $q \not \in \BroadcastSet$ and consider {\it arbitrary} $P_i, P_j \in (\W_q \cap \Hon)$. Since $P_i$ and $P_j$ are part of a clique,
     it follows that $P_i$ and $P_j$ have broadcasted the messages $\OK(i, j)$ and $\OK(j, i)$ respectively. Moreover, these messages were broadcasted, latest by time $T - \TimeReveal - \TimeBC$, since
     it takes at least $\TimeBC$ time to compute any output in an instance of $\BC$ in the {\it synchronous} network (see Theorem \ref{thm:BC}).
     Now since $P_i$ and $P_j$ have broadcasted $\OK(i, j)$ and $\OK(j, i)$ messages, it implies that they have verified that $s_{qi} = s_{qj}$ holds, where
     $s_{qi}$ and $s_{qj}$ are the shares, received by $P_i$ and $P_j$ respectively, from $\D$. Let $s_{qi} = s_{qj} = s^{\star}_q$. 
     We define
     \[ s^{\star} \defined \displaystyle \sum_{q = 1, \ldots, |\Z_s|} s^{\star}_q.    \]
     Till now we have shown that there exists some $s^{\star} \in \F$, such that by time $T - \TimeReveal + \Delta$, all the parties in $\Hon$ will have the core-sets $\W_1, \ldots, \W_{|\Z_s|}$, where
     $S_q \setminus \W_q \in \Z_a$, for $q = 1, \ldots, |\Z_s|$ and where each $P_i \in (\W_q \cap \Hon)$ will have a common share $[s^{\star}]_q$.
     We need to show that by the time $T - \TimeReveal + \Delta$, the parties in $\Hon$ will have the required IC-signatures, as part of $[s^{\star}]$, satisfying the linearity property.
     Namely, each $P_i \in \Hon$ will have $\ICSig(P_j, P_i, P_k, [s^{\star}]_q)$, corresponding to
     every $S_q \in \ShareSpec_{\Z_s}$ where $P_i \in \W_q$, of every $P_j \in \W_q$ and for every $P_k \in \PartySet$. There are two possible cases with respect to $q$.
     If $q \in \BroadcastSet$, then from the protocol steps, $\W_q$ is publicly set to $S_q$ and $\ICSig(P_j, P_i, P_k, [s^{\star}]_q)$ is set to the default value of
     $\ICSig(P_j, P_i, P_k, s_q)$, where $s_q$ is made public by $\D$. As per our notations, $s_q = s^{\star}_q$ for every 
     $q \in \BroadcastSet$.
     On the other hand, if $q \not \in \BroadcastSet$, then also $P_i$ will possess $\ICSig(P_j, P_i, P_k, [s^{\star}]_q)$. This is because as per the protocol steps, 
     since $P_i, P_j \in \W_q$, it follows that 
      $P_i$ must have verified that the edge $(P_i, P_j) \in G^{(i)}$, which further implies that $P_i$ has received $\ICSig(P_j, P_i, P_k, s_{qj})$ from $P_j$, where $s_{qj} = s_{qi}$ holds.
      Here $s_{qi}$ is the share received by $P_i$ from $\D$ and as per our notation, $s_{qi} = s^{\star}_q$. 
      Hence $\ICSig(P_j, P_i, P_k, s_{qj})$ is the same as $\ICSig(P_j, P_i, P_k, [s^{\star}]_q)$.
      On the other hand, in the protocol, $P_i$ gives $\ICSig(P_i, P_j, P_k, s_{qi})$ to every $P_j \in S_q$ for every $P_k \in \PartySet$, where $s_{qi} = [s^{\star}]_q$ holds.
      Hence if any corrupt $P_j \in \W_{q}$ has $\ICSig(P_i, P_j, P_k, s'_q)$, of any honest $P_i \in \W_{q}$
     for any $P_k \in \PartySet$, then
    $s'_q = [s^{\star}]_q$ holds. It is easy to see that all the underlying IC-signatures will be linear since the parties follow the linearity principle while generating IC-signatures.
    
 Finally, we now show that every honest party $P_i \in S_q \setminus W_q$ gets $[s^\star]_q$, except with a probability of $\Order(|\ShareSpec_{\Z_s}| \cdot n^2 \cdot \errorAICP)$, at time $T + \Delta$. For this, we first show that $P_i$ computes some share on the behalf of $S_q$. Since  $\Z_s$ satisfies the $\Q^{(1)}(\W_q, \Z_s)$ condition, the set
  $\W_q$ contains {\it at least} one {\it honest} party, say $P_j$. Since $P_j$ follows the protocol steps honestly, it reveals $\ICSig(P_k, P_j, P_i, [s^\star]_q)$ of every $P_k \in \W_q$ to $P_i$, 
  at time $T - \TimeReveal + \Delta$. From 
  the {\it $\Z_s$-correctness} of ICP in the {\it synchronous} network
   (see Theorem \ref{thm:ICP}), it follows that $P_i$ will accept the IC-signatures $\ICSig(P_k, P_j, P_i, [s^\star]_q)$, revealed by $P_j$,
    after time $\TimeReveal$. On the other hand, even if $P_k \in \W_q$ is {\it corrupt}, then also from the {\it $\Z_s$-non-repudiation} property of ICP in the {\it synchronous} network (see Theorem \ref{thm:ICP}), it follows that $P_i$ accepts $\ICSig(P_k, P_j, P_i, [s^\star]_q)$, except with a probability $\errorAICP$, after time $\TimeReveal$. As there can be $\Order(n)$ {\it corrupt} parties in $\W_q$, from the union bound, it follows that 
except with a probability $\Order(n \cdot \errorAICP)$, party $P_i$ will find a candidate party from $\W_q$, who reveals $[s^\star]_q$, along with the IC-signature of all the parties in $\W_q$, after time $\TimeReveal$. Now as there can be $\Order(n)$ parties in $S_q \setminus \W_q$, it follows that except with probability $\Order(n^2 \cdot \errorAICP)$, every honest party $P_i \in S_q \setminus \W_q$ will find a candidate party from $\W_q$, who reveals $[s^\star]_q$ along with the IC-signature of all the parties in $\W_q$ at time $T + \Delta$. 

We next show that $P_i \in S_q \setminus W_q$ indeed sets $[s^\star]_q$ as the share corresponding to $S_q$.
 Suppose that $P_i$ sets the share to some value $s'$. From the protocol steps, this implies that there exists some $P_j \in \W_q$, such that $P_i$ has accepted
   the IC-signatures $\ICSig(P_k, P_j, P_i, s')$ of every $P_k \in \W_q$, revealed by $P_j$.  If $P_j$ is {\it honest}, then indeed $s' = [s^\star]_q$, as one of the IC-signatures
    $\ICSig(P_k, P_j, P_i, s')$ is the same as $\ICSig(P_k, P_j, P_i, [s^\star]_q)$, corresponding to the {\it honest} $P_k \in \W_q$, which is guaranteed to exist.
    So consider the case when $P_j$ is {\it corrupt}.
    Moreover, let $P_k \in \W_q$ be an {\it honest} party (which is guaranteed to exist).
    In order that $s' \neq [s^\star]_q$, it must be the case that $P_i$ accepts  $\ICSig(P_k, P_j, P_i, s')$, revealed by $P_j$.  However, from
    the {\it unforgeability} property of ICP (see Theorem \ref{thm:ICP}), this can happen only with probability $\errorAICP$. Now as there can be up to $\Order(n)$ corrupt parties in $\W_q$, from the union bound, it follows that the probability that $P_i$ outputs $s' \neq [s^\star]_q$ is at most $\Order(n \cdot \errorAICP)$. Since there can be up to $\Order(n)$ parties in $S_q \setminus \W_q$, it follows that except with probability at most 
    $\Order(n^2 \cdot \errorAICP)$, the output of every honest party in $S_q$ is indeed $[s^\star]_q$. Now, there can be $|\ShareSpec_{\Z_s}|$ possibilities for $S_q$. From the union bound, it follows that corresponding to each $S_q \in \ShareSpec_{\Z_s}$, every honest party in $S_q$ outputs $[s^\star]_q$, except with a probability of $\Order(|\ShareSpec_{\Z_s}| \cdot n^2 \cdot \errorAICP)$.
\end{proof}

We next consider an {\it asynchronous} network and prove the analogue of Lemma \ref{lemma:VSSSynchronousHonestDealer} by showing that if $\D$ is {\it honest}, then the parties eventually output $[s]$ 
 except with a probability of $\Order(|\ShareSpec_{\Z_s}| \cdot n^2 \cdot \errorAICP)$. Further, the adversary learns no additional information about $s$. The proof of the lemma follows closely the proof of Lemma \ref{lemma:VSSSynchronousHonestDealer},
  except that we now rely on the properties of $\BC$, and ICP in the {\it asynchronous} network.
 \begin{lemma}
 \label{lemma:VSSAsynchronousHonestDealer}
If the network is asynchronous and $\D$ is {\it honest}, 
  participating  in $\VSS$ with input $s$, then the honest parties eventually output $[s]$ except with a probability of $\Order(|\ShareSpec_{\Z_s}| \cdot n^2 \cdot \errorAICP)$, with the view of the adversary remaining independent of $s$.
 \end{lemma}
 \begin{proof}
  Let  $Z^\star \in \Z_a$ be the set of {\it corrupt} parties and let $\Hon = \PartySet \setminus Z^{\star}$ be the set of {\it honest} parties. 
   We note that $\Hon \in \ShareSpec_{\Z_s}$, since $Z^{\star} \in \Z_s$ as $\Z_a \subset \Z_s$.
   Since $\D$ is honest, it picks the shares $s_1,\dots, s_{|\Z_s|}$ such that $s = s_1 + \dots + s_{|\Z_s|}$ and sends $s_q$ to each party $P_i \in S_q$, corresponding to
   every $S_q \in \ShareSpec_{\Z_s}$. These shares are eventually delivered. Now consider an {\it arbitrary} $S_q \in \ShareSpec_{\Z_s}$.
  After receiving the share $s_{qi}$ from $\D$, each party $P_i \in (S_q \cap \Hon)$ starts giving
   $\ICSig(P_i, P_j, P_k, s_{qi})$ to every $P_j \in S_q$, for every $P_k \in \PartySet$, where $s_{qi} = s_q$ holds,
   such that the linearity principle is followed while generating these IC-signatures.
      Then from the {\it $\Z_a$-correctness} of $\Auth$ in the {\it asynchronous} network (Theorem \ref{thm:ICP}), 
   it follows that each party $P_i \in (S_q \cap \Hon)$ will eventually receive 
   $\ICSig(P_j, P_i, P_k, s_{qj})$ from every $P_j \in (S_q \cap \Hon)$, for every $P_k \in \PartySet$, such that $s_{qj} = s_{qi} = s_q$ holds. 
   Since $S_q$ is arbitrary, it follows that eventually, every party $P_i \in \Hon$ broadcasts an $\OK(i, j)$ message, corresponding to every 
   $P_j \in \Hon$. From the {\it $\Z_a$-weak validity} and {\it $\Z_a$-fallback validity}  of $\BC$ in the {\it asynchronous} network (Theorem \ref{thm:BC}), it follows that
   these $\OK(i, j)$ messages are eventually received by every party in $\Hon$.
   Consequently, the set $\Hon$ eventually constitutes a clique in the consistency graph of every party in $\Hon$.
   Now since the set $\Hon \in \ShareSpec_{\Z_s}$, it follows that eventually, there exists {\it some}
   $S_p \in \ShareSpec_{\Z_s}$, such that $\D$ computes the core-sets $\{ \W_{q}\}_{q = 1, \ldots, |\Z_s|}$ and broadcast-set $\BroadcastSet$, followed by
   broadcasting  $(\CandidateCoreSets, \D, S_p, \{ \W_{q}\}_{q = 1, \ldots, |\Z_s|}, \BroadcastSet, \{ s_q\}_{q \in \BroadcastSet})$.
    Moreover, since $\D$ is {\it honest}, it computes the sets $\{ \W_{q}\}_{q = 1, \ldots, |\Z_s|}$ and $\BroadcastSet$ honestly, satisfying the following conditions, for
   $q = 1, \ldots, |\Z_s|$.
        \begin{myitemize}
	 \item[--] If $S_q$ constitutes a clique in the graph $G^{(\D)}$, then $\W_{q} = S_q$.
	 \item[--] Else if $(S_p \cap S_q)$ constitutes a clique in $G^{(\D)}$ and 
	 $\Z_s$ satisfies the $\Q^{(1)}(S_p \cap S_q, \Z_s)$ condition, 
	  then $\W_{q} = (S_p \cap S_q)$.
	 \item[--] Else $\W_{q} = S_q$ and $q \in \BroadcastSet$. 	 
	 \end{myitemize}
      Note that for each $\W_q$, the condition $\Q^{(1)}(\W_q, \Z_s)$  holds. This is obviously true if $\W_{q} = (S_p \cap S_q)$, since in this case $\D$ also checks that
      $\Q^{(1)}(S_p \cap S_q, \Z_s)$ holds. On the other hand, even if $\W_q = S_q$, 
      the condition $\Q^{(1)}(\W_q, \Z_s)$ holds, since $\Q^{(1)}(S_q, \Z_s)$ holds.    
      We also note that if $S_q = \Hon$, then $q \notin \BroadcastSet$ and consequently, $\D$ does not make public the share $s_q$. This is because
      the parties in $S_p \cap S_q$ will constitute a clique and the $\Q^{(1)}(S_p \cap S_q, \Z_s)$ condition will be satisfied due to the $\Q^{(2, 1)}(\PartySet, \Z_s, \Z_a)$
      condition. Hence $\W_q$ will be set to $S_p \cap S_q$.
      
       Since $\D$ eventually broadcasts $(\CandidateCoreSets, \D, S_p, \{ \W_{q}\}_{q = 1, \ldots, |\Z_s|}, \BroadcastSet, \{ s_q\}_{q \in \BroadcastSet})$,
      from the {\it $\Z_a$-weak validity} and {\it $\Z_a$-fallback validity} of $\BC$ in the {\it asynchronous} network, it follows that all the parties in $\Hon$ will eventually receive
      $(\CandidateCoreSets, \D, S_p, \{ \W_{q}\}_{q = 1, \ldots, |\Z_s|}, \BroadcastSet, \{ s_q\}_{q \in \BroadcastSet})$ from the broadcast of
      $\D$. We next show that every $P_i \in \Hon$ will eventually include $(\D, S_p)$ to the set $\C_i$. For this, we need to show that all the conditions which hold
      for $\D$ in its graph $G^{(\D)}$ when it broadcasts the $\CandidateCoreSets$ message,
       are bound to eventually hold for every $P_i \in \Hon$.    
      However, this simply follows from the {\it $\Z_a$-weak validity, $\Z_a$-fallback validity, $\Z_a$-weak consistency} and {\it $\Z_a$-fallback consistency} of 
      $\BC$ in the {\it asynchronous} network (see Theorem \ref{thm:BC})
      and the fact that edges are added to consistency graphs, based on the receipt of $\OK(\star, \star)$ messages, which are broadcasted through various $\BC$ instances.
      Consequently, any edge $(i, j)$ which is included in $G^{(\D)}$, is bound to be eventually included in the graph $G^{(i)}$ of every
      $P_i \in \Hon$.
   
     Note that corresponding to every $S_q \in \ShareSpec_{\Z_s}$, every honest $P_i \in \W_{q}$ will have the share $s_q$, which is either made public by $\D$
      as part of the $\CandidateCoreSets$ message
      or received from $\D$.
     Next, we note that each $P_i \in \Hon$ will have $\ICSig(P_j, P_i, P_k, s_q)$, corresponding to
     every $S_q \in \ShareSpec_{\Z_s}$ where $P_i \in \W_q$, from every $P_j \in \W_q$ and for every $P_k \in \PartySet$. 
     On the other hand, if any corrupt $P_j \in \W_{q}$ has $\ICSig(P_i, P_j, P_k, s'_q)$ of any honest $P_i \in \W_{q}$ for
    any $P_k  \in \PartySet$, then
    $s'_q = s_q$ holds.
     The proof for this is exactly the same as that of Lemma \ref{lemma:VSSSynchronousHonestDealer}.
     
     We now show that every honest party $P_i \in S_q \setminus W_q$ eventually sets $[s]_q$ to $s_q$, except with a probability of $\Order(|\ShareSpec_{\Z_s}| \cdot n^2 \cdot \errorAICP)$. We first show that $P_i$ sets $[s]_q$ to {\it some} value. Since 
     $\Z_s$ satisfies the $\Q^{(1)}(\W_q, \Z_s)$ condition, 
     this means that $\W_q$ contains {\it at least} one {\it honest} party, say $P_j$. Since $P_j$ follows the protocol steps honestly, it reveals 
     $\ICSig(P_k, P_j, P_i, [s]_q)$ of every $P_k \in \W_q$ to $P_i$. From the 
     {\it $\Z_a$-correctness} property of ICP in the {\it asynchronous} network
      (see Theorem \ref{thm:ICP}), it follows that $P_i$ will eventually accept these signatures.
       On the other hand, even if 
      $P_k \in \W_q$ is {\it corrupt}, then from the {\it $\Z_a$-non-repudiation} property of ICP in the {\it asynchronous} network
       (see Theorem \ref{thm:ICP}), it follows that $P_i$ eventually accepts $\ICSig(P_k, P_j, P_i, [s]_q)$, except with a probability $\errorAICP$. From the union bound, it follows that except with probability $\Order(n^2 \cdot \errorAICP)$, each party $P_i \in S_q$ outputs some value for $[s]_q$. Further, this value must be $s_q$. The proof of this follows from what was shown in \ref{lemma:VSSSynchronousHonestDealer}. Now, there can be $|\ShareSpec_{\Z_s}|$ possibilities for $S_q$. From the union bound, it follows that corresponding to each $S_q \in \ShareSpec_{\Z_s}$, every honest party in $S_q$ eventually outputs $[s]_q$, except with a probability of $\Order(|\ShareSpec_{\Z_s}| \cdot n^2 \cdot \errorAICP)$.

   Finally, the privacy for $s$ follows from the fact that throughout the protocol, the view of the adversary remains independent of the share $s_q$, corresponding to the group $S_q$, where
   $S_q = \Hon$. This is because as shown above, $q \notin \BroadcastSet$ and consequently, $\D$ does not make public the share $s_q$.
   Moreover, during the pairwise consistency tests, the view of the adversary remains independent of $s_q$, when the parties in $\Hon$ exchange IC-signed $s_q$, which
   follows from the {\it privacy} property of ICP (see Theorem \ref{thm:ICP}). Further, while computing the output, $\ICSig(P_k, P_j, P_i, s_q)$ is revealed by party $P_j \in W_q$ only to each party $P_i \in S_q \setminus W_q$. Hence, the adversary does not learn $s_q$.     
 \end{proof}
     
 Finally, we consider a {\it corrupt} $\D$ in the {\it asynchronous} network and prove the analogue of Lemma \ref{lemma:VSSSynchronousCorruptDealer}, whose proof is very similar to that
  of Lemma \ref{lemma:VSSSynchronousCorruptDealer}.
 \begin{lemma}
 \label{lemma:VSSAsynchronousCorruptDealer} 
 If the network is asynchronous and $\D$ is corrupt and if any honest party computes an output, then there exists some $s^{\star} \in \F$, such that the honest parties eventually output
  $[s^{\star}]$, except with a probability of $\Order(|\ShareSpec_{\Z_s}| \cdot n^2 \cdot \errorAICP)$. 
 \end{lemma}
 \begin{proof}
 Let $Z^{\star} \in \Z_a$ be the set of {\it corrupt} parties and let $\Hon = \PartySet \setminus Z^{\star}$ be the set of {\it honest} parties. Let $P_m \in \Hon$ be the {\it first} honest party,
  who computes an output in the protocol $\VSS$. This implies that there exists some  $S_p \in \ShareSpec_{\Z_s}$, such that $P_m$ receives a message
     $(\CandidateCoreSets, \D, S_p, \{ \W_{q}\}_{q = 1, \ldots, |\Z_s|}, \BroadcastSet, \{ s_q\}_{q \in \BroadcastSet})$ from the broadcast
     of $\D$, and includes $(\D, S_p)$ to $\C_m$.
     This further implies that $P_m$ has verified that the following hold, for 
   $q = 1, \ldots, |\Z_s|$.
    \begin{myitemize}
       \item[--] If $q \not \in \BroadcastSet$, then $\W_{q}$ is either $S_q$ or $(S_p \cap S_q)$. Moreover, 
       $\Z_s$ satisfies the $\Q^{(1)}(\W_q, \Z_s)$ condition 
       and the parties in 
       $\W_q$ constitute a clique in the graph $G^{(m)}$.
    \item[--] If  $q  \in \BroadcastSet$, then $\D$ has made public $s_q$, as part of the $\CandidateCoreSets$ message. Moreover, $\W_q$ is set to $S_q$.
    \end{myitemize}
From the {\it $\Z_a$-weak consistency} and {\it $\Z_a$-fallback consistency} of $\BC$ in the {\it asynchronous} network (see Theorem \ref{thm:BC}),
 it follows that every party in $\Hon$ will eventually receive
  $(\CandidateCoreSets, \D, S_p, \{ \W_{q}\}_{q = 1, \ldots, |\Z_s|}, \BroadcastSet, \{ s_q\}_{q \in \BroadcastSet})$ from the broadcast
     of $\D$. We next show that each $P_i \in \Hon$ will eventually include $(\D, S_p)$ to $\C_i$. 
     For this, it is enough to show that all the edges which are present in $G^{(m)}$ when $(\D, S_p)$ is included in $\C_i$, 
     are bound to be eventually present in the
      graph $G^{(i)}$ of every $P_i \in \Hon$.
      However, this simply follows from the {\it $\Z_a$-weak validity, $\Z_a$-fallback validity, $\Z_a$-weak consistency} and {\it $\Z_a$-fallback consistency} of $\BC$ in the 
      {\it asynchronous} network (see Theorem \ref{thm:BC})
      and the fact that edges are added to consistency graphs, based on the receipt of $\OK(\star, \star)$ messages, which are broadcasted through various $\BC$ instances.
           
     Next, it can be shown that corresponding to every $S_q \in \ShareSpec_{\Z_s}$, every $P_i \in \W_q$ will have a common share, say $s^{\star}_q$.
     The proof for this is the same as Lemma \ref{lemma:VSSSynchronousCorruptDealer}.
         We define
     \[ s^{\star} \defined \displaystyle \sum_{q = 1, \ldots, |\Z_s|} s^{\star}_q.    \]
     Now similar to the proof of Lemma \ref{lemma:VSSSynchronousCorruptDealer}, it can be shown that each party $P_i$ in $(\Hon \cap \W_q)$ will eventually have the required IC-signatures
     on $[s^{\star}]_q$ as part of $[s^{\star}]$ and 
     will reveal these to parties in $S_q \setminus W_q$. Consequently, each party in $\Hon$ will eventually set $[s^\star]_q$ as the share corresponding to $S_q$
      and hence, $s^{\star}$ will eventually be secret-shared, except with a probability of $\Order(|\ShareSpec_{\Z_s}| \cdot n^2 \cdot \errorAICP)$.
\end{proof}

We next derive the communication complexity of the protocol $\VSS$.
\begin{lemma}
\label{lemma:VSSCommunicationComplexity}
Protocol $\VSS$ incurs a communication of $\Order(|\Z_s| \cdot n^8 \cdot \log{|\F|} \cdot |\sigma|)$ bits.
\end{lemma}
 \begin{proof}
 In the protocol, $\D$ needs to send the share $s_q$ to {\it all} the parties in $S_q$. This incurs a total communication of $\Order(|\Z_s| \cdot n \cdot \log{|\F|})$ bits. There are
   $\Order(|\Z_s| \cdot n^3)$ instances of $\Auth$ invoked, to exchange IC-signed values, during the pairwise consistency tests. From Theorem \ref{thm:ICP},
   this incurs a total communication of 
   $\Order(|\Z_s| \cdot n^8 \cdot \log{|\F|} \cdot |\sigma|)$ bits. There are $\Order(n^2)$ $\OK$ messages which need to be broadcasted, which from Theorem \ref{thm:BC},
   incurs a total communication of
   $\Order(n^6 \cdot \log{n} \cdot |\sigma|)$ bits, since each $\OK$ message encodes the identity of two parties, requiring $\log{n}$ bits.
   Finally, $\D$ needs to broadcast a candidate $(\CandidateCoreSets, \D, S_p, \{ \W_{q}\}_{q = 1, \ldots, |\Z_s|}, \BroadcastSet, \{ s_q\}_{q \in \BroadcastSet})$ message, where
   $S_p, \BroadcastSet$ and each $\W_q$ can be represented by $\Order(n)$ bits. And corresponding to the indices in $\BroadcastSet$, the dealer $\D$ may end up broadcasting
   $\Order(|\Z_s|)$ shares. From Theorem \ref{thm:BC}, this incurs a total communication of $\Order(|\Z_s| \cdot (n^5 \cdot |\sigma| + n^4 \cdot \log{|\F|} \cdot |\sigma|))$ bits. While computing the output, $\Order(n^3 \cdot |\Z_s|)$ instances of $\Reveal$ are involved, which incur a communication of $\Order(|\Z_s| \cdot n^4 \cdot \log{|\F|})$ bits.
 \end{proof}
 
 Theorem \ref{thm:VSS} now follows from Lemma \ref{lemma:VSSSynchronousHonestDealer}-\ref{lemma:VSSCommunicationComplexity}. \\~\\
 \noindent {\bf Theorem \ref{thm:VSS}.}
 {\it Protocol $\VSS$ achieves the following, except with a probability of $\Order(|\ShareSpec_{\Z_s}| \cdot n^2 \cdot \errorAICP)$, where $\D$ has input $s \in \F$ for $\VSS$ and where 
 $\TimeVSS = \Delta + \TimeAuth + 2\TimeBC + \TimeReveal$. 
   \begin{myitemize}
   \item[--] If $\D$ is honest, then the following hold.
         \begin{myitemize}
         \item[--] {\bf $\Z_s$-correctness}: In a synchronous network, the honest parties output $[s]$ at time $\TimeVSS$. 
         \item[--] {\bf $\Z_a$-correctness}: In an asynchronous network, the honest parties eventually output  $[s]$.
         \item[--] {\bf Privacy}: Adversary's view remains independent of $s$ in any network.  
         \end{myitemize}   
   \item[--] If $\D$ is corrupt, then the following hold.
          \begin{myitemize}      
           \item[--]  {\bf $\Z_s$-commitment}: In a synchronous network, either no honest party obtains any output or there exists some
             $s^{\star} \in \F$, such that the parties output $[s^{\star}]$.
           Moreover, if any honest party computes its output corresponding to $[s^{\star}]$ at time $T$, then all honest parties
		       compute their output corresponding to $[s^{\star}]$ by time $T + \Delta$.  
	          \item[--]  {\bf $\Z_a$-commitment}: In an asynchronous network, either no honest party obtains any output or there exists some
             $s^{\star} \in \F$, such that the honest parties eventually output $[s^{\star}]$.      
          \end{myitemize}
   \item[--] {\bf Communication Complexity}: $\Order(|\Z_s| \cdot n^8 \cdot \log{|\F|} \cdot |\sigma|)$ bits are communicated by the honest parties.
   \end{myitemize}
}

%% file: AppRec.tex
\section{Network Agnostic Reconstruction Protocols and Secure Multicast}
\label{app:Rec}
This section presents our network-agnostic reconstruction protocols and secure multicast protocol, along with their
  properties. We start with the protocol $\RecShare$ 
 for reconstructing a designated share,
  presented in Fig \ref{fig:RecShare}.
\begin{protocolsplitbox}{$\RecShare([s], S_q, \ReceiverSet)$}{Network agnostic
 reconstruction protocol to reconstruct a designated share $[s]_q$. The above code is executed by each $P_i \in \PartySet$.}{fig:RecShare}
\justify
\begin{myitemize}
\item[--] {\bf Sending IC-signed Share to the Parties}: If $P_i \in \W_q$, then 
reveal $\ICSig(P_j, P_i, P_k, [s]_q)$ of every 
 $P_j \in \W_q$ to every $P_k \in \ReceiverSet \setminus \W_q$. Here $\W_q$ denotes the publicly known core-set corresponding to $S_q \in \ShareSpec_{\Z_s}$, as part of
  $[s]$.
\item[--] {\bf Computing Output}: If $P_i \in (\ReceiverSet \cap \W_q)$, then output $[s]_q$. 
 Else if $P_i \in \ReceiverSet \setminus \W_q$, then 
 check
 if there exists any $P_j \in \W_q$ and a value $s_{qj}$, such that $P_i$ has accepted
  $\ICSig(P_k, P_j, P_i, s_{qj})$ of every $P_k \in \W_q$. Upon finding such a $P_j$,
   output $[s]_q = s_{qj}$.\footnote{If there are multiple such parties $P_j$, then consider the one with the smallest index.}
\end{myitemize}
\end{protocolsplitbox}

We next prove the properties of the protocol $\RecShare$. \\~\\
\noindent {\bf Lemma \ref{lemma:RecShare}.}
{\it 
Let $s$ be a value which is linearly secret-shared with IC signatures, let $S_q \in \ShareSpec_{\Z_s}$ be a designated set and let $\ReceiverSet \subseteq \PartySet$ be a designated
 set of receivers. Then protocol $\RecShare$ achieves the following.
  \begin{myitemize}
  \item[--] {\bf $\Z_s$-correctness}: In a synchronous network, all honest parties in $\ReceiverSet$ output $[s]_q$ at time $\TimeRecShare = \TimeReveal$, 
   except with a probability of $\Order(n^2 \cdot \errorAICP)$. 
   \item[--]  {\bf $\Z_a$-correctness}: In an asynchronous network, all honest parties in $\ReceiverSet$ eventually output $[s]_q$, 
   except with a probability of $\Order(n^2 \cdot \errorAICP)$. 
     \item[--] {\bf Privacy}: If $\ReceiverSet$ consists of only honest parties, then the view of the adversary remains independent of $[s]_q$.
   \item[--] {\bf Communication Complexity}: $\Order(|\ReceiverSet| \cdot n^3 \cdot \log{|\F|})$ bits are communicated.   
  \end{myitemize}
}
\begin{proof}
 We first note that all honest parties in $(\ReceiverSet \cap \W_q)$ output $[s]_q$ correctly. So consider an {\it arbitrary} honest $P_i \in \ReceiverSet \setminus \W_q$.  
 We first show that $P_i$ indeed computes an output in the protocol, irrespective of the network type.

 Since $\Z_s$ satisfies
  the $\Q^{(1)}(\W_q, \Z_s)$ condition, it contains {\it at least} one {\it honest} party, say
   $P_j$. Since $P_j$ follows the protocol steps honestly, it reveals $\ICSig(P_k, P_j, P_i, [s]_q)$ of 
   every $P_k \in \W_q$  to $P_i$.
   From the {\it correctness} properties of ICP (see Theorem \ref{thm:ICP}), 
       it follows that $P_i$ will accept the IC-signatures
    $\ICSig(P_k, P_j, P_i, [s]_q)$, revealed by $P_j$, after time 
    $\TimeReveal$ in a {\it synchronous} network, or eventually in an {\it asynchronous} network.
    On the other hand, even if $P_k \in \W_q$ is {\it corrupt}, then also
   from the {\it non-repudiation} properties of ICP, it follows
   that $P_i$
    accepts $\ICSig(P_k, P_j, P_i, [s]_q)$, except with a probability $\errorAICP$,
    after time 
    $\TimeReveal$ in a {\it synchronous} network, or eventually in an {\it asynchronous} network.
    As there can be $\Order(n)$ {\it corrupt} parties in $\W_q$, from the union bound, it follows that 
    except with a probability $\Order(n \cdot \errorAICP)$, party $P_i$ will find a candidate party from $\W_q$, who reveals $[s]_q$, along with the IC-signature of all the parties in $\W_q$, 
     after time 
    $\TimeReveal$ in a {\it synchronous} network, or eventually in an {\it asynchronous} network.
    This is because the {\it honest} party in $\W_q$ always constitutes a candidate party.
        Now as there can be $\Order(n)$ parties in $\ReceiverSet \setminus \W_q$, it follows that 
    except with probability $\Order(n^2 \cdot \errorAICP)$, every honest party in $\ReceiverSet \setminus \W_q$ will 
    find a candidate party from $\W_q$, who reveals $[s]_q$, along with the IC-signature of all the parties in $\W_q$, 
     after time 
    $\TimeReveal$ in a {\it synchronous} network, or eventually in an {\it asynchronous} network.
    Hence all honest parties in $\ReceiverSet$ compute an output,  after time 
    $\TimeReveal$ in a {\it synchronous} network, or eventually in an {\it asynchronous} network, except 
    with a probability $\Order(n^2 \cdot \errorAICP)$.
    
    We next show that the output computed by all the honest parties is indeed correct. While this is trivially true for the parties in $(\ReceiverSet \cap \W_q)$, 
    consider an arbitrary {\it honest} party $P_i \in \ReceiverSet \setminus \W_q$.
    The above argument shows that $P_i$ computes an output in the protocol, irrespective of the network type. So let $P_i$ output $s'$. 
    We wish to show that $s' =  [s]_q$. From the protocol steps, since $P_i$ outputs $s'$, it implies that there exists some $P_j \in \W_q$, such that $P_i$ has accepted
    $\ICSig(P_k, P_j, P_i, s')$ of every $P_k \in \W_q$, revealed by $P_j$.  If $P_j$ is {\it honest}, then indeed $s' = [s]_q$, as one of the IC-signatures
    $\ICSig(P_k, P_j, P_i, s')$ is the same as $\ICSig(P_k, P_j, P_i, [s]_q)$, corresponding to the {\it honest} $P_k \in \W_q$, which is guaranteed to exist.
    So consider the case when $P_j$ is {\it corrupt}.
    Moreover, let $P_k \in \W_q$ be an {\it honest} party (which is guaranteed to exist).
    In order that $s' \neq [s]_q$, it must be the case that $P_i$ accepts  $\ICSig(P_k, P_j, P_i, s')$, revealed by $P_j$.  However, from
    the {\it unforgeability} property of ICP (see Theorem \ref{thm:ICP}), this can happen only with probability $\errorAICP$. 
    Now as there can be up to $\Order(n)$ corrupt parties in $\W_q$, from the union bound, it follows that the probability that $P_i$ outputs
    $s' \neq [s]_q$ is at most $\Order(n \cdot \errorAICP)$.
    And since there can be up to $\Order(n)$ parties in $\ReceiverSet \setminus \W_q$, it follows that except with probability at most 
    $\Order(n^2 \cdot \errorAICP)$, the output of every honest party in $\ReceiverSet$ is indeed $[s]_q$.
    
    Communication complexity follows from the communication complexity of $\Reveal$ (Theorem \ref{thm:ICP}) and the fact that $\Order(|\ReceiverSet| \cdot n^2)$ instances of $\Reveal$ are involved.
    And privacy follows from the privacy of ICP.
\end{proof}

Protocol $\Rec$ for reconstructing $s$ by a designated set of receivers is presented in Fig \ref{fig:Rec}.

\begin{protocolsplitbox}{$\Rec([s], \ReceiverSet)$}{ Network agnostic reconstruction protocol to 
 reconstruct a secret-shared value with IC signatures. The above code is executed by each party $P_i \in \PartySet$.}{fig:Rec}
\justify
\begin{myitemize}
\item[--] {\bf Reconstructing Individual Shares}: Corresponding to each $S_q \in \ShareSpec_{\Z_s}$, participate in an instance $\RecShare([s], S_q, \ReceiverSet)$
 of $\RecShare$, to let the parties in $\ReceiverSet$ reconstruct $[s]_q$.
\item[--] {\bf Computing Output}: If $P_i \in \ReceiverSet$, then output $\displaystyle s = \sum_{S_q \in \ShareSpec} [s]_q$.
\end{myitemize}
\end{protocolsplitbox}

The properties of the protocol $\Rec$ are stated in Lemma \ref{lemma:Rec}. \\~\\
\noindent {\bf Lemma \ref{lemma:Rec}.}
 {\it Let $s$ be a value which is linearly secret-shared with IC signatures and let $\ReceiverSet \subseteq \PartySet$ be a set of designated receivers.
 Then protocol $\Rec$ achieves the following.
  \begin{myitemize}
  \item[--] {\bf $\Z_s$-correctness}: In a synchronous network, all honest parties in $\ReceiverSet$ output $s$ at time $\TimeRec = \TimeRecShare$, 
   except with probability $\Order(|\ShareSpec_{\Z_s}| \cdot n^2 \cdot \errorAICP)$. 
   \item[--]  {\bf $\Z_a$-correctness}: In an asynchronous network, all honest parties in $\ReceiverSet$ eventually output $s$, 
   except with probability $\Order(|\ShareSpec_{\Z_s}| \cdot n^2 \cdot \errorAICP)$. 
    \item[--] {\bf Privacy}: If $\ReceiverSet$ consists of only honest parties, then the view of the adversary remains independent of $s$.
   \item[--] {\bf Communication Complexity}: $\Order(|\Z_s| \cdot |\ReceiverSet| \cdot n^3 \cdot \log{|\F|})$ bits are communicated.   
  \end{myitemize}
}
\begin{proof}
The proof follows from Lemma \ref{lemma:RecShare}, and the fact that $|\ShareSpec_{\Z_s}|$ instances of $\RecShare$ are invoked.
\end{proof}
\subsection{Network Agnostic Secure Multicast}
Protocol $\SVM$ is presented in Fig \ref{fig:SVM}.
\begin{protocolsplitbox}{$\SVM(\Sender, v, \ReceiverSet)$}{The network agnostic SVM protocol}{fig:SVM}
\justify
\begin{myitemize}
\item[--] {\bf Sending the Value to the Parties}: $\Sender$ on having the input $v$, invokes an instance of $\VSS$ with input $v$ and the parties in $\PartySet$ participates in this instance.
\item[--] {\bf Verifying if Sender has Committed Any Value}: Each $P_i \in \PartySet$ {\color{red} waits till its local time becomes $\TimeVSS$}, initializes a Boolean variable $\flag^{(\Sender, \ReceiverSet)}$ to $0$ and
 then do the following. 
   \begin{myitemize}
   \item[--] Upon computing an output in the $\VSS$ instance, set $\flag^{(\Sender, \ReceiverSet)}$ to $1$.
   \item[--] Upon setting $\flag^{(\Sender, \ReceiverSet)}$ to $1$, participate in an instance $\Rec([v], \ReceiverSet)$ of $\Rec$ to let the parties in $\ReceiverSet$ reconstruct $v$
   \end{myitemize}
\item[--] {\bf Computing Output}: Each $P_i \in \ReceiverSet$ upon computing an output $v$ during the instance $\Rec([v], \ReceiverSet)$, outputs $v$.
\end{myitemize}
\end{protocolsplitbox}

We next prove the properties of the protocol $\SVM$. \\~\\
\noindent {\bf Lemma \ref{lemma:SVM}.}
{\it Protocol $\SVM$ achieves the following, where $\Sender$ participates with input $v$ and where each honest party initializes $\flag^{(\Sender, \ReceiverSet)}$ to $0$.
\begin{myitemize}
\item[--] {\bf Synchronous Network}: If $\Sender$ is honest, then all honest parties set $\flag^{(\Sender, \ReceiverSet)}$ to $1$ at time $\TimeVSS$
 and except with probability $\Order(n^3 \cdot \errorAICP)$, all honest parties in $\ReceiverSet$ output $v$, after time
 $\TimeSVM = \TimeVSS + \TimeRec$.  Moreover, if $\ReceiverSet$ consists of only honest parties, then the view of $\Adv$ remains independent of $v$.
  If $\Sender$ is corrupt and some honest party sets $\flag^{(\Sender, \ReceiverSet)}$ to $1$, then there exists some $v^{\star}$ such that, 
  except with probability $\Order(n^3 \cdot \errorAICP)$, all honest parties in $\ReceiverSet$ output $v^{\star}$. Moreover, if any honest party sets 
  $\flag^{(\Sender, \ReceiverSet)}$ to $1$ at time $T$, 
  then all honest parties in $\ReceiverSet$ output $v^{\star}$ by time $T + 2\Delta$.
  \item[--] {\bf Asynchronous Network}: If $\Sender$ is honest, then except with probability $\Order(n^3 \cdot \errorAICP)$, all honest parties in $\ReceiverSet$ eventually output $v$.
   Moreover, if $\ReceiverSet$ consists of only honest parties, then the view of the adversary remains independent of $v$.
  If $\Sender$ is corrupt and some honest party sets $\flag^{(\Sender, \ReceiverSet)}$ to $1$, then there exists some $v^{\star}$ such that, 
  except with probability $\Order(n^3 \cdot \errorAICP)$, all honest parties in $\ReceiverSet$ eventually output $v^{\star}$.   
 \item[--] {\bf Communication Complexity}:  $\Order(|\Z_s| \cdot n^8 \cdot \log{|\F|} \cdot |\sigma|)$ bits are communicated.
\end{myitemize}
}
\begin{proof}
Let us first consider an {\it honest} $\Sender$. If the network is {\it synchronous}, then from the {\it $\Z_s$-correctness} of $\VSS$ in the {\it synchronous} network (Theorem \ref{thm:VSS}),
 at time $\TimeVSS$, all {\it honest} parties will output $[v]$. Consequently, each honest party will set $\flag^{(\Sender, \ReceiverSet)}$ to $1$ and start participating in the instance of $\Rec$.
  Hence, from the {\it $\Z_s$-correctness} of $\Rec$ in the {\it synchronous} network
   (Lemma \ref{lemma:Rec}), coupled with the modifications presented in Section \ref{sec:SuperPolynomial}, it follows that all honest
  parties  in $\ReceiverSet$ output $v$, except with probability $\Order(n^3 \cdot \errorAICP)$, after time $\TimeSVM = \TimeVSS + \TimeRec$.
  The privacy of $v$ follows from the {\it privacy} of $\VSS$ (Theorem \ref{thm:VSS}) and {\it privacy} of $\Rec$ (Lemma \ref{lemma:Rec}).
  The proof for the case of honest $\Sender$ in an {\it asynchronous} network is the same as above, except that we now rely on the 
   {\it $\Z_a$-correctness} of $\VSS$ in the {\it asynchronous} network (Theorem \ref{thm:VSS}) and the 
   {\it $\Z_a$-correctness} of $\Rec$ in the {\it asynchronous} network (Lemma \ref{lemma:Rec}).
   
   Next consider a {\it corrupt} $\Sender$. Let us first consider a {\it synchronous} network. Let $P_i$ be the {\it first} honest party who sets
    $\flag^{(\Sender, \ReceiverSet)}$ to $1$. This implies that there exists some $v^{\star}$, such that $P_i$ outputs $[v^{\star}]$. Let $T$ be the time when
    $P_i$ outputs $[v^{\star}]$ during the instance of $\VSS$ (and hence sets $\flag^{(\Sender, \ReceiverSet)}$ to $1$). 
    From the {\it $\Z_s$-commitment} of $\VSS$ in the {\it synchronous} network, it follows that {\it all} honest parties will output $[v^{\star}]$ 
    (and hence set $\flag^{(\Sender, \ReceiverSet)}$ to $1$), latest by time $T + \Delta$. 
    Hence all honest parties will start participating in the instance of $\Rec$, latest by time $T + \Delta$.
    Hence, from the {\it $\Z_s$-correctness} of $\Rec$ in the {\it synchronous} network
   (Lemma \ref{lemma:Rec}), coupled with the modifications presented in Section \ref{sec:SuperPolynomial}, it follows that all honest
  parties  in $\ReceiverSet$ output $v^{\star}$, except with probability $\Order(n^3 \cdot \errorAICP)$, by time $T + 2\Delta$.

  The proof for the case of a  {\it corrupt} $\Sender$ in an {\it asynchronous} network is the same as above, except that we now rely on the 
  {\it $\Z_a$-commitment} of $\VSS$ in the {\it asynchronous} network  (Theorem \ref{thm:VSS}) and the {\it $\Z_a$-correctness} of $\Rec$ in the {\it asynchronous} network (Lemma \ref{lemma:Rec}).
  
  The communication complexity follows from the communication complexity of $\VSS$ and $\Rec$.
\end{proof}

%% file: AppMDVSS.tex
\section{Properties of the Protocol $\MDVSS$}
\label{app:MDVSS}
In this section, we prove the properties of the protocol $\MDVSS$ (see Fig \ref{fig:MDVSS} for the formal description).

We first show that if the network is {\it synchronous}, then all honest parties will compute a {\it common} candidate set of committed dealers $\CD$ set by time $\TimeSVM + 2\TimeBA$, such that {\it all} honest dealers 
 are guaranteed to be present in $\CD$.
 \begin{lemma}
 \label{lemma:MVSSSynchronousCD}
 If the network is synchronous and $P_{\ell} \in \PartySet$ is an honest dealer participating with input $s^{(\ell)}$, then all the following hold in $\MDVSS$, where $\Hon$ is the set of
  honest parties.
     \begin{myitemize}
     \item[--] Except with probability $\Order(n^3 \cdot \errorAICP)$, all the parties in $\Hon$ will have a common $\CD$ set by time $\TimeSVM + 2\TimeBA$, where 
     $\Hon \subseteq \CD$. \footnote{This automatically implies that $\PartySet \setminus \CD \in \Z_s$.}
     \item[--] Except with probability $\Order(n^3 \cdot \errorAICP)$, corresponding to every dealer $P_{\ell} \in \CD$ and every $S_q \in \ShareSpec_{|\Z_s|}$,
      every party in $(\Hon \cap S_q)$ will have a common share, say ${s^{\star}}^{(\ell)}_q$, which is the same as $s^{(\ell)}_q$, for an honest $P_{\ell}$.
         \end{myitemize}
 \end{lemma}
 \begin{proof}
 Let $Z^{\star} \in \Z_s$ be the set of {\it corrupt} parties and let $\Hon = \PartySet \setminus Z^{\star}$ be the set of {\it honest} parties. 
  From the properties of $\SVM$ in the {\it synchronous} network (Lemma \ref{lemma:SVM}),
   it follows that every $P_i \in \Hon$ will set $\flag^{(P_{\ell}, S_q)}$ to $1$ at time $\TimeVSS$, corresponding to every
  $P_{\ell} \in \Hon$ and every $S_q \in \ShareSpec_{\Z_s}$. Moreover, at time $\TimeSVM$, corresponding to every
  $P_{\ell} \in \Hon$ and every $S_q \in \ShareSpec_{\Z_s}$, each $P_i \in (\Hon \cap S_q)$ computes the output $s^{(\ell)}_{qi}$ during the instance $\SVM^{(P_{\ell}, S_q)}$.
  Furthermore, except with probability $\Order(n^3 \cdot \errorAICP)$, the value $s^{(\ell)}_{qi}$ will be the same as $s^{(\ell)}_q$.
  Hence, corresponding to every $P_{\ell} \in \Hon$, every $P_i \in \Hon$ will start participating with input $1$ during the instance $\BA^{(\ell)}$, at time $\TimeSVM$.
  Hence from the {\it $\Z_s$-validity} of $\BA$ in the {\it synchronous} network (Theorem \ref{thm:BA}), at time $\TimeSVM + \TimeBA$,
   all the parties in $\Hon$ will obtain the output $1$ from the $\BA^{(\ell)}$ instances, corresponding to each $P_{\ell} \in \Hon$.
   Since $\PartySet \setminus \Hon \in \Z_s$, it follows that
   at time $\TimeSVM + \TimeBA$, all the parties in $\Hon$ start participating with input $0$ in any remaining instance $\BA^{(\star)}$ of $\BA$, for which no input is provided yet.
   Hence from the {\it $\Z_s$-security} of $\BA$ in the {\it synchronous} network (Theorem \ref{thm:BA}), all the parties in $\Hon$ will compute some output in all the $n$
   instances of $\BA^{(\star)}$ by time $\TimeSVM + 2\TimeBA$. Moreover, the outputs will be common for the parties in $\Hon$.
   Consequently, all the parties in $\Hon$ will have a common $\CD$ set at the time $\TimeSVM + 2\TimeBA$. 
   Moreover, $\Hon \subseteq \CD$, since $\CD$ includes all the dealers $P_{\ell}$ such that $\BA^{(\ell)}$ outputs $1$. And as shown above
   the $\BA^{(\ell)}$ instances corresponding to $P_{\ell} \in \Hon$ outputs $1$.
   
   Next, consider an {\it arbitrary} $P_{\ell} \in \CD$. This implies that at time $\TimeSVM + \TimeBA$, at least one party from $\Hon$, say $P_k$,
   has participated with input $1$ during the instance $\BA^{(\ell)}$. If not, then from the {\it $\Z_s$-validity} of $\BA$ in the {\it synchronous} network (Theorem \ref{thm:BA}),
   all the parties in $\Hon$ would have obtained the output $0$ from the instance $\BA^{(\ell)}$ at the time $\TimeSVM + 2\TimeBA$ and hence $P_{\ell} \not \in \CD$, which is a 
   {\it contradiction}. This implies that by the time $\TimeSVM + \TimeBA$, party $P_k$ has set $\flag^{(P_{\ell}, S_q)}$ to $1$ during the instance $\SVM(P_{\ell}, s^{(\ell)}_q, S_q)$,
    for $q = 1, \ldots, |\Z_s|$.
   So consider an {\it arbitrary} $S_q \in \ShareSpec_{\Z_s}$. From the properties of $\SVM$ in the {\it synchronous} network (Lemma \ref{lemma:SVM}),
   it follows that there exists some value ${s^{\star}}^{(\ell)}_q$, which is the same as $s^{(\ell)}_q$ for an {\it honest} $P_{\ell}$,
    such that except with probability $\Order(n^3 \cdot \errorAICP)$,
   all the parties in $\Hon$ output ${s^{\star}}^{(\ell)}_q$ during the instance $\SVM^{(P_{\ell}, S_q)}$ by time $\TimeSVM + \TimeBA + 2\Delta < \TimeSVM + 2\TimeBA$.     
 \end{proof}
 
We next show that if the network is {\it synchronous} and if an {\it honest} dealer from $\CD$
  broadcasts any set of candidate core-sets, then all honest parties will ``accept" the core-sets.
   Moreover, the dealer will {\it never} make public the share corresponding to the group from $\ShareSpec_{\Z_s}$ consisting of only honest parties while making public these core-sets.
 Furthermore, each honest dealer in $\CD$ will start making 
 public {\it at least} one candidate set of core-sets, namely the one computed with respect to the group from  $\ShareSpec_{\Z_s}$, 
  consisting of only honest parties. A consequence of all these properties is that if the dealer is honest, the adversary will not learn any information about the dealer's input. 
\begin{lemma}
\label{lemma: MVSSHonestDealerSynchronous}
If the network is synchronous and $P_{\ell} \in \CD$ is an {\it honest} dealer participating with input $s^{(\ell)}$, then all the following hold in $\MDVSS$ except with
 probability $\Order(n^3 \cdot \errorAICP)$, where $\Hon$ is the set of honest parties.
\begin{myitemize}
 \item[--]  If $S_p = \Hon$, then $P_{\ell}$ will broadcast
   $(\CandidateCoreSets, P_{\ell}, S_p, \{ \W^{(\ell)}_{p, q}\}_{q = 1, \ldots, |\Z_s|},  \BroadcastSet^{(\ell)}_p, \allowbreak \{ s^{(\ell)}_q\}_{q \in \BroadcastSet^{(\ell)}_p})$ at time
    $\TimeSVM + 2\TimeBA + \TimeAuth + \TimeBC$.
\item[--] If $P_{\ell}$ broadcasts any $(\CandidateCoreSets, P_{\ell}, S_p, \{ \W^{(\ell)}_{p, q}\}_{q = 1, \ldots, |\Z_s|}, \BroadcastSet^{(\ell)}_p, \{ s^{(\ell)}_q\}_{q \in \BroadcastSet^{(\ell)}_p})$ at time
 $T$, then every honest $P_i \in \PartySet$ will include $(P_{\ell}, S_p)$ to $\C_i$ at time $T + \TimeBC$. 
  Moreover, the following will hold.
    \begin{myitemize}
    \item[--] If $S_q = \Hon$, then $q \not \in \BroadcastSet^{(\ell)}_p$.
    \item[--] For $q = 1, \ldots, |\Z_s|$, each $\W^{(\ell)}_{p, q}$ will be either $S_q$ or $(S_p \cap S_q)$. Moreover, 
    $\Z_s$ will satisfy the $\Q^{(1)}(\W^{(\ell)}_{p, q}, \Z_s)$ condition.
    \item[--] If  $q \not \in \BroadcastSet^{(\ell)}_p$, then every honest $P_i \in S_q$ will have the share $s^{(\ell)}_q$.
    Moreover, every honest $P_i \in \W^{(\ell)}_{p, q}$ will have $\ICSig(P_j, P_i, P_k, s^{(\ell)}_q)$ of every $P_j \in \W^{(\ell)}_{p, q}$ for every $P_k \in \PartySet$.
    Furthermore, if any corrupt $P_j \in \W^{(\ell)}_{p, q}$ have $\ICSig(P_i, P_j, P_k, s'^{(\ell)}_q)$ of any honest $P_i \in \W^{(\ell)}_{p, q}$ for any $P_k \in \PartySet$, then
    $s'^{(\ell)}_q = s^{(\ell)}_q$ holds. Also, all the underlying IC-signatures will satisfy the linearity property.
    \end{myitemize}
 \item[--] The view of the adversary will be independent of $s^{(\ell)}$.
\end{myitemize}
\end{lemma}
\begin{proof}
 Let  $Z^\star \in \Z_s$ be the set of {\it corrupt} parties and let $\Hon = \PartySet \setminus Z^{\star}$ be the set of {\it honest} parties. 
   Since the dealer $P_{\ell}$ is {\it honest}, from Lemma \ref{lemma:MVSSSynchronousCD} it follows that corresponding to
   every $S_q \in \ShareSpec_{\Z_s}$, every party in $(\Hon \cap S_q)$ will have the share $s^{(\ell)}_{qi}$, by time $\TimeSVM + 2\TimeBA$, except with probability $\Order(n^3 \cdot \errorAICP)$, where
   $s^{(\ell)}_{qi} = s^{(\ell)}_q$.
  Now consider an {\it arbitrary} $S_q \in \ShareSpec_{\Z_s}$.
  At time $\TimeSVM + 2\TimeBA$, each party $P_i \in (S_q \cap \Hon)$ starts giving
   $\ICSig(P_i, P_j, P_k, s^{(\ell)}_{qi})$ to every $P_j \in S_q$, for every $P_k \in \PartySet$, where $s^{(\ell)}_{qi} = s^{(\ell)}_q$ holds.
   Then from the {\it $\Z_s$-correctness} of $\Auth$  in the {\it synchronous} network (Theorem \ref{thm:ICP}), 
   it follows that at time $\TimeSVM + 2\TimeBA + \TimeAuth$, each party $P_i \in (S_q \cap \Hon)$ will receive 
   $\ICSig(P_j, P_i, P_k, s^{(\ell)}_{qj})$ from every $P_j \in (S_q \cap \Hon)$, for every $P_k \in \PartySet$, such that $s^{(\ell)}_{qj} = s^{(\ell)}_{qi} = s^{(\ell)}_q$ holds. 
   Since $S_q$ is arbitrary, it follows that at time $\TimeSVM + 2\TimeBA + \TimeAuth$, every party $P_i \in \Hon$ broadcasts an $\OK^{(\ell)}(i, j)$ message, corresponding to every 
   $P_j \in \Hon$. From the {\it $\Z_s$-validity} of $\BC$ in the {\it synchronous} network (Theorem \ref{thm:BC}), it follows that
   these $\OK^{(\ell)}(i, j)$ messages are received by every party in $\Hon$ through {\it regular-mode} at time $\TimeSVM + 2\TimeBA + \TimeAuth + \TimeBC$.
   Consequently, the set $\Hon$ constitutes a clique in the consistency graph $G^{(\ell, i)}$ of every party $P_i \in \Hon$ at time $\TimeSVM + 2\TimeBA + \TimeAuth + \TimeBC$.
   Note that the set $\Hon \in \ShareSpec_{\Z_s}$. Let $S_p$ be the set from $\ShareSpec_{\Z_s}$, such that $S_p = \Hon$.
   From the protocol steps, it then follows that at time $\TimeSVM + 2\TimeBA + \TimeAuth + \TimeBC$, the dealer $P_{\ell}$  will compute core-sets $\W^{(\ell)}_{p, q}$ 
   for $q = 1, \ldots, |\Z_s|$ and broadcast-set $\BroadcastSet^{(\ell)}_p$ with respect to $S_p$ as follows.
	 \begin{myitemize}
	 \item[--] If $S_q$ constitutes a clique in the graph $G^{(\ell, \ell)}$, then $\W^{(\ell)}_{p, q}$ is set as $S_q$. 
	 \item[--] Else if $(S_p \cap S_q)$ constitutes a clique in $G^{(\ell, \ell)}$ and $\Z_s$ satisfies the $\Q^{(1)}(S_p \cap S_q, \Z_s)$ condition, 
	 then $\W^{(\ell)}_{p, q}$ is set as $(S_p \cap S_q)$.
	 \item[--] Else  $\W^{(\ell)}_{p, q}$ is set to $S_q$ and $q$ is included to $\BroadcastSet^{(\ell)}_p$. 	 
	 \end{myitemize}
After computing the core-sets and  broadcast-set, $P_{\ell}$ will broadcast 
 $(\CandidateCoreSets, P_{\ell}, S_p, \{ \W^{(\ell)}_{p, q}\}_{q = 1, \ldots, |\Z_s|}, \BroadcastSet^{(\ell)}_p, \{ s^{(\ell)}_q\}_{q \in \BroadcastSet^{(\ell)}_p})$ at time
  $\TimeSVM + 2\TimeBA + \TimeAuth + \TimeBC$. This proves the first part of the lemma.
  
 We next proceed to prove the second part of the lemma. So consider an {\it arbitrary} $S_p \in \ShareSpec_{\Z_s}$, such that $P_{\ell}$ compute core-sets $\W^{(\ell)}_{p, q}$ 
   for $q = 1, \ldots, |\Z_s|$ and broadcast-set $\BroadcastSet^{(\ell)}_p$ with respect to $S_p$ and broadcasts 
   $(\CandidateCoreSets, P_{\ell}, S_p, \{ \W^{(\ell)}_{p, q}\}_{q = 1, \ldots, |\Z_s|}, \BroadcastSet^{(\ell)}_p, \{ s^{(\ell)}_q\}_{q \in \BroadcastSet^{(\ell)}_p})$ at time $T$. 
   This means at time $T$, the parties in $S_p$ constitute a clique in the graph $G^{(\ell, \ell)}$. 
   We also note that $T \geq \TimeSVM + 2\TimeBA + \TimeAuth + \TimeBC$. This is because any instance of $\BC$ takes at least $\TimeBC$ time in a {\it synchronous} network to generate an output. 
   And the parties in $\Hon$ start participating in any $\BC$ instance invoked for broadcasting any $\OK^{(\ell)}(\star, \star)$ message, only {\it after} time $\TimeSVM + 2\TimeBA + \TimeAuth$.
    Consequently,
    any $\OK^{(\ell)}(\star, \star)$ message received by $P_{\ell}$, must be {\it after} time $\TimeSVM + 2\TimeBA + \TimeAuth + \TimeBC$.
    We also note that any edge $(P_j, P_k)$ which is present in the graph $G^{(\ell, \ell)}$ of $P_{\ell}$ at time $T$, is bound to be present in the graph $G^{(\ell, i)}$ of every $P_i \in \Hon$, latest
    by time $T + \Delta$. This is because the edge $(P_j, P_k)$ is added to $G^{(\ell, \ell)}$ upon the receipt of $\OK^{(\ell)}(j, k)$ and $\OK^{(\ell)}(k, j)$ messages from the broadcast of $P_j$
    and $P_k$ respectively. And from the {\it $\Z_s$-validity, $\Z_s$-consistency} and {\it $\Z_s$-fallback consistency} of $\BC$ in the {\it synchronous} network, these $\OK^{(\ell)}(\star, \star)$
    messages will be received by every party $P_i \in \Hon$, latest by time $T + \Delta$.
        Since $P_{\ell}$ is assumed to be {\it honest}, it follows that the sets $\{ \W^{(\ell)}_{p, q}\}_{q = 1, \ldots, |\Z_s|}$ and $\BroadcastSet^{(\ell)}_p$ satisfy the following properties.
    \begin{myitemize}
   	\item[--] If $S_q$ constitutes a clique in the graph $G^{(\ell, \ell)}$, then $\W^{(\ell)}_{p, q}$ is set as $S_q$. 
	 \item[--] Else if $(S_p \cap S_q)$ constitutes a clique in $G^{(\ell, \ell)}$ and $\Z_s$ satisfies the $\Q^{(1)}(S_p \cap S_q, \Z_s)$ condition,
	  then $\W^{(\ell)}_{p, q}$ is set as $(S_p \cap S_q)$.
	 \item[--] Else  $\W^{(\ell)}_{p, q}$ is set to $S_q$ and $q$ is included to $\BroadcastSet^{(\ell)}_p$. 	 
    \end{myitemize}
    We also note that if $S_q = \Hon$, then $q \not \in \BroadcastSet^{(\ell)}_p$ and consequently, $P_{\ell}$ will {\it not} make the share $s^{(\ell)}_q$ public.
    This is because $T \geq \TimeSVM + 2\TimeBA + \TimeAuth + \TimeBC$. And as shown in the proof of the first part, the set $\Hon$ will constitute a clique in the graph $G^{(\ell, \ell)}$ at time 
    $ \TimeSVM + 2\TimeBA + \TimeAuth + \TimeBC$. Since $P_{\ell}$ is {\it honest}, from the {\it $\Z_s$-validity} of $\BC$ in the {\it synchronous} network, it follows that all the parties in
    $\Hon$ will receive $(\CandidateCoreSets, P_{\ell}, S_p, \{ \W^{(\ell)}_{p, q}\}_{q = 1, \ldots, |\Z_s|}, \BroadcastSet^{(\ell)}_p, \{ s^{(\ell)}_q\}_{q \in \BroadcastSet^{(\ell)}_p})$
    through the regular-mode at time $T + \TimeBC$. Moreover, each party $P_i \in \Hon$ will include $(P_{\ell}, S_p)$ to the set $\C_i$ at time $T + \TimeBC$. 
    This is because since $P_{\ell}$ has computed the sets 
     $\{ \W^{(\ell)}_{p, q}\}_{q = 1, \ldots, |\Z_s|}$ and $\BroadcastSet^{(\ell)}_p$ {\it honestly}, these sets will pass all the verifications for each $P_i \in \Hon$ at time $T + \Delta$.
    
     Next consider an {\it arbitrary} $q \not \in \BroadcastSet^{(\ell)}_p$. This implies that $P_{\ell}$ has set $\W^{(\ell)}_{p, q}$ as $(S_p \cap S_q)$ because
     the parties in $(S_p \cap S_q)$ constitutes a clique in the graph $G^{(\ell, \ell)}$. Now consider an {\it arbitrary} $P_i \in (\Hon \cap S_q)$. 
        This implies that $P_i$ has computed $s^{(\ell)}_{qi}$ during the instance $\SVM(P_{\ell}, s^{(\ell)}_q, S_q)$
         at time $\TimeSVM$, which will be the same as $s^{(\ell)}_q$, since $P_{\ell}$ is {\it honest}.
     Next consider {\it arbitrary} $P_i, P_j \in \W^{(\ell)}_{p, q}$, such that $P_j \neq P_i$. This implies that the edge $(i, j)$ is present in the graph $G^{(\ell, \ell)}$, which further implies that
     $P_i$ has broadcasted the message $\OK^{(\ell)}(i, j)$. This further implies that $P_i$ must have received 
     $\ICSig(P_j, P_i, P_k, s^{(\ell)}_{qj})$ from $P_j$, for every $P_k \in \PartySet$, such that $s^{(\ell)}_{qj} = s^{(\ell)}_{qi}$ holds.
     Since $s^{(\ell)}_{qi} = s^{(\ell)}_q$, it follows that $\ICSig(P_j, P_i, P_k, s^{(\ell)}_{qj})$ is the same as $\ICSig(P_j, P_i, P_k, s^{(\ell)}_{q})$.
     On the other hand, consider an {\it arbitrary} $P_j \in \W^{(\ell)}_{p, q}$, such that $P_j$ is {\it corrupt} and where $P_j$ has received
      $\ICSig(P_i, P_j, P_k, s'^{(\ell)}_{q})$ from $P_i$, for any $P_k  \in \PartySet$.
      Then from the protocol steps, it follows that $s'^{(\ell)}_{q} = s^{(\ell)}_{qi}$, since $P_i$ gives the IC-signature on the share $s^{(\ell)}_{qi}$, received from $P_{\ell}$.
      And since $s^{(\ell)}_{qi} = s^{(\ell)}_q$, it follows that  $\ICSig(P_i, P_j, P_k, s'^{(\ell)}_{q})$ is the same as  $\ICSig(P_i, P_j, P_k, s^{(\ell)}_{q})$.
      The linearity of the underlying IC-signatures follow from the fact the parties follow the linearity principle while generating IC-signatures.
      
      Finally, the privacy of $s^{(\ell)}$ follows from the fact that throughout the protocol, adversary {\it does not} learn anything about the share $s^{(\ell)}_q$, provided $S_q = \Hon$.
      Namely, during the instance $\SVM(P_{\ell}, s^{(\ell)}_q, S_q)$ where
       $(S_q \cap Z^{\star}) = \emptyset$, the view of the adversary remains independent of $s^{(\ell)}_q$, which follows from the privacy of $\SVM$ (Lemma \ref{lemma:SVM}). 
      Moreover, as shown above, $P_{\ell}$ never makes public the share  $s^{(\ell)}_q$, as $q \not \in \BroadcastSet^{(\ell)}_p$.
      Furthermore, since the set $\Hon$ will consists of {\it only} honest parties, from the {\it privacy} of ICP (see Theorem \ref{thm:ICP}), it follows that the adversary does not learn any additional information about
      $s^{(\ell)}_q$, when the parties in $\Hon$ exchange IC-signed $s^{(\ell)}_q$ during the pairwise consistency tests. 
\end{proof}

We next show that if the network is {\it synchronous}, then any candidate set of core-sets ``accepted" on the behalf of a {\it corrupt} dealer
 by any {\it honest} party at the time $T$, is bound to be accepted by all honest parties, latest by time $T + \Delta$.
\begin{lemma}
\label{lemma: MVSSHonestPartiesSynchronous}
If the network is synchronous and if in $\MDVSS$ any honest party $P_i$ receives 
 $(\CandidateCoreSets, P_{\ell}, S_p, \{ \W^{(\ell)}_{p, q}\}_{q = 1, \ldots, |\Z_s|},  \BroadcastSet^{(\ell)}_p,  \{ s^{(\ell)}_q\}_{q \in \BroadcastSet^{(\ell)}_p})$ from the broadcast of any 
  corrupt dealer $P_{\ell} \in \CD$
  and includes $(P_{\ell}, S_p)$ to $\C_i$ at time $T$, then all honest parties $P_j$ will receive 
   $(\CandidateCoreSets, P_{\ell}, S_p, \{ \W^{(\ell)}_{p, q}\}_{q = 1, \ldots, |\Z_s|},  \BroadcastSet^{(\ell)}_p,  \{ s^{(\ell)}_q\}_{q \in \BroadcastSet^{(\ell)}_p})$ from the broadcast of $P_{\ell}$
   and include  $(P_{\ell}, S_p)$ to $\C_j$ by time $T + \Delta$.
   Moreover, for $q = 1, \ldots, |\Z_s|$, the following holds, except with probability $\Order(n^3 \cdot \errorAICP)$.
     \begin{myitemize}
     \item[--] $\W^{(\ell)}_{p, q}$ is either $S_q$ or $(S_p \cap S_q)$. Moreover, $\Z_s$ satisfies the
     $\Q^{(1)}(\W^{(\ell)}_{p, q} , \Z_s)$ condition.     
     \item[--] If  $q \not \in \BroadcastSet^{(\ell)}_p$, then every honest $P_i \in S_q$ will have a common share, say ${s^{\star}}^{(\ell)}_q$.
    Moreover, every honest $P_i \in \W^{(\ell)}_{p, q}$ will have $\ICSig(P_j, P_i, P_k, {s^{\star}}^{(\ell)}_q)$ of every $P_j \in \W^{(\ell)}_{p, q}$ and for every $P_k \in \PartySet$.
    Furthermore, if any corrupt $P_j \in \W^{(\ell)}_{p, q}$ has $\ICSig(P_i, P_j, P_k, s'^{(\ell)}_q)$ of any honest $P_i \in \W^{(\ell)}_{p, q}$ for any $P_k \in \PartySet$, then
    $s'^{(\ell)}_q = {s^{\star}}^{(\ell)}_q$ holds. Also, all the underlying IC-signatures will satisfy the linearity principle.
     \end{myitemize}   
\end{lemma}
\begin{proof}
The proof follows very closely the proof of Lemma \ref{lemma: MVSSHonestDealerSynchronous}. Let $Z^{\star} \in \Z_s$ be the set of {\it corrupt} parties and let $\Hon = \PartySet \setminus Z^{\star}$
 be the set of {\it honest} parties. Now consider 
 an {\it arbitrary corrupt} dealer $P_{\ell} \in \CD$ and an {\it arbitrary} $P_i \in \Hon$, such that 
  $P_i$ receives 
  $(\CandidateCoreSets, P_{\ell}, S_p, \{ \W^{(\ell)}_{p, q}\}_{q = 1, \ldots, |\Z_s|},  \BroadcastSet^{(\ell)}_p,  \{ s^{(\ell)}_q\}_{q \in \BroadcastSet^{(\ell)}_p})$ from the broadcast of
   $P_{\ell}$
  and includes $(P_{\ell}, S_p)$ to $\C_i$ at time $T$. Now consider another {\it arbitrary} $P_j \in \Hon$, such that $P_j \neq P_i$.
  From the {\it $\Z_s$-consistency} and {\it $\Z_s$-fallback consistency} of $\BC$ in the {\it synchronous} network, it follows that $P_j$ is bound to receive
  $(\CandidateCoreSets, P_{\ell}, S_p, \{ \W^{(\ell)}_{p, q}\}_{q = 1, \ldots, |\Z_s|},  \BroadcastSet^{(\ell)}_p,  \{ s^{(\ell)}_q\}_{q \in \BroadcastSet^{(\ell)}_p})$
  from the broadcast of $P_{\ell}$, latest by time $T + \Delta$. We wish to show that $P_j$ will include $(P_{\ell}, S_p)$ to $\C_j$, by time $T + \Delta$.
  For this, we note that since $P_i$ has included $(P_{\ell}, S_p)$ to $\C_i$ at time $T$, 
  {\it all} the following conditions hold for $P_i$ at time $T$, 
    for $q = 1, \ldots, |\Z_s|$.
      \begin{myitemize}
      \item[--] If $q \in \BroadcastSet^{(\ell)}_p$, then the set $\W^{(\ell)}_{p, q} = S_q$.
      \item[--] If $(q \not \in \BroadcastSet^{(\ell)}_p)$, then $\W^{(\ell)}_{p, q}$ is either $S_q$ or  $(S_p \cap S_q)$, such that:
         \begin{myitemize}
         \item[--] If $\W^{(\ell)}_{p, q} = S_q$, then $S_q$ constitutes a clique in $G^{(\ell, i)}$.
         \item[--] Else if $\W^{(\ell)}_{p, q} = (S_p \cap S_q)$, then 
      $(S_p \cap S_q)$ constitutes a clique in $G^{(\ell, i)}$ and 
      $\Z_s$ satisfies the $\Q^{(1)}(S_p \cap S_q, \Z_s)$ condition.      
         \end{myitemize}
           \end{myitemize}
We claim that all the above conditions will hold {\it even} for $P_j$ by time $T + \Delta$. This is because {\it all} the edges which are present in the consistency graph $G^{(\ell, i)}$ 
 at time $T$ are bound to be present in the consistency graph $G^{(\ell, j)}$ by time $T + \Delta$.
  This follows from the {\it $\Z_s$-validity, $\Z_s$-consistency} and {\it $\Z_s$-fallback consistency} of $\BC$ in the {\it synchronous} network (see Theorem \ref{thm:BC})
  and the fact that the edges in the graph $G^{(\ell, i)}$ are based on $\OK^{(\ell)}(\star, \star)$ messages, which are received through various $\BC$ instances.

Next consider an {\it arbitrary} $q \not \in \BroadcastSet^{(\ell)}_p$. This implies that $\W^{(\ell)}_{p, q}$ is set as $(S_p \cap S_q)$ and
     all the parties in $(S_p \cap S_q)$ constitute a clique in the consistency graph of every party in $\Hon$. 
    From the properties of $\SVM$ in the {\it synchronous} network, all honest parties in $S_q$ compute a common output, say ${s^{\star}}^{(\ell)}_q$, during the instance 
    $\SVM(P_{\ell}, s^{(\ell)}_q, S_q)$.      
     Next consider an {\it arbitrary} $P_i \in (\Hon \cap \W^{(\ell)}_{p, q})$ and any {\it arbitrary} $P_j \in \W^{(\ell)}_{p, q}$.
     This implies that $P_i$ has broadcasted the message $\OK^{(\ell)}(i, j)$, after receiving 
      $\ICSig(P_j, P_i, P_k, s^{(\ell)}_{qj})$ from $P_j$, for every $P_k \in \PartySet$,
      and verifying that the share $s^{(\ell)}_{qj}$ is the {\it same} as the one, computed during the instance $\SVM^{(P_{\ell}, s^{(\ell)}_q, S_q)}$.
     Since the share computed by $P_i$ during $\SVM(P_{\ell}, s^{(\ell)}_q, S_q)$ is ${s^{\star}}^{(\ell)}_q$, it follows that
      $\ICSig(P_j, P_i, P_k, s^{(\ell)}_{qj})$ is the same as $\ICSig(P_j, P_i, P_k, {s^{\star}}^{(\ell)}_q)$.
      On the other hand, since $P_i$ gives its IC-signature on ${s^{\star}}^{(\ell)}_q$ to {\it every} $P_j \in S_q$, it follows that
      if any {\it corrupt} $P_j \in \W^{(\ell)}_{p, q}$ has $\ICSig(P_i, P_j, P_k, s'^{(\ell)}_q)$ from $P_i$ for any $P_k \in \PartySet$, then
      $s'^{(\ell)}_q = {s^{\star}}^{(\ell)}_q$ holds. The linearity of the underlying IC-signatures follow from the fact that the parties follow the linearity principle while generating the IC-signatures.
\end{proof}

Now based on the previous two lemmas, we show that in a {\it synchronous} network, all honest parties will output a ``legitimate"
 set of parties $\Core$ after time $\TimeSVM + \TimeAuth + 2\TimeBC + 6\TimeBA$, such that {\it at least} one honest party is present in $\Core$.
  And corresponding to every party in $\Core$, there exists some value, which is linearly secret-shared with IC-signatures. Moreover, the values corresponding to the honest parties remain private.
 \begin{lemma}
 \label{lemma:MVSSSynchronousProperties}
  If the network is synchronous, then in $\MDVSS$, except with probability $\Order(n^3 \cdot \errorAICP)$,
   at the time $\TimeMDVSS = \TimeSVM  + \TimeAuth + 2\TimeBC + 6\TimeBA$, all honest parties output a common set $\Core$, 
   such that at least one honest party will be present in $\Core$.
   Moreover, corresponding to every $P_{\ell} \in \Core$, there exists some ${s^{\star}}^{(\ell)}$, where ${s^{\star}}^{(\ell)} = s^{(\ell)}$ for an honest $P_{\ell}$,
   which is the input of $P_{\ell}$ for $\MDVSS$,
       such that the values $\{  {s^{\star}}^{(\ell)}  \}_{P_{\ell} \in \Core}$ are linearly secret-shared with IC-signatures. Furthermore, 
       if $P_{\ell}$ is honest, then adversary's view is independent of $s^{(\ell)}$.
 \end{lemma}
\begin{proof}
Let $Z^{\star} \in \Z_s$ be the set of {\it corrupt} parties and let $\Hon = \PartySet \setminus Z^{\star}$ be the set of honest parties. 
 From Lemma \ref{lemma:MVSSSynchronousCD},  except with probability $\Order(n^3 \cdot \errorAICP)$, all the parties in $\Hon$ will have a common set of committed dealers
  $\CD$ by the time $\TimeSVM + 2\TimeBA$, where 
     $\Hon \subseteq \CD$. Moreover, corresponding to every dealer $P_{\ell} \in \CD$ and every $S_q \in \ShareSpec_{|\Z_s|}$,
      every party in $(\Hon \cap S_q)$ will have a common share, say ${s^{\star}}^{(\ell)}_q$, which is the same as $s^{(\ell)}_q$, for an $P_{\ell}$.
   We begin by showing that once the set of committed dealers $\CD$ is decided, then all the $|\Z_s|$ instances $\BA^{(1, \star)}$ of $\BA$
   and then all the $|\CD|$ instances $\BA^{(2, \star)}$ of $\BA$ will produce some output, for all the parties in $\Hon$, by time $\TimeMDVSS$.

 Consider the set $S_p \in \ShareSpec_{\Z_s}$, such that $S_p = \Hon$. 
 Then corresponding to each $P_{\ell} \in (\Hon \cap \CD)$, every $P_i \in \Hon$ will receive
  $(\CandidateCoreSets, P_{\ell}, S_p, \{ \W^{(\ell)}_{p, q}\}_{q = 1, \ldots, |\Z_s|},  \BroadcastSet^{(\ell)}_p,  \{ s^{(\ell)}_q\}_{q \in \BroadcastSet^{(\ell)}_p})$
  from the broadcast of $P_{\ell}$  
   and includes $(P_{\ell}, S_p)$ to $\C_i$ at time  $\TimeSVM + 2\TimeBA + \TimeAuth + 2\TimeBC$ (see Lemma \ref{lemma: MVSSHonestDealerSynchronous}). 
    Since $\PartySet \setminus \Hon \in \Z_s$ and $\Hon \subseteq \CD$,
     it follows that at the time  $\TimeSVM + 2\TimeBA+ \TimeAuth + 2\TimeBC$, every party $P_i \in \Hon$ will have a set ${\cal A}_{p, i}$ (namely ${\cal A}_{p, i} = \Hon$), 
    where $\CD \setminus {\cal A}_{p, i} \in \Z_s$
      and where $(P_{\ell}, S_p) \in \C_i$ for every $P_{\ell} \in {\cal A}_{p, i}$.
      Consequently, each $P_i \in \Hon$ starts participating in the instance $\BA^{(1, p)}$ 
      with input $1$, at the time $\TimeSVM + 2\TimeBA + \TimeAuth + 2\TimeBC$.
       From the {\it $\Z_s$-security} of $\BA$ in the {\it synchronous} network (see Theorem \ref{thm:BA}),
        it follows that at the time $\TimeSVM + \TimeAuth + 2\TimeBC + 3 \TimeBA$, every $P_i \in \Hon$ obtains the output $1$ from the instance $\BA^{(1, p)}$. 
    Consequently, at the time $\TimeSVM +  \TimeAuth + 2\TimeBC + 3 \TimeBA$,
   every party in $\Hon$ will start participating in the remaining $\BA^{(1, \star)}$ instances for which no input has been provided yet (if there are any), with input $0$.  
  And from the {\it $\Z_s$-security} of $\BA$ in the {\it synchronous} network,
   these $\BA^{(1, \star)}$  instances will produce common outputs, for every party in $\Hon$, at the time $\TimeSVM + \TimeAuth + 2\TimeBC + 4\TimeBA$.
  As a result, at the time $\TimeSVM + \TimeAuth + 2\TimeBC + 4\TimeBA$, all the parties in $\Hon$ will compute a $\qcore$. Moreover, $\qcore$ will be {\it common} for all the
  parties in $\Hon$, since it corresponds to the least-indexed $\BA^{(1, \star)}$ instance among $\BA^{(1, 1)}, \ldots, \BA^{(1, |\Z_s|)}$, which produces output $1$. And from the
  {\it $\Z_s$-security} of $\BA$ in the {\it synchronous} network, each $\BA^{(1, q)}$ instance produces a common output for every party in $\Hon$.
  We also note that $\qcore$ will be indeed set to some value from the set $\{1, \ldots, |\Z_s| \}$. This is because as shown above, the index $p$ where $S_p = \Hon$ always constitute a candidate
  $\qcore$.

  We next claim that corresponding to $S_{\qcore}$, 
  there exists a subset of parties ${\cal B}_{\qcore}$, where $\CD \setminus {\cal B}_{\qcore} \in \Z_s$, such that corresponding to every $P_j \in {\cal B}_{\qcore}$,
  the ordered pair $(P_j, S_{\qcore})$ is present in the set $\C_i$ of every $P_i \in \Hon$, at the time $\TimeSVM  + \TimeAuth + 2\TimeBC + 4\TimeBA$.
  Assuming that the claim is true, it implies that all the parties in $\Hon$ will participate with input $1$ in the instances $\BA^{(2, j)}$, corresponding to every
  $P_j \in {\cal B}_{\qcore}$, at the time $\TimeSVM +  \TimeAuth + 2\TimeBC + 4\TimeBA$. And hence from the {\it $\Z_s$-security} of the $\BA$ in the {\it synchronous} network (Theorem \ref{thm:BA}),
  all the parties will obtain the output $1$ in the $\BA^{(2, j)}$ instances, corresponding to every
  $P_j \in {\cal B}_{\qcore}$, at the time $\TimeSVM  + \TimeAuth + 2\TimeBC + 5\TimeBA$. As a result, at time $\TimeSVM + \TimeAuth + 2\TimeBC + 5\TimeBA$, all the parties in $\Hon$ will start participating 
  in the remaining $\BA^{(2, \star)}$ instances for which no input has been provided yet (if there are any), with input $0$. 
  Consequently, from the {\it $\Z_s$-security} of the $\BA$ in the {\it synchronous} network (Theorem \ref{thm:BA}), 
  at the time $\TimeMDVSS$, all the parties in $\Hon$ will have some output from all the $|\CD|$ instances of $\BA^{(2, \star)}$. 
  Moreover, the outputs will be {\it common} for all the parties in $\Hon$. Furthermore, the parties in $\Hon$ will have a subset $\Core$, which corresponds to all the 
   $\BA^{(2, j)}$ instances, which have produced output $1$. Note that $(\Hon \cap \Core) \neq \emptyset$. This is because ${\cal B}_{\qcore} \subseteq \Core$ and $\Hon \subseteq \CD$.
   Consequently, $(\Hon \cap {\cal B}_{\qcore}) \neq \emptyset$, as otherwise $\Z_s$ {\it does not} satisfy the $\Q^{(2)}(\PartySet, \Z_s)$ condition, which is a contradiction.
    
   We next proceed to prove our claim. Since the instance $\BA^{(1, \qcore)}$ has produced output $1$, it follows that at least one party from $\Hon$, say $P_k$, have participated with input $1$
   in the instance $\BA^{(1, \qcore)}$. This is because if {\it all} the parties in $\Hon$ participates with input $0$ in the instance  $\BA^{(1, \qcore)}$, then from the {\it $\Z_s$-validity} of $\BA$ in the
   {\it synchronous} network (Theorem \ref{thm:BA}), all the parties in $\Hon$ would have obtained the output $0$ from the instance $\BA^{(1, \qcore)}$, which is a contradiction.
   We also note that $P_k$ would have started participating with input $1$ in the instance $\BA^{(1, \qcore)}$, latest by time $\TimeSVM + \TimeAuth + 2\TimeBC + 3 \TimeBA$.
   This is because as argued above, by time $\TimeSVM   + \TimeAuth + 2\TimeBC + 3\TimeBA$, all the parties in $\Hon$ would have started participating in {\it all} the $|\Z_s|$ instances
   of $\BA^{(1, \star)}$, with {\it some} input. Now since $P_k$ has participated with input $1$
   in the instance $\BA^{(1, \qcore)}$, it follows that at the time $\TimeSVM + \TimeAuth + 2\TimeBC + 3\TimeBA$, there exists a subset of parties ${\cal A}_{\qcore, k}$, where 
   $\CD \setminus {\cal A}_{\qcore, k} \in \Z_s$, such that $(P_{\ell}, S_{\qcore})$ is present in the set $\C_k$, corresponding to every $P_{\ell} \in {\cal A}_{\qcore, k}$.
   We show that the set ${\cal A}_{\qcore, k}$ constitutes the candidate ${\cal B}_{\qcore}$. For this, note that for any $P_{\ell} \in {\cal A}_{\qcore, k}$,
   party $P_k$ includes $(P_{\ell}, S_{\qcore})$ to $\C_i$, only after receiving
   a message $(\CandidateCoreSets, P_{\ell}, S_{\qcore}, \{ \W^{(\ell)}_{\qcore, q}\}_{q = 1, \ldots, |\Z_s|},  \BroadcastSet^{(\ell)}_{\qcore},  \{ s^{(\ell)}_q\}_{q \in \BroadcastSet^{(\ell)}_{\qcore}})$
   from the broadcast of $P_{\ell}$
   and verifying it. Moreover, $P_k$ must have received the message
    $(\CandidateCoreSets, P_{\ell}, S_{\qcore}, \{ \W^{(\ell)}_{\qcore, q}\}_{q = 1, \ldots, |\Z_s|},  \BroadcastSet^{(\ell)}_{\qcore},  \{ s^{(\ell)}_q\}_{q \in \BroadcastSet^{(\ell)}_{\qcore}})$ from 
    each $P_{\ell} \in {\cal A}_{\qcore, k}$, latest by time $\TimeSVM  + \TimeAuth + 2\TimeBC + 3\TimeBA$. 
    It then follows from Lemma \ref{lemma:MVSSSynchronousProperties} that by time $\TimeSVM + \TimeAuth + 2\TimeBC + 3\TimeBA + \Delta < \TimeSVM + \TimeAuth + 2\TimeBC + 4\TimeBA$,
    {\it every} party in $\Hon$ would have received  
    $(\CandidateCoreSets, P_{\ell}, S_{\qcore}, \{ \W^{(\ell)}_{\qcore, q}\}_{q = 1, \ldots, |\Z_s|},  \BroadcastSet^{(\ell)}_{\qcore},  \{ s^{(\ell)}_q\}_{q \in \BroadcastSet^{(\ell)}_{\qcore}})$ from 
    each $P_{\ell} \in {\cal A}_{\qcore, k}$. And hence each party $P_i \in \Hon$ would include $(P_{\ell}, S_{\qcore})$ to the set $\C_i$, corresponding to every $P_{\ell} \in {\cal A}_{\qcore, k}$,
    by time $\TimeSVM + \TimeAuth + 2\TimeBC + 4\TimeBA$. This proves our claim. 
    
    We next claim that at the time $\TimeMDVSS$, corresponding to every $P_{\ell} \in \Core$, every $P_i \in \Hon$ would have received a message 
    $(\CandidateCoreSets, P_{\ell}, S_{\qcore}, \{ \W^{(\ell)}_{\qcore, q}\}_{q = 1, \ldots, |\Z_s|},  \BroadcastSet^{(\ell)}_{\qcore},  \{ s^{(\ell)}_q\}_{q \in \BroadcastSet^{(\ell)}_{\qcore}})$ from 
    the broadcast of $P_{\ell}$. The proof for this is very similar to the proof of the previous claim and relies on the properties of $\BA$. So consider an {\it arbitrary} $P_{\ell} \in \Core$.
    This implies that the instance  $\BA^{(2, \ell)}$ has produced output $1$ for {\it all} the parties in $\Hon$, which further implies that at least one party from $\Hon$, say $P_m$, has participated with input
    $1$ during the instance $\BA^{(2, \ell)}$. If not, then from the {\it $\Z_s$-validity} of $\BA$ in the {\it synchronous} network (Theorem \ref{thm:BA}), the instance $\BA^{(2, \ell)}$
    would have produced output $0$ for all the parties in $\Hon$, which is a contradiction. 
    We also note that $P_m$ must have started participating in the instance  $\BA^{(2, \ell)}$, latest by time $\TimeSVM + \TimeAuth + 2\TimeBC + 5\TimeBA$.
    This is because as shown above, by time $\TimeSVM + \TimeAuth + 2\TimeBC + 5\TimeBA$, all the parties in $\Hon$ would have started participating in {\it all} the $|\CD|$
    instances of $\BA^{(2, \star)}$, with some input. Now since $P_m$ participates with input $1$ in the instance $\BA^{(2, \ell)}$, it follows that by time 
    $\TimeSVM + \TimeAuth + 2\TimeBC + 5\TimeBA$, party $P_m$ must have received a message 
     $(\CandidateCoreSets, P_{\ell}, S_{\qcore}, \{ \W^{(\ell)}_{\qcore, q}\}_{q = 1, \ldots, |\Z_s|},  \BroadcastSet^{(\ell)}_{\qcore},  \{ s^{(\ell)}_q\}_{q \in \BroadcastSet^{(\ell)}_{\qcore}})$
     from the broadcast of $P_{\ell}$ and included $(P_{\ell}, S_{\qcore})$ to $\C_m$. 
      It then follows from Lemma \ref{lemma:MVSSSynchronousProperties} that by time $\TimeSVM + \TimeAuth + 2\TimeBC + 5\TimeBA + \Delta < \TimeSVM + \TimeAuth + 2\TimeBC + 6\TimeBA$,
      {\it every} party $P_i$ in $\Hon$ would have received  
    $(\CandidateCoreSets, P_{\ell}, S_{\qcore}, \{ \W^{(\ell)}_{\qcore, q}\}_{q = 1, \ldots, |\Z_s|},  \BroadcastSet^{(\ell)}_{\qcore},  \{ s^{(\ell)}_q\}_{q \in \BroadcastSet^{(\ell)}_{\qcore}})$ from 
    $P_{\ell}$ and would include $(P_{\ell}, S_{\qcore})$ to $\C_i$.
    
     Till now we have shown that all the all the $|\Z_s|$ instances $\BA^{(1, \star)}$ of $\BA$
      and then all the $|\CD|$ instances $\BA^{(2, \star)}$ of $\BA$ will produce some output, for all the parties in $\Hon$, by time $\TimeMDVSS$. Moreover, at the time
     $\TimeMDVSS$, all the parties in $\Hon$ will have a common $\qcore \in \{1, \ldots, |\Z_s| \}$ and a common set $\Core \subseteq \PartySet$, where $\Core$ has at least one {\it honest} party.
     Furthermore, corresponding to every $P_{\ell} \in \Core$, each $P_i \in \Hon$ would have received a message 
     $(\CandidateCoreSets, P_{\ell}, S_{\qcore}, \{ \W^{(\ell)}_{\qcore, q}\}_{q = 1, \ldots, |\Z_s|},  \BroadcastSet^{(\ell)}_{\qcore},  \{ s^{(\ell)}_q\}_{q \in \BroadcastSet^{(\ell)}_{\qcore}})$
     from the broadcast of $P_{\ell}$. Furthermore, from Lemma \ref{lemma:MVSSSynchronousProperties}, corresponding to each $P_{\ell} \in \Core$, the set 
     $\W^{(\ell)}_{\qcore, q}$ will be either the set $S_q$ or $(S_{\qcore} \cap S_q)$, for $q = 1, \ldots, |\Z_s|$. If 
     $\W^{(\ell)}_{\qcore, q} = S_q$ for {\it every} $P_{\ell} \in \Core$, then all the parties in $\Hon$ will set $\W_q$ to $S_q$.
     On the other hand, if $\W^{(\ell)}_{\qcore, q} = (S_{\qcore} \cap S_q)$ for {\it any} $P_{\ell} \in \Core$, then all the parties in 
     $\Hon$ will set $\W_q$ to $(S_{\qcore} \cap S_q)$. Irrespective of the case, all the parties in $\Hon$ would set $\W_q$ to a common subset. 
     We also note that irrespective of the case, $S_q \setminus \W_q \in \Z_a$ holds. This is because from Lemma \ref{lemma:MVSSSynchronousProperties}
     and Lemma \ref{lemma: MVSSHonestPartiesSynchronous}, the condition $S_q \setminus \W^{(\ell)}_{\qcore, q} \in \Z_a$ holds, corresponding to {\it every}
     $P_{\ell} \in \Core$.
     
     Finally consider an {\it arbitrary} $P_{\ell} \in \Core$ and an {\it arbitrary} $S_q \in \ShareSpec_{\Z_s}$. We claim that at the time $\TimeMDVSS$, all the parties in $(\Hon \cap S_q)$
     will have a common share, say ${s^{\star}}^{(\ell)}_q$, where ${s^{\star}}^{(\ell)}_q = s^{(\ell)}_q$ for an {\it honest} $P_{\ell}$. For this, we consider two possible cases.
      If $q \in \BroadcastSet^{(\ell)}_{\qcore}$, then each $P_i \in (\Hon \cap S_q)$ would have received  $s^{(\ell)}_q$ from the broadcast of $P_{\ell}$, as part of the 
     $(\CandidateCoreSets, P_{\ell}, S_{\qcore}, \{ \W^{(\ell)}_{\qcore, q}\}_{q = 1, \ldots, |\Z_s|},  \BroadcastSet^{(\ell)}_{\qcore},  \{ s^{(\ell)}_q\}_{q \in \BroadcastSet^{(\ell)}_{\qcore}})$ 
     message. Consequently, in this case ${s^{\star}}^{(\ell)}_q$ is the same as $s^{(\ell)}_q$, received from the broadcast of $P_{\ell}$.
     On the other hand, if  $q \not \in \BroadcastSet^{(\ell)}_{\qcore}$, then from Lemma \ref{lemma: MVSSHonestPartiesSynchronous}, 
     each $P_i \in (\Hon \cap S_q)$ would have a common share, say ${s^{(\star)}}^{(\ell)}_q$; moreover, from Lemma \ref{lemma: MVSSHonestDealerSynchronous},
     if $P_{\ell}$ is {\it honest}, then ${s^{(\star)}}^{(\ell)}_q = s^{(\ell)}_q$  holds. We define
     \[ {s^{\star}}^{(\ell)} \defined \displaystyle \sum_{q = 1, \ldots, |\Z_s|}  {s^{(\star)}}^{(\ell)}_q,     \]
     where ${s^{\star}}^{(\ell)} = s^{(\ell)}$ for an {\it honest} $P_{\ell}$.
      Hence at time $\TimeMDVSS$, each party in $(\Hon \cap S_q)$ has $[{s^{\star}}^{(\ell)}]_q$. 
      We also note that if $q \in \BroadcastSet^{(\ell)}_{\qcore}$, then {\it every} $P_i \in \W_q$ sets $\ICSig(P_j, P_i, P_k, [{s^{\star}}^{(\ell)}]_q)$
      to the default value, corresponding to every $P_j \in \W_q$ and every $P_k \in \PartySet$.
      On the other hand, if $q \not \in  \BroadcastSet^{(\ell)}_{\qcore}$, 
      then every $P_i \in (\Hon \cap S_q)$ will have $\ICSig(P_j, P_i, P_k, {s^{\star}}^{(\ell)}_q)$ of every $P_j \in \W_q$ and for every $P_k \in \PartySet$.
    Furthermore, if any corrupt $P_j \in \W_q$ has $\ICSig(P_i, P_j, P_k, s'^{(\ell)}_q)$ of any $P_i \in (\Hon \cap S_q)$ for any $P_k  \in \PartySet$, then
    $s'^{(\ell)}_q = {s^{\star}}^{(\ell)}_q$ holds.
    Moreover, from Lemma \ref{lemma: MVSSHonestPartiesSynchronous}, if $P_{\ell}$ is {\it honest}, then ${s^{\star}}^{(\ell)}_q$ in the IC-signatures mentioned above will be the same as
    $s^{(\ell)}_q$. It then follows that ${s^{\star}}^{(\ell)}$ will be linearly secret-shared; the linearity of the underlying IC-signatures
    follows since the (honest) parties follow the linearity principle, while generating the IC-signatures.
    
    The privacy of $s^{(\ell)}$ for an {\it honest} $P_{\ell}$ follows from Lemma \ref{lemma: MVSSHonestPartiesSynchronous}.
\end{proof}

We next consider an {\it asynchronous} network. We first prove an analogue of Lemma \ref{lemma:MVSSSynchronousCD} in the asynchronous network.
 \begin{lemma}
 \label{lemma:MVSSAsynchronousCD}
 If the network is asynchronous and $P_{\ell} \in \PartySet$ is an honest dealer participating with input $s^{(\ell)}$, then all the following hold in $\MDVSS$, where $\Hon$ is the set of
  honest parties.
     \begin{myitemize}
     \item[--] Except with probability $\Order(n^3 \cdot \errorAICP)$, almost-surely, all the parties in $\Hon$ will eventually have a common $\CD$ set, where  
     $\PartySet \setminus \CD \in \Z_s$.
     \item[--] Except with probability $\Order(n^3 \cdot \errorAICP)$, corresponding to every dealer $P_{\ell} \in \CD$ and every $S_q \in \ShareSpec_{|\Z_s|}$,
      every party in $(\Hon \cap S_q)$ will eventually have a common share, say ${s^{\star}}^{(\ell)}_q$, which is the same as $s^{(\ell)}_q$, for an honest $P_{\ell}$.
     \end{myitemize}
 \end{lemma}
 \begin{proof}
 Let $Z^{\star} \in \Z_a$ be the set of {\it corrupt} parties and let $\Hon = \PartySet \setminus Z^{\star}$ be the set of {\it honest} parties. Note that $\Hon \in \ShareSpec_{\Z_s}$, since $\Z_a \subset \Z_s$. 
  From the properties of $\SVM$ in the {\it asynchronous} network (Lemma \ref{lemma:SVM}),
   it follows that every $P_i \in \Hon$ will eventually set $\flag^{(P_{\ell}, S_q)}$ to $1$ during the instance $\SVM(P_{\ell}, s^{(\ell)}_q, S_q)$, corresponding to every
  $P_{\ell} \in \Hon$ and every $S_q \in \ShareSpec_{\Z_s}$. Moreover, corresponding to every
  $P_{\ell} \in \Hon$ and every $S_q \in \ShareSpec_{\Z_s}$, each $P_i \in (\Hon \cap S_q)$ eventually computes an output $s^{(\ell)}_{qi}$ during the instance $\SVM((P_{\ell}, s^{(\ell)}_q, S_q)$.
  Furthermore, except with probability $\Order(n^3 \cdot \errorAICP)$, the value $s^{(\ell)}_{qi}$ will be the same as $s^{(\ell)}_q$.

   We first claim that there always exists a subset of parties ${\cal D}$, where $\PartySet \setminus {\cal D} \in \Z_s$, such that the $\BA$ instance $\BA^{(\ell)}$ eventually produces output $1$ for all the
   parties in $\Hon$, corresponding to every $P_{\ell} \in {\cal D}$. Assuming that the claim is true, it implies that all the parties in $\Hon$ will eventually participate with some input in the $\BA$ instances
   $\BA^{(1)}, \ldots, \BA^{(n)}$. This is because from the protocol steps, once the $\BA^{(\star)}$ instances corresponding to the parties in ${\cal D}$ produce output $1$, all the parties in $\Hon$ will start participating with input
   $0$ in the remaining instances $\BA^{(\star)}$ of $\BA$ (if any), for which no input has been provided yet.
    And hence from the {\it $\Z_a$-security} of $\BA$ in the {\it asynchronous} network, it follows that almost-surely, all these $\BA$ instances will eventually produce some output
   for all the parties in $\Hon$. Moreover, the outputs will be the same for all the parties in $\Hon$. Consequently, all the parties in $\Hon$ will eventually obtain a common $\CD$ set. Moreover, 
   $\PartySet \setminus \CD \in \Z_s$, since $\CD$ consists of all those parties $P_{\ell}$, such that the instance $\BA^{(\ell)}$ produces output $1$. And according to our claim, ${\cal D} \subseteq \CD$ holds. We now
   proceed to prove our claim.
   
   There are two possible cases. Consider the case when {\it some} $P_i \in \Hon$ starts participating with input $0$ in any $\BA^{(\star)}$ instance. This implies that for $P_i$, there exists a subset of parties
   $\CD_i$ where $\PartySet \setminus \CD_i \in \Z_s$, such that corresponding to every $P_{\ell} \in \CD_i$, the instance $\BA^{(\ell)}$ has produced output $1$ for $P_i$. In this case, the set $\CD_i$ is the
   candidate ${\cal D}$ set, whose existence we want to prove. Next, consider the case when {\it no} party in $\Hon$ has started participating with input $0$ in {\it any} of the $\BA^{(\star)}$ instances.
   In this case, the set $\Hon$ constitutes the candidate ${\cal D}$ set. This is because as shown above, every $P_i \in \Hon$ will eventually set $\flag^{(P_{\ell}, S_q)}$ to $1$, corresponding to every
  $P_{\ell} \in \Hon$ and every $S_q \in \ShareSpec_{\Z_s}$. And hence every $P_i \in \Hon$ will eventually start participating with input $1$ in the $\BA^{(\ell)}$ instances, corresponding to $P_{\ell} \in \Hon$.
  Consequently, the {\it $\Z_s$-validity} of $\BA$ in the {\it asynchronous} network (Theorem \ref{thm:BA}) will guarantee that the $\BA^{(\ell)}$ instances, corresponding to $P_{\ell} \in \Hon$ eventually
  produce output $1$ for all the parties in $\Hon$.

   Next, consider an {\it arbitrary} $P_{\ell} \in \CD$. This implies that at least one party from $\Hon$, say $P_k$,
   has participated with input $1$ during the instance $\BA^{(\ell)}$. If not, then from the {\it $\Z_s$-validity} of $\BA$ in the {\it asynchronous} network (Theorem \ref{thm:BA}),
   all the parties in $\Hon$ would have obtained the output $0$ from the instance $\BA^{(\ell)}$ and hence $P_{\ell} \not \in \CD$, which is a 
   {\it contradiction}. This implies that party $P_k$ has set $\flag^{(P_{\ell}, S_q)}$ to $1$ during the instance $\SVM(P_{\ell}, s^{(\ell)}_q, S_q)$,
    for $q = 1, \ldots, |\Z_s|$.
   So consider an {\it arbitrary} $S_q \in \ShareSpec_{\Z_s}$. From the properties of $\SVM$ in the {\it asynchronous} network (Lemma \ref{lemma:SVM}),
   it follows that there exists some value ${s^{\star}}^{(\ell)}_q$, which is the same as $s^{(\ell)}_q$ for an {\it honest} $P_{\ell}$,
    such that except with probability $\Order(n^3 \cdot \errorAICP)$,
   all the parties in $\Hon$ eventually output ${s^{\star}}^{(\ell)}_q$ during the instance $\SVM(P_{\ell}, s^{(\ell)}_q, S_q)$.   
 \end{proof}

We next prove the analogue of Lemma \ref{lemma: MVSSHonestDealerSynchronous} in the {\it asynchronous} network.
\begin{lemma}
\label{lemma: MVSSHonestDealerAsynchronous}
If the network is asynchronous and $P_{\ell} \in \CD$ is an {\it honest} dealer participating with input $s^{(\ell)}$, then all the following hold in $\MDVSS$ except with probability $\Order(n^3 \cdot \errorAICP)$,
 where $\Hon$ is the set of honest parties.
\begin{myitemize}
 \item[--]  If $S_p = \Hon$, then $P_{\ell}$ will eventually broadcast
   $(\CandidateCoreSets, P_{\ell}, S_p, \{ \W^{(\ell)}_{p, q}\}_{q = 1, \ldots, |\Z_s|}, \allowbreak  \BroadcastSet^{(\ell)}_p, \{ s^{(\ell)}_q\}_{q \in \BroadcastSet^{(\ell)}_p})$.
\item[--] If $P_{\ell}$ broadcasts any $(\CandidateCoreSets, P_{\ell}, S_p, \{ \W^{(\ell)}_{p, q}\}_{q = 1, \ldots, |\Z_s|}, \BroadcastSet^{(\ell)}_p, \{ s^{(\ell)}_q\}_{q \in \BroadcastSet^{(\ell)}_p})$ 
 then every honest $P_i \in \PartySet$ will eventually include $(P_{\ell}, S_p)$ to $\C_i$. 
  Moreover, the following will hold.
    \begin{myitemize}
    \item[--] If $S_q = \Hon$, then $q \not \in \BroadcastSet^{(\ell)}_p$.
    \item[--] For $q = 1, \ldots, |\Z_s|$, each $\W^{(\ell)}_{p, q}$ will be either $S_q$ or $(S_p \cap S_q)$ such that 
    $\Z_s$ satisfies the $\Q^{(1)}(\W^{(\ell)}_{p, q}, \Z_s)$ condition.    
    \item[--] If  $q \not \in \BroadcastSet^{(\ell)}_p$, then every honest $P_i \in S_q$ will have the share $s^{(\ell)}_q$.
    Moreover, every honest $P_i \in \W^{(\ell)}_{p, q}$ will have $\ICSig(P_j, P_i, P_k, s^{(\ell)}_q)$ of every $P_j \in \W^{(\ell)}_{p, q}$ for every $P_k \in \PartySet$.
    Furthermore, if any corrupt $P_j \in \W^{(\ell)}_{p, q}$ have $\ICSig(P_i, P_j, P_k, s'^{(\ell)}_q)$ of any honest $P_i \in \W^{(\ell)}_{p, q}$ for any $P_k \in \PartySet$, then
    $s'^{(\ell)}_q = s^{(\ell)}_q$ holds. Also, all the underlying IC-signatures will satisfy the linearity property.
    \end{myitemize}
 \item[--] The view of the adversary will be independent of $s^{(\ell)}$.
\end{myitemize}
\end{lemma}
\begin{proof}
Let $Z^{\star} \in \Z_a$ be the set of {\it corrupt} parties and let $\Hon = \PartySet \setminus Z^{\star}$ be the set of {\it honest} parties. We first note that $\Hon \in \ShareSpec_{\Z_s}$,
 since $Z^{\star} \in \Z_s$ as $\Z_a \subset \Z_s$. The proof for the first part of the lemma is similar to the proof of the first part of Lemma \ref{lemma: MVSSHonestDealerSynchronous}, except that
  all the ``favourable" conditions hold for an {\it honest} $P_{\ell}$ {\it eventually}. In more detail, consider an {\it arbitrary} $S_q \in \ShareSpec_{\Z_s}$.
   Then each $P_i \in (S_q \cap \Hon)$ eventually computes the share $s^{(\ell)}_{qi}$ during the instance $\SVM(P_{\ell}, s^{(\ell)}_q, S_q)$, where
   $s^{(\ell)}_{qi} = s^{(\ell)}_q$ holds, except with probability $\Order(n^3 \cdot \errorAICP)$. 
  Consequently, $P_i$
   starts giving
   $\ICSig(P_i, P_j, P_k, s^{(\ell)}_{qi})$ to every $P_j \in S_q$, for every $P_k \in \PartySet$.
   Then from the {\it $\Z_a$-correctness} of $\Auth$  in the {\it asynchronous} network (Theorem \ref{thm:ICP}), 
   it follows that each party $P_i \in (S_q \cap \Hon)$ will eventually receive 
   $\ICSig(P_j, P_i, P_k, s^{(\ell)}_{qj})$ from every $P_j \in (S_q \cap \Hon)$, for every $P_k \in \PartySet$, such that $s^{(\ell)}_{qj} = s^{(\ell)}_{qi} = s^{(\ell)}_q$ holds. 
   Since $S_q$ is arbitrary, it follows that eventually, every party $P_i \in \Hon$ broadcasts an $\OK^{(\ell)}(i, j)$ message, corresponding to every 
   $P_j \in \Hon$. From the {\it $\Z_a$-weak validity} and  {\it $\Z_a$-fallback validity}  of $\BC$ in the {\it asynchronous} network (Theorem \ref{thm:BC}), it follows that
   these $\OK^{(\ell)}(i, j)$ messages are eventually received by every party in $\Hon$.
   Consequently, the set $\Hon$ eventually becomes a clique in the consistency graph $G^{(\ell, i)}$ of every party $P_i \in \Hon$.
    Let $S_p$ be the set from $\ShareSpec_{\Z_s}$, such that $S_p = \Hon$.
   From the protocol steps, it then follows that the dealer $P_{\ell}$  will eventually compute core-sets $\W^{(\ell)}_{p, q}$ 
   for $q = 1, \ldots, |\Z_s|$ and broadcast-set $\BroadcastSet^{(\ell)}_p$ with respect to $S_p$ as follows,
	 \begin{myitemize}
	 \item[--] If $S_q$ constitutes a clique in the graph $G^{(\ell, \ell)}$, then $\W^{(\ell)}_{p, q}$ is set as $S_q$. 
	 \item[--] Else if $(S_p \cap S_q)$ constitutes a clique in $G^{(\ell, \ell)}$ and
	 $\Z_s$ satisfies the $\Q^{(1)}(S_p \cap S_q, \Z_s)$ condition, 
	 then $\W^{(\ell)}_{p, q}$ is set as $(S_p \cap S_q)$.
	 \item[--] Else  $\W^{(\ell)}_{p, q}$ is set to $S_q$ and $q$ is included to $\BroadcastSet^{(\ell)}_p$. 	 
	 \end{myitemize}
After computing the core-sets and  broadcast-set, $P_{\ell}$ will eventually broadcast 
 $(\CandidateCoreSets, P_{\ell}, S_p, \{ \W^{(\ell)}_{p, q}\}_{q = 1, \ldots, |\Z_s|}, \BroadcastSet^{(\ell)}_p, \{ s^{(\ell)}_q\}_{q \in \BroadcastSet^{(\ell)}_p})$. 
  
We next proceed to prove the second part of the lemma, whose proof is again similar to the proof of the second part of the Lemma \ref{lemma: MVSSHonestDealerSynchronous}, except that all the
  ``favourable" conditions which hold for $P_{\ell}$, are guaranteed to hold eventually for all the parties in $\Hon$. In more detail, 
 consider an {\it arbitrary} $S_p \in \ShareSpec_{\Z_s}$, such that $P_{\ell}$ compute core-sets $\W^{(\ell)}_{p, q}$ 
   for $q = 1, \ldots, |\Z_s|$ and broadcast-set $\BroadcastSet^{(\ell)}_p$ with respect to $S_p$ and broadcasts 
   $(\CandidateCoreSets, P_{\ell}, S_p, \{ \W^{(\ell)}_{p, q}\}_{q = 1, \ldots, |\Z_s|}, \BroadcastSet^{(\ell)}_p, \{ s^{(\ell)}_q\}_{q \in \BroadcastSet^{(\ell)}_p})$. 
   This means the parties in $S_p$ constitute a clique in the graph $G^{(\ell, \ell)}$. 
    We note that all the edges which are present in the graph $G^{(\ell, \ell)}$ when $S_p$ constitute a clique in $G^{(\ell, \ell)}$ are bound to be eventually
    included in the graph $G^{(\ell, i)}$ of every party $P_i \in \Hon$. This is because the edges are included by $P_{\ell}$ based on various $\OK^{(\ell)}(\star, \star)$
    messages, which are received by $P_{\ell}$ through various $\BC$ instances. Consequently, due to the various properties of $\BC$ in the {\it asynchronous} network,
    these $\OK^{(\ell)}(\star, \star)$ messages are bound to be eventually delivered to every party in $\Hon$. As a result, all the properties which hold for $P_{\ell}$ in the graph
    $G^{(\ell, \ell)}$ when $S_p$ constitute a clique in $G^{(\ell, \ell)}$ are bound to hold eventually for every $P_i \in \Hon$ in the graph $G^{(\ell, \ell)}$. 
       Since $P_{\ell}$ is assumed to be {\it honest}, it computes the sets $\{ \W^{(\ell)}_{p, q}\}_{q = 1, \ldots, |\Z_s|}$ and $\BroadcastSet^{(\ell)}_p$, satisfying the following properties.
    \begin{myitemize}
   	\item[--] If $S_q$ constitutes a clique in the graph $G^{(\ell, \ell)}$, then $\W^{(\ell)}_{p, q}$ is set as $S_q$. 
	 \item[--] Else if $(S_p \cap S_q)$ constitutes a clique in $G^{(\ell, \ell)}$ and 
	  $\Z_s$ satisfies the $\Q^{(1)}(S_p \cap S_q, \Z_s)$ condition, then $\W^{(\ell)}_{p, q}$ is set as $(S_p \cap S_q)$.
	 \item[--] Else  $\W^{(\ell)}_{p, q}$ is set to $S_q$ and $q$ is included to $\BroadcastSet^{(\ell)}_p$. 	 
    \end{myitemize}
    We also note that if $S_q = \Hon$, then $q \not \in \BroadcastSet^{(\ell)}_p$ and consequently, $P_{\ell}$ will {\it not} make the share $s^{(\ell)}_q$ public.
    This is because $P_{\ell}$ will set $\W^{(\ell)}_{p, q}$ to $(S_p \cap S_q)$. 
     In more detail, the parties in $(S_p \cap S_q)$ will constitute a clique in $G^{(\ell, \ell)}$, since
    $S_p$ constitutes a clique in $G^{(\ell, \ell)}$, when $P_{\ell}$ starts computing the core-sets $\{ \W^{(\ell)}_{p, q}\}_{q = 1, \ldots, |\Z_s|}$.
    Moreover,  $\Z_s$ will satisfy the $\Q^{(1)}(S_p \cap S_q, \Z_s)$ condition, due to the $\Q^{(2, 1)}(\PartySet, \Z_s, \Z_a)$ condition.
    
    Since $P_{\ell}$ is {\it honest}, from the {\it $\Z_a$-weak validity} and {\it $\Z_a$-fallback validity}  of $\BC$ in the {\it asynchronous} network (see Theorem \ref{thm:BC}),
     it follows that all the parties in
    $\Hon$ will eventually receive $(\CandidateCoreSets, P_{\ell}, S_p, \{ \W^{(\ell)}_{p, q}\}_{q = 1, \ldots, |\Z_s|}, \BroadcastSet^{(\ell)}_p, \{ s^{(\ell)}_q\}_{q \in \BroadcastSet^{(\ell)}_p})$ from the
    broadcast of $P_{\ell}$. Moreover, each party $P_i \in \Hon$ will eventually include $(P_{\ell}, S_p)$ to the set $\C_i$. 
    This is because since $P_{\ell}$ has computed the sets 
     $\{ \W^{(\ell)}_{p, q}\}_{q = 1, \ldots, |\Z_s|}$ and $\BroadcastSet^{(\ell)}_p$ {\it honestly}, these sets will eventually pass all the verifications for each $P_i \in \Hon$.
     The last statement is true because as shown above, all the properties which hold for $P_{\ell}$ in the graph
    $G^{(\ell, \ell)}$ when $S_p$ constitute a clique in $G^{(\ell, \ell)}$ are bound to hold eventually for every $P_i \in \Hon$ in the graph $G^{(\ell, \ell)}$. 
    
    The proof for the rest of the properties stated in the lemma is similar to that of Lemma \ref{lemma: MVSSHonestDealerSynchronous}, except that we now rely on the security properties of ICP
    in the {\it asynchronous} network; to avoid repetition we do not produce the details here.
    The linearity of the underlying IC-signatures is ensured since the parties follow the linearity principle while generating the IC-signatures.
\end{proof}

We next prove an analogue of Lemma \ref{lemma: MVSSHonestPartiesSynchronous} in the {\it asynchronous} network.
\begin{lemma}
\label{lemma: MVSSHonestPartiesAsynchronous}
If the network is asynchronous and if in $\MDVSS$ any honest party $P_i$ receives 
 $(\CandidateCoreSets, P_{\ell}, S_p, \{ \W^{(\ell)}_{p, q}\}_{q = 1, \ldots, |\Z_s|},  \BroadcastSet^{(\ell)}_p,  \{ s^{(\ell)}_q\}_{q \in \BroadcastSet^{(\ell)}_p})$ from the broadcast of any 
  corrupt dealer $P_{\ell} \in \CD$
  and includes $(P_{\ell}, S_p)$ to $\C_i$, then all honest parties $P_j$ will eventually receive 
   $(\CandidateCoreSets, P_{\ell}, S_p, \{ \W^{(\ell)}_{p, q}\}_{q = 1, \ldots, |\Z_s|},  \BroadcastSet^{(\ell)}_p,  \{ s^{(\ell)}_q\}_{q \in \BroadcastSet^{(\ell)}_p})$ from the broadcast of $P_{\ell}$
   and include  $(P_{\ell}, S_p)$ to $\C_j$.
   Moreover, for $q = 1, \ldots, |\Z_s|$, the following holds, except with probability $\Order(n^3 \cdot \errorAICP)$.
     \begin{myitemize}
     \item[--] $\W^{(\ell)}_{p, q}$ is either $S_q$ or $(S_p \cap S_q)$. Moreover, 
       $\Z_s$ satisfies the $\Q^{(1)}(\W^{(\ell)}_{p, q}, \Z_s)$ condition.    
     \item[--] If  $q \not \in \BroadcastSet^{(\ell)}_p$, then every honest $P_i \in S_q$ will have a common share, say ${s^{\star}}^{(\ell)}_q$.
    Moreover, every honest $P_i \in \W^{(\ell)}_{p, q}$ will have $\ICSig(P_j, P_i, P_k, {s^{\star}}^{(\ell)}_q)$ of every $P_j \in \W^{(\ell)}_{p, q}$ for every $P_k \in \PartySet$.
    Furthermore, if any corrupt $P_j \in \W^{(\ell)}_{p, q}$ has $\ICSig(P_i, P_j, P_k, s'^{(\ell)}_q)$ of any honest $P_i \in \W^{(\ell)}_{p, q}$ for any $P_k \in \PartySet$, then
    $s'^{(\ell)}_q = {s^{\star}}^{(\ell)}_q$ holds. Also, all the underlying IC-signatures will satisfy the linearity property.
     \end{myitemize}   
\end{lemma}
\begin{proof}
The proof is very similar to the proof of Lemma \ref{lemma: MVSSHonestPartiesSynchronous}, except that we now rely on the properties of $\BC$ in the {\it asynchronous} network.
  Let $Z^{\star} \in \Z_a$ be the set of {\it corrupt} parties and let $\Hon = \PartySet \setminus Z^{\star}$
 be the set of {\it honest} parties. Now consider 
 an {\it arbitrary corrupt} dealer $P_{\ell} \in \CD$ and an {\it arbitrary} $P_i \in \Hon$, such that 
  $P_i$ receives 
  $(\CandidateCoreSets, P_{\ell}, S_p, \{ \W^{(\ell)}_{p, q}\}_{q = 1, \ldots, |\Z_s|},  \BroadcastSet^{(\ell)}_p,  \{ s^{(\ell)}_q\}_{q \in \BroadcastSet^{(\ell)}_p})$ from the broadcast of
   $P_{\ell}$
  and includes $(P_{\ell}, S_p)$ to $\C_i$. Now consider another {\it arbitrary} $P_j \in \Hon$, such that $P_j \neq P_i$.
  From the {\it $\Z_a$-weak consistency} and {\it $\Z_a$-fallback consistency} of $\BC$ in the {\it asynchronous} network, it follows that $P_j$ is bound to eventually receive
  $(\CandidateCoreSets, P_{\ell}, S_p, \{ \W^{(\ell)}_{p, q}\}_{q = 1, \ldots, |\Z_s|},  \BroadcastSet^{(\ell)}_p,  \{ s^{(\ell)}_q\}_{q \in \BroadcastSet^{(\ell)}_p})$
  from the broadcast of $P_{\ell}$. We wish to show that $P_j$ will eventually include $(P_{\ell}, S_p)$ to $\C_j$.
  For this, we note that since $P_i$ has included $(P_{\ell}, S_p)$ to $\C_i$, 
  {\it all} the following conditions hold for $P_i$, 
    for $q = 1, \ldots, |\Z_s|$.
      \begin{myitemize}
      \item[--] If $q \in \BroadcastSet^{(\ell)}_p$, then the set $\W^{(\ell)}_{p, q} = S_q$.
      \item[--] If $(q \not \in \BroadcastSet^{(\ell)}_p)$, then $\W^{(\ell)}_{p, q}$ is either $S_q$ or  $(S_p \cap S_q)$, such that:
         \begin{myitemize}
         \item[--] If $\W^{(\ell)}_{p, q} = S_q$, then $S_q$ constitutes a clique in $G^{(\ell, i)}$.
         \item[--] Else if $\W^{(\ell)}_{p, q} = (S_p \cap S_q)$, then 
      $(S_p \cap S_q)$ constitutes a clique in $G^{(\ell, i)}$ and   $\Z_s$ satisfies the $\Q^{(1)}(S_p \cap S_q, \Z_s)$ condition.    
         \end{myitemize}
           \end{myitemize}
We claim that all the above conditions will hold eventually {\it even} for $P_j$. This is because {\it all} the edges which are present in the consistency graph $G^{(\ell, i)}$ 
 when $P_i$ includes $(P_{\ell}, S_p)$ to $\C_i$ are bound to be eventually present in the consistency graph $G^{(\ell, j)}$.
 This follows from the {\it $\Z_a$-weak validity, $\Z_a$-fallback validity, $\Z_a$-weak consistency} and {\it $\Z_a$-fallback consistency} 
  of $\BC$ in the {\it asynchronous} network (see Theorem \ref{thm:BC})
  and the fact that the edges in the graph $G^{(\ell, i)}$ are based on $\OK^{(\ell)}(\star, \star)$ messages, which are received through various $\BC$ instances.

The proof regarding the IC-signatures is exactly the same as Lemma \ref{lemma: MVSSHonestPartiesSynchronous} and to avoid repetition, we do not produce the formal details here.
\end{proof}

Finally, based on the previous two lemmas, we prove an analogue of Lemma \ref{lemma:MVSSSynchronousProperties} and show that 
  in an {\it asynchronous} network, all honest parties will eventually output a ``legitimate"
 set of parties $\Core$.
 \begin{lemma}
 \label{lemma:MVSSAsynchronousProperties}
  If the network is asynchronous, then in $\MDVSS$, except with probability $\Order(n^3 \cdot \errorAICP)$,
  almost-surely all honest parties eventually output a common set $\Core$, 
   such that at least one honest party will be present in $\Core$.
   Moreover, corresponding to every $P_{\ell} \in \Core$, there exists some ${s^{\star}}^{(\ell)}$, where ${s^{\star}}^{(\ell)} = s^{(\ell)}$ for an honest $P_{\ell}$,
   which is the input of $P_{\ell}$ for $\MDVSS$,
       such that the values $\{  {s^{\star}}^{(\ell)}  \}_{P_{\ell} \in \Core}$ are linearly secret-shared with IC-signatures. Furthermore, 
       if $P_{\ell}$ is honest, then the adversary's view is independent of $s^{(\ell)}$.
 \end{lemma}
\begin{proof}
The proof structure is very similar to that of Lemma \ref{lemma:MVSSSynchronousProperties}, except that we now rely on the properties of $\BA$ and $\BC$ in the {\it asynchronous} network and
 Lemma \ref{lemma:MVSSAsynchronousCD}-\ref{lemma: MVSSHonestPartiesAsynchronous}. 
  Let $Z^{\star} \in \Z_a$ be the set of {\it corrupt} parties and let $\Hon = \PartySet \setminus Z^{\star}$ be the set of honest parties. 
   From Lemma \ref{lemma:MVSSAsynchronousCD},  except with probability $\Order(n^3 \cdot \errorAICP)$, almost-surely all the parties in $\Hon$ will eventually 
   have a common set of committed dealers
  $\CD$, where 
     $\PartySet \setminus \CD \in \Z_s$. Moreover, corresponding to every dealer $P_{\ell} \in \CD$ and every $S_q \in \ShareSpec_{|\Z_s|}$,
      every party in $(\Hon \cap S_q)$ will eventually have a common share, say ${s^{\star}}^{(\ell)}_q$, which is the same as $s^{(\ell)}_q$, for an $P_{\ell}$.
   We begin by showing that once the set of committed dealers $\CD$ is decided, then almost-seurely, all the $|\Z_s|$ instances $\BA^{(1, \star)}$ of $\BA$
   and then all the $|\CD|$ instances $\BA^{(2, \star)}$ of $\BA$ will eventually produce some output, for all the parties in $\Hon$.
 
 We first claim that irrespective of way messages are scheduled and the order in which the parties in $\Hon$ participate in various $\BA^{(1, \star)}$ instances, there will be
  some instance $\BA^{(1, p)}$ corresponding to some $S_p \in \ShareSpec_{\Z_s}$, which will eventually produce output $1$ for all the parties in $\Hon$.
  For this, consider the set $S_p \in \ShareSpec_{|\Z_s|}$, such that $S_p = \Hon$ (such an $S_p$ is bound to exist since $Z^{\star} \in \Z_s$ also holds, as $\Z_a \subset \Z_s$). 
  If there exists {\it some} $P_i \in \Hon$ which starts participating with input $0$ in the instance $\BA^{(1, p)}$, then the claim is true, because $P_i$ participates with input  
  $0$ during $\BA^{(1, p)}$, only after receiving the output $1$ from some other instance of $\BA^{(1, \star)}$, say $\BA^{(1, q)}$.
  And hence from the {\it $\Z_a$-consistency} of $\BA$ in the {\it asynchronous} network (Theorem \ref{thm:BA}), all the parties in $\Hon$ will eventually obtain the output $1$ from the instance
  $\BA^{(1, q)}$, thus proving our claim. On the other hand, consider the case when {\it no} party has yet started participating with any input in the instance $\BA^{(1, p)}$. 
   Then corresponding to each $P_{\ell} \in (\Hon \cap \CD)$, every $P_i \in \Hon$ will eventually receive
  $(\CandidateCoreSets, P_{\ell}, S_p, \{ \W^{(\ell)}_{p, q}\}_{q = 1, \ldots, |\Z_s|},  \BroadcastSet^{(\ell)}_p,  \{ s^{(\ell)}_q\}_{q \in \BroadcastSet^{(\ell)}_p})$
  from the broadcast of $P_{\ell}$ 
   and includes $(P_{\ell}, S_p)$ to $\C_i$ (see Lemma \ref{lemma: MVSSHonestDealerAsynchronous}). 
    Since $\CD \setminus \Hon  \subseteq Z^{\star}$ and $Z^{\star} \in \Z_s$ (due to condition $\Z_s \subset \Z_a$), 
    it follows that 
    every party $P_i \in \Hon$ will eventually have a set ${\cal A}_{p, i}$ (namely ${\cal A}_{p, i} = \Hon$), 
    where $\CD \setminus {\cal A}_{p, i} \in \Z_s$
      and where $(P_{\ell}, S_p) \in \C_i$ for every $P_{\ell} \in {\cal A}_{p, i}$.
      Consequently, each $P_i \in \Hon$ will eventually start participating in the instance $\BA^{(1, p)}$ 
      with input $1$, if they have not done so. And from the {\it $\Z_s$-validity} of $\BA$ in the {\it asynchronous} network (see Theorem \ref{thm:BA}), it follows that 
      all the parties in $\Hon$ eventually obtain the output $1$ from the instance $\BA^{(1, p)}$, thus proving our claim in this case as well.
      
      Now from the above claim, it follows that every party in $\Hon$ will eventually 
      start participating in the remaining $\BA^{(1, \star)}$ instances for which no input has been provided yet (if there are any), with input $0$.  
  And from the {\it $\Z_a$-security} of $\BA$ in the {\it asynchronous} network, almost-surely,
   these $\BA^{(1, \star)}$  instances will eventually produce common outputs, for every party in $\Hon$.
  As a result, all the parties in $\Hon$ will eventually compute a $\qcore \in \{1, \ldots, |\Z_s| \}$. Moreover, $\qcore$ will be {\it common} for all the
  parties in $\Hon$, since it corresponds to the least-indexed $\BA^{(1, \star)}$ instance among $\BA^{(1, 1)}, \ldots, \BA^{(1, |\Z_s|)}$, which produces output $1$. And from the
  {\it $\Z_a$-security} of $\BA$ in the {\it asynchronous} network, each $\BA^{(1, \star)}$ instance produces a common output for every party in $\Hon$.
  We also note that $\qcore$ will be indeed set to some value from the set $\{1, \ldots, |\Z_s| \}$. This is because as shown above, the index $p$ where $S_p = \Hon$ always constitute a candidate
  $\qcore$.
  
   We next claim that corresponding to $S_{\qcore} \in \ShareSpec_{\Z_s}$, 
  there exists a subset of parties ${\cal B}_{\qcore}$, where $\CD \setminus {\cal B}_{\qcore} \in \Z_s$, such that corresponding to every $P_j \in {\cal B}_{\qcore}$,
  the ordered pair $(P_j, S_{\qcore})$ is eventually included in the set $\C_i$ of every $P_i \in \Hon$.
  Assuming that the claim is true, we next show that there always exists a set of parties ${\cal B}$ where $\CD \setminus {\cal B} \in \Z_s$, such that the 
  $\BA^{(2, j)}$ instance eventually produce output $1$ for all the parties in $\Hon$, corresponding to every $P_j \in {\cal B}$. For this, we consider two possible cases.
  If any $P_i \in \Hon$ has started participating with input $0$ in any instance $\BA^{(2, \star)}$, then it implies that for $P_i$, there exists a subset ${\cal B}_i$ where
  $\CD \setminus {\cal B}_i \in \Z_s$ and where corresponding to each $P_j \in {\cal B}_i$, the instance $\BA^{(2, j)}$ has produced output $1$. Hence from the 
  {\it $\Z_a$-consistency} of $\BA$ in the {\it asynchronous} network (see Theorem \ref{thm:BA}), all these $\BA^{(2, j)}$ instances will eventually produce output $1$ for all the parties in $\Hon$.
  On the other hand, consider the case when {\it no} party in $\Hon$ has started participating with input $0$ in any instance $\BA^{(2, \star)}$.
  Then as per the claim, all the parties in $\Hon$ will eventually start participating with input $1$ in the instances $\BA^{(2, j)}$, corresponding to every
  $P_j \in {\cal B}_{\qcore}$. Hence from the {\it $\Z_a$-security} of the $\BA$ in the {\it synchronous} network (Theorem \ref{thm:BA}),
  all the parties will eventually obtain the output $1$ in the $\BA^{(2, j)}$ instances, corresponding to every
  $P_j \in {\cal B}_{\qcore}$. Thus irrespective of the case, the set ${\cal B}$ is guaranteed. 
   As a result, all the parties in $\Hon$ will eventually start participating 
  in the remaining $\BA^{(2, \star)}$ instances for which no input has been provided yet (if there are any), with input $0$. 
  Consequently, from the {\it $\Z_a$-security} of the $\BA$ in the {\it asynchronous} network (Theorem \ref{thm:BA}), 
  almost-surely, all the parties in $\Hon$ will eventually have some output from all the $|\CD|$ instances of $\BA^{(2, \star)}$. 
  Moreover, the outputs will be {\it common} for all the parties in $\Hon$. Furthermore, the parties in $\Hon$ will have a subset $\Core$, which corresponds to all the 
   $\BA^{(2, j)}$ instances, which produces output $1$. Note that $\CD \setminus \Core \in \Z_s$ holds, since we have shown that the 
   $\BA^{(2, j)}$ instances, corresponding to the parties $P_j \in {\cal B}$ will produce output $1$, implying ${\cal B} \subseteq \Core$.
    And $\CD \setminus {\cal B} \in \Z_s$ holds. Now since $\PartySet \setminus \CD \in \Z_s$ and $\CD \setminus \Core \in \Z_s$, it follows that
    $(\Hon \cap \Core) \neq \emptyset$, since the $\Q^{(2, 1)}(\PartySet, \Z_s, \Z_a)$ condition is satisfied.
   
   We next proceed to prove our claim about the existence of ${\cal B}_{\qcore}$. 
   Since the instance $\BA^{(1, \qcore)}$ has produced output $1$, it follows that at least one party from $\Hon$, say $P_k$, must have participated with input $1$
   in the instance $\BA^{(1, \qcore)}$.  This is because if {\it all} the parties in $\Hon$ participates with input $0$ in the instance  $\BA^{(1, \qcore)}$, then from the {\it $\Z_a$-validity} of $\BA$ in the
   {\it asynchronous} network (Theorem \ref{thm:BA}), all the parties in $\Hon$ would have obtained the output $0$ from the instance $\BA^{(1, \qcore)}$, which is a contradiction.
   Now since $P_k$ has participated with input $1$
   in the instance $\BA^{(1, \qcore)}$, it follows there exists a subset of parties ${\cal A}_{\qcore, k}$, where 
   $\CD \setminus {\cal A}_{\qcore, k} \in \Z_s$, such that $(P_{\ell}, S_{\qcore})$ is present in the set $\C_k$, corresponding to every $P_{\ell} \in {\cal A}_{\qcore, k}$.
   We show that the set ${\cal A}_{\qcore, k}$ constitutes the candidate ${\cal B}_{\qcore}$. For this, note that for any $P_{\ell} \in {\cal A}_{\qcore, k}$,
   party $P_k$ includes $(P_{\ell}, S_{\qcore})$ to $\C_i$, only after receiving
   a message $(\CandidateCoreSets, P_{\ell}, S_{\qcore}, \{ \W^{(\ell)}_{\qcore, q}\}_{q = 1, \ldots, |\Z_s|},  \BroadcastSet^{(\ell)}_{\qcore},  \{ s^{(\ell)}_q\}_{q \in \BroadcastSet^{(\ell)}_{\qcore}})$
   from the broadcast of $P_{\ell}$
   and verifying it. 
    It then follows from Lemma \ref{lemma:MVSSAsynchronousProperties} that eventually,
    {\it every} party in $\Hon$ would have received  
    $(\CandidateCoreSets, P_{\ell}, S_{\qcore}, \{ \W^{(\ell)}_{\qcore, q}\}_{q = 1, \ldots, |\Z_s|},  \BroadcastSet^{(\ell)}_{\qcore},  \{ s^{(\ell)}_q\}_{q \in \BroadcastSet^{(\ell)}_{\qcore}})$ from 
    each $P_{\ell} \in {\cal A}_{\qcore, k}$. Hence each party $P_i \in \Hon$ would eventually include $(P_{\ell}, S_{\qcore})$ to the set $\C_i$, corresponding to every $P_{\ell} \in {\cal A}_{\qcore, k}$.

   We next claim that corresponding to every $P_{\ell} \in \Core$, every $P_i \in \Hon$ eventually receives a message 
    $(\CandidateCoreSets, P_{\ell}, S_{\qcore}, \{ \W^{(\ell)}_{\qcore, q}\}_{q = 1, \ldots, |\Z_s|},  \BroadcastSet^{(\ell)}_{\qcore},  \{ s^{(\ell)}_q\}_{q \in \BroadcastSet^{(\ell)}_{\qcore}})$ from 
    the broadcast of $P_{\ell}$. The proof for this is very similar to the proof of the previous claim and relies on the properties of $\BA$ in the {\it asynchronous} network.
     So consider an {\it arbitrary} $P_{\ell} \in \Core$.
    This implies that the instance  $\BA^{(2, \ell)}$ have produced output $1$, which further implies that at least one party from $\Hon$, say $P_m$, have participated with input
    $1$ during the instance $\BA^{(2, \ell)}$. If not, then from the {\it $\Z_a$-validity} of $\BA$ in the {\it asynchronous} network (Theorem \ref{thm:BA}), the instance $\BA^{(2, \ell)}$
    would have produced output $0$, which is a contradiction. 
    Now since $P_m$ participates with input $1$ in the instance $\BA^{(2, \ell)}$, it follows that
    $P_m$ must have received a message 
     $(\CandidateCoreSets, P_{\ell}, S_{\qcore}, \{ \W^{(\ell)}_{\qcore, q}\}_{q = 1, \ldots, |\Z_s|},  \BroadcastSet^{(\ell)}_{\qcore},  \{ s^{(\ell)}_q\}_{q \in \BroadcastSet^{(\ell)}_{\qcore}})$
     from the broadcast of $P_{\ell}$ and included $(P_{\ell}, S_{\qcore})$ to $\C_m$. 
      It then follows from Lemma \ref{lemma:MVSSAsynchronousProperties} that eventually,
      {\it every} party $P_i$ in $\Hon$ would receive
    $(\CandidateCoreSets, P_{\ell}, S_{\qcore}, \{ \W^{(\ell)}_{\qcore, q}\}_{q = 1, \ldots, |\Z_s|},  \BroadcastSet^{(\ell)}_{\qcore},  \{ s^{(\ell)}_q\}_{q \in \BroadcastSet^{(\ell)}_{\qcore}})$ from 
    $P_{\ell}$ and includes $(P_{\ell}, S_{\qcore})$ to $\C_i$.
    
     Till now we have shown that almost-surely, all the $|\Z_s|$ instances $\BA^{(1, \star)}$ of $\BA$
      and then all the $|\CD|$ instances $\BA^{(2, \star)}$ of $\BA$ will eventually produce some output, for all the parties in $\Hon$. Moreover, 
      all the parties in $\Hon$ will eventually have a common $\qcore \in \{1, \ldots, |\Z_s| \}$ and a common set $\Core \subseteq \PartySet$, where $\Core$ has at least one {\it honest} party.    
     Furthermore, corresponding to every $P_{\ell} \in \Core$, each $P_i \in \Hon$ would eventually received a message 
     $(\CandidateCoreSets, P_{\ell}, S_{\qcore}, \{ \W^{(\ell)}_{\qcore, q}\}_{q = 1, \ldots, |\Z_s|},  \BroadcastSet^{(\ell)}_{\qcore},  \{ s^{(\ell)}_q\}_{q \in \BroadcastSet^{(\ell)}_{\qcore}})$
     from the broadcast of $P_{\ell}$. The rest of the proof will be now same as Lemma \ref{lemma:MVSSSynchronousProperties}, except that we 
     now rely on Lemma \ref{lemma: MVSSHonestDealerAsynchronous} and Lemma \ref{lemma: MVSSHonestPartiesAsynchronous}.
     To avoid repetition, we do not produce the formal details here.        
\end{proof}

We finally derive the communication complexity of the protocol.
\begin{lemma}
\label{lemma:MDVSSCC}
Protocol $\MDVSS$ incurs a communication of $\Order(|\Z_s|^2 \cdot n^9 \cdot \log{|\F|} \cdot |\sigma|)$ bits. In addition, 
 $\Order(|\Z_s| + n)$ instances of $\BA$ are invoked.
\end{lemma}
\begin{proof}
The number of $\BA$ instances follows easily from the protocol inspection. The communication complexity of the protocol is dominated by the instances of $\SVM$ and $\Auth$ invoked in the protocol.
 There are $\Order(|\Z_s| \cdot n)$ instances of $\SVM$ and $\Order(|\Z_s| \cdot n^4)$ instances of $\Auth$ invoked. The communication complexity now follows from the communication complexity of
  $\SVM$ (Lemma \ref{lemma:SVM}) and $\Auth$ (Theorem \ref{thm:ICP}).
\end{proof}

Theorem \ref{thm:MDVSS} now follows from Lemma \ref{lemma:MVSSSynchronousProperties}, Lemma \ref{lemma:MVSSAsynchronousProperties} and Lemma \ref{lemma:MDVSSCC}.

%% file: AppTriples.tex
\section{Properties of the Triple-Generation Protocol}
\label{app:Triples}
In this section, we prove the properties of the triple-generation protocol and related sub-protocols.
\subsection{Properties of the Protocol $\LSh$}
We first prove the properties of the protocol $\LSh$ (see Fig \ref{fig:LSh} for the formal description). \\~\\
\noindent {\bf Lemma \ref{lemma:LSh}.}
{\it Let $r$ be a random value which is linearly secret-shared with IC-signatures with ${\GW}_1, \ldots, {\GW}_{|\Z_s|}$ being the underlying core-sets.
  Then protocol $\LSh$ achieves the following where $\D$ participates with the input $s$.
\begin{myitemize}  
\item[--] If $\D$ is honest, then the following hold, where $\TimeLSh = \TimeRec + \TimeBC$.
 \begin{myitemize}
    \item[--] {\bf $\Z_s$-Correctness}: If the network is synchronous, then except with probability $\Order(n^3 \cdot \errorAICP)$, the honest parties output $[s]$ at the time $\TimeLSh$, with
     ${\GW}_1, \ldots, {\GW}_{|\Z_s|}$ being the underlying core-sets.
     \item[--] {\bf $\Z_a$-Correctness}: If the network is asynchronous, then except with probability $\Order(n^3 \cdot \errorAICP)$, the honest parties eventually
      output $[s]$, with
     ${\GW}_1, \ldots, {\GW}_{|\Z_s|}$ being the underlying core-sets.
    \item[--] {\bf Privacy}: Irrespective of the network type, the view of the adversary remains independent of $s$.
 \end{myitemize}
  \item[--] If $\D$ is corrupt then either no honest party computes any output or there exists some value, say $s^{\star}$, such that the following hold.
    \begin{myitemize}
     \item[--] {\bf $\Z_s$-Commitment}: If the network is synchronous, then
      except with probability $\Order(n^3 \cdot \errorAICP)$, the honest parties output $[s^{\star}]$, with
     ${\GW}_1, \ldots, {\GW}_{|\Z_s|}$ being the underlying core-sets. Moreover, if any honest party computes its output at the time $T$, then all honest parties will have their respective output by
     the time $T + \Delta$.
        \item[--] {\bf $\Z_a$-Commitment}: If the network is asynchronous, then
      except with probability $\Order(n^3 \cdot \errorAICP)$, the honest parties eventually output $[s^{\star}]$, with
     ${\GW}_1, \ldots, {\GW}_{|\Z_s|}$ being the underlying core-sets.     
    \end{myitemize}
  \item[--] {\bf Communication Complexity}: $\Order(|\Z_s| \cdot n^3 \cdot \log{|\F|} + n^4 \cdot \log{|\F|} \cdot |\sigma|)$ bits are communicated by the honest parties.
\end{myitemize}
}
\begin{proof}
Let us first consider an {\it honest} dealer. Moreover, we consider a {\it synchronous} network. Form the {\it $\Z_s$-correctness} of $\Rec$ in the {\it synchronous} network (Lemma \ref{lemma:Rec}),
  after time $\TimeRec$, the dealer $\D$ will reconstruct
 $r$, except with probability $\Order(n^3 \cdot \errorAICP)$. Moreover, from the {\it privacy} property of $\Rec$, $r$ will be random from the point of the view of the adversary. Since $\D$ is {\it honest}, from the
  {\it $\Z_s$-validity} of $\BC$ in the {\it synchronous} network (Theorem \ref{thm:BC}), all honest parties will receive $\s$ from the broadcast of $\D$ at the time $\TimeRec + \TimeBC$, 
   where $\s = s + r$. The parties then take the
  default linear secret-sharing $[\s]$ of $\s$ with the IC-signatures, with ${\GW}_1, \ldots, {\GW}_{|\Z_s|}$ being the underlying core-sets. Since
  $r$ is also linearly secret-shared with IC-signatures, with ${\GW}_1, \ldots, {\GW}_{|\Z_s|}$ being the underlying core-sets, from the linearity property of the secret sharing, it follows that
  $[\s - r]$ will be the same as a linear secret-sharing of $s$ with IC-signatures, with ${\GW}_1, \ldots, {\GW}_{|\Z_s|}$ being the underlying core-sets.
  This proves the {\it $\Z_s$-correctness}.
  The {\it privacy} of $s$ follows since $r$ remains random for the adversary and hence $\s$ {\it does not} reveal any information about $s$ to the adversary.
  The {\it $\Z_a$-correctness} and {\it privacy} for an {\it honest} dealer in an {\it asynchronous} network follows using similar arguments as above, except that we now rely on the 
  {\it $\Z_a$-correctness} of $\Rec$ in the {\it asynchronous} network (Lemma \ref{lemma:Rec}) and the  {\it $\Z_a$-validity} of $\BC$ in the {\it asynchronous} network (Theorem \ref{thm:BC}).
  
  We next consider a {\it corrupt} dealer and a {\it synchronous} network. Let $P_h$ be the {\it first} honest party which computes some output in the protocol, at the time $T$. This implies that $P_h$ has received some value
   $\s$ from the broadcast of $\D$. From the {\it $\Z_s$-consistency} and {\it $\Z_s$-fallback consistency} of $\BC$ in the {\it synchronous} network (Theorem \ref{thm:BC}), it follows that all the honest parties will receive
   $\s$ from the broadcast of $\D$, latest by the time $T + \Delta$. We define
   \[ s^{\star} \defined \s - r.  \]
   Since the (honest) parties take the
  default linear secret-sharing $[\s]$ of $\s$ with the IC-signatures, with ${\GW}_1, \ldots, {\GW}_{|\Z_s|}$ being the underlying core-sets
  and since $r$ is also linearly secret-shared with IC-signatures, with ${\GW}_1, \ldots, {\GW}_{|\Z_s|}$ being the underlying core-sets, from the 
  linearity property of the secret sharing, it follows that
  $[\s - r]$ will be the same as a linear secret-sharing of $s^{\star}$ with IC-signatures, with ${\GW}_1, \ldots, {\GW}_{|\Z_s|}$ being the underlying core-sets.
  This completes the proof of the {\it $\Z_s$-commitment} in the {\it synchronous} network. The proof of the {\it $\Z_a$-commitment} in the {\it asynchronous} network is similar as above, except that
  we now rely on the {\it $\Z_a$-weak consistency} and {\it $\Z_a$-fallback consistency} of $\BC$ in the {\it asynchronous} network (Theorem \ref{thm:BC})
  and {\it $\Z_a$-correctness} of $\Rec$ in the {\it asynchronous} network (Lemma \ref{lemma:Rec}).
  
  In the protocol, one instance of $\Rec$ with $|\ReceiverSet| = 1$ and one instance of $\BC$ with $\ell = \log{|\F|}$ bits are invoked. 
    The communication complexity now follows from the communication complexity of $\Rec$ (Lemma \ref{lemma:Rec}) and communication complexity of $\BC$ (Theorem \ref{thm:BC}).
\end{proof}
\subsection{Properties of the Protocol $\BasicMult$}
In this section, we prove the properties of the protocol $\BasicMult$ (see Fig \ref{fig:BasicMult} for the formal details).
 While proving these properties, we will assume that {\it no honest} party is present in the set $\Discarded$; 
 looking ahead, this will be ensured in the protocol $\RandMultCI$ (presented in Section \ref{ssec:RandMultCI}), where the set $\Discarded$ is maintained. 
 
 We begin by showing that no summand-sharing party during the first two phases are from the discarded set of parties. \\~\\
 \noindent {\bf Lemma \ref{lemma:BasicMultFuture}.}
{\it During any instance $\BasicMult(\Z_s, \Z_a, \ShareSpec_{\Z_s}, [a], [b], {\GW}_1, \ldots, \allowbreak {\GW}_{|\Z_s|}, \Discarded, \iter)$ of $\BasicMult$,
 if $P_j \in \Selected_{\iter}$ then $P_j \not \in \Discarded$, irrespective of the network type. 
 }  
\begin{proof}
Let $P_j$ be an {\it arbitrary} party belonging to the set $\Selected_{\iter}$. This implies that $P_j$ is included to $\Selected_{\iter}$, either during Phase I or Phase II.
 If $P_j \in \Discarded$, then
  no honest party will participate with input $1$ in the instance $\BA^{(\phaseone, \hop, j)}$ during Phase I for any value of $\hop$ and 
   instance $\BA^{(\phasetwo, j)}$ during Phase II. Consequently, from the 
  {\it $\Z_s$-validity} and {\it $\Z_a$-validity} of $\BA$ (Theorem \ref{thm:BA}),
  party $P_j$ will not be added to $\Selected_\iter$, which is a contradiction.
\end{proof}

 We next show that the first phase will get over for the honest parties
  after a fixed time in a {\it synchronous} network and eventually in an asynchronous network. Towards this we show that if the honest parties start any hop during the first phase, then they will complete it after a fixed time
   in a {\it synchronous} network and eventually in an asynchronous network. \\~\\
\noindent {\bf Lemma \ref{lemma:BasicMultHopTermination}.}
{\it Suppose that no honest party is present in $\Discarded$. If the honest parties start participating during hop number $\hop$ of Phase I of $\BasicMult$ with iteration number $\iter$, 
 then except with probability $\Order(n^3 \cdot \errorAICP)$,
 the hop takes $\TimeLSh + 2\TimeBA$ time to complete in a synchronous network, or almost-surely completes eventually in an asynchronous network.
}
\begin{proof}
Let $Z^{\star}$ be the set of {\it corrupt} parties and let $\Hon = \PartySet \setminus Z^{\star}$ be the set of {\it honest} parties. 
 We note that since $\Z_a \subset \Z_s$, {\it irrespective} of the network type,
 there exists some set in $\ShareSpec_{\Z_s}$, say $S_h$, such that $S_h \in \ShareSpec_{\Z_s}$.
  Since the honest parties participate in hop number $\hop$, it implies that there exists {\it no} $S_q \in \ShareSpec_{\Z_s}$,
   such that $\Products^{(S_q)}_\iter = \emptyset$.
   Particularly, this implies that $\Products^{(S_h)}_\iter \neq \emptyset$.
   Hence there exists some $P_j \in \Hon$, such that $P_j \notin \Selected_\iter$.
   This is because if $\Hon \subseteq \Selected_\iter$, then clearly $\Products^{(S_h)}_\iter =  \emptyset$ and hence the parties in $\Hon$ will {\it not} participate in hop number $\hop$.
   
    Let us first consider a {\it synchronous} network. During the hop number $\hop$, {\it every} $P_j \in \Hon$ such that $P_j \notin \Selected_\iter$ will invoke an instance of $\LSh$ with the input 
    $c^{(j)}_\iter$. Then from the {\it $\Z_s$-correctness} of $\LSh$ in the {\it synchronous} network (Lemma \ref{lemma:LSh}),
     after time $\TimeLSh$, the parties in $\Hon$ output $[c^{(j)}_\iter]$, except with probability $\Order(n^3 \cdot \errorAICP)$. 
     Since $(\Hon \cap \Discarded) = \emptyset$, 
     all the parties in $\Hon$ will participate in the instance $\BA^{(\phaseone, \hop, j)}$ with input $1$. Hence
     from the {\it $\Z_s$-validity} of $\BA$ in the {\it synchronous} network (Theorem \ref{thm:BA}), all the parties in $\Hon$
      output $1$ during the instance $\BA^{(\phaseone, \hop, j)}$,
       after time $\TimeLSh + \TimeBA$. 
       Consequently, the parties in $\Hon$ start participating with input $0$ in the remaining $\BA$ instances 
       $\BA^{(\phaseone, \hop, k)}$ (if any), for which no input is provided yet. Hence from the {\it $\Z_s$-security} of $\BA$ in the {\it synchronous} network, at the 
       time $\TimeLSh + 2\TimeBA$, the parties in $\Hon$ compute some output during the $\BA$ instances invoked during hop number $\hop$       
       and hence complete the hop number $\hop$.

Next, consider an {\it asynchronous} network. 
 From the {\it $\Z_a$-correctness} of $\LSh$ in the {\it asynchronous} network (Lemma \ref{lemma:LSh}),
 corresponding to  {\it every} $P_j \in \Hon$ such that $P_j \notin \Selected_\iter$, 
  the parties in $\Hon$ eventually output $[c^{(j)}_\iter]$, except with probability $\Order(n^3 \cdot \errorAICP)$, during the instance
   $\LSh^{(\phaseone, \hop, j)}$. Now there are now two cases.
\begin{myitemize}
\item[--] {\bf Case 1 - There exists some $P_j \in \Hon$ where $P_j \not \in \Selected_\iter$ and some party $P_i \in \Hon$, 
 such that 
 $P_i$ has started participating with input $0$ in the instance $\BA^{(\phaseone, \hop, j)}$}: This implies that
  $P_i$ has computed the output $1$ in some $\BA$ instance during the Phase I, say $\BA^{(\phaseone, \hop, m)}$. And hence $P_i$ starts participating in all the {\it remaining} $\BA$ instances of 
   hop number $\hop$ of Phase I
   (if any) with the input $0$. From the {\it $\Z_a$-consistency} of $\BA$ in the {\it asynchronous} network, all the parties in $\Hon$ will also eventually compute the output $1$ during the instance
   $\BA^{(\phaseone, \hop, m)}$ and will start participating in all the {\it remaining} $\BA$ instances of hop number $\hop$ of Phase I
   (if any) with the input $0$. Consequently, from the {\it $\Z_a$-security} of $\BA$ in the {\it asynchronous} network (Theorem \ref{thm:BA}), almost-surely, the parties in $\Hon$ eventually compute some
   output during all the $\BA$ instances of hop number $\hop$  
       and hence complete the hop number $\hop$.
\item[--] {\bf Case 2 - No honest party has yet started participating with input $0$ in any of the BA instances of Phase I corresponding to the honest parties}:
  In this case, the honest parties will eventually start participating with input $1$ in the $\BA$ instances $\BA^{(\phaseone, \hop, j)}$,
  corresponding to the parties $P_j \in \Hon$ where $P_j \not \in \Selected_\iter$. Hence from the 
  {\it $\Z_a$-validity} of $\BA$ in the {\it asynchronous} network (Theorem \ref{thm:BA}),
  it follows that eventually, there will be some $P_j \in \Hon$ where $P_j \not \in \Selected_\iter$, such that
  all the parties in $\Hon$ compute the output $1$ during the instance $\BA^{(\phaseone,\hop, j)}$. The rest of the proof is similar to what is argued in the previous case.
\end{myitemize}
\end{proof}

We next show that in protocol $\BasicMult$, the honest parties compute some output, after a fixed time in a {\it synchronous} network and eventually in an {\it asynchronous} network. \\~\\
{\bf Lemma \ref{lemma:BasicMultTermination}.}
{\it If no honest party is present in $\Discarded$, then in protocol $\BasicMult$, except with probability $\Order(n^3 \cdot \errorAICP)$,
 all honest parties compute some output by the time $\TimeBasicMult = (2n + 1) \cdot \TimeBA + (n + 1) \cdot \TimeLSh + \TimeRec$ in a synchronous network, or almost-surely,
  eventually in an asynchronous network.
}
\begin{proof}
In the protocol $\BasicMult$, to compute an output, the (honest) parties need to complete the three phases. We first show that 
 {\it irrespective} of the network type, the honest parties will complete Phase II, provided they complete Phase I.
 This is because 
 the honest parties will participate in {\it all} the $\BA$ instances 
  $\BA^{(\phasetwo, j)}$ of phase II with {\it some} input, 
   after waiting {\it exactly} for time $\TimeLSh$. Consequently, once phase I is completed, from the 
   {\it $\Z_s$-security} of $\BA$ in the {\it synchronous} network (Theorem \ref{thm:BA}), it takes $\TimeLSh + \TimeBA$ time
   for the honest parties to compute outputs in the $\BA$ instances during phase II in a {\it synchronous} network. 
   On the other hand, the {\it $\Z_a$-security} of $\BA$ in the {\it asynchronous} network (Theorem \ref{thm:BA})
    guarantees that almost-surely, all honest parties 
   eventually compute some output during the $\BA$ instances during phase II in an asynchronous network. Now once Phase I and Phase II are completed,
   it takes $\TimeRec$ time for the parties to complete Phase III in a {\it synchronous} network (follows from Lemma \ref{lemma:Rec}),
   while in an {\it asynchronous} network, it gets completed eventually (see Lemma \ref{lemma:Rec}).

We now show that Phase I always gets completed for the honest parties.
  Lemma \ref{lemma:BasicMultHopTermination} guarantees that if the honest parties start any hop during Phase I, then it gets completed for all the honest parties
 after time $\TimeLSh + 2\TimeBA$ in a {\it synchronous} network or eventually in an {\it asynchronous} network. 
  From the protocol steps, it follows that there can be {\it at most} $n$ hops during Phase I. This is because once the set of {\it honest} parties are included in the set of summand-sharing parties
  $\Selected_\iter$ during Phase I, then the parties will exit Phase I. 
\end{proof}

We next show that the adversary does not learn anything additional about $a$ and $b$ during the protocol. \\~\\
{\bf Lemma \ref{lemma:BasicMultPrivacy}.}
{\it If no honest party is present in $\Discarded$, then the view of the adversary remains independent of $a$ and $b$ throughout the protocol, irrespective of the network type.
}
\begin{proof}
Let $Z^{\star}$ be the set of {\it corrupt} parties and let $\Hon = \PartySet \setminus Z^{\star}$ be the set of {\it honest} parties. 
 The view of the adversary remains independent of $a$ and $b$ during Phase I and Phase II.
  This is because the {\it privacy} of $\LSh$ (Lemma \ref{lemma:LSh}) ensures that the view of the adversary remains independent of the summand-sums shared by the {\it honest}
  summand-sharing parties during Phase I and Phase II. Let $S_h \in \ShareSpec_{\Z_s}$ be the group consisting of {\it only} honest parties; i.e. $(S_h \cap Z^{\star}) = \emptyset$.
  To prove that the view of the adversary remains independent of $a$ and $b$ during Phase III, we show that {\it irrespective} of the network type, 
    any summand of the form $(h, q)$ or $(p, h)$
    will {\it not} be present in $\Products_\iter$ during this phase corresponding to any $p, q \in \{1, \ldots, |\Z_s| \}$, implying that the shares $[a]_h$ and $[b]_h$ {\it does not} get publicly reconstructed.
    Note that $(h, q) \not \in \Products^{(j)}_\iter$ for any $P_j \in Z^{\star}$, as otherwise
 it would imply that a {\it corrupt} $P_j \in (S_h \cap Z^{\star})$, which is a contradiction. 
  Similarly, $(p, h) \not \in \Products^{(j)}_\iter$ for any $P_j \in Z^{\star}$

Let us first consider a {\it synchronous} network and consider an {\it arbitrary} ordered pair $(h, q)$. 
  Consider the case when 
  $(h, q)$ has {\it not been} removed from $\Products_\iter$ during Phase I in any of the hops. 
  Then, during Phase II, the pair $(h, q)$ will get {\it statically} re-assigned to the $\Products^{(j)}_\iter$ set of some $P_j \in \Hon$, such that $P_j \not \in \Selected_{\iter}$.
  From the protocol steps, $P_j$ will include the summand $[a]_h \cdot [b]_q$ while computing $c^{(j)}_\iter$ and share
  $c^{(j)}_\iter$ through an instance of $\LSh$. From the {\it $\Z_s$-correctness} of $\LSh$ in the {\it synchronous} network, the parties in $\Hon$ will output $[c^{(j)}_\iter]$ during the instance of
  $\LSh$ invoked by $P_j$. Since $P_j \notin (\Selected_\iter \cup \Discarded)$, it follows that all the parties in $\Hon$
   will participate in the instance $\BA^{(\phasetwo, j)}$ with input $1$ and compute the output $1$. Consequently, $(h, q)$ will be removed from the updated $\Products_\iter$ during Phase II, if not removed
   during Phase I. By the same logic, any ordered pair of the form $(p, h)$ will also be removed from $\Products_\iter$, by the end of Phase II.

Next, consider an {\it asynchronous} network and an {\it arbitrary} $(h, q)$.  In this case, we show
 that $(h, q)$ will eventually be removed from $\Products_\iter$ during Phase I itself. 
  Let the parties in $\Hon$ complete Phase I. 
  This implies that there exists some $S_{\ell} \in \ShareSpec_{\Z_s}$ such that $\Products^{(S_{\ell})}_\iter = \emptyset$ when Phase I gets over. 
  We show that $(h, q)$ was present in $\Products^{(S_\ell)}_\iter$ at the beginning of Phase I, when the parties initialize 
  $\Products^{(S_\ell)}_\iter$. That is, at the time of initialization, there was some $P_j \in (S_{\ell} \cap \Hon \cap S_q)$ such that 
  $(h, q) \in \Products^{(j)}_{\iter}$. For this, it is enough to show that $(S_{\ell} \cap \Hon \cap S_q) \neq \emptyset$, which follows from the fact that
  $\Z_s$ and $\Z_a$ satisfy the condition $\Q^{(2, 1)}(\PartySet, \Z_s, \Z_a)$ and 
  $(S_{\ell} \cap \Hon \cap S_q) = \PartySet \setminus (Z_{\ell} \cup Z^{\star} \cup Z_q)$, where 
  $Z^{\star} \in \Z_a$ (because we are considering an {\it asynchronous} network) and $Z_{\ell}, Z_q \in \Z_s$ (follows from the construction of $\ShareSpec_{\Z_s}$). 
  Let $P_j \in (S_{\ell} \cap \Hon \cap S_q)$.
   Hence at the time of the initialization, $(h, q) \in \Products^{(j)}_\iter$, which further implies that 
    $(h, q) \in \Products^{(j)}_\iter \subseteq \Products^{(S_\ell)}_\iter$ at the time of initialization.
    And hence $(h, q)$ must have been removed from $\Products_\iter$ by the end of Phase I, since $P_j$ would have been selected as the summand-sharing party in one of the hops during Phase I.
     The same logic also applies to any arbitrary $(p, h)$, implying that $(p, h)$ would have been removed from $\Products_\iter$ by the end of Phase I itself.
\end{proof}

We next show that if all the parties behave honestly in the protocol, the parties output a linear-secret sharing of $a \cdot b$ with IC-signatures. \\~\\
{\bf Lemma \ref{lemma:BasicMultCorrectness}.}
{\it If no honest party is present in $\Discarded$ and if all parties in $\PartySet \setminus \Discarded$ behave honestly, then in protocol $\BasicMult$, the honest parties output
 a linear secret-sharing of $a \cdot b$ with IC-signatures, with ${\GW}_1, \ldots, {\GW}_{|\Z_s|}$ being the underlying core-sets, irrespective of the network type.
}
\begin{proof}
Note that $(\Selected_\iter \cap \Discarded) = \emptyset$, which follows from Lemma \ref{lemma:BasicMultFuture}.
Now
 Consider an {\it arbitrary} $(p, q) \in \Products_\iter$. We claim that {\it irrespective} of the network type, if all the parties behave honestly, 
 then the summand $[a]_p \cdot [b]_q$ is considered on behalf of {\it exactly} one party while
 secret-sharing the summand-sums. That is there is {\it exactly} one $P_j$ across the three phases,
  such that the following hold:
  \[ c^{(j)}_\iter = \ldots +  [a]_p[b]_q + \ldots.\] 
  Assuming that the claim is true, the proof then follows from the linearity property of secret-sharing and the fact that $c_\iter = c^{(1)}_\iter + \ldots + c^{(n)}_\iter$ holds. That is, corresponding to each $P_j \in \Selected_\iter$, the value
  $c^{(j)}_\iter$ is secret-shared through an instance of $\LSh$, with ${\GW}_1, \ldots, {\GW}_{|\Z_s|}$ being the underlying core-sets, which guarantees that 
   $c^{(j)}_\iter$ is linearly secret-shared with IC-signatures, with ${\GW}_1, \ldots, {\GW}_{|\Z_s|}$ being the underlying core-sets.
   Moreover, by the end of Phase II, the parties will be publicly knowing the set $\Selected_\iter$, which follows from the security properties of $\BA$.
   Furthermore, after Phase II, the parties will be publicly reconstructing the shares $[a]_p$ and $[b]_q$ corresponding to every $(p, q)$, which are still present in $\Products_\iter$.
   Since corresponding to each 
   $P_j \in \PartySet \setminus \Selected_\iter$, the parties take the default linear secret-sharing of 
   $c^{(j)}_\iter$ with IC-signatures and ${\GW}_1, \ldots, {\GW}_{|\Z_s|}$ being the underlying core-sets, it follows that 
  $c$ will be linearly secret-shared with IC-signatures and ${\GW}_1, \ldots, {\GW}_{|\Z_s|}$ being the underlying core-sets.
   Moreover, $c = a \cdot b$ holds, since each summand $[a]_p \cdot [b]_q$ is considered exactly once, as per our claim.
   We now proceed to prove our claim. 
   
   We first show that there exists {\it at least one} $P_j$ in one of the three phases, such that the following holds:
   \[c^{(j)}_\iter = \ldots + [a]_p[b]_q + \ldots .\]
 For this, consider the following cases.
\begin{myitemize}
	\item[--] {\bf Case 1 --- During Phase I, there is some $P_j \in \Selected_\iter$, such that $(p, q)$ was present in the set $\Products^{(j)}_\iter$ when $P_j$ 
	was added to $\Selected_\iter$}:  In this case,  there is nothing to show.
	\item[--] {\bf Case 2 --- At the end of Phase I, $(p, q)$ is still {\it present} in $\Products_\iter$}: 
	This implies that at the end of Phase I, $(\Selected_\iter \cap S_p \cap S_q) = \emptyset$. Since $\Z_s$ and $\Z_a$ satisfy the condition $\Q^{(2, 1)}(\PartySet, \Z_s, \Z_a)$,
	it follows that $(S_p \cap S_q) \neq \emptyset$. Let $P_j \defined \min(S_p, S_q)$.
	Note that $P_j \not \in \Selected_\iter$ at the end of Phase I, otherwise $(p, q)$ would have been {\it removed} from $\Products_\iter$, which is a contradiction. Now if
	$P_j \not \in \Discarded$, then clearly the summand $[a]_p \cdot [b]_q$ will be present in $c^{(j)}_\iter$, shared by $P_j$ during the Phase II.
	Else, during Phase III, the parties will publicly reconstruct $[a]_p$ and $[b]_q$ and consequently $[a]_p \cdot [b]_q$ will be present in the default secret-sharing of $c^{(j)}_\iter$, taken on the
	behalf of $P_j$. \\[.2cm]
\end{myitemize}

To complete the proof of our claim, we next show that $(p, q)$ {\it cannot} be present in the summand-sum of {\it more than one} party across the three phases. 
 On {\it contrary}, let $P_j$ and $P_k$ be two {\it distinct} parties, such that the following holds across the three phases:
 \[ c^{(j)}_\iter = \ldots +  [a]_p[b]_q + \ldots  \quad \wedge  \quad c^{(k)}_\iter = \ldots +  [a]_p[b]_q + \ldots  \]
Now there are three following cases. 
\begin{myitemize}
	\item[--] {\bf Case 1 - $P_j, P_k \in \Selected_\iter$ at the end of Phase I}: From the security properties of $\BA$, the parties will 
	agree on which party to add to $\Selected_\iter$ during every hop during Phase I.
	Moreover, from the protocol steps, exactly one party is selected as a summand-sharing party and added to $\Selected_\iter$ in each hop.
     Suppose that $P_j$ was added to $\Selected_\iter$ during $\hop^{(j)}$, and that $P_k$ was added to $\Selected_\iter$ during $\hop^{(k)}$. 
     Moreover, without loss of generality, let $\hop^{(j)} < \hop^{(k)}$. From the protocol steps, it follows that $P_j, P_k \not \in \Discarded$, as otherwise no honest party would have voted for $P_j$
     and $P_k$ as a candidate summand-sharing party and consequently, $P_j, P_k \not \in \Selected_\iter$. 
     Now as per the lemma conditions, all the parties (including $P_k$) behave honestly. Hence, the ordered pair 
     $(p, q)$ would be removed from $\Products^{(k)}_\iter$ at the end of $\hop^{(j)}$.
     Consequently, $P_k$ will {\it not include} the summand $[a]_p \cdot [b]_q$ while computing $c^{(k)}_\iter$ during hop number $\hop^{(k)}$, which is a contradiction.
	\item[--] {\bf Case 2 --- $P_j \in \Selected_\iter$ at the end of Phase I and $P_k \notin \Selected_\iter$ at the end of Phase I}:
	In this case, $(p, q)$ will be {\it removed} from $\Products_\iter$ at the end of Phase I.
	 Hence, it {\it cannot} get re-assigned to any other party $P_k$ after Phase I and hence {\it cannot} belong to $\Products^{(k)}_\iter$.
	  Consequently, the summand $[a]_p \cdot [b]_q$ will {\it not} be considered while computing $c^{(k)}_\iter$, which is a contradiction,
	\item[--] {\bf Case 3 - $P_j \notin \Selected_\iter$ and $P_k \notin \Selected_\iter$ at the end of Phase I}: 
	In this case, the summand $[a]_p \cdot [b]_p$ is {\it deterministically} and {\it statically} re-assigned to the least-indexed party from the set $(S_p \cap S_q)$, as per the 
	$\min$ function. Hence $(p, q)$ will be re-assigned to either $P_j$ or $P_k$, but {\it not both}.
	Consequently, the summand $[a]_p \cdot [b]_q$ will be considered while computing either $c^{(j)}_\iter$ or  $c^{(k)}_\iter$, but not both, which is a contradiction.
\end{myitemize}
\end{proof}

{\bf Lemma \ref{lemma:BasicMultCommunication}.}
{\it Protocol $\BasicMult$ incurs a communication of $\Order(|\Z_s| \cdot n^5 \cdot \log{|\F|} + n^6 \cdot \log{|\F|} \cdot |\sigma|)$
 bits and makes $\Order(n^2)$ calls to $\BA$.
}
\begin{proof}
During Phase I, there can be $\Order(n)$ hops, where during each hop, a party from $\PartySet \setminus \Discarded$ secret-shares a field element through an instance of $\LSh$. Moreover,
 $n$ instances of $\BA$ are invoked to agree upon the summand-sharing party of the hop.
  Hence, $\Order(n^2)$ instances of $\LSh$ and $\Order(n^2)$ instances of $\BA$ are required during Phase I.
  During Phase II, $\Order(n)$ instances of $\LSh$ and $\Order(n)$ instances of $\BA$ are required.
  Finally during Phase III, up to $\Order(|\Z_s|)$ shares need to be publicly reconstructed. 
  The communication complexity now follows from the communication complexity of $\LSh$ (Lemma \ref{lemma:LSh}) and 
  communication complexity of $\RecShare$ with $|\ReceiverSet| = n$ (Lemma \ref{lemma:RecShare}).
\end{proof}

As a corollary of Lemma \ref{lemma:BasicMultCommunication}, we derive the following corollary, which determines the {\it maximum} number of instances of $\LSh$ which are invoked during an instance
 of $\BasicMult$. Looking ahead, this will be useful to later calculate the maximum number of instances of $\LSh$ which need to be invoked as part of our final multiplication protocol.
 This will be further useful to determine the number of linearly secret-shared values with IC-signatures and core-sets ${\GW}_1, \ldots, {\GW}_{|\Z_s|}$, which need to be generated through the
  protocol $\Rand$ beforehand. \\~\\
{\bf Corollary \ref{cor:BasicMultUpperBound}.}
{\it During any instance of $\BasicMult$, there can be at most $n^2 + n$ instances of $\LSh$ invoked.
}
\begin{proof}
The proof follows from the fact that during Phase I, there can be up to $n^2$ instances of $\LSh$, if $\Discarded = \emptyset$
 and during Phase II, there can be up to $n - 1$ instances of $\LSh$, if only one party is added to $\Selected_\iter$ during Phase I.
\end{proof}
\subsection{Properties of the Protocol $\RandMultCI$}
In this section, we prove the properties of the protocol $\RandMultCI$ (see Fig \ref{fig:RandMultCI} for the formal description).
 We begin by showing that irrespective of the network type, the (honest) parties compute linearly secret-shared $a_\iter, b_\iter, b'_\iter$
 and $r_\iter$ with IC-signatures, which are random from the point of view of the adversary. \\~\\
 \noindent {\bf Lemma \ref{lemma:RandMultCIACS}.}
{\it  In protocol $\RandMultCI$, the following hold.
 \begin{myitemize}
 \item[--] {\bf Synchronous Network}: Except with probability $\Order(n^3 \cdot \errorAICP)$, honest parties will have linearly secret-shared $a_\iter, b_\iter, b'_\iter$
 and $r_\iter$ with IC-signatures, with ${\GW}_1, \ldots, {\GW}_{|\Z_s|}$ being the underlying core-sets, by the time $\TimeLSh + 2\TimeBA$. Moreover, adversary's view is independent of
  $a_\iter, b_\iter, b'_\iter$ and $r_\iter$. 
 \item[--] {\bf Asynchronous Network}: Except with probability $\Order(n^3 \cdot \errorAICP)$, almost-surely, honest parties will eventually have linearly secret-shared $a_\iter, b_\iter, b'_\iter$
 and $r_\iter$ with IC-signatures, with ${\GW}_1, \ldots, {\GW}_{|\Z_s|}$ being the underlying core-sets. Moreover, adversary's view is independent of
  $a_\iter, b_\iter, b'_\iter$ and $r_\iter$. 
 \end{myitemize}
 }
 \begin{proof}
 We first consider a {\it synchronous} network. Let $Z^{\star} \in \Z_s$ be the set of {\it corrupt} parties and
  let $\Hon =  \PartySet \setminus Z^{\star}$
  be the set of {\it honest} parties.
     Corresponding to each $P_j \in \Hon$, 
     the honest parties compute the output $[a^{(j)}_\iter], [b^{(j)}_\iter], [b'^{(j)}_\iter]$ and $[r^{(j)}_\iter]$ during the instances of $\LSh$ invoked by $P_j$ at the time
     $\TimeLSh$, except with probability $\Order(n^3 \cdot \errorAICP)$. This follows from the {\it $\Z_s$-correctness} of $\LSh$ in the {\it synchronous} network (Lemma \ref{lemma:LSh}). 
      Consequently, at the time $\TimeLSh$, all the parties in $\Hon$ will be present in the set $\CSet_i$ of every $P_i \in \Hon$. 
      Hence corresponding to each $P_j \in \Hon$, each $P_i \in \Hon$
    starts participating with input $1$ in the instance $\BA^{(j)}$ at the time
  $\TimeLSh$. Hence from the {\it $\Z_s$-validity} and {\it $\Z_s$-guaranteed liveness} properties of
   $\BA$ in the {\it synchronous} network (Theorem \ref{thm:BA}), it follows that at the time $\TimeLSh + \TimeBA$,
    all the parties in $\Hon$ compute the output $1$ during the instance $\BA^{(j)}$,
  corresponding to every $P_j \in \Hon$. Consequently, at the time $\TimeLSh + \TimeBA$,
  all the parties in $\Hon$ will start participating in the remaining $\BA$ instances for which no input has been provided yet (if there are any). 
  And from the {\it $\Z_s$-guaranteed liveness} and {\it $\Z_s$-consistency} properties
   of $\BA$ in the {\it synchronous} network (Theorem \ref{thm:BA}),
   these $\BA$ instances will produce common outputs for every honest party by the time $\TimeLSh + 2 \TimeBA$.
  Hence, at the time $\TimeLSh + 2 \TimeBA$, the honest parties will have a common $\CoreSet$, where $\PartySet \setminus \CoreSet \in \Z_s$
   and where $\Hon \subseteq \CoreSet$.
  We next wish to show that corresponding to {\it every} $P_j \in \CoreSet$, 
  there exists some quadruplet of values, which are linearly secret-shared with IC-signatures, with ${\GW}_1, \ldots, {\GW}_{|\Z_s|}$ being the underlying core-sets.
    
 Consider an {\it arbitrary} party $P_j \in \CoreSet$. If $P_j \in \Hon$,
  then whatever we wish to show is correct, as shown above.
  Next, consider a {\it corrupt} $P_j \in \CoreSet$. Since $P_j \in \CoreSet$, it follows that the instance $\BA^{(j)}$ produces the output $1$
  for all honest parties.
   This further implies that at least one 
  {\it honest} $P_i$ must have computed some output during the instances of $\LSh$ invoked by $P_j$, by the time
   $\TimeLSh + \TimeBA$ (implying that $P_j \in \CSet_i$) and participated
  with input $1$ in the instance $\BA^{(j)}$. Otherwise, 
  {\it all honest} parties would participate with input
   $0$ in the instance $\BA^{(j)}$ at the time $\TimeLSh + \TimeBA$ and 
   then from the {\it $\Z_s$-validit}y of $\BA$ in the {\it synchronous} network, every honest party would have
   computed the output $0$ in the instance
  $\BA^{(j)}$ and hence $P_j$ will not be present in $\CoreSet$, which is a contradiction. 
  Now if $P_i$ has computed
   some output during the instances of $\LSh$ invoked by $P_j$ at the time $\TimeLSh + \TimeBA$, then from the {\it $\Z_s$-commitment} of $\LSh$ in the
   {\it synchronous} network (Lemma \ref{lemma:LSh}), it follows that except with probability $\Order(n^3 \cdot \errorAICP)$,
   there exist values $(a^{(j)}_\iter, b^{(j)}_\iter, b'^{(j)}_\iter, r^{(j)}_\iter)$, which will be 
   linearly secret-shared with IC-signature, with ${\GW}_1, \ldots, {\GW}_{|\Z_s|}$ being the underlying core-sets, 
   by the time $\TimeLSh + \TimeBA + \Delta$.
  Since $\Delta < \TimeBA$, it follows that by the time $\TimeLSh + 2\TimeBA$, 
  the honest parties will have $[a^{(j)}_\iter], [b^{(j)}_\iter], [b'^{(j)}_\iter]$ and $[r^{(j)}_\iter]$. 
  From the linearity property of secret-sharing, it then follows that by the time $\TimeLSh + 2\TimeBA$,
  the values $a_\iter, b_\iter, b'_\iter$
 and $r_\iter$ will be linearly secret-shared with IC-signatures, with ${\GW}_1, \ldots, {\GW}_{|\Z_s|}$ being the underlying core-sets.

  From the {\it privacy} of $\LSh$ (Lemma \ref{lemma:LSh}), the view of the adversary will be independent of the values 
   $(a^{(j)}_\iter, b^{(j)}_\iter, b'^{(j)}_\iter, r^{(j)}_\iter)$, corresponding to the parties $P_j \in \Hon$. As $(\Hon \cap \CoreSet) \neq \emptyset$,
   it follows that $a_\iter, b_\iter, b'_\iter$ and $r_\iter$ will be indeed random from the point of view of the adversary. This completes the proof for the case of {\it synchronous} network.
   
   We next consider an {\it asynchronous} network. 
   Let $Z^{\star} \in \Z_a$ be the set of {\it corrupt} parties and
 let $\Hon = \PartySet \setminus Z^{\star}$
  be the set of {\it honest} parties. Notice that $\PartySet \setminus \Hon \in \Z_s$, since $\Z_a \subset \Z_s$. Now irrespective of the way messages are scheduled, there
    will be eventually a subset of parties $\PartySet \setminus Z$ for some $Z \in \Z_s$, such that
    all the parties in $\Hon$ participate with input $1$ in the instances of $\BA$, corresponding to the parties in
    $\PartySet \setminus Z$. This is because,
  corresponding to every $P_j \in \Hon$, all the parties in $\Hon$ {\it eventually} compute
  some output during the instances of $\LSh$ invoked by $P_j$ except with probability $\Order(n^3 \cdot \errorAICP)$, 
   which follows from the 
  {\it $\Z_a$-correctness} of $\LSh$ in the {\it asynchronous} network (Lemma \ref{lemma:LSh}). 
   So even if the {\it corrupt} parties $P_j$ do not invoke their respective instances of $\LSh$, 
   there will be a set of $\BA$ instances corresponding to the parties in $\PartySet \setminus Z$ for some $Z \in \Z_s$, 
    in which all
   the parties in $\Hon$ eventually participate with input $1$.
   Consequently, from the {\it $\Z_a$-almost-surely liveness} and {\it $\Z_a$-consistency} properties of 
   $\BA$ in the {\it asynchronous} network (Theorem \ref{thm:BA}), these $\BA$ instances
  eventually produce the output $1$ for all the parties in $\Hon$.
   Hence, all the parties in $\Hon$ eventually participate with some input in the remaining $\BA$ instances,
  which almost-surely produce some output for every honest party eventually. 
   From the properties of $\BA$ in the {\it asynchronous} network,
   it then follows that all the honest
  parties output the same $\CoreSet$.
  
  Now consider an {\it arbitrary} 
  $P_j \in \CoreSet$. It implies that the honest parties computed the output $1$ during
   the instance $\BA^{(j)}$, which further implies that at least one
   {\it honest} $P_i$ participated with input $1$ in $\BA^{(j)}$ after computing
    its output in the instances of $\LSh$ invoked by $P_j$. If $P_j$ is {\it honest}, then the
   {\it $\Z_a$-correctness} of $\LSh$ in the {\it asynchronous} network 
    guarantees that except with probability $\Order(n^3 \cdot \errorAICP)$, the values
        $(a^{(j)}_\iter, b^{(j)}_\iter, b'^{(j)}_\iter, r^{(j)}_\iter)$ 
     will be {\it eventually} linearly secret-shared with IC-signatures, with 
     ${\GW}_1, \ldots, {\GW}_{|\Z_s|}$ being the underlying core-sets. 
       On the other hand, even if $P_j$ is {\it corrupt}, the {\it $\Z_a$-commitment} 
   of $\LSh$ in the {\it asynchronous} network (Lemma \ref{lemma:LSh})
    guarantees that there exist values $(a^{(j)}_\iter, b^{(j)}_\iter, b'^{(j)}_\iter, r^{(j)}_\iter)$ which are 
   {\it eventually} linearly secret-shared with IC-signatures, with 
     ${\GW}_1, \ldots, {\GW}_{|\Z_s|}$ being the underlying core-sets, except with probability $\Order(n^3 \cdot \errorAICP)$.
      From the linearity property of secret-sharing, it then follows that eventually,
  the values $a_\iter, b_\iter, b'_\iter$
 and $r_\iter$ will be linearly secret-shared with IC-signatures, with ${\GW}_1, \ldots, {\GW}_{|\Z_s|}$ being the underlying core-sets.
    The privacy of $a_\iter, b_\iter, b'_\iter$ and $r_\iter$ is similar as for the synchronous communication network and the fact that 
    $(\Hon \cap \CoreSet) \neq \emptyset$ still holds in the {\it asynchronous} network.
 \end{proof}

We next claim that all honest parties will eventually agree on whether the instances of $\BasicMult$ in $\RandMultCI$ have succeeded or failed.\\~\\
{\bf Lemma \ref{lemma:RandMultCITermination}.}
{\it Consider an arbitrary $\iter$, 
  such that all honest parties participate in the instance $\RandMultCI(\PartySet, \Z_s, \Z_a, \ShareSpec_{\Z_s}, {\GW}_1, \ldots, {\GW}_{|\Z_s|}, \Discarded, \iter)$,
   where
   no honest party is present in $\Discarded$. 
   Then except with probability $\Order(n^3 \cdot \errorAICP)$, all honest parties
   reconstruct a (common) value $d_\iter$ and set $\flag_\iter$ to a common Boolean value, 
    at the time $\TimeLSh + 2\TimeBA + \TimeBasicMult + 3\TimeRec$ in a synchronous network, or eventually in an asynchronous network. 
}
\begin{proof}
From Lemma \ref{lemma:RandMultCIACS}, the honest parties have $[a_\iter], [b_\iter], [b'_\iter]$ and $[r_\iter]$ at the time $\TimeLSh + 2\TimeBA$ in a synchronous network, or eventually in an asynchronous network, 
 except with probability $\Order(n^3 \cdot \errorAICP)$, with ${\GW}_1, \ldots, {\GW}_{|\Z_s|}$ being the underlying core-sets.
 From Lemma \ref{lemma:BasicMultTermination}, it follows that the honest parties have
  $\{[c^{(1)}_\iter] , \ldots, [c^{(n)}_\iter], [c_\iter] \}$ and 
         $\{[c'^{(1)}_\iter], \ldots, \allowbreak  [c'^{(n)}_\iter], [c'_\iter] \}$ from the 
          corresponding instances of $\BasicMult$, either after time $\TimeBasicMult$ or eventually, based on the network type, where
          ${\GW}_1, \ldots, {\GW}_{|\Z_s|}$ are the underlying core-sets.
            From Lemma \ref{lemma:Rec},
          the honest parties reconstruct $r_\iter$ from the corresponding instance of $\Rec$ after time $\TimeRec$ in a {\it synchronous} network or eventually in an {\it asynchronous} network,
          except with probability $\Order(n^3 \cdot \errorAICP)$. 
           From the linearity property of secret-sharing, it then follows that the honest parties
          compute $[e_\iter]$ and hence reconstruct $e_{\iter}$ 
          from the corresponding instance of $\Rec$, after time $\TimeRec$ in a {\it synchronous} network or eventually in an {\it asynchronous} network, except with probability 
          $\Order(n^3 \cdot \errorAICP)$. Moreover,
          ${\GW}_1, \ldots, {\GW}_{|\Z_s|}$ will be the underlying core-sets for $[e_\iter]$. 
         Again, from the linearity property of secret-sharing, it follows that
          the honest parties compute $[d_\iter]$ with ${\GW}_1, \ldots, {\GW}_{|\Z_s|}$ being
           the underlying core-sets, followed by reconstructing $d_\iter$ from the corresponding instance of $\Rec$,
          which takes $\TimeRec$ time in a {\it synchronous} network or happens eventually in an {\it asynchronous} network.
         Thus, the honest parties will have  $d_\iter$ either at the time $\TimeLSh + 2\TimeBA + \TimeBasicMult + 3\TimeRec$ 
         in a synchronous network, or eventually in an asynchronous network. 
         Now depending upon the value of $d_\iter$, the honest parties set $\flag_\iter$ to either $0$ or $1$. 
\end{proof}

We next claim that if no cheating occurs, then the honest parties output a multiplication-triple, which is linearly secret-shared with IC-signatures.\\~\\
\noindent {\bf Lemma \ref{lemma:RandMultCIHonestBehaviour}.}
{\it Consider an arbitrary $\iter$, 
  such that all honest parties participate in the instance $\RandMultCI(\PartySet, \Z_s, \Z_a, \ShareSpec_{\Z_s}, {\GW}_1, \ldots, {\GW}_{|\Z_s|}, \Discarded, \iter)$,
   where
   no honest party is present in $\Discarded$.
   If no party in $\PartySet \setminus \Discarded$ behaves maliciously, then $d_\iter = 0$
   and 
    the honest parties output $([a_\iter], [b_\iter], [c_\iter])$     
     at the time $ \TimeLSh + 2\TimeBA + \TimeBasicMult + 3\TimeRec$ in a synchronous network or eventually in an asynchronous network, where $c_\iter = a_\iter \cdot b_\iter$ 
     and where ${\GW}_1, \ldots, {\GW}_{|\Z_s|}$ are
           the underlying core-sets
}   
\begin{proof}
If no party in $\PartySet \setminus \Discarded$ behaves maliciously, then from Lemma \ref{lemma:BasicMultCorrectness}, the honest parties compute
 $[c_\iter]$ and $[c'_\iter]$ from the respective instances of $\BasicMult$, such that $c_\iter = a_\iter \cdot b_\iter$ and 
 $c'_\iter = a_\iter \cdot b'_\iter$ holds
 and where ${\GW}_1, \ldots, {\GW}_{|\Z_s|}$ are
           the underlying core-sets.
  Moreover, from Lemma \ref{lemma:RandMultCITermination}, the honest parties will compute
   $d_\iter$ at the time $\TimeLSh + 2\TimeBA + \TimeBasicMult + 3\TimeRec$ in a synchronous network or eventually in an asynchronous network.
  Furthermore, if $c_\iter = a_\iter \cdot b_\iter$ and 
   $c'_\iter = a_\iter \cdot b'_\iter$ holds, the value $d_\iter$ will be $0$ and consequently, the honest parties will output 
  $([a_\iter], [b_\iter], [c_\iter])$. Furthermore, it is easy to see that 
   ${\GW}_1, \ldots, {\GW}_{|\Z_s|}$ will be
           the underlying core-sets. This is because all the secret-shared values in the protocol are linearly secret-shared with IC-signatures, 
           with ${\GW}_1, \ldots, {\GW}_{|\Z_s|}$ being
           the underlying core-sets.
\end{proof} 
 We next show that if $d_\iter \neq 0$, then the honest parties include at least one {\it new} maliciously-corrupt party in the set $\Discarded$. \\~\\
\noindent {\bf Lemma \ref{lemma:RandMultCICorruptBehaviour}.}
{\it Consider an arbitrary $\iter$, 
  such that all honest parties participate in the instance $\RandMultCI(\PartySet, \Z_s, \Z_a, \ShareSpec_{\Z_s}, {\GW}_1, \ldots, {\GW}_{|\Z_s|}, \Discarded, \iter)$,
   where
   no honest party is present in $\Discarded$. 
   If $d_\iter \neq 0$, then except with probability $\Order(n^3 \cdot \errorAICP)$, 
   the honest parties update $\Discarded$ by adding a new maliciously-corrupt
   party in $\Discarded$, either at the time $\TimeRandMultCI = \TimeLSh + 2\TimeBA + \TimeBasicMult + 4\TimeRec$ in a synchronous network or eventually in an asynchronous network.
  }
   \begin{proof}
   Let $d_\iter \neq 0$ and let $\Selected_\iter$ be the set of summand-sharing parties across the two instances of $\BasicMult$ executed in $\RandMultCI$.
   That is:
   \[ \Selected_\iter \defined \Selected_{\iter, c} \cup \Selected_{\iter, c'}.\] 
   Note that
   there exists no $P_j \in \Selected_\iter$ such that $P_j \in \Discarded$, which follows from Lemma \ref{lemma:BasicMultFuture}.
   We claim that there exists at least one party $P_j \in \Selected_\iter$, such that corresponding to
    $c^{(j)}_\iter$ and $c'^{(j)}_\iter$, the following holds:
\[   r_\iter \cdot c^{(j)}_\iter + c'^{(j)}_\iter \neq r_\iter \cdot \sum_{(p,q) \in \Products^{(j)}_{\iter, c}} [a_\iter]_p [b_\iter]_q  
    + \sum_{(p, q) \in \Products^{(j)}_{\iter, c'}} [a_\iter]_p [b'_\iter]_q. \]
   Assuming the above holds, the proof now follows from the fact that once the parties reconstruct 
   $d_\iter \neq 0$, they proceed to reconstruct the shares
   $\{ [a_\iter]_q, [b_\iter]_q,  [b'_\iter]_q\}_{S_q \in \ShareSpec_{\Z_s}}$ through appropriate instances of $\RecShare$ 
   and the values $c^{(1)}_\iter, \ldots, c^{(n)}_\iter,  c'^{(1)}_\iter, \ldots, c'^{(n)}_\iter$ through 
   appropriate instances of $\Rec$. From Lemma \ref{lemma:RandMultCITermination}, Lemma \ref{lemma:RecShare} and Lemma \ref{lemma:Rec},
    this happens by the time $\TimeLSh + 2\TimeBA + \TimeBasicMult + 4\TimeRec$ in a synchronous network or eventually in an asynchronous network,
    except with probability $\Order(n^3 \cdot \errorAICP)$.
     Upon reconstructing these values, party $P_j$ will be included in the set
   $\Discarded$. Moreover, it is easy to see that $P_j$ is a maliciously-corrupt party since, for every {\it honest}
   $P_j \in \Selected_\iter$, the following conditions hold:
   \[ \displaystyle  c^{(j)}_\iter = \sum_{(p,q) \in \Products^{(j)}_{\iter, c}}[a_\iter]_p [b_\iter]_q \quad \mbox{ and } \quad 
     c'^{(j)}_\iter = \sum_{(p,q) \in \Products^{(j)}_{\iter,c'}}[a_\iter]_p [b'_\iter]_q.\] 
    We prove the above claim through a contradiction. Let the following condition hold for {\it each} $P_j \in \Selected_\iter$:
    \[r_\iter \cdot c^{(j)}_\iter + c'^{(j)}_\iter = r_\iter \cdot \sum_{(p,q) \in \Products^{(j)}_{\iter, c}} [a_\iter]_p [b_\iter]_q + 
    \sum_{(p, q) \in \Products^{(j)}_{\iter, c'}} [a_\iter]_p [b'_\iter]_q. \]
   Next, summing the above equation over all $P_j \in \Selected_\iter$, we get that the following holds:
  \[
   \sum_{P_j \in  \Selected_\iter} r_\iter \cdot c^{(j)}_\iter + c'^{(j)}_\iter = \sum_{P_j \in  \Selected_\iter} r_\iter \cdot \sum_{(p,q) \in \Products^{(j)}_{\iter, c}} [a_\iter]_p [b_\iter]_q 
   + \sum_{(p, q) \in \Products^{(j)}_{\iter, c'}} [a_\iter]_p [b'_\iter]_q.
 \]
   This implies that the following holds:
  \[
   		r_\iter \cdot  \sum_{P_j \in  \Selected_\iter} c^{(j)}_\iter + c'^{(j)}_\iter =  r_\iter \cdot \sum_{P_j \in  \Selected_\iter}  \sum_{(p,q) \in \Products^{(j)}_{\iter, c}} [a_\iter]_p [b_\iter]_q 
    + \sum_{(p, q) \in \Products^{(j)}_{\iter, c'}} [a_\iter]_p [b'_\iter]_q.
\]
    Now based on the way $a_{\iter}, b_{\iter}, b'_{\iter}, c_{\iter}$ and $c'_{\iter}$  are defined, the above implies that the following holds:
    \[r_\iter \cdot c_\iter + c'_\iter = r \cdot a_\iter \cdot b_\iter + a_\iter \cdot b'_\iter \]
   This further implies that 
   \[ r_\iter \cdot c_\iter + c'_\iter = (r_\iter \cdot b_\iter + b'_\iter) \cdot a_\iter \]
   Since in the protocol $e_\iter \defined r_\iter \cdot b_\iter + b'_\iter$, the above implies that
   \[r_\iter \cdot c_\iter + c'_\iter  = e_\iter \cdot a_\iter  \quad \Rightarrow  \; e_\iter \cdot a_\iter - r_\iter \cdot c_\iter - c'_\iter = 0 \quad \Rightarrow \; d_\iter = 0,   \]
   where the last equality follows from the fact that in the protocol, $d_\iter \defined e_\iter \cdot a_\iter - r_\iter \cdot c_\iter - c'_\iter$.
   However $d_\iter = 0$ is a contradiction since, according to the hypothesis of the lemma, we are given that
    $d_\iter \neq 0$.
   \end{proof}

We next show that if the honest parties output a secret-shared triple in the protocol, then except with probability $\frac{1}{|\F|}$, the triple is a multiplication-triple.
 Moreover, the triple will be random for the adversary. \\~\\
\noindent {\bf Lemma \ref{lemma:RandMultCICorrectness}.}
{\it Consider an arbitrary $\iter$, 
  such that all honest parties participate in the instance $\RandMultCI(\PartySet, \Z_s, \Z_a, \ShareSpec_{\Z_s}, {\GW}_1, \ldots, {\GW}_{|\Z_s|}, \Discarded, \iter)$,
   where
   no honest party is present in $\Discarded$. 
   If $d_\iter = 0$, then the honest parties output linearly secret-shared
  $([a_\iter], [b_\iter], [c_\iter])$ with IC-signatures with ${\GW}_1, \ldots, {\GW}_{|\Z_s|}$ being
           the underlying core-sets, 
  at the time $\TimeLSh + 2\TimeBA + \TimeBasicMult + 3\TimeRec$ in a synchronous network or eventually in an asynchronous network where, 
   except with probability  $\frac{1}{|\F|}$, the condition $c_\iter = a_\iter \cdot b_\iter$ holds. 
   Moreover, the view of $\Adv$ will be independent of $(a_\iter, b_\iter, c_\iter)$.
}
\begin{proof}
Let $d_\iter = 0$. From Lemma \ref{lemma:BasicMultTermination}, all honest parties will agree that $d_\iter = 0$, either at the time $\TimeLSh + 2\TimeBA + \TimeBasicMult + 3\TimeRec$ in a synchronous network or eventually in an asynchronous network. Then, from the protocol steps, the honest parties output $([a_\iter], [b_\iter], [c_\iter])$, with ${\GW}_1, \ldots, {\GW}_{|\Z_s|}$ being
           the underlying core-sets.
  In the protocol $d_{\iter} \defined e_\iter \cdot a_\iter - r_\iter \cdot c_\iter - c'_\iter$, where
  $e_\iter \defined r_\iter \cdot b_\iter + b'_\iter$. Since $d_\iter = 0$ holds, it implies that the honest parties have verified that the following holds:
   \[r_\iter (a_\iter \cdot b_\iter - c_\iter) = (c'_\iter - a_\iter \cdot b'_\iter). \]
  We note that $r_\iter$ will be a random element from $\F$ and will be unknown to $\Adv$ till it is publicly reconstructed, which follows from Lemma \ref{lemma:RandMultCIACS}
   We also note that $r_\iter$ will be unknown to $\Adv$, till the outputs for the underlying instances
  of $\BasicMult$ are computed, and the honest parties have $[c_\iter]$ and $[c'_\iter]$. This is because, in the protocol, 
   the honest parties start participating in the instance of $\Rec$ to reconstruct $r_\iter$, only after they compute
    $[c_\iter]$ and $[c'_\iter]$. 
  Now we have the following cases with respect to
 whether any party from $\PartySet \setminus \Discarded$ behaved maliciously during the underlying instances of $\BasicMult$.
 \begin{myitemize}
 \item[--] {\bf Case I: $c_\iter = a_\iter \cdot b_\iter$ and $c'_\iter = a_\iter \cdot b'_\iter$} --- In this case, $(a_\iter, b_\iter, c_\iter)$ is a multiplication-triple.
 \item[--] {\bf Case II: $c_\iter = a_\iter \cdot b_\iter$, but $c'_\iter \neq a_\iter \cdot b'_\iter$} --- This case is never possible, as this will lead to the contradiction
  that $r_\iter (a_\iter \cdot b_\iter - c_\iter) \neq (c'_\iter - a_\iter \cdot b'_\iter)$ holds.
 \item[--] {\bf Case III: $c_\iter \neq a_\iter \cdot b_\iter$, but $c'_\iter = a_\iter \cdot b'_\iter$} --- This case is possible only if $r_\iter = 0$, as otherwise 
   this will lead to the contradiction
  that $r_\iter (a_\iter \cdot b_\iter - c_\iter) \neq (c'_\iter - a_\iter \cdot b'_\iter)$ holds. However, since $r_\iter$ is a random element from $\F$, it implies that this case
  can occur only with probability at most $\frac{1}{|\F|}$.
  \item[--] {\bf Case IV: $c_\iter \neq a_\iter \cdot b_\iter$ as well as $c'_\iter \neq a_\iter \cdot b'_\iter$} --- This case is possible only if 
   $r_\iter = (c'_\iter - a_\iter \cdot b'_\iter) \cdot (a_\iter \cdot b_\iter - c_\iter)^{-1}$, as otherwise 
   this will lead to the contradiction
  that $r_\iter (a_\iter \cdot b_\iter - c_\iter) \neq (c'_\iter - a_\iter \cdot b'_\iter)$ holds. However, since $r_\iter$ is a random element from $\F$, it implies that this case
  can occur only with probability at most $\frac{1}{|\F|}$.
 \end{myitemize}
Hence, we have shown that except with probability at most $\frac{1}{|\F|}$, the triple $(a_\iter, b_\iter, c_\iter)$ is a multiplication-triple. 
 To complete the proof, we need to argue that the view of $\Adv$ in the protocol will be independent of the triple 
 $(a_\iter, b_\iter, c_\iter)$. For this, we first note that $a_\iter, b_\iter$ and $b'_\iter$ will be random for the adversary
 at the time of their generation, which follows from Lemma \ref{lemma:RandMultCIACS}.
   From Lemma \ref{lemma:BasicMultPrivacy}, $\Adv$ learns nothing additional 
  about $a_\iter$, $b_\iter$ and $b'_\iter$ during the two instances of $\BasicMult$. 
   While $\Adv$ learns the value of $e_\iter$, since $b'_\iter$ is a uniformly distributed for $\Adv$, 
    for every candidate value of $b'_\iter$ from the view-point of $\Adv$,
    there is a corresponding value of $b_\iter$ consistent with the $e_\iter$ learnt by $\Adv$.
    Hence, learning $e_\iter$ does not add any new information about $(a_\iter, b_\iter, c_\iter)$ to the view of $\Adv$. 
    Moreover, $\Adv$ will be knowing beforehand that $d_\iter$ will be $0$ and hence, learning this value does not change
    the view of $\Adv$ regarding $(a_\iter, b_\iter, c_\iter)$.     
\end{proof}

We next derive the communication complexity of the protocol $\RandMultCI$. \\~\\
\noindent {\bf Lemma \ref{lemma:RandMultCICommunication}.}
{\it Protocol $\RandMultCI$ incurs a communication of $\Order(|\Z_s| \cdot n^5 \cdot \log{|\F|} + n^6 \cdot \log{|\F|} \cdot |\sigma|)$
 bits and makes $\Order(n^2)$ calls to $\BA$.
}
\begin{proof}
To generate $[a_\iter], [b_\iter], [b'_\iter]$ and $[r_\iter]$, $\Order(n)$ instances of $\LSh$ and $\BA$ are invoked. To compute $[c_\iter]$ and $[c'_\iter]$, two instances of $\BasicMult$
 are invoked. To publicly reconstruct $e_\iter$ and $d_\iter$, two instances of $\Rec$ are invoked with $|\ReceiverSet| = n$. Finally, if $d_\iter \neq 0$, then $3 \cdot |\ShareSpec_{\Z_s}|$ instances of
  $\RecShare$ and $2n$ instances of $\Rec$ are invoked, with $|\ReceiverSet| = n$. The communication complexity now follows from the communication complexity of
   $\LSh$ (Lemma \ref{lemma:LSh}), communication complexity of $\BasicMult$ (Lemma \ref{lemma:BasicMultCommunication}),
   communication complexity of $\RecShare$ (Lemma \ref{lemma:RecShare}) and  communication complexity of $\Rec$ (Lemma \ref{lemma:Rec}). 
\end{proof}
\subsection{Properties of the Protocol $\TripGen$}
In this section, we prove the properties of the protocol $\TripGen$ (see Fig \ref{fig:TripGen} for the formal details.)
 We begin by showing that each party computes an output in the protocol. \\~\\
\noindent {\bf Lemma \ref{lemma:TripGenTermination}.}
{\it Let $t$ be the size of the largest set in $\Z_s$. Then except with probability $\Order(n^3 \cdot \errorAICP)$, the honest parties compute an output during 
  $\TripGen$, by the time $\TimeTripGen = (t + 1) \cdot \TimeRandMultCI$ in a synchronous network, or almost-surely, eventually in an asynchronous network, where
  $\TimeRandMultCI = \TimeLSh + 2\TimeBA + \TimeBasicMult + 4\TimeRec$.
}
\begin{proof}
From Lemma \ref{lemma:RandMultCITermination}, except with probability $\Order(n^3 \cdot \errorAICP)$,
  the honest parties will know the outcome of each iteration $\iter$, since all honest parties set the Boolean variable 
 $\flag_\iter$ to a common value. For every iteration
  $\iter$ where $\flag_\iter$ is set to $1$, from Lemma \ref{lemma:RandMultCICorruptBehaviour}, a {\it new} corrupt party is added to $\Discarded$. Thus, after at most $t$ iterations, 
  all the corrupt parties will be included in $\Discarded$ and the parties will set 
  $\flag_\iter$ to $0$ in the next iteration. Moreover, they will output 
  ${({\GW}_1,  \ldots,  {\GW}_{|\Z_s|}, [a_\iter], [b_\iter], [c_\iter])}$, computed during the corresponding instance of $\RandMultCI$.
\end{proof}

We next claim that the output computed by the honest parties is indeed a multiplication-triple.\\~\\
\noindent {\bf Lemma \ref{lemma:TripGenCorrectnessPrivacy}.}
{\it If the honest parties output ${({\GW}_1,  \ldots,  {\GW}_{|\Z_s|}, [a_\iter], [b_\iter], [c_\iter])}$ during the protocol $\TripGen$,
 then $a_\iter, b_\iter$ and $c_\iter$ are linearly secret-shared with IC-signatures, with
  ${\GW}_1,  \ldots,  {\GW}_{|\Z_s|}$ being the underlying core-sets. Moreover, $c_\iter = a_\iter b_\iter$ holds, except with probability $\frac{1}{|\F|}$. 
   Furthermore, the view of the adversary remains independent of $a_\iter, b_\iter$ and $c_\iter$.
}
\begin{proof}
Follows from Lemma \ref{lemma:RandMultCICorrectness}.
\end{proof}

We finally derive the communication complexity of the protocol $\TripGen$.\\~\\
\noindent {\bf Lemma \ref{lemma:TripGenCommunication}.}
{\it Protocol $\TripGen$ incurs a communication of $\Order(|\Z_s| \cdot n^6 \cdot \log{|\F|} + n^7 \cdot \log{|\F|} \cdot |\sigma|)$
 bits and makes $\Order(n^3)$ calls to $\BA$.
}
\begin{proof}
The proof follows from the communication complexity of $\RandMultCI$ (Lemma \ref{lemma:RandMultCICommunication}) and the fact that $\Order(n)$ instances of 
 $\RandMultCI$ are invoked in the protocol.
\end{proof}

%% file: AppMPC.tex
\section{Properties of the Circuit Evaluation Protocol}
\label{app:MPC}
In this section, we prove the properties of the protocol $\PiMPC$ (see Fig \ref{fig:MPC} for the formal description).
 We begin by showing that the honest parties compute some output during the pre-processing phase.
\begin{lemma}
\label{lemma:MPCPreProcessing}
Protocol $\PiMPC$ achieves the following during the pre-processing phase.
\begin{myitemize}
\item[--] {\bf Synchronous Network}: Except with probability $\Order(n^3 \cdot \errorAICP)$,  at the time $\TimeRand$, the honest parties have
$\{r^{(\mathfrak{l})} \}_{\mathfrak{l} = 1, \ldots, L}$, which are linearly secret-shared with IC-signatures, with 
  ${\GW}_1, \ldots, {\GW}_{|\Z_s|}$ being the underlying core-sets, where $L \defined n^3 \cdot c_M + 4n^2 \cdot c_M + n^2 + n$. Moreover, the view of the adversary remains independent of
   the values $\{r^{(\mathfrak{l})} \}_{\mathfrak{l} = 1, \ldots, L}$.
   At the time $\TimeRand + \TimeTripGen$, the honest parties have triples $\{(a^{(\ell)}, b^{(\ell)}, c^{(\ell)}) \}_{\ell = 1, \ldots, c_M}$, 
   which are linearly secret-shared with IC-signatures, with 
  ${\GW}_1, \ldots, {\GW}_{|\Z_s|}$ being the underlying core-sets, where $c^{(\ell)} = a^{(\ell)} \cdot b^{(\ell)}$, except with probability $\frac{1}{|\F|}$.
  The view of the adversary will be independent of the multiplication-triples. 
\item[--] {\bf Asynchronous Network}: Except with probability $\Order(n^3 \cdot \errorAICP)$,  almost-surely, the honest parties eventually have
$\{r^{(\mathfrak{l})} \}_{\mathfrak{l} = 1, \ldots, L}$, which are linearly secret-shared with IC-signatures, with 
  ${\GW}_1, \ldots, {\GW}_{|\Z_s|}$ being the underlying core-sets, where $L \defined n^3 \cdot c_M + 4n^2 \cdot c_M + n^2 + n$. Moreover, the view of the adversary remains independent of the
  values $\{r^{(\mathfrak{l})} \}_{\mathfrak{l} = 1, \ldots, L}$. Furthermore, the honest parties eventually have
  triples $\{(a^{(\ell)}, b^{(\ell)}, c^{(\ell)}) \}_{\ell = 1, \ldots, c_M}$, 
   which are linearly secret-shared with IC-signatures, with 
  ${\GW}_1, \ldots, {\GW}_{|\Z_s|}$ being the underlying core-sets, where $c^{(\ell)} = a^{(\ell)} \cdot b^{(\ell)}$, except with probability $\frac{1}{|\F|}$.
  The view of the adversary will be independent of the multiplication-triples. 
\end{myitemize}
\end{lemma}
\begin{proof}
The proof follows from the {\it $\Z_s$-correctness}, {\it $\Z_a$-correctness} and {\it privacy} of the protocol $\Rand$ (Theorem \ref{thm:Rand}) and
 from the properties of $\TripGen$ in the asynchronous and asynchronous network (Lemma \ref{lemma:TripGenTermination} and Lemma \ref{lemma:TripGenCorrectnessPrivacy}).
  We also note that the multiplication-triples will be linearly secret-shared with IC-signatures, with 
  ${\GW}_1, \ldots, {\GW}_{|\Z_s|}$ being the underlying core-sets. This is because there will be at most $n^3 \cdot c_M + 4n^2 \cdot c_M + n^2$ instances of $\LSh$
  invoked as part of $\TripGen$ for generating $c_M$ multiplication-triples. And prior to invoking the instance of $\TripGen$, the honest parties would have already generated
   $n^3 \cdot c_M + 4n^2 \cdot c_M + n^2 + n$ number of linearly secret-shared random pads with IC-signatures with ${\GW}_1, \ldots, {\GW}_{|\Z_s|}$ being the underlying 
   core-sets through the instance of $\Rand$, which
   can serve $n^3 \cdot c_M + 4n^2 \cdot c_M + n^2 + n$ instances of $\LSh$.  
\end{proof}

We next show that during the input phase, the inputs of all honest parties will be linearly secret-shared with IC-signatures in a {\it synchronous} network
 and in an {\it asynchronous} network, the inputs of a subset of the parties will be linearly secret-shared with IC-signatures.
 \begin{lemma}
 \label{lemma:MPCInputPhase}
Protocol $\PiMPC$ achieves the following during the input phase.
\begin{myitemize}
\item[--] {\bf Synchronous Network}: Except with probability $\Order(n^3 \cdot \errorAICP)$,  at the time $\TimeRand + \TimeTripGen + \TimeLSh + 2\TimeBA$, 
 the honest parties will have a common subset $\CoreSet$ where $\PartySet \setminus \CoreSet \in \Z_s$, such that all honest parties will be present in $\CoreSet$.
  Moreover, corresponding to every $P_j \in \CoreSet$, there will be some value, say ${x^{\star}}^{(j)}$, which is the same as $x^{(j)}$, which will be linearly secret-shared with IC-signatures,
  with ${\GW}_1, \ldots, {\GW}_{|\Z_s|}$ being the underlying core-sets. Furthermore, the view of the adversary will be independent of the $x^{(j)}$ values, corresponding to the
  honest parties $P_j \in \CoreSet$.
\item[--] {\bf Asynchronous Network}: Except with probability $\Order(n^3 \cdot \errorAICP)$,  almost-surely, 
 the honest parties will eventually have a common subset $\CoreSet$ where $\PartySet \setminus \CoreSet \in \Z_s$.
  Moreover, corresponding to every $P_j \in \CoreSet$, there will be some value, say ${x^{\star}}^{(j)}$, which is the same as $x^{(j)}$, which will be eventually linearly secret-shared with IC-signatures,
  with ${\GW}_1, \ldots, {\GW}_{|\Z_s|}$ being the underlying core-sets. Furthermore, the view of the adversary will be independent of the $x^{(j)}$ values, corresponding to the
  honest parties $P_j \in \CoreSet$.
\end{myitemize}
\end{lemma}
\begin{proof}
 We first consider a {\it synchronous} network. 
 From the protocol steps, the honest parties start participating in the input phase, only after computing output during the instance of $\TripGen$, which 
 from Lemma \ref{lemma:MPCPreProcessing} happens at the time $\TimeRand + \TimeTripGen$. We also note that there can be at most
 $n^3 \cdot c_M + 4n^2 \cdot c_M + n^2$ instances of $\LSh$, invoked as part of the instance of $\TripGen$, which will utilize the 
 secret-shared values $\{r^{(\mathfrak{l})} \}_{\mathfrak{l} = 1, \ldots, n^3 \cdot c_M + 4n^2 \cdot c_M + n^2}$ as pads. 
 Consequently, the remaining linearly secret-shared values $\{r^{(\mathfrak{l})} \}_{\mathfrak{l} = n^3 \cdot c_M + 4n^2 \cdot c_M + n^2 + 1, \ldots, n^3 \cdot c_M + 4n^2 \cdot c_M + n^2 + n}$
 will be still available to the honest parties for being used as pads in up to $n$ instances of $\LSh$, since these pads will be random from the adversary's point of view.
 
 Let $Z^{\star} \in \Z_s$ be the set of {\it corrupt} parties and
  let $\Hon =  \PartySet \setminus Z^{\star}$
  be the set of {\it honest} parties.
  We claim that by the time $\TimeRand + \TimeTripGen + \TimeLSh + 2\TimeBA$, except with probability $\Order(n^3 \cdot \errorAICP)$,
   all the parties in $\Hon$ will have a common subset $\CoreSet$, where $\PartySet \setminus \CoreSet \in \Z_s$
   and where $\Hon \subseteq \CoreSet$. The proof for this is exactly the same as that of Lemma \ref{lemma:RandMultCIACS}. Namely, at the time  $\TimeRand + \TimeTripGen + \TimeLSh$, all
   the parties in $\Hon$ would start participating with input $1$ in the BA instances $\BA^{(j)}$, corresponding to the parties $P_j \in \Hon$, since by this time, the instances of $\LSh$ invoked by the
   parties in $\Hon$ will produce output for all the parties in $\Hon$. Consequently, these BA instances 
    will produce output $1$ at the time 
    $\TimeRand + \TimeTripGen + \TimeLSh + \TimeBA$, after which all the parties in $\Hon$ will start participating with input $0$ in the remaining BA instances (if any).
    Consequently, by the time  $\TimeRand + \TimeTripGen + \TimeLSh + 2\TimeBA$, all the $n$ instances of $\BA$ will produce some output and the parties in $\Hon$ will have a common $\CoreSet$.
    Next, it can be shown that corresponding to every $P_j \in \CoreSet$, there exists some value, say ${x^{\star}}^{(j)}$, which is the same as $x^{(j)}$, 
    such that the parties in $\Hon$ have a linear secret-sharing with IC-signatures of ${x^{\star}}^{(j)}$,
     with ${\GW}_1, \ldots, {\GW}_{|\Z_s|}$ being the underlying core-sets. The proof for this follows similar lines as that of Lemma \ref{lemma:RandMultCIACS}.
     It is easy to see that ${x^{\star}}^{(j)}$ will be linearly secret-shared with IC-signatures with ${\GW}_1, \ldots, {\GW}_{|\Z_s|}$ being the underlying core-sets.
     This is because the instance of $\LSh$ invoked by $P_j$ utilizes $r^{(n^3 \cdot c_M + 4n^2 \cdot c_M + n^2 + j)}$ as the pad, which is linearly secret-shared
     with IC-signatures with ${\GW}_1, \ldots, {\GW}_{|\Z_s|}$ being the underlying core-sets.
     
     The proof of the lemma for an {\it asynchronous} network is almost the same as above and follows using similar arguments as used to prove Lemma \ref{lemma:RandMultCIACS}
     for the case of {\it asynchronous} network. 
     
     Finally, the privacy of the inputs $x^{(j)}$ of the parties $P_j \in (\Hon \cap \CoreSet)$ 
     follows from the {\it privacy} of $\LSh$ and the fact that the underlying pads $r^{(n^3 \cdot c_M + 4n^2 \cdot c_M + n^2 + j)}$ used in the corresponding
     instances of $\LSh$ are still random for the adversary, after the instance of $\TripGen$.    
 \end{proof}

We next show that the honest parties compute an output during the circuit-evaluation phase.

 \begin{lemma}
 \label{lemma:MPCCircuitEvaluationPhase}
Protocol $\PiMPC$ achieves the following during the circuit-evaluation phase, where $D_M$ denotes the multiplicative depth of $\ckt$.
\begin{myitemize}
\item[--] {\bf Synchronous Network}: Except with probability $\Order(n^3 \cdot \errorAICP)$,  at the time $\TimeRand + \TimeTripGen + \TimeLSh + 2\TimeBA + (D_M + 1) \cdot \TimeRec$, 
 the honest parties will have $y$, where $y = f({x^{\star}}^{(1)}, \ldots, {x^{\star}}^{(n)})$, such that ${x^{\star}}^{(j)} = x^{(j)}$ for every honest party $P_j \in \CoreSet$ and where
 ${x^{\star}}^{(j)} = 0$ for every $P_j \not \in \CoreSet$. Moreover, all honest parties will be present in $\CoreSet$. 
 Furthermore, the view of the adversary will be independent of the $x^{(j)}$ values, corresponding to the
  honest parties $P_j \in \CoreSet$.
\item[--] {\bf Asynchronous Network}: Except with probability $\Order(n^3 \cdot \errorAICP)$,  almost-surely, 
 the honest parties will eventually have $y$, where $y = f({x^{\star}}^{(1)}, \ldots, {x^{\star}}^{(n)})$, where ${x^{\star}}^{(j)} = x^{(j)}$ for every honest party $P_j \in \CoreSet$ and where
 ${x^{\star}}^{(j)} = 0$ for every $P_j \not \in \CoreSet$. Furthermore, the view of the adversary will be independent of the $x^{(j)}$ values, corresponding to the
  honest parties $P_j \in \CoreSet$.
\end{myitemize}
\end{lemma}
\begin{proof}
Let us first consider a {\it synchronous} network. Let $Z^{\star} \in \Z_s$ be the set of {\it corrupt} parties and
  let $\Hon =  \PartySet \setminus Z^{\star}$
  be the set of {\it honest} parties.
 From Lemma \ref{lemma:MPCPreProcessing}, at the time $\TimeRand + \TimeTripGen$, 
 the parties in $\Hon$ will have the triples $\{(a^{(\ell)}, b^{(\ell)}, c^{(\ell)}) \}_{\ell = 1, \ldots, c_M}$, 
   which are linearly secret-shared with IC-signatures, with 
  ${\GW}_1, \ldots, {\GW}_{|\Z_s|}$ being the underlying core-sets and where $c^{(\ell)} = a^{(\ell)} \cdot b^{(\ell)}$, except with probability $\frac{1}{|\F|}$.
  Moreover, from Lemma \ref{lemma:MPCInputPhase}, at the time $\TimeRand + \TimeTripGen + \TimeLSh + 2\TimeBA$,
  the honest parties will have a common subset $\CoreSet$ where $\PartySet \setminus \CoreSet \in \Z_s$, such that all honest parties will be present in $\CoreSet$.
  Furthermore, corresponding to every $P_j \in \CoreSet$, there will be some value, say ${x^{\star}}^{(j)}$, which is the same as $x^{(j)}$ for an
   {\it honest} $P_j$, which will be linearly secret-shared with IC-signatures,
  with ${\GW}_1, \ldots, {\GW}_{|\Z_s|}$ being the underlying core-sets. At the end of the input phase, the parties take $0$ as the input on behalf of the parties $P_j \not \in \CoreSet$
  and take the default linear secret-sharing of $0$ with IC-signatures, with ${\GW}_1, \ldots, {\GW}_{|\Z_s|}$ being the underlying core-sets.
  To prove the lemma, we show that all the gates in $\ckt$ are correctly evaluated. Namely, for every gate in $\ckt$, given the gate-inputs in a linearly secret-shared fashion
  with IC-signatures with ${\GW}_1, \ldots, {\GW}_{|\Z_s|}$ being the underlying core-sets, the parties compute the gate-output in a linearly secret-shared fashion
  with IC-signatures with ${\GW}_1, \ldots, {\GW}_{|\Z_s|}$ being the underlying core-sets. 
  While this is true for the linear gates in $\ckt$, which follows from the linearity of the secret-sharing, 
  the same is true even for the multiplication gates, except with probability $\Order(n^3 \cdot \errorAICP)$.
  This is because, for every multiplication gate, the parties deploy a linearly secret-shared multiplication-triple from the pre-processing phase and apply Beaver's method.
  And the masked gate-inputs are correctly reconstructed through instances of $\Rec$, except with probability $\Order(n^3 \cdot \errorAICP)$.
  Since all the independent multiplication gates at the same multiplicative depth can be evaluated in parallel, to evaluate the multiplication gates, it will take a total $D_M \cdot \TimeRec$ time. 
    Finally, once the circuit-output is ready in a secret-shared fashion, it is publicly reconstructed through an instance of $\Rec$, which takes $\TimeRec$ time
    and produces the correct output, except with probability $\Order(n^3 \cdot \errorAICP)$.
    The {\it privacy} of the inputs of the honest parties in $\CoreSet$ follows from the {\it privacy} of $\LSh$ (Lemma \ref{lemma:LSh}) 
    and the fact no additional information is revealed during the evaluation of multiplication gates. This is because the underlying multiplication-triples which are deployed while
    applying Beaver's method are random for the adversary.
  
  The proof for the case of {\it asynchronous} network follows similar arguments as above and depends upon the properties of the pre-processing phase and input phase
  in the {\it asynchronous} network (Lemma \ref{lemma:MPCPreProcessing} and Lemma \ref{lemma:MPCInputPhase}).
\end{proof}

We finally show that the honest parties terminate the protocol.
\begin{lemma}
\label{lemma:MPCTermination}
If the network is synchronous, then except with probability $\Order(n^3 \cdot \errorAICP)$,
  the honest parties terminate the protocol at the time $\TimeRand + \TimeTripGen + \TimeLSh + 2\TimeBA + (D_M + 1) \cdot \TimeRec + \Delta$.
  If the network is asynchronous, then except with probability $\Order(n^3 \cdot \errorAICP)$, almost-surely,
  the honest parties eventually terminate the protocol.
\end{lemma}
\begin{proof}
Let us first consider a {\it synchronous} network. Let $Z^{\star} \in \Z_s$ be the set of {\it corrupt} parties and
  let $\Hon =  \PartySet \setminus Z^{\star}$
  be the set of {\it honest} parties.
  From Lemma \ref{lemma:MPCCircuitEvaluationPhase}, except with probability $\Order(n^3 \cdot \errorAICP)$,  at the time $\TimeRand + \TimeTripGen + \TimeLSh + 2\TimeBA + (D_M + 1) \cdot \TimeRec$, 
  all the parties in $\Hon$ will have $y$. Hence
  every party in $\Hon$ will send a $\ready$ message for $y$ to all the parties, which gets delivered within $\Delta$ time, while
  the parties in $Z^{\star}$ may send a $\ready$ message for some $y'$ where $y' \neq y$. Now since
  $\Z_s$ {\it does not} satisfy the $\Q^{(1)}(Z^{\star}, \Z_s)$ condition, it follows that no party in $\Hon$ will ever send a $\ready$ message for any $y'$ where $y' \neq y$. 
   As $\PartySet \setminus \Hon \in \Z_s$, it follows that at the time 
   $\TimeRand + \TimeTripGen + \TimeLSh + 2\TimeBA + (D_M + 1) \cdot \TimeRec + \Delta$, all the parties in $\Hon$ will have sufficient number of $\ready$ messages for $y$ and hence they terminate
   with output $y$.
   
   Next, consider an {\it asynchronous} network. Let $Z^{\star} \in \Z_a$ be the set of {\it corrupt} parties and
  let $\Hon =  \PartySet \setminus Z^{\star}$
  be the set of {\it honest} parties. Note that $\PartySet \setminus Z^{\star} \in \Z_s$, since $\Z_a \subset \Z_s$. 
   From Lemma \ref{lemma:MPCCircuitEvaluationPhase}, except with probability $\Order(n^3 \cdot \errorAICP)$,  almost-surely,   
  all the parties in $\Hon$ will eventually compute $y$. Hence
  every party in $\Hon$ will eventually send some $\ready$ message. We claim that no party in $\Hon$ will ever send a $\ready$ message for any $y' \neq y$.
  On the contrary, let $P_i \in \Hon$ be the {\it first} party, which sends a $\ready$ message for $y' \neq y$.
  From the protocol steps, it follows that $P_i$ sends the $\ready$ message for $y'$ after computing $y'$ during the circuit-evaluation phase.
  Otherwise, there should exist a subset of parties ${\cal A}$ where $\Z_s$ satisfies  $\Q^{(1)}({\cal A}, \Z_s)$ condition (implying that ${\cal A}$ has at least one party from $\Hon$),
  who should have sent the $\ready$ message for $y'$ to $P_i$, which is {\it not} possible, since we are assuming $P_i$ to be {\it first} party from $\Hon$ to send a $\ready$
  message for $y'$. From Lemma \ref{lemma:MPCCircuitEvaluationPhase}, $P_i$ will {\it not} compute $y'$ and hence will {\it not} send a $\ready$ message for $y'$.
  Now since every party in $\Hon$ eventually computes $y$ in the circuit-evaluation phase, it eventually sends a $\ready$ message for $y$. And since
  $\PartySet \setminus Z^{\star} \in \Z_s$ and $\PartySet \setminus Z^{\star} \not \in \Z_s$, it follows that irrespective of the behaviour of the corrupt parties, the parties in $\Hon$ will eventually
  receive a sufficient number of $\ready$ messages for $y$, to terminate with output $y$. 
  
  Let $P_h$ be the {\it first} party from $\Hon$, who terminates with output $y$. This implies that
  there exists a subset of parties ${\cal W}$ with $\PartySet \setminus {\cal W} \in \Z_s$, who sends a $\ready$ message for $y$ to $P_h$. Now consider the set $(\Hon \cap {\cal W})$.
  The set satisfies the $\Q^{(1)}(\Hon \cap {\cal W}, \Z_s)$ condition, due to the $\Q^{(2, 1)}(\PartySet, \Z_s, \Z_a)$ condition.
   The $\ready$ messages of these parties (for $y$) get eventually delivered to every party in $\Hon$.
  Consequently, every party in $\Hon$ (including $P_h$) who has {\it not} yet sent any $\ready$ message will eventually send the $\ready$ message for $y$, which gets eventually delivered to all the parties.
  And as a result, every party in $\Hon$ will eventually have a sufficient number of $\ready$ messages for $y$, to terminate with the output $y$.
\end{proof}

We next derive the communication complexity of the protocol.
\begin{lemma}
\label{lemma:MPCComplexity}
Protocol $\PiMPC$ incurs a communication of 
  $\Order(|\Z_s|^2 \cdot n^{12} \cdot \log{|\F|} \cdot |\sigma|)$
   bits and makes $\Order(n^3)$ calls to $\BA$.
\end{lemma}
\begin{proof}
The communication complexity is dominated by the instance of $\Rand$ to generate  $L \defined n^3 \cdot c_M + 4n^2 \cdot c_M + n^2 + n$ random secret-shared values
 and the instance of $\TripGen$ to generate $L = c_M$ number of secret-shared multiplication-triples. The proof now follows from the communication complexity of $\Rand$ (Theorem \ref{thm:Rand})
  and the communication complexity of the (generalized) $\TripGen$ protocol (Lemma \ref{lemma:TripGenCommunication}).
\end{proof}

Theorem \ref{thm:MPC} now easily follows from Lemma \ref{lemma:MPCPreProcessing}-\ref{lemma:MPCComplexity}.